\documentclass[a4paper,openany,oneside]{memoir}
\usepackage{latexsym,amsmath,amssymb,mathrsfs,amsthm}
\usepackage{fixme}
\usepackage{todonotes}

\usepackage[pdfauthor={Pieter Naaijkens},
            pdftitle={Quantum spin systems on infinite lattices},
			bookmarksnumbered,
			colorlinks=true
			]{hyperref}

\usepackage[T1]{fontenc}
\usepackage{titlesec, blindtext, color}
\definecolor{gray75}{gray}{0.75}
\newcommand{\hsp}{\hspace{20pt}}
\titleformat{\chapter}[hang]{\Huge\bfseries}{\thechapter\hsp\textcolor{gray75}{|}\hsp}{0pt}{\Huge\bfseries}

\usepackage{makeidx}
\makeindex
\newcommand{\idx}[1]{\index{#1}\emph{#1}}
\newcommand{\idxr}[1]{\index{#1}#1}

\usepackage{tikz}
\usetikzlibrary{arrows}
\usetikzlibrary{positioning}

\renewcommand\printtoctitle[1]{}

\newcommand{\timeso}{\overset{\mathrm{out}}{\times}}

\numberwithin{equation}{section}

\theoremstyle{plain}
\newtheorem{theorem}{Theorem}

\numberwithin{theorem}{section}
\newtheorem{lemma}[theorem]{Lemma}

\newtheorem{proposition}[theorem]{Proposition}
\newtheorem{corollary}[theorem]{Corollary}

\newtheorem{definition}[theorem]{Definition}

\newtheorem{example}[theorem]{Example}
\newtheorem{assumption}[theorem]{Assumption}
\newtheorem{exercise}[theorem]{Exercise}

\theoremstyle{definition}
\newtheorem{remark}[theorem]{Remark}
\newcommand{\Aut}{\operatorname{Aut}}

\newcommand{\alg}[1]{\mathfrak{#1}}

\newcommand{\Tr}{\operatorname{Tr}}

\newcommand{\supp}{\operatorname{supp}}

\newcommand{\mc}[1]{\mathcal{#1}}

\newcommand{\ket}[1]{\ensuremath{\left|#1\right\rangle}}
\newcommand{\bra}[1]{\ensuremath{\left\langle#1\right|}}
\newcommand{\ketbra}[1]{\ensuremath{\left|#1\right\rangle\!\left\langle#1\right|}}
\newcommand{\braket}[2]{\ensuremath{\left\langle#1\middle|#2\right\rangle}}

\title{Quantum spin systems on infinite lattices}
\author{Pieter Naaijkens\\Leibniz Universit\"at Hannover}

\newcommand{\HRule}{\rule{\linewidth}{0.5mm}}

\begin{document}
\noindent\HRule\\
\begin{center}
{\huge Quantum Spin Systems on Infinite Lattices}\\\vspace{\baselineskip}
{\large Summer 2013\\
Extended and corrected Summer 2016}\\
\vspace{1cm}
{\Large Pieter Naaijkens\footnote{To be published as a volume in the \emph{Lecture Notes in Physics} series (Springer).}\\RWTH Aachen University and University of California, Davis}
\end{center}
\HRule\\
\vspace{-1cm}

\tableofcontents*
\thispagestyle{empty}

\chapter{Introduction}
Quantum mechanics arguably is one of the most successful physical theories of the past century. Not only has its development led to important results in physics, also many fields in mathematics have their origin in or progressed significantly alongside the development of quantum mechanics, and vice versa~\cite{CPA:CPA3160130102}. Of particular interest also from a mathematical point of view are systems with infinitely many degrees of freedom. This includes quantum field theory, but one can also take infinitely many copies of a \emph{finite} system, for example a spin-1/2 system. This is the type of systems that we will study in these lecture notes. That is, we will almost exclusively defined on a discrete space (a lattice), rather than a continuum theory.

There are many different ways of studying such systems, varying among other things on their level of mathematical rigour. Here we follow the so-called the operator-algebraic approach, and aim to introduce the necessary mathematical tools for this. In this approach the observables are modelled by elements of some operator algebra, in our case this is usually a $C^*$-algebra. This is essentially a generalisation of the usual setting in quantum mechanics, where one considers the bounded (and sometimes also unbounded) operators acting on a Hilbert space. It should become clear in due course why this generalisation is useful, and not just some inessential mathematical machinery.

We give a brief overview of the contents of these lectures. In the first part of the notes the necessary mathematics is introduced: we give a short introduction to operator algebras and $C^*$-algebras in particular. In addition, the relation between this abstract mathematics and quantum mechanics is explained. Then we move on to the main topic of interest: quantum systems with infinitely many sites. In the remainder, we will see how various concepts in physics, such as ground states, can be described in this setting. Some of this concepts will be illustrated with the toric code, a simple example that has surprisingly interesting features.

Another fundamental tool that will be discussed are the Lieb-Robinson bounds. Such bounds limit the speed at which information can propagate in the system, not unlike the speed of light in special relativity. This has many profound implications, some of which we will discuss here. For example, it allows us to translate many properties from relativistic systems to quantum spin systems, even though there are also essential differences between the two.

Chapters~\ref{ch:opalg} and~\ref{ch:infinite} discuss the essential mathematical background and the results there are used throughout the book. Sections~\ref{sec:lrfinite} to~\ref{sec:lrlocality} show how Lieb-Robinson bounds can be derived and give various estimates that will be used in the remainder of the book. The remaining chapters and sections can be read independently of each other.

\section{Prerequisites}
The aim of these lecture notes is to provide an introduction to a mathematically rigorous treatment of infinite quantum spin systems. Although we do not always provide full proofs, they do exist for almost all statements in these notes, in particular those labelled as a Lemma or Theorem. Naturally, this requires different mathematical tools. To keep the book accessible to both physics as well as mathematics students, some techniques that we will need often are introduced here in some detail. An example is the construction of the completion of a metric space. Many mathematics students will have seen this (for example in the construction of the real numbers $\mathbb{R}$ from $\mathbb{Q}$), but it is less common to see it in a physics curriculum. Conversely, the analysis of unbounded operators on Hilbert space makes an appearance (although sometimes somewhat disguised) in any introduction to quantum mechanics, but in a mathematics curriculum it typically is discussed in a course on functional analysis, which commonly is an elective course.

Although many of the mathematical constructs are discussed here in some level of detail, a certain level of mathematical maturity is assumed. The reader should be familiar with the basic concepts of analysis, topology and functional analysis. In particular, convergence in and completeness of metric spaces and the notion of Hilbert space are important. They are all discussed briefly here, but some familiarity with them would be helpful. At some points we also have to use measure theory and the theory of integration, although the heuristics are understandable without being familiar with measure theory. An excellent book on advanced analysis treating all these topics in detail (and more) is \emph{Analysis now} by G.K.\ Pedersen~\cite{MR971256}, or the even more comprehensive~\cite{MR751959}.

Finally, basic knowledge of quantum mechanics is assumed. Although not strictly necessary to understand the mathematical details, it provides motivation for many of the important definitions and concepts we will introduce. For example, we will not go into the details of \emph{why} observables in quantum mechanical systems can be described by self-adjoint operators acting on some Hilbert space, or why the time evolution is generated by a positive, self-adjoint operator, the Hamiltonian.

\section{Quantum spin systems}
A single spin-1/2 particle is one of the simplest quantum mechanical systems, especially if we forget about the spatial degrees of freedom and only think about the spin part. Nevertheless, such systems play an important part in quantum information. They are the quantum analogue, called a qubit in that context, of a bit in computer science. As such they are one of the fundamental building blocks of quantum computation. There is not much to say about a single qubit, but if we combine a number of qubits (even finitely many), there are a lot of interesting questions to be studied.

Slightly more general, simple quantum mechanical systems such as qubits can be the constituents of complex quantum systems. Recall that in statistical mechanics the goal is to describe systems made up of a large number of ``simple'' objects such as a gas which is made up of individual molecules. It then turns out that precise knowledge on each individual molecule (e.g.\ its position and momentum) is not so useful any more to say anything about the system as a whole. In quantum statistical mechanics the situation is similar. Here one considers systems made up of a large number of simple \emph{quantum} systems. In typical models one can think for example of a lattice, with at each site of the lattice an atom. The atoms are considered to be pinned down at the lattice site, and hence cannot move. However, the quantum degrees of freedom the atoms can interact with nearby atoms. Such models are often encountered in condensed matter, for example.

In dealing with such systems it can be convenient to consider the so-called \idx{thermodynamic limit}, which is the limit of infinite system size. This has a few advantages.\footnote{An extensive discussion can be found in the introduction of~\cite{MR887100}.} First of all, one can make a clear distinction between \emph{local} and \emph{global} properties, without having to keep track of the system size. Moreover one does not have to worry about boundary effects and boundary conditions, which may simplify life a lot. Besides that, considering the thermodynamic limit is even \emph{necessary} to consider, for example, phase transitions. Depending on your interests, you may also find the fact that it leads to interesting mathematics a compelling reason to study such systems.

There is more than one way to study the thermodynamic limit. One could for example only look at finite systems, and find the quantities of interest as a function of system size, and then see how they behave in the infinite volume limit. This however can become cumbersome, and it is not always clear if one is not missing some essential feature by only considering finite systems. Another approach would be to work directly in the thermodynamic limit. This is what we prefer here.

In these lectures we will develop the tools to deal with such systems in a mathematically rigorous way. Since the goal is to create familiarity with the tools and methods of the operator algebraic approach to quantum mechanics, some proofs are only sketched. Moreover, to simplify matters most results are often not stated in the most general way possible. Rather, the emphasis will be on quantum spin systems (rather than, say, continuous systems) on lattices. 

In the next sections we give a brief overview of the main physics questions that we want to address. All the mathematics behind it will be explained in detail in the next chapters, so those readers who have not encountered, for example, tensor products of Hilbert spaces before may choose to skip this section.

\subsection{Infinite systems and locality}
Consider a single spin-1/2 system. The spin degree of freedom is described by the Hilbert space $\mathbb{C}^2$. Furthermore, if we have two quantum mechanical systems, we can get the Hilbert space of the composite system by taking the tensor product. Hence, a reasonable assumption of what the Hilbert space of infinitely many copies of the spin system (say, a spin chain) would be is
\[	
	\mathcal{H} := \bigotimes_{n = -\infty}^{\infty} \mathbb{C}^2.
\]
This, unfortunately, is not well defined. In particular, the natural choice of an inner product on this space does not converge, and hence is ill-defined.

One way to get out of this is to look at what the \emph{observables}, or more generally \emph{operators} on the system would be. If we consider a single spin, the Hilbert space is $\mathbb{C}^2$, so all linear maps on this space are given by $2 \times 2$ matrices in $M_2(\mathbb{C})$. Such a linear map is also called an operator. Then, if we have a finite number of spins, we can take the tensor product of the single site operators. By also taking linear combinations, we get the algebra of all operators on $n$ spins, in this case it is isomorphic to $M_{2^n}(\mathbb{C})$. 

Note that there is a natural way to consider an operator acting on a single site, as an operator acting on $n$ sites, by postulating that it does not do anything on the added sites. Hence, we can take the union over $n$ of all operators acting on $n$ sides, with the proper identifications. This generates what is called a $C^*$-algebra. In this case it describes all operations on the spin chain that can be approximated \emph{arbitrarily well} by \emph{local} operators. An observable or operator is called \emph{local}\index{local!observable}\index{observable!local} if it acts only on finitely many sites.

Another observation that will be important is that operators that act on \emph{distinct} sites commute with each other. This is easy to see from the construction above, since we can regard them both as operators acting on $n$ sites, where $n$ is sufficiently large. But according to the description above, the operators act trivially on the added sites, from which it follows that the two operators commute. This is called \idx{locality}: observables acting on disjoint parts of the system commute with each other.

\subsection{Local interactions}
The construction above is very general. To define a quantum spin model, one should also define \idx{dynamics}, or in other words, the time evolution. It will be convenient to work in the \idx{Heisenberg picture}, where the states are regarded as fixed, and one looks at how the observables evolve over time. In any introductory course to quantum mechanics, one learns that the dynamics are generated by the \idx{Hamiltonian}. Since we are dealing with infinite systems, it is somewhat tricky to directly define it.

It turns out that in many interesting models, spins only directly interact with neighbouring spins, and not with spins far away.\footnote{At least, usually this gives a very good approximation of the ``true'' dynamics which may be much more complicated. One can compare this for example with gravity: if we are sufficiently far away from earth, the gravitational field due to the earth mass is so small that we can effectively pretend it does not exist.} We can then define a self-adjoint (= Hermitian) operator $\Phi(X)$, where $X$ is some finite number of sites, to describe the energy due to interactions between the spins in $X$. Typically, $\Phi(X) = 0$ if the spins in $X$ are not nearest-neighbours. Assigning $\Phi(X)$ for every finite site $X$ is then called giving a \emph{local interaction}\idx{local!interaction}.

The total energy due to \emph{all} the interactions in a finite set $\Lambda$ is then given by the \emph{local Hamiltonian}\index{local!Hamiltonian}, defined as
\[
	H_\Lambda := \sum_{X \subset \Lambda} \Phi(X).
\]
The \emph{local} time evolution, that is the time evolution if only the terms in $\Lambda$ are present, is then given by
\[
	\alpha_t^{\Lambda}(A) := e^{i t H_\Lambda} A e^{-i t H_\Lambda}.
\]
To get the dynamics of the entire system, one can take $\Lambda$ to be bigger and bigger, and hope that this converges. As we will see, this is indeed the case for a large class of local interactions.

Once dynamics have been introduced, one can talk about ground states and equilibrium states, and for example show that they exist (which is not immediately obvious).

\subsection{Superselection sectors}
Recall that a reasonable assumption of any physical system is that it has charge conservation. For example, it is not possible to create a single electron, but it is possible that (say, a photon) disintegrates into an electron and a positron. Since the charges are opposite, the total charge is still neutral.

There is another way to think about this: it means that with a \emph{local} (physical, i.e.\ gauge invariant) operation it is not possible to create a single charge. Or, for the same reasons, remove one. This means that with such local operations, we can never change a state with a single charge to a state without any charge. Another example is an infinite quantum spin chain: we cannot change a state with all spins in the down position to a state with all of them in the up position, if we are only allowed to act locally (that is, on finitely many spins only at a time). A consequence is that if we take the superposition of those two states with different charges, the relative phase between them cannot be observed. We say there is a \idx{superselection rule}, since we cannot coherently superpose arbitrary states. A superselection \emph{sector} is then an equivalence class of such states.

On the mathematical level, the description is much more natural in the thermodynamic limit than for finite systems: it is related to the existence of inequivalent representations of the observables. In finite systems, however, there are no inequivalent representations, and one has to put in additional constraints by hand if one wants to discuss superselection rules.

Sometimes there is much more structure than can be recovered: by studying the properties of the inequivalent representations associated with superselection sectors, it is possible to find all different ``types'' of charges there can exist in the system, and study what happens for example when we bring two of them close together. We will outline this analysis for the toric code, which is an important example for quantum information and quantum computation.

\subsection{Lieb-Robinson bounds and applications}
It is useful to know how fast correlations spread through the system. In other words, if I do something at a site $x$, how long will it take before I can see an effect at a distant site $y$? In relativistic theories this is bounded by the speed of light, and space-like separated points in space-time cannot influence each other. For quantum spin systems \emph{a priori} there is no such bounds. In many cases of interest, however, one \emph{can} prove that such a bound exist, and there is an upper limit on the speed with which correlations can propagate through the system. This velocity is called the \emph{Lieb-Robinson velocity}\index{Lieb-Robinson!velocity}, and the effect on observables are given by \emph{Lieb-Robinson bounds}.\index{Lieb-Robinson!bound}

It is useful to, at least heuristically, understand why one could expect such a bound to hold. Recall that we often deal with strictly local interactions. This means that for small times, the largest effect of an operation $A$ will be at nearby sites. More concretely, one can consider the time evolved observable $A$, which is given by (for the local dynamics)
\[
	e^{i t H_{\Lambda}} A e^{-it H_{\Lambda}} = A + i t [H_\Lambda, A] - \frac{t^2}{2}[ H_\Lambda, [H_\Lambda, A]] + \dots.
\]
Since $H_\Lambda$ is a sum of local interactions, the commutator will kill most of the interaction terms if $A$ is also local. That is, the support of $[H_\Lambda,A]$ is only slightly larger. Since for small $t$ the remainder terms are also small, the time evolved observable can be well approximated by an observable with slightly larger support. Lieb-Robinson bounds give a precise estimate of how good this approximation is.

This has many interesting applications, some of which we will discuss in these lecture notes. One example application is that they can be used to describe scattering of ``excitations'' of quantum spin systems. Lieb-Robinson bounds are used there as follows. The excitations can be created using (almost) local operators. Then, as time flows, the excitations will move. Lieb-Robinson bounds tell us that they will still be local. This has as a consequence that, if the excitations have different velocities, they will be far away from each other at large time scales, and one can show that they essentially behave as free particles. This is essential in the construction of scattering states. It can also be shown that the Lieb-Robinson velocity provides an upper bound on the velocity of the excitations, showing once again that it is the speed of light of quantum spin systems.

\section{Topics not covered}
With respect to the topics discussed in these lecture notes, there are many important subjects left virtually untouched. One of the reasons is to keep the material short enough to be able to cover most of it in a one-semester course. The selection itself is of course heavily influenced by the author's own interests and more importantly his expertise, and should not be taken as an indication of their importance. Since the field goes back at least fifty years, by now even an overview of major topics will necessarily leave things out, but we nevertheless give some examples of related topics that are not discussed in these lecture notes. These lecture notes should provide enough background material to dive into these topics head first.

On the mathematical side we only introduce the very basics. In particular, we will mainly discuss $C^*$-algebras. For physical applications, the von Neumann algebras, which form a special class of operator algebras, also play an important role. Indeed, the field essentially started with Murray and von Neumann~\cite{10.2307/1968693}, who was interested the mathematical description of quantum mechanics. And indeed, throughout the history of the theory of operator algebras, there have been many interactions with mathematical physics. One example of this are the KMS states. Originally introduced to describe thermal equilibrium states, as we will see, they have become very useful in the study of operator algebras. This is related the so-called Tomita-Takesaki theory, which in turn has applications in physics.

In physics it has become increasingly important to do computer simulations, and quantum many body systems form no exception. Simulating such systems is hard in general, essentially because the dimension of the Hilbert space grows exponentially in the number of sites. Hence it takes a huge amount of resources to simulate such systems, even for relatively small system sizes. The techniques we will discuss here are not of much help there. However, one can look for certain subsets of states for which this problem is easier. One important example are the \idx{finitely correlated states} or \index{matrix product state}\emph{matrix product states}~\cite{MR1158756}. Such states admit a compact description, and hence reducing the number of parameters that have to be stored, that could in principle be stored on a computer. This make them much more amenable to simulations that an arbitrary state on an infinite lattice. The review~\cite{mpsreview} discusses many applications, and since this review has appeared, new and improved applications have been found.

Our go-to example in these lectures will be the toric code. This is a relatively simple model that nevertheless has many interesting properties. It is however somewhat atypical, in the sense that many of its features are not shared with a general quantum spin system, which of course is one of the reason why it is interesting. There are many other interesting examples that do not make an appearance (or only briefly) in these lecture notes. Perhaps the best-known example is the quantum Heisenberg model, which models magnetism. A specific instance of this model is the 1D Ising spin chain, which models spins in the presence of a transverse magnetic field. A typical question that one is interested in when studying such models is if there exist phase transitions: abrupt changes in the value of physical quantities. For example, a system can become spontaneously magnetised as the temperature drops below a critical temperature~\cite{Dyson1978}.

Another example, which in a sense is an extension of the Heisenberg model, is the AKLT (after Affleck, Kennedy, Lieb and Tasaki) model~\cite{affleck1988}. The ground states of this model can be described using finitely correlated states methods~\cite{MR1158756}. There are many interesting questions one can ask about this model, for example what phases there are: depending on the choice of parameters in the model, distinct physical behaviour is observed. One question that one can ask is what the energy spectrum of the Hamiltonian is. In general this question is too difficult to answer explicitly, instead one considers the easier question if the system is gapped or not. That is, if we normalise the Hamiltonian such that the ground state of energy zero, if there is a $\gamma > 0$ such that there are no states with energy in $(0,\gamma)$, independent of the size of the system. Although Affleck, Kennedy, Lieb and Tasaki showed that there is a spectral gap, it was difficult to give good bounds on the gap. In later work, Nachtergaele introduced a more general method to proof the existence of a spectral gap (and give good bounds) in certain finitely correlated state models~\cite{nachtergaele1996}.

Leaving aside specific examples, there are many more structural results on the type of models that we consider. For example, we will only briefly discuss the role of symmetries. Symmetries are often a very helpful tool in studying physics. One way they can be used is to discuss order and disorder in large systems. This is related to symmetry breaking: it is possible that even though a Hamiltonian admits a certain symmetry, this symmetry is broken in the ground state~\cite{MR1919619,MR2374090}. This is the mechanism that is behind the Higgs boson, but it can also appear in quantum spin systems. Symmetries can also be used to decompose invariant states into \emph{extremal} invariant states, which cannot be written as a combination of distinct invariant states. This is called \emph{ergodic decomposition}. This helps in studying, for example, the structure of the set of translation invariant KMS states, which as we mentioned earlier, describe states at thermal equilibrium.

\section{Acknowledgements}
These lecture notes grew out a course given at the Leibniz Universit\"at Hannover in the \emph{Sommersemester} 2013. I thank the participants of the course for their helpful comments and feedback, and the university for giving me this opportunity to teach. In addition, I would like to thank Tobias Osborne and Reinhard Werner for useful advice. The many positive comments and feedback of my colleagues kept me motivated to correct and extend an earlier version of these lecture notes. Thanks also to Leopold Kellers and Alessandro Michelangeli for pointing out a gap in one of the proofs in an earlier version.

These lecture notes have evolved over the course of almost four years, during which the author has been supported first by the Dutch Organisation for Scientific Research (NWO) through a Rubicon grant and later has received funding from the European Union's Horizon 2020 research and innovation program under the Marie Sklodowska-Curie grant agreement No 657004. Partial support through the EU project QFTCMPS and the cluster of excellence EXC 201 Quantum Engineering and Space-Time Research is also gratefully acknowledged.

And finally, the author would like to thank Christian Caron of Springer for his editorial support.

\chapter{Operator algebras}\label{ch:opalg}
Linear maps acting on Hilbert spaces play in important role in quantum mechanics. Such a linear map is called an \idx{operator}. Well-known examples are the Pauli matrix $\sigma_z$, which measures the spin in the $z$ direction, and the position and momentum operators $P$ and $Q$ of a single particle on the line. The first acts on the Hilbert space $\mathbb{C}^2$, while the latter are defined on dense subspaces of $L^2(\mathbb{R})$. An \idx{operator algebra} is an algebra of such operators, usually with additional conditions such as closure in a certain topology. Here we introduce some of the basic concepts in the theory of operator algebras. The material here is standard, and by now there is a huge body of textbooks on the subject, most of which cover a substantially bigger part of the subject than these notes. Particularly recommended are the two volumes by Bratteli and Robinson~\cite{MR887100,MR1441540}, which contain many applications to physics. Many of the topics covered here are studied there \emph{in extenso}. The books by Kadison \& Ringrose provide a very thorough introduction to the field, with many exercises~\cite{MR1468230,MR719020}. The book by Takesaki~\cite{MR1873025} (and the subsequent volumes II and III~\cite{MR1943006,MR1943007}) is a classic, but is more technical. Volume I of Reed and Simon's \emph{Methods of Modern Mathematical Physics}~\cite{MR751959} or Pedersen's \emph{Analysis Now}~\cite{MR971256} cover the necessary tools of functional analysis (and much more), but do not cover most of the material on operator algebras we present here.

Before going into the mathematical details, let us take a step back, and try to motivate the study of operator algebras. Is it really necessary, or even useful, to develop this whole mathematical machinery? What is wrong with the usual picture of a Hilbert space of wave functions, with, say, position and momentum observables? There are a few reasons one can give.

The operator algebraic approach generally helps to obtain results at the level of mathematical rigour. Especially for systems with infinitely many degrees of freedom, the algebraic approach is often more suitable. These play an important role in quantum field theory and quantum statistical mechanics. In the latter, one tries to obtain properties from systems with very many microscopic degrees of freedom. The operator algebraic approach gives a natural playground for the study of structural questions, such as the (non)-existence of translationally invariant equilibrium states in certain classes of models. Finally, the subject is of interest from a purely mathematical point of view as well, and has led to interesting collaborations between (mathematical) physicists and mathematicians.

We first recall some basic notions of functional analysis, regarding Hilbert and Banach spaces, and linear maps on Hilbert spaces. Then the notion of a $C^*$-algebra is introduced. They provide a central role in modelling the quantum systems in which we are interested. It is possible to define the spectrum of an element of a $C^*$-algebra. This provides a generalisation of the notion of the eigenvalue of a matrix. Using the spectrum a powerful technique, the functional calculus, can be developed: it allows one to apply suitable continuous functions (for example, the square root) to certain classes of operators. We will only need most of the results in that section near the end of this book, so it can be safely skipped on first reading.

To make the connection between the abstract $C^*$-algebraic picture and physical reality it is necessary to have a way to describe expected values of measurement outcomes. This naturally leads to the notion of a \emph{state} on a $C^*$-algebras. Using such states there is then a canonical way, the \emph{GNS construction} to represent the system on a Hilbert space, bringing us back to the more familiar picture of quantum mechanics. This construction is outlined in Section~\ref{sec:gns}. Finally, at the end of this chapter we will see how the theory of quantum mechanics can be phrased in this mathematical framework.  

\section{Functional analysis: a quick review}\label{sec:funcan}
We first recall the definition of a \idx{Hilbert space}. Let $\mc{H}$ be a vector space over $\mathbb{C}$.\footnote{One can also define Hilbert spaces over $\mathbb{R}$, but these will play no role in these lectures.} An inner product on $\mc{H}$ is a function $\langle \cdot, \cdot\rangle: \mc{H} \times \mc{H} \to \mathbb{C}$ such that the following conditions hold:
\begin{enumerate}
	\item \emph{Conjugate symmetry}: $\langle \eta, \psi \rangle = \overline{\langle \psi, \eta\rangle}$ for all $\eta,\psi \in \mc{H}$,
	\item \emph{Linearity in the second variable}:\footnote{In the mathematics literature the inner product is often linear in the \emph{first} variable, and anti-linear in the second. We adopt the physics conventions.} $\langle \psi, \lambda \eta + \xi \rangle = \lambda \langle \psi, \varepsilon \rangle + \langle \psi, \xi \rangle$ for all $\psi, \eta, \xi \in \mc{H}$ and $\lambda \in \mathbb{C}$,
	\item \emph{Positive definiteness}: if $\langle \psi, \psi \rangle \geq 0$ for all $\psi \in \mc{H}$, and $\langle \psi, \psi \rangle = 0$ if and only if $\psi = 0$.
\end{enumerate}
If the condition that $\langle \psi, \psi \rangle = 0$ for $\psi = 0$ only is not satisfied, $\mc{H}$ is called a \emph{pre-inner product space}. One can always divide out the subspace of vectors with norm zero to obtain a proper inner product space from a pre-inner product space. Note that from the first and second condition we obtain $\langle \lambda \psi, \eta \rangle = \overline{\lambda} \langle \psi, \eta \rangle$.

With the help of an inner product a norm\index{norm!inner product space} $\| \cdot \|$ can be defined on $\mc{H}$, by $\|\psi\|^2 = \langle \psi, \psi \rangle$ for all $\psi \in \mc{H}$. That this indeed is a norm follows from the well-known Cauchy-Schwarz inequality, which states that $|\langle \psi,\xi \rangle|^2 \leq \langle \psi,\psi \rangle \langle \xi, \xi \rangle$, and the other properties of an inner product. Hence each inner product space is a normed vector space in a canonical way.

A normed space can be endowed with a \emph{topology}\index{norm topology}, since the norm defines a distance function (or \emph{metric}), defined by $d(\psi,\xi) := \|\psi - \xi\|$. In other words, we can talk about convergence of sequences and of continuity. Recall that a sequence ${(\psi_n)}_n$ in a metric space converges to some element $\psi$ if and only if $d(\psi_n, \psi)$ tends to zero as $n$ goes to infinity. A special role is played by the \emph{Cauchy sequences}\index{Cauchy sequence}. A sequence ${(\psi_n)}_n$ in a metric space is a Cauchy sequence if for each $\varepsilon > 0$, there is some integer $N$ such that $d(\psi_n, \psi_m) < \varepsilon$ for all $n,m \geq N$. This can be interpreted as follows: if we look far enough in the sequence, all elements will be close to each other. 

It is not so difficult to show that if a sequence converges in a metric space, it must necessarily be a Cauchy sequence. The converse is not true: there are metric spaces in which Cauchy sequences do not converge. For example, consider the space $X = (0,1)$ with the usual metric of $\mathbb{R}$. Then the sequence given by $x_n = 1/n$ is a Cauchy sequence, but it does not converge (since $0 \notin X$). A metric space in which all Cauchy sequences converge is called \emph{complete}\index{complete metric space}\index{metric space!complete}. The most familiar example is likely the real line $\mathbb{R}$, but note that a subspace of a complete space is not necessarily complete, as the example above shows. 

In these notes we will mainly be concerned with a special class of complete spaces: the Banach spaces. A \idx{Banach space} is a normed vector space that is complete with respect to the norm. For quantum mechanics the natural playground is the \emph{Hilbert space}, whose definition we recall here.
\begin{definition}
A \idxr{Hilbert space} $\mc{H}$ is an inner product space over $\mathbb{C}$ that is also a Banach space. That is, it is complete with respect to the norm induced by the inner product.
\end{definition}

\subsection{Examples of Hilbert spaces}
The easiest examples of Hilbert spaces are the finite dimensional ones, where $\mc{H} = \mathbb{C}^N$ for some integer $N$. The inner product is the usual one, that is if $\xi, \eta \in \mc{H}$ then $\langle \xi, \eta \rangle = \sum_{i=1}^N \overline{\xi}_i \eta_i$. It follows from basic linear algebra that this indeed is an inner product. That $\mc{H}$ is complete with respect to the metric induced by the inner product essentially follows from the completeness of $\mathbb{C}$. This can be seen by noting that a sequence $\xi_n \in \mc{H}$ is Cauchy if and only if the sequences of coordinates $n \mapsto (\xi_i)_n$, with $i = 1, \dots N$ is Cauchy. Hence $\mathbb{C}^N$ is a Hilbert space.

Next we consider infinite dimensional Hilbert spaces. Write $\ell^2(\mathbb{N})$\index{_l2N@$\ell^2(\mathbb{N})$} for the set of all sequences $x_n$, such that $\sum_{i=1}^\infty |x_i|^2 < \infty$. For two such sequences $x_n$ and $y_n$, define
\begin{equation}
	\label{eq:elltwo}
	\langle x_n, y_n \rangle_{\ell^2} = \sum_{i=1}^\infty \overline{x}_n y_n.
\end{equation}
The first thing to check is that this definition makes sense: a priori it is not clear that the sum on the right hand side converges. Note that since $(|x|-|y|)^2 \geq 0$ for all complex numbers $x$ and $y$, it follows that $|\overline{x}_n y_n| \leq \frac{1}{2}(|x_n|^2+|y_n|^2)$. Hence for $x_n$ and $y_n$ elements of $\ell^2(\mathbb{N})$, the right hand side of equation~\eqref{eq:elltwo} converges absolutely, and hence converges. It follows that $\ell^2(\mathbb{N})$ is a vector space over $\mathbb{C}$ and that equation~\eqref{eq:elltwo} defines an inner product.

\begin{exercise}
Show that $\ell^2(\mathbb{N})$ is complete with respect to the norm induced by the inner product (and hence a Hilbert space).
\end{exercise}

A familiar example from quantum mechanics is the Hilbert space $L^2(\mathbb{R})$\index{_l2R@$L^2(\mathbb{R})$}, consisting of square integrable functions $f: \mathbb{R} \to \mathbb{C}$ with the inner product
\[
\langle f, g \rangle_{L^2} = \int \overline{f(x)} g(x) dx.
\]
The definition of such integrals is a subtle matter. It is usually implicit that the integration is with respect to the Lebesgue measure on $\mathbb{R}$.\footnote{The appropriate setting for a mathematical theory of integration is \emph{measure theory}. We will briefly come back to this when we discuss the functional calculus, but for the most part of this book it will play no role. Details can be found in most graduate level texts on analysis or functional analysis.} In calculations, on the other hand, one usually has to deal with integrals of explicitly given functions and one does not have to worry about such issues. The proof that the vector space of all square integrable functions is indeed a Hilbert space requires some results of measure theory, and is outside of the scope of these notes. It can be found in most textbooks on functional analysis.

The definition of $L^2$ spaces can be generalised to any measure space $(X, \mu)$ by using the corresponding integral in the definition of the inner product. In the definition of the $L^2$ space one has to identify functions whose values only differ on a set of measure zero, since otherwise we would not obtain a proper inner product, because each function that is non-zero on a set of measure zero has norm zero, but is not equal to the zero function. So strictly speaking the $L^2$ spaces consist of \emph{equivalence classes} of functions. The resulting Hilbert space is then denoted by $L^2(X, d\mu)$.

\begin{remark}\label{rem:ell2}
	The Hilbert space $\ell^2(\mathbb{N})$ can be seen as a special case of an $L^2$ space. To see this, note that for each sequence $x_n$ we can find a function $f: \mathbb{N} \to \mathbb{C}$, defined by $f(n) = x_n$. Conversely, each such function defines a sequence via the same formula. Next we define a measure $\mu$ on $\mathbb{N}$ by $\mu(X) = |X|$, where $|X|$ is the number of elements for any $X \subset \mathbb{N}$. This measure is usually called the \emph{counting measure} for obvious reasons. For two functions $f$ and $g$ the $L^2$ inner product then becomes
\[
\langle f, g \rangle_{L^2} = \int \overline{f(n)} g(n) d\mu(n) = \sum_{i=1}^\infty \overline{f(n)} g(n) = \langle f, g \rangle_{\ell^2},
\]
where in the last equality we identified the functions $f$ and $g$ with their corresponding sequences, as above. It follows that $\ell^2(\mathbb{N}) \equiv L^2(\mathbb{N}, d\mu)$.
\end{remark}

\subsection{Linear maps}
Once one has defined a class of mathematical objects (in our case, Hilbert spaces), the next question to answer is what kind of maps between Hilbert spaces one considers.\footnote{In abstract nonsense language: one has to say in which category one is working.} For Hilbert spaces linear maps\index{linear map} $L$ are an obvious choice. That is, $L: \mc{H}_1 \to \mc{H}_2$ is a function such that for all $\lambda \in \mathbb{C}$ and $\xi, \eta \in \mc{H}_1$ we have
\[
L(\lambda \xi) = \lambda L(\xi) \quad \textrm{and}\quad L(\xi + \eta) = L(\xi) + L(\eta).
\]
We will usually omit the brackets and simply write $L\xi$ for $L(\xi)$.

Since Hilbert spaces carry a topology induced by the norm, it is natural to restrict to those linear maps that are continuous with respect to this topology. On the other hand, it is often convenient to work with bounded linear maps\index{linear map!bounded}, for which there is a constant $M > 0$ such that $\|L\xi\| \leq M \|\xi\|$, independent of the vector $\xi$. We should note that not all operators one encounters in quantum mechanics are bounded, for example the momentum operator of a particle on the line is not. The treatment of unbounded linear maps is a topic on its own, to which we will only return briefly when discussing Hamiltonians.

The next proposition says that bounded and continuous linear maps are actually the same. Note that this means that unbounded linear maps are not continuous. In general, they can only be defined on a \emph{subset} of a Hilbert space.
\begin{proposition}
\label{prop:lincont}
Let $L: V \to W$ be a linear map between two normed vector spaces. Then the following are equivalent:
\begin{enumerate}
	\item $L$ is continuous with respect to the norm topology,
	\item $L$ is bounded.
\end{enumerate}
\end{proposition}
\begin{proof}
First suppose that $L$ is bounded, so that there is a constant $M > 0$ such that $\|Lv\| \leq M \|v\|$. Let $\varepsilon > 0$ and choose $\delta = \varepsilon/M$. Then for any $v_1, v_2 \in V$ such that $\| v_1 - v_2 \| < \delta$, we have (with the help of linearity)
\[
	\|L v_1-L v_2\| = \| L(v_1-v_2) \| \leq M \|v_1-v_2\| < \delta M = \varepsilon.
\]
This proves that $L$ is continuous.

Conversely, set $\varepsilon = 1$ and let $\delta > 0$ be the corresponding maximal distance as in the definition of continuity at the point $v_0 = 0$. Then if $\|v\| < \delta$ (hence the distance of $v$ to $0$ is smaller than $\delta$), it follows that $\|Lv\| \leq 1$. If $v$ is any non-zero vector, the norm of $\frac{\delta \cdot v}{2\|v\|}$ is smaller than $\delta$. By linearity it follows that
\[
\left\| \delta \frac{Lv}{2 \|v\|} \right\| \leq 1, \quad \Leftrightarrow \quad \| L v \| \leq \frac{2}{\delta} \|v\|.
\]
This concludes the proof, since for $v = 0$ the inequality is satisfied trivially.
\end{proof}
Hence continuous linear maps and bounded linear maps between Hilbert spaces are the same thing. This is often very convenient. For example, in many situations it is easy to construct a linear map that is bounded on some dense subset of a Hilbert space. Then this result allows us to extend the definition of the linear map to the whole Hilbert space, and obtain a bounded linear map again. We will see examples of this later.

\begin{definition}\index{_b(H)@$\alg{B}(\mc{H})$}
Let $\mc{H}_1$ and $\mc{H}_2$ be Hilbert spaces. We write $\alg{B}(\mc{H}_1, \mc{H}_2)$ for the set of bounded linear maps from $\mc{H}_1$ to $\mc{H}_2$. As a shorthand, $\alg{B}(\mc{H}) := \alg{B}(\mc{H}, \mc{H})$.
\end{definition}
An element of $\alg{B}(\mc{H}_1, \mc{H}_2)$ will also be called a \emph{(bounded) operator}\index{operator!bounded}. Usually we will restrict to maps in $\alg{B}(\mc{H})$, and call such elements simply \emph{operators on $\mc{H}$}. Note that $\alg{B}(\mc{H}_1, \mc{H}_2)$ is a vector space over $\mathbb{C}$. We will later study more properties of these spaces. Essentially, the theory of operator algebras amounts to studying these spaces (and subspaces thereof).

\label{page:adjoint}
If $A$ is a bounded linear map $A: \mc{H}_1 \to \mc{H}_2$, the \idx{adjoint} $A^*$ of $A$ is the unique linear map $A^*: \mc{H}_2 \to \mc{H}_1$\index{_astar@$A^*$} such that
\begin{equation}
	\label{eq:adjoint}
\langle \psi, A^* \eta \rangle_{\mc{H}_1} = \langle A \psi, \eta \rangle_{\mc{H}_2}, \quad \textrm{for all } \psi \in \mc{H}_1,\, \eta \in \mc{H}_2.
\end{equation}
Of course one has to show that such a linear map $A^*$ exists and indeed is unique. To do this rigorously requires some work,\footnote{\label{ftn:adjoint}The proof requires the \idx{Riesz representation theorem}. First note that the map $\xi \mapsto \overline{\langle A \xi, \psi \rangle}$ is a bounded linear functional (that is, a linear map to $\mathbb{C})$ on $\mc{H}$. By the Riesz representation theorem there is a $\eta \in \mc{H}$ such that $\langle \eta, \xi \rangle = \overline{\langle A \xi, \psi\rangle}$ for all $\xi \in \mc{H}$. Set $A^*\psi = \xi$. This uniquely defines a bounded operator $A^*$ with the right properties.} but a heuristic argument goes as follows. Choose orthonormal bases $e_i$ and $f_j$ of $\mc{H}_{1}$ and $\mc{H}_2$. Then we can represent $A^*$ as a matrix (an ``infinite'' matrix if either of the Hilbert spaces is infinite dimensional) with respect to these bases. The elements $\langle e_i, A^* f_j \rangle_{\mc{H}_1}$ give the matrix coefficients in this basis, and hence determine the linear map $A^*$. In the physics literature the adjoint is usually called the \emph{Hermitian conjugate}\index{Hermitian conjugate|see {adjoint}} and denoted by a $\dagger$. We will write $*$ for the adjoint, as is common in mathematics and mathematical physics. We will come back to the adjoint of linear maps on Hilbert spaces when we discuss $C^*$-algebras.

\begin{exercise}\label{ex:adjoint}
Verify that the adjoint has the following properties: 
\begin{enumerate}
	\item $(A^*)^* = A$ for any linear map $A$,
	\item the adjoint is conjugate linear, i.e. $(\lambda A + B)^* = \overline{\lambda}A^* + B^*$, with $\lambda \in \mathbb{C}$,
	\item and $(AB)^* = B^* A^*$.
\end{enumerate}
Here $AB$ stands for the composition $A \circ B$, that is, $ABv = A(B(v))$. In these notes we will always omit the composition sign of two linear maps.
\end{exercise}

An operator $A$ is called \emph{self-adjoint}\index{operator!self-adjoint} if $A= A^*$ and \emph{normal} if $A A^* = A^* A$. It is an \idx{isometry} if $V^* V = I$, where $I$ is the identity operator. The name ``isometry'' comes from the fact that this implies that $V$ is an isometric map, since
\[
	\|V \psi\|^2 = \langle V \psi, V \psi \rangle = \langle \psi, V^*V \psi \rangle = \langle \psi, \psi \rangle = \|\psi\|^2.
\]
Note that an isometry is always injective. A \emph{projection $P$}\index{projection} is an operator such that $P^2 = P^* = P$. Projectors project onto closed subspaces of the Hilbert space $\mc{H}$ on which they are defined. If $V$ is an isometry, it is easy to check that $VV^*$ is a projection.

Finally we come to the \emph{isomorphisms}\index{isomorphism!of Hilbert spaces} of Hilbert spaces. An isomorphism $U: \mc{H}_1 \to \mc{H}_2$ is a bounded linear map, such that its inverse $U^{-1}$ exists and is also a linear map. We also require that $U$ is an isometry, $\|U \psi \| = \|\psi\|$.
\begin{theorem}
	Let $U: \mc{H}_1 \to \mc{H}_2$ be a bounded linear map between two Hilbert spaces. Then the following are equivalent:
	\begin{enumerate}
		\item\label{it:isomorph} $U$ is an isomorphism of Hilbert spaces,
		\item\label{it:innerpd} $U$ is surjective and $\langle U \xi, U \eta \rangle = \langle \xi, \eta \rangle$ for all $\eta,\xi \in \mc{H}_1$,
		\item\label{it:unitary} $U^*U = U U^* = I$, where $I$ is the identity map. (Note that it would be more accurate to distinguish between the identity map of $\mc{H}_1$ and of $\mc{H}_2$).
	\end{enumerate}
\end{theorem}
\begin{proof}
\eqref{it:isomorph}~$\Rightarrow$~\eqref{it:innerpd} Since $U$ is invertible it must be surjective. Because $U$ is in addition an isometry, it follows that $\langle U\xi, U\xi\rangle = \|U\xi\|^2 = \langle \xi, \xi \rangle$ for all $\xi \in \mc{H}_1$. Now let $\xi,\eta \in \mc{H}_1$. Then the \emph{polarisation identity}, which can be easily verified, says that
\begin{equation}
	\label{eq:polarisation}
\langle U \xi, U \eta \rangle = \frac{1}{4} \sum_{k=0}^3 i^k \langle U(\xi + i^k \eta), U(\xi + i^k \eta)\rangle,
\end{equation}
and the claim follows.

\eqref{it:innerpd} $\Rightarrow$~\eqref{it:unitary} Since $U$ is an isometry it follows that $U^*U = I$ and that $U$ is injective. Hence $U$ is a bijection and therefore invertible, so $U^* U U^{-1} = U^* = U^{-1}$. It follows that $UU^* = I$.

\eqref{it:unitary} $\Rightarrow$~\eqref{it:isomorph} The assumptions say that $U^*$ is the inverse of $U$. Because $U^*U = I$ it is also clear that $U$ is isometric.
\end{proof}
A continuous linear map between Hilbert spaces satisfying assumption~\ref{it:unitary} of the Theorem will be called a \idx{unitary}.

The theorem can be used to show that, in a sense, there are not many different Hilbert spaces. To make this precise, recall that an orthonormal basis of a Hilbert space $\mc{H}$ is a set $\{ \xi_\alpha \}$ of vectors of norm one such that $\langle \xi_\alpha, \xi_\beta \rangle = 0$ if $\alpha \neq \beta$ and such that the linear span of these vectors is \emph{dense}\index{dense linear space} in $\mc{H}$. That is, for each vector $\xi$ and $\varepsilon > 0$, one can find a linear combination of \emph{finitely} many of the basis vectors that is close to $\xi$, i.e.\ there exist finitely many $\lambda_\alpha \neq 0$ such that
\[
	\left\| \sum_{i=1}^k \lambda_{\alpha_i} \xi_{\alpha_i} - \xi \right\| < \varepsilon.
\]
The cardinality of a set of orthonormal basis vectors is called the \idx{dimension} of a Hilbert space. One can show that this is independent of the choice of orthonormal basis. This leads to the following result.
\begin{proposition}
Two Hilbert spaces $\mc{H}_1$ and $\mc{H}_2$ are isomorphic if and only if they have the same dimension.
\end{proposition}
\begin{proof}[Proof (sketch)]
	If $U: \mc{H}_1 \to \mc{H}_2$ is a unitary and $\{ \xi_\alpha \}$ is an orthonormal basis, then $\{ U \xi_\alpha \}$ is again orthonormal. Moreover, because $U$ is surjective, it also is a basis of $\mc{H}_2$.

	Conversely, suppose that we have bases $\{\xi_\alpha\}$ and $\{\eta_\alpha\}$ of $\mc{H}_1$ and $\mc{H}_2$ respectively. Define a linear map $U$ by $U\xi_\alpha = \eta_\alpha$. This can be linearly extended to the span of $\xi_\alpha$, and is clearly bounded. It follows that $U$ can be extended to a map of $\mc{H}_1$ to $\mc{H}_2$. One can then verify that this map must be a unitary.
\end{proof}

\subsection{Completion}\index{completion|(}
There is no reason why a metric space should be complete. In many constructions we will encounter in these notes this will be the case. Nevertheless, it is possible to embed a (non-complete) metric space into a space that \emph{is} complete. A well-known example is the completion of the rational numbers $\mathbb{Q}$: this is precisely the set of real numbers $\mathbb{R}$.\footnote{This is essentially how the real numbers are usually \emph{defined}. Namely, they are limits of Cauchy sequences of rational numbers.} There is a certain ``minimal'' way to do this, in the sense that we \emph{only} add limits of Cauchy sequences, and nothing more. Essentially this means that the original metric space is \emph{dense} in the completion. Moreover, it makes sense to demand that the metric on the new space restricts to the metric on the original space (embedded in the completion).

In this section we outline this construction with an example: the construction of a Hilbert space from a (not necessarily complete) inner product space $H$. This construction is primarily of theoretical interest. It shows that a completion always exists, but the description of the completion is not very convenient in concrete calculations, and in practice it is usually sufficient to know that a completion exists.

So let $H$ be a vector space over the complex numbers, and let $\langle \cdot,\cdot\rangle_H$ be an inner product on $H$. We will define a Hilbert space $\mc{H}$ containing $H$, and show that it is complete. First, let $(\xi_n)_{n=1}^\infty$ and $(\psi_n)_{n=1}^\infty$ be two Cauchy sequences in $H$. We say that the sequences are \emph{equivalent}, and write $(\xi_n)_n \sim (\psi_n)_n$, if for each $\varepsilon > 0$ there is some $N > 0$ such that $\|\xi_n - \psi_n\| < \varepsilon$ for all $n > N$. Hence two Cauchy sequences are equivalent if they are arbitrarily close for sufficiently large $N$. 

It is possible to define a scalar multiplication and addition on the set of Cauchy sequences. Let $(\psi_n)_n$ and $(\xi_n)_n$ be two Cauchy sequences. If $\lambda \in \mathbb{C}$, then $(\lambda \psi_n)_{n=1}^{\infty}$ is again a Cauchy sequence. Similarly, the sequence $(\psi_n + \xi_n)_n$ is Cauchy (by the triangle inequality). These operations are well-defined with respect to the equivalence relation defined above, so that $\lambda (\psi_n)_n \sim \lambda (\xi_n)_n$ if $(\psi_n)_n \sim (\xi_n)_n$ (and similarly for addition). This means that we can take the vector space $V$ of Cauchy sequences, and divide out by the equivalence relation. We will denote the corresponding vector space by $\mc{H}$.

There is a different way to think about $\mc{H}$. First of all, note that $(\psi_n)_n \sim (\xi_n)_n$ if and only if $(\psi_n-\xi_n)_n \simeq (0)_n$, where $(0)_n$ is the sequence of all zeros. The set of sequences that are equivalent to the zero sequence forms a vector space $V_0$. Then $\mc{H}$ is equal to $V/V_0$, the quotient of $V$ by the vector space of ``null sequences''. Hence an element of $\mc{H}$ is an equivalence class $[(\psi_n)_n]$, and two representatives of the same class are related by $(\psi_n)_n = (\xi_n)_n + (\eta_n)_n$, where $(\eta_n)_n \in V_0$. Sometimes one writes $(\psi_n)_n + V_0$ for this equivalence class.

At this point $\mc{H}$ is a vector space over $\mathbb{C}$. The goal is to define an inner product on $\mc{H}$ that is compatible, in a sense to be made precise later, with the inner product of $H$. To this end, define
\[
\langle [(\psi_n)_n], [(\xi_n)_n] \rangle_{\mc{H}} := \lim_{n \to \infty} \langle \psi_n, \xi_n \rangle_{H}.
\]
There are a few things to check. First of all, the right hand side should converge for this definition to even make sense. Secondly, it should be well defined in the sense that it does not depend on the representative of the equivalence class. The first property follows with the Cauchy-Schwartz inequality: because $\psi_n$ and $\xi_n$ are Cauchy sequences, $\langle \psi_n, \xi_n\rangle$ will be a Cauchy sequence in $\mathbb{C}$, and hence converge by completeness of $\mathbb{C}$. As for the second part, consider for example $(\xi_n)_n \sim (\eta_n)_n$. Then
\[
|\langle \psi_n, \xi_n \rangle - \langle \psi_n, \eta_n\rangle| = |\langle \psi_n, (\xi_n-\eta_n)\rangle| \leq \| \psi_n \| \| \xi_n -\eta_n \|.
\]
Since $(\psi_n)_n$ is a Cauchy sequence, the sequence $\| \psi_n \|$ is uniformly bounded. The terms $\| \xi_n - \eta_n\|$ go to zero as $n$ goes to infinity (by assumption). Hence we conclude that $\langle [(\psi_n)_n], [(\xi_n)_n] \rangle_{\mc{H}} = \langle [(\psi_n)_n], [(\eta_n)_n] \rangle_{\mc{H}}$.

The original space $H$ can be embedded into $\mc{H}$ in the following way. Let $\psi \in \mc{H}$. Then $(\psi_n)_n$, where $\psi_n = \psi$ for each $n$, clearly is a Cauchy sequence. Then the embedding is given by a map $\iota: H \to \mc{H}$, with $\iota(\psi) = [(\psi_n)_n]$, the equivalence class of the constant sequence with values $\psi$. Note that $\langle \eta, \psi \rangle_H = \langle \iota(\eta), \iota(\psi) \rangle_{\mc{H}}$. In particular, the map $\iota$ is an isometry.

\begin{exercise}
Verify the details, for example show that $\iota$ is linear and that $\mc{H}$ indeed is a complete metric space.
\end{exercise}

The results in this section, together with the exercise above, can be summarised as follows.
\begin{theorem}
	Let $H$ be an inner product space. Then there is a Hilbert space $\mc{H}$ and a linear embedding $\iota: H \hookrightarrow \mc{H}$ such that $\iota(H)$ is dense in $\mc{H}$ and $\langle \psi, \xi \rangle_H = \langle \iota(\psi), \iota(\xi) \rangle_{\mc{H}}$. 
\end{theorem}
The completion is in fact unique in the following sense. Let $\mc{H}'$ be another completion, that is there is an isometric linear map $\iota' : H \hookrightarrow \mc{H}'$ whose image is dense in $\mc{H}'$. Then there is a unitary $U: \mc{H} \to \mc{H}'$.\index{completion!uniqueness}

Usually $H$ is identified with its image in $\mc{H}$, that is, we think of $H \subset \mc{H}$ and every element of $\mc{H}$ can be approximated arbitrarily well by an element in $H$. Later we will need the completion of algebras, with respect to some norm. These can be constructed in a similar manner.
\index{completion|)}

\subsection{Tensor products and direct sums}\label{sec:tensprod}
There are few basic constructions to obtain new Hilbert spaces from (pairs of) Hilbert spaces. The simplest construction is taking subspaces. Let $\mc{H}$ be a Hilbert space and let $\mc{K} \subset \mc{H}$ be a linear subspace. If $\mc{K}$ is \emph{closed} in norm (that is, the limit of any Cauchy sequence in $\mc{K}$ is also in $\mc{K}$), then $\mc{K}$ is a Hilbert space, by restricting the inner product of $\mc{H}$ to $\mc{K}$. Note that one can define a projection $P: \mc{H} \to \mc{H}$ onto the subspace $\mc{K}$. This projection is sometimes written as $[\mc{K}]$.

Next, suppose that we have two Hilbert spaces $\mc{H}_1$ and $\mc{H}_2$. The \idx{direct sum} $\mc{H}_1 \oplus \mc{H}_2$ consists of all tuples $(\xi, \eta)$ with $\xi \in \mc{H}_1$ and $\eta \in \mc{H}_2$. This is a Hilbert space, if we define the following addition rules and inner product:
\begin{align*}
	\lambda_1 (\xi_1, \eta_1) + \lambda_2 (\xi_2, \eta_2) &= (\lambda_1 \xi_1 + \lambda_2 \xi_2, \lambda_1 \eta_1 + \lambda_2 \eta_2), \\
	\langle (\xi_1, \eta_1), (\xi_2, \eta_2) \rangle_{\mc{H}_1 \oplus \mc{H}_2} &= \langle \xi_1, \xi_2 \rangle_{\mc{H}_1} + \langle \eta_1, \eta_2 \rangle_{\mc{H}_2},
\end{align*}
where the $\lambda_i$ are scalars. Note there is a natural isometry $V_1 : \mc{H}_1 \to \mc{H}_1 \oplus \mc{H}_2$ defined by $V_1 \xi = (\xi, 0)$, and similarly for $V_2$. If $L_i : \mc{H}_i \to \mc{K}_i$ for $i=1,2$ and $\mc{K}_i$ Hilbert spaces are linear maps, we can define a linear map $L_1 \oplus L_2: \mc{H}_1 \oplus \mc{H}_2 \to \mc{K}_1 \oplus \mc{K}_2$ by\index{linear map!direct sum}
\[
	(L_1 \oplus L_2)(\xi, \eta) = (L_1 \xi, L_2 \eta).
\]
If $L_1$ and $L_2$ are bounded, then so is $L_1 \oplus L_2$. Hence in that case it is continuous, by Proposition~\ref{prop:lincont}. One can show that $(L_1 \oplus L_2)^* = L_1^* \oplus L_2^*$. This construction readily generalises to finite direct sums. It is also possible to define \emph{infinite} direct sums, but this requires a bit more care to make sure all expressions converge.

\begin{exercise}\label{ex:infsum}
Define the direct sum of a countable collection of Hilbert spaces, and show that this indeed defines a Hilbert space.
\end{exercise}

Finally we discuss the tensor product of two Hilbert spaces.\index{tensor product!of Hilbert spaces} This has a direct physical interpretation: if two quantum systems are described by two Hilbert spaces, the tensor product describes the combination of these two systems as a whole. Its construction is a bit more involved than that for direct sums. One way to do it is as follows. First consider the vector space $V$ consisting of formal (finite) linear combinations of elements of the form $\xi \otimes \eta$ for some $\xi \in \mc{H}_1$ and $\eta \in \mc{H}_2$. That is, an element of $V$ is of the form $\sum_{i=1}^n \lambda_i \xi_i \otimes \eta_i$. Now this is a huge space, and we want to impose some relations. That is, we will identify elements of the following form:
\begin{equation}
	\label{eq:tensrel}
\begin{split}
		\lambda (\xi \otimes \eta) = (&\lambda \xi) \otimes \eta = \xi \otimes (\lambda \eta),\\
		(\xi_1 + \xi_2) \otimes \eta &= \xi_1 \otimes \eta + \xi_2 \otimes \eta,\\
		\xi \otimes (\eta_1 + \eta_2) &= \xi \otimes \eta_1 + \xi \otimes \eta_2.\\
\end{split}
\end{equation}
If we quotient $V$ by this relations we obtain a vector space $H$. The image of $\xi \otimes \eta$ in this quotient will again be written as $\xi \otimes \eta$. The vector space $H$ can be made into an pre-Hilbert space by setting
\[
	\langle \xi_1 \otimes \eta_1, \xi_2 \otimes \eta_2\rangle_H = \langle \xi_1, \eta_1 \rangle_{\mc{H}_1}\langle \xi_2, \eta_2 \rangle_{\mc{H}_2}
\]
and extending by linearity. This induces a norm on $H$, but in general $H$ is \emph{not} a Hilbert space, since it will not be complete with respect to this norm. However, by the results of the previous section we can complete $H$ to obtain a Hilbert space $\mc{H} = \mc{H}_1 \otimes \mc{H}_2$. Again, finite tensor products are defined analogously. We will discuss infinite tensor products in the next chapter.

Suppose again that $L_i : \mc{H}_i \to \mc{K}_i$ are bounded linear maps.\index{tensor product!of bounded linear maps} We can define a map $L_1 \otimes L_2 : H \to \mc{K}_1 \otimes \mc{K}_2$ by
\[
	(L_1 \otimes L_2)(\xi \otimes \eta) = L_1 \xi \otimes L_2 \eta.
\]
This is well-defined since the $L_i$ are linear, and hence this definition is compatible with the relations in~\eqref{eq:tensrel}. Since $L_1$ and $L_2$ are bounded, $L_1 \otimes L_2$ is bounded \emph{on the subspace} $H$. To define $L_1 \otimes L_2$ on all of $\mc{H}_1 \otimes \mc{H}_2$ we use that $H$ is dense in the tensor product, together with Proposition~\ref{prop:lincont}. Since $H$ is dense in $\mc{H} := \mc{H}_1 \otimes \mc{H}_2$, every $\xi \in \mc{H}$ is the norm limit of a sequence $\xi_n$ of elements in $H$. We then define
\[
(L_1 \otimes L_2)\xi := \lim_{n \to \infty} (L_1 \otimes L_2)\xi_n.
\]
Since $L_1 \otimes L_2$ is bounded, the image of a Cauchy sequence is again Cauchy. It follows that $(L_1 \otimes L_2)\xi_n$ converges in the Hilbert space $\mc{K}_1 \otimes \mc{K}_2$. This defines $L_1 \otimes L_2$ on all of $\mc{H}$. One can check that the definition of $L_1 \otimes L_2$ does not depend on the choice of sequence, and that it is indeed a linear map. This procedure is sometimes called \emph{extension by continuity}, where we extend a bounded linear map defined on a dense subset of a complete linear space, and taking values in a complete linear space, to the whole space. 

For adjoints of tensor products of linear maps we have $(L_1 \otimes L_2)^* = L_1^* \otimes L_2^*$.

\begin{exercise}
	Verify the details of the \emph{extension by continuity} construction. Also show that indeed the adjoint indeed satisfies $(L_1 \otimes L_2)^* = L_1^* \otimes L_2^*$.
\end{exercise}

\section{Banach and $C^*$-algebras}
The set $\alg{B}(\mc{H})$ of bounded operators on a Hilbert space has more structure than we have discussed above. For example, it is an \idx{algebra}. Recall that an algebra $\alg{A}$ is a vector space\footnote{We will deal almost exclusively with vector spaces over $\mathbb{C}$.} on which a multiplication operation is defined. This multiplication should be compatible with addition, in the sense that $A(B+C) = AB + AC$ for all $A,B,C \in \alg{A}$, and similarly for $(B+C)A$. Moreover, it should be compatible with multiplication by scalars in the obvious sense. We do not require our algebras to have a unit, although in most of our applications this will be the case. In this section we will try to abstract the properties of $\alg{B}(\mc{H})$ to obtain the notion of a $C^*$-algebra. 

A \emph{$*$-algebra}\index{_staralgebra@$*$-algebra} is an algebra on which an (anti-linear) involution $*$\index{involution} is defined. That is, there is a map $*: \alg{A} \to \alg{A}$ with the following properties:
\begin{enumerate}
	\item $(A^*)^* = A$ for all $A \in \alg{A}$,
	\item $(AB)^* = B^* A^*$ for $A,B \in \alg{A}$,
	\item $(\lambda A + B)^*  = \overline{\lambda} A^* + B^*$ for $A,B \in \alg{A}$ and $\lambda \in \mathbb{C}$.
\end{enumerate}
The first property says that the $*$-operation is involutive, and the last property describes anti-linearity. A familiar example of a $*$-algebra are the $n \times n$ matrices where the $*$-operation is given by taking the adjoint of the matrix. 

So far the discussion has been purely algebraic. In our prototypical example of $\alg{B}(\mc{H})$ topology enters in a natural way. As we will see later, it is for example possible to define a norm on $\alg{B}(\mc{H})$. This makes it an example of a Banach algebra:
\begin{definition}
	A \idx{Banach algebra} $\alg{A}$ is an algebra which is complete with respect to a norm $\|\cdot\|$. Moreover, the norm should satisfy $\|A B\| \leq \|A\| \|B\|$. A \emph{Banach $*$-algebra} is a Banach algebra which is also a $*$-algebra and for which $\|A\| = \|A^*\|$ for all $A \in \alg{A}$.
\end{definition}
As before, completeness means that any Cauchy sequence $A_n$ in $\alg{A}$ converges to some $A \in \alg{A}$. The inequality on the norm of a product $AB$ guarantees that the multiplication in the algebra is continuous with respect to the metric topology, and since the $*$-operation is isometric, it is continuous as well. Finally, if $\alg{A}$ has a unit $I$, then $\| I \| \geq 1$. A unit is an element $I$ such that $IA = AI = A$ for all $A \in \alg{A}$. A unit is necessarily unique (if it exists).

\begin{exercise}
Prove these claims.
\end{exercise}

A Banach algebra already has a rich structure. However, it turns out that if we demand one extra condition on the norm, we can do much more. This extra condition is what defines a $C^*$-algebra. Such algebras will be used to model the observables in quantum mechanical systems.
\begin{definition}
A \emph{$C^*$-algebra}\index{C*algebra@$C^*$-algebra} $\alg{A}$ is a Banach $*$-algebra such that the norm satisfies the $C^*$-property: $\|A^*A\| = \|A\|^2$ for all $A \in \alg{A}$. 
\end{definition}
Note that the condition that $\|A^*\| = \|A\|$ for a Banach algebra follows from the $C^*$-property and submultiplicativity ($\|AB\| \leq \|A\|\|B\|$) of the norm. The converse however is not true, so not any Banach $*$-algebra is a $C^*$-algebra. If a $C^*$-algebra has a unit $I$, it follows automatically that $\| I \| = 1$.\index{C*algebra@$C^*$-algebra!unital}

Most (if not all) of the $C^*$-algebras in these notes are unital. In cases where it makes a difference, we will only give proofs in the unital case. Nevertheless, let us briefly mention two useful techniques in dealing with non-unital algebras. First of all, let $\alg{A}$ be a non-unital $C^*$-algebra. Then $\alg{A}$ can always be embedded into a unital algebra $\widetilde{\alg{A}}$, defined as
\[
	\widetilde{\alg{A}} := \{ (A, \lambda) : A \in \alg{A}, \quad \lambda \in \mathbb{C} \}.
\]
Addition and the $*$-operation act componentwise, multiplication is defined by $(A, \lambda)(B, \mu) = (\mu A + \lambda B + AB, \lambda\mu)$. A norm is defined by
\[
\|(A, \lambda)\| = \sup_{B \in \alg{A}, \|B\| = 1} \| A + \lambda B\|. 
\]
It can be shown that $\widetilde{\alg{A}}$ is a unital $C^*$-algebra with this norm. The unit is the element $(0,1)$. Note that $\alg{A}$ clearly is a subalgebra of $\widetilde{\alg{A}}$. This construction is called \idx{adjoining a unit}. The other technique is that of an \idx{approximate identity}. This is a net $E_\lambda$ of positive elements (i.e., elements of the form $A^*A$ for some $A$, see also Section~\ref{sec:positive}), such that $\| E_\lambda \| \leq 1$ and $\lim_\lambda E_\lambda A = \lim_\lambda A E_\lambda = A$ for all $A \in \alg{A}$. An approximate identity always exists in a $C^*$-algebra.

\subsection{Example: the bounded operators on a Hilbert space}\index{C*algebra@$C^*$-algebra!B(H)@$\alg{B}(\mc{H})$}
Let $\mc{H}$ be a Hilbert space. An \idx{operator} $A$ on $\mc{H}$ is a linear map $A: \mc{H} \to \mc{H}$ (or sometimes only defined on a linear subspace of $\mc{H}$). It is called \emph{bounded}\index{operator!bounded} if there is some constant $C > 0$ such that $\| A \psi \| \leq C \| \psi \|$ for all $\psi \in \mc{H}$. If this is not the case, then the operator is called \emph{unbounded}. Note that the bounded operators on a Hilbert space are a special case of the linear maps we studied earlier (c.f. Proposition~\ref{prop:lincont}). We will show that $\alg{B}(\mc{H})$ is a $C^*$-algebra, where the $*$-operation is the adjoint of linear maps on Hilbert spaces defined earlier.

First, it is clear that $\alg{B}(\mc{H})$ is an algebra, since we can add (or multiply with a scalar) bounded linear maps to obtain new bounded maps. The product operation is the composition of linear maps. Moreover, the adjoint satisfies all the conditions of a $*$-operation by Exercise~\ref{ex:adjoint}. It remains to define a suitable norm and show that it has the right properties. The norm can be defined as follows:
\begin{equation}
	\| A \| := \sup_{\psi \in \mc{H}, \psi \neq 0} \frac{\|A \psi \|}{\|\psi\|} =  \sup_{\psi \in \mc{H}, \| \psi \| = 1} \| A \psi \|.
	\label{eq:opnorm}
\end{equation}
The last equality immediately follows from linearity of $A$. Note that we use the same notation for the norm on $\mathcal{H}$ and the norm on the algebra $\alg{B}(\mc{H})$. That equation~\eqref{eq:opnorm} is indeed a norm follows directly because $\| \cdot \|$ is a norm on the Hilbert space $\mc{H}$. Submultiplicativity follows because for any $\psi \in \mc{H}$, 
\[
\|AB\psi\| = \|A(B\psi)\| \leq \|A\| \|B\psi\| \leq \|A\| \|B\| \|\psi\|.
\]
To show that $\alg{B}(\mc{H})$ is complete with respect to this norm, let $A_n \in \alg{B}(\mc{H})$ be a Cauchy sequence of operators. Let $\psi \in \mc{H}$. Then $n \mapsto A_n \psi$ is a Cauchy sequence of vectors in the Hilbert space, hence this converges to a vector which we will denote by $A\psi$. That is, for every $\psi \in \mc{H}$ we define
\[
	A \psi := \lim_{n \to \infty} A_n \psi.
\]
Since each $A_n$ is linear, this defines a linear map $A$. The map $A$ is bounded and $\|A-A_n\| \to 0$ if $n \to \infty$. 
\begin{exercise}
Verify these claims.
\end{exercise}

It remains to show that the $C^*$-identity for the norm holds. Note that for $A \in \alg{B}(\mc{H})$ we have
\[
	\| A \psi \|^2 = \langle A \psi, A \psi \rangle = \langle \psi, A^*A \psi \rangle \leq \| A^*A \| \|\psi\|^2
\]
for all $\psi \in \mc{H}$. Hence $\|A\|^2 \leq \| A^*A \| \leq \|A\| \|A^*\|$ and $\|A\| \leq \|A^*\|$. Reversing the roles of $A$ and $A^*$ shows that $\|A\| = \|A^*\|$, from which the claim follows.

\begin{remark}
The completeness proof works for any Banach space. For the adjoint, however, one needs the Hilbert space structure to define it in the first place.
\end{remark}

\subsection{Example: commutative $C^*$-algebras}\index{C*algebra@$C^*$-algebra!commutative}
\label{subsec:commalg}
Let $X$ be a locally compact topological space. Write $C_0(X)$ for the space of all continuous functions $f: X \to \mathbb{C}$ that vanish at infinity. That is, $f \in C_0(X)$ if and only if for each $\varepsilon > 0$, there is a \emph{compact} set $K_\varepsilon \subset X$ such that $|f(x)| < \varepsilon$ for all $x \in X \setminus K_\varepsilon$. To get some feeling for this condition, consider the case $X = \mathbb{R}$. Then $f \in C_0(\mathbb{R})$ if and only if $f$ is continuous and $|f(x)| \to 0$ if $x \to \pm \infty$. Examples are
\[
f(x) = \frac{1}{1+x^2},\,\, \textrm{ and } f(x) = \exp(-x^2).
\]
Note that $C_0(X)$ is an algebra if we define multiplication and addition pointwise:
\[
	(f+g)(x) = f(x) + g(x), \quad (f\cdot g)(x) = f(x) g(x).
\]
An involution can be defined by setting $f^*(x) = \overline{f(x)}$. Hence $C_0(X)$ is a $*$-algebra. We can define a norm as follows:
\[
	\| f \| := \sup_{x \in X} |f(x)|.
\]
Because continuous functions on compact sets are bounded, and $f$ is small outside some compact set, it follows that $\| f \|$ is finite if $f \in C_0(X)$. Note that the $C^*$-property, $\| f^*f \| = \|f \|^2$ follows immediately from the definition. A standard result in topology says that if $\|f_n - f\| \to 0$ for some function $f: X \to \mathbb{C}$ and $f_n$ a sequence in $C_0(X)$, then $f$ is also continuous (and it follows that $f \in C_0(X)$). Hence $C_0(X)$ is a $C^*$-algebra.

Note that $C_0(X)$ is a commutative algebra. The unit would be the function $f(x) = 1$ for all $x$. But this is only in $C_0(X)$ if $X$ is compact, because otherwise it does not vanish at infinity. The commutative Gel'fand-Naimark theorem in $C^*$-algebras says that in fact any commutative $C^*$-algebra is of the form $C_0(X)$ for some locally compact space $X$, and $X$ is compact if and only if the $C^*$-algebra is unital.

\section{Spectrum and positive operators}\label{sec:positive}
Many problems in physics (and countless other fields) can be reduced to solving an eigenvalue problem. For example, the energy levels of a quantum system are given by the spectrum of the Hamiltonian. It is therefore no surprise that spectral theory is an important part of functional analysis and one of the most important tools in the toolbox of the mathematical physicist. 

Recall that $H$ if is a self-adjoint matrix acting on a finite dimensional Hilbert space, there is a complete set of eigenvalues. This set is called the \idx{spectrum} of $H$. In general, for any matrix $A$, we will call the set of eigenvalues of this matrix the spectrum of $A$. This is the notion that we want to generalise to $C^*$-algebras. We immediately run into problems: it might be that the operators are not given as linear maps on a vector space, so that the notion of an eigenvalue does not apply. Even if $A$ is an operator on an infinite-dimensional vector space, it may be the case that $A$ has no eigenvalues at all. This is undesirable, so that the notion of the spectrum has to be extended.

To this end, consider again a matrix $A \in M_n(\mathbb{C})$. If $\lambda$ is an eigenvalue of $A$, by definition there is a non-zero vector $\psi$ such that $(A-\lambda I)\psi = 0$. Hence the kernel of $A$ is non-trivial, and by the rank-nullity theorem it follows that $A-\lambda I$ is not surjective, and hence is not invertible. This motivates the following definition.
\begin{definition}
	Let $\alg{A}$ be a unital $C^*$-algebra and $A \in \alg{A}$. Then the \idx{spectrum}\index{_sigma(A)@$\sigma(A)$} of $A$ is defined as
	\[
		\sigma(A) := \{ \lambda \in \mathbb{C} :  A - \lambda I \,\,\textrm{is not invertible}\}.
	\]
	The complement of $\sigma$ is called the \idx{resolvent set}.
\end{definition}
Note that if $A$ acts on a Hilbert space and $\lambda$ is an eigenvalue of $A$, it follows that $\lambda \in \sigma(A)$, so this definition indeed extends the notion of an eigenvalue.

The rest of this section will be related to study of the spectrum of various classes of operators. Although these methods are very powerful, in most of the book they will only play a minor role. The reader therefore can safely skip the following subsections, which are of a more technical nature, and refer back when necessary.

\subsection{Functional calculus}\label{sec:funccalc}
It is often convenient to consider functions on the spectrum of an operator to obtain new operators. For simplicity we consider only the case of polynomials in detail. Suppose that $p(z) = \sum_{k=0}^n c_n z^k$ is a polynomial on $\mathbb{C}$ and let $A \in \alg{A}$ be an operator in a unital $C^*$-algebra. Then we can define a new operator
\[
	p(A) = \sum_{k=0}^n c_k A^k.
\]
What can we say about the spectrum of $p(A)$? Let $\lambda \in \mathbb{C}$. Then
\[
	p(A) - p(\lambda) I = \sum_{k=1}^n c_k (A^k - \lambda^k I).
\]
Recall that $x^k - y^k = (x-y) \sum_{m=0}^{k-1} x^m y^{k-m-1}$ if $k \geq 1$. Hence
\[
	p(A) - p(\lambda) I = (A-\lambda I) \left[ \sum_{k=1}^n \sum_{m=0}^{k-1} c_k A^k \lambda^{k-m-1} \right].
\]
Write $B$ for the operator in square brackets. Suppose that $p(\lambda) \notin \sigma(p(A))$. Then with the calculation above it follows that $B (p(A)-p(\lambda)I)^{-1}$ is the inverse of $A -\lambda I$. Hence $\lambda \notin \sigma(A)$. It follows that $p(\sigma(A)) \subset \sigma(p(A))$.

The reverse inclusion is also true. Let $\mu \in \sigma(p(A))$. By the fundamental theorem of algebra there are $\lambda_i$ such that
\[
	p(z) - \mu = (z-\lambda_1)\cdots(z-\lambda_n).
\]
Since $\mu$ is in the spectrum of $p(A)$, there must be a $\lambda_i$ such that $A - \lambda_i I$ is not invertible, or in other words, $\lambda_i \in \sigma(A)$. Setting $z = \lambda_i$ it follows that $p(\lambda_i) = \mu$, and it follows that $\sigma(p(A)) \subset p(\sigma(A))$.

The procedure above is an example of what is called \idx{functional calculus}. It can be extended to larger classes of functions (not just polynomials) on the spectrum. To study this it is useful to first better understand what the spectrum of bounded operators is like. To this end, define the \idx{spectral radius} as
\begin{equation}
	r(A) := \sup_{\lambda \in \sigma(A)} |\lambda|.
\end{equation}
For bounded operators the spectral radius is always finite:

\begin{proposition}\label{prop:specradius}
	Let $\alg{A}$ be a $C^*$-algebra and $A \in \alg{A}$. Then $r(A) \leq \|A\|$ and $\sigma(A)$ is compact. If $A = A^*$, this subset is contained in the real line.
\end{proposition}
\begin{proof}
Suppose that $|\lambda| > \|A\|$. Then the series $\sum_{n=0}^\infty \lambda^{-1} A^n$ converges to some element in $\alg{A}$, since $\| \lambda^{-1} A \| < 1$. Call this element $B$. Then clearly
\[
	\lambda^{-1} AB = \sum_{n=1}^\infty \lambda^{-1} A^n = B-I = \lambda^{-1} BA.
\]
Rearranging terms gives
\[
	\lambda^{-1} (\lambda I - A) B = \lambda^{-1} B (\lambda I - A) = I.
\]
Hence $(A-\lambda I)$ is invertible and $\lambda \notin \sigma(A)$, and it follows that $r(A) \leq \|A\|$.

From the first part of the proof it follows that the spectrum is bounded, so to show compactness it is enough to show that it is closed. Since by definition it is the complement of the resolvent set, it is enough to show that this set is open. Suppose that $\lambda \notin \sigma(A)$. Then $(A-\lambda I)$ is (by definition) invertible, hence $\varepsilon := \|(A-\lambda I)^{-1} \|^{-1} > 0$ (and finite). Suppose that $|\mu| < \varepsilon$. Then by a similar argument as above it follows that $I - \mu I (A-\lambda I)^{-1}$ is invertible. Multiplying with $(A-\lambda I)$ shows that $A - (\lambda + \mu) I$ is invertible. Hence $\lambda + \mu \notin \sigma(A)$ and it follows that the resolvent set of $A$ is open.

Finally, suppose that $A = A^*$ and that $\lambda = \lambda_1 + i \lambda_2 \in \sigma(A)$, with $\lambda_i \in \mathbb{R}$. Define $A_n = A - \lambda_1 I + i n \lambda_2 I$. From the spectral calculus above it follows that
\[
	\sigma(A_n) = \{ \lambda - \lambda_1 + i \lambda_2 : \lambda \in \sigma(A) \}.
\]
In particular it follows that $i(n+1) \lambda_2 \in \sigma(A_n)$. Since the spectral radius is bounded by $\|A_n\|$, with the $C^*$-property of the norm we have the estimate
\[
		|(n+1) \lambda_2|^2 \leq \| A_n^* A_n \| = \| (A-\lambda_1)^2 + n^2 \lambda_2^2 \| \leq \|A-\lambda_1\|^2 + n^2 \lambda_2^2.
\]
Since this must hold for any positive integer $n$, it follows that $\lambda_2 = 0$ and hence $\sigma(A) \subset \mathbb{R}$.
\end{proof}
Again this is a generalization of what is true for matrices in $M_n(\mathbb{C})$. For example, it is a basic fact in linear algebra that a self-adjoint (that is, hermitian) matrix $A$ has real eigenvalues. If we diagonalise $A$ it is also not that hard to see that the operator norm $\|A\|$, defined in equation~\eqref{eq:opnorm}, is the absolute value of the largest eigenvalue. That is, it is equal to the spectral radius.

\begin{remark}\label{rem:spectralradius}
One can prove a stronger result, namely that 
\begin{equation}
	\label{eq:specradformula}
	r(A) = \lim_{n \to \infty} \|A^n\|^{1/n}.
\end{equation}
In case $A$ is normal (i.e., $AA^* = A^*A$) this implies that $r(A) = \|A\|$. To see this, note that $\|A^2\|^2 = \| (A^*)^2 A^2\| = \|(A^*A)^2\| = \|A^*A\|^2 = (\|A\|^2)^2$, where we made use of the $C^*$-property of the norm. By induction it follows that $\|A^{2^n}\|^2 = \|A\|^{2^{n+1}}$. With the above limit expression it follows that $r(A) = \|A\|$. 
\end{remark}

\begin{exercise}\label{ex:specab}
	Let $A,B \in \alg{A}$ be operators in some $C^*$-algebra $\alg{A}$. Show that $\sigma(AB) \subset \sigma(BA) \cup \{0\}$. \emph{Hint:} Consider $\lambda \notin \sigma(AB)$ with $\lambda \neq 0$. It is sufficient to show that $(BA-\lambda I)$ is invertible. To this end, consider the operator $\lambda^{-1}(B(AB-\lambda I)^{-1}A-I)$.
\end{exercise}

The above results give us the tools to consider more general functions than just polynomials. For example, in the holomorphic functional calculus we consider a function $f$ that is analytic with radius of convergence at least $\|A\|$. Then we have a power expansion $f(z) = \sum_{n=0}^\infty c_n z^n$ and we can define
\[
	f(A) := \sum_{n=0}^\infty c_n A^n.
\]
This expression is well-defined: $\| f(A) \| \leq \sum_{n=0}^\infty |c_n| \|A\|^n$, and since $\|A\|$ is assumed to be smaller than the radius of convergence of $f$, the right hand side is finite and the series for $f(A)$ converges. Since $\alg{A}$ is a Banach space it follows that $f(A) \in \alg{A}$. One can again prove that $f(\sigma(A)) = \sigma(f(A))$.

Even more general one can consider just continuous functions on the spectrum that do not necessarily have a power series expansion. In the case that $A$ is self-adjoint, there is an elegant away to state all the properties of the continuous functional calculus.\index{functional calculus} Since the spectrum of $A \in \alg{A}$ is compact, the set of continuous functions on $\sigma(A)$, $C(\sigma(A))$, is a commutative $C^*$-algebra by the example on page~\pageref{subsec:commalg}. The functional calculus can then be understood as a $*$-homomorphism between these two algebras. We will discuss homomorphisms in more detail later, but the point is that it is a linear map that preserves the algebraic structure of the $C^*$-algebra.
\begin{theorem}
	\label{thm:funccalc}
	Let $A = A^*$ be an operator in a $C^*$-algebra $\alg{A}$. Then there is an isometric $*$-homomorphism $\Phi: C(\sigma(A)) \to \alg{A}$ such that $\Phi(f) = f(A)$. In particular, $f(A)^* = \overline{f}(A)$ and $f(A)g(A) = (fg)(A)$.
\end{theorem}
\begin{proof}[sketch]
	One way to prove this result is to first verify the identities for polynomials, as outlined above. Since $A$ is self-adjoint, its spectrum is contained in the real line. In fact, it is contained in a compact interval since the spectral radius is bounded by the norm of $\alg{A}$. But the Weierstrass approximation theorem says that each continuous function on a compact interval can be approximated uniformly by a sequence of polynomials. Hence we can define $\Phi$ on the dense subset of polynomials, and extend the results to all of $C(\sigma(A))$ by continuity if we can show that $\Phi$ is bounded on the polynomials. This follows from Remark~\ref{rem:spectralradius}: let $p$ be a polynomial, and note that  $\| \Phi(p) \| = r(\Phi(p)) = \sup_{\lambda \in \sigma(A)} |p(\lambda)|$, where we used that for polynomials $p(\sigma(A)) = \sigma(p(A))$. The right hand side is the sup-norm on $C(\sigma(A))$, so that $\Phi$ is indeed isometric on the polynomials.
\end{proof}

\subsection{Positive operators}
In physics or quantum mechanics in particular, it is often the case that observables of interest can only take positive values. For example, it is common to normalise the Hamiltonian of a system such that the ground state has zero energy, and all other states have positive energy. In general, consider any operator $A \in \alg{B}(\mathcal{H})$ acting on some Hilbert space. Then for each $\psi \in \mathcal{H}$, we have
\[
	\langle \psi, A^*A \psi \rangle = \langle A \psi, A \psi \rangle = \|A \psi\|^2 \geq 0.
\]
We say that $A^*A$ is a \idx{positive operator}. This can be straightforwardly generalised to $C^*$-algebras: we say that $A \in \alg{A}$ is a positive operator if there is some $B \in \alg{A}$ such that $A = B^*B$. It turns out that also in this case, the adjective ``positive'' is appropriate when we look at the spectrum of such operators. Before we discuss this result, we need two technical lemmas.

\begin{lemma}\label{lem:posnorm}
	If $A,B \in \alg{A}$ are self-adjoint and both their spectra are contained in $[0,\infty)$, then $\sigma(A+B) \subset [0,\infty)$.
\end{lemma}
\begin{proof}
	Let $S$ be a self-adjoint element of $\alg{A}$ and let $\lambda \geq \|A\|$. Then we claim that $\|S- \lambda I\| \leq \lambda$ if and only if $\sigma(S) \subset [0, \infty)$. By Proposition~\ref{prop:specradius}, $\sigma(S) \subset [-\lambda, \lambda]$. Since $S - \lambda I$ is self-adjoint, $\|S-\lambda I\| = r(S-\lambda I)$ by the remark following that proposition. Hence
\[
	\| S- \lambda I\| = \sup_{\mu \in \sigma(S)} |\mu - \lambda| = \sup_{\mu \in \sigma(S)} (\lambda-\mu),
\]
where the first equality follows from the spectral calculus and the second because $\sigma(S) \subset [-\lambda, \lambda]$. The term on the right is smaller than or equal to $\lambda$ if and only if $\sigma(S) \subset [0,\infty)$.

Applying this to $A$ and $B$, it follows that $\|A- \|A\| I\| \leq \|A\|$ and similar for $B$. But then
\[
	\|A+B-(\|A\|+\|B\|) I\| \leq \|A \| + \|B\|.
\]
Because $\|A + B\| \leq \|A\| + \|B\|$, the result follows from the first paragraph.
\end{proof}

\begin{lemma}\label{lem:negspec}
	Let $A \in \alg{A}$ be an element of some $C^*$-algebra such that for the spectrum of $-A^*A$ we have $\sigma(-A^*A) \subset [0, \infty)$. It follows that $A = 0$.
\end{lemma}
\begin{proof}
It is easy to see that $A = H + i K$ for some self-adjoint $H$ and $K$. Then we can calculate
\[
	A^*A = 2H^2 + 2 K^2 - A A^*.
\]
By Exercise~\ref{ex:specab}, $\sigma(-AA^*) \subset \sigma(-A^*A) \cup \{0\} \subset [0,\infty)$. Because $H$ and $K$ are self-adjoint, their spectrum is real and the spectrum of $H^2$ and $K^2$ is also contained in the positive real line. By the previous Lemma, it follows that $\sigma(A^*A) \subset [0,\infty)$. Since $\sigma(-A^*A) = \{ -\lambda : \lambda \in \sigma(A^*A) \}$, it follows from the premises that $\sigma(A^*A) = \{0\}$. By the spectral radius formula and the $C^*$-property of the norm, $0=r(A^*A) = \|A^*A\| = \|A\|^2$, so that $A = 0$.
\end{proof}

We are now in a position to prove that an operator is positive if and only if it is self-adjoint and its spectrum is contained in the positive real line. This characterisation is sometimes more convenient than the more algebraic condition $A = B^*B$. In addition, as we will see, the spectrum of self-adjoint operators has a clear physical meaning in measurements in quantum mechanics. 
\begin{theorem}
	\label{thm:positiveop}
	Let $\alg{A}$ be a $C^*$-algebra and $A \in \alg{A}$. Then the following are equivalent:
	\begin{enumerate}
		\item\label{it:posbb} $A$ is positive, i.e. there is $B \in \alg{A}$ with $A = B^*B$;
		\item\label{it:posspec} $\sigma(A) \subset [0, \infty)$ and $A$ is self-adjoint;
		\item\label{it:possq} $A = H^2$ for some $H \in \alg{A}$ with $H = H^*$ and $\sigma(H) \subset [0, \infty)$.
	\end{enumerate}
\end{theorem}
\begin{proof}
	Assume that $A$ is positive, then $A$ is clearly self-adjoint. We want to show that its spectrum is contained in the positive real line. To this end, define the functions $f_{\pm}(\lambda) = \pm \frac{1}{2}(\lambda \pm |\lambda|)$. Then we can define
\[
	A_+ = f_+(A), \quad A_{-} = f_{-}(A).
\]
By the functional calculus, Theorem~\ref{thm:funccalc}, it follows that $A = A_{+} - A_{-}$ and that $\sigma(A_{\pm}) \subset [0, \infty)$. In addition, $A_+A_{-} = A_{-}A_{+} = f_+(A) f_{-}(A) = (f_+ f_-)(A) = 0$, since $f_+ f_- = 0$. Note that by construction, $\sigma(A_+) \subset [0, \infty)$, so it is enough to show that $A_{-} = 0$. Define $C = BA_{-}$. Then a simple calculation shows that $C^*C = -(A_{-})^3$. By construction the spectrum of $A_{-}$ is contained in the positive real line, hence $\sigma(-A_{-}^3) \subset -[0,\infty)$. With the Lemma above it follows that $C = 0$, or $A_{-}^3 = 0$. Let $f(t) = t^{1/3}$. Then $f(A_{-}^3) = A_{-}$, which can be seen from Theorem~\ref{thm:funccalc}. Hence $A = A_{+}$ and we have shown that~\ref{it:posbb} $\Rightarrow$~\ref{it:posspec}.

If the spectrum is positive, we can define $H := \sqrt{A}$ as discussed above. Since the square root is a real function, $\sqrt{A^*} = (\sqrt{A})^*$ and hence $H$ is self-adjoint. Moreover, $\sigma(\sqrt{A}) = \sqrt{\sigma(A)}$ proving that \ref{it:posspec} implies~\ref{it:possq}. The implication~\ref{it:possq} $\Rightarrow$ \ref{it:posbb} is trivial.
\end{proof}

\begin{definition}
	Operators in a $C^*$-algebra $\alg{A}$ satisfying any of the equivalent conditions in Theorem~\ref{thm:positiveop} are called \emph{positive}\index{positive operator}\index{operator!positive}. The set of all positive operators in $\alg{A}$ is denoted $\alg{A}_+$\index{_aplus@$\alg{A}_+$}.
\end{definition}

Positive operators have some nice properties. For example, consider a positive operator $A$. Because the spectrum is contained in the positive real line, $f(t)= \sqrt{t}$ is continuous on the spectrum. Hence we can apply the spectral calculus and set $\sqrt{A} := f(A)$. Since $f(A) g(A) = (fg)(A)$ for $f,g \in C(\sigma(A))$, it follows that $(\sqrt{A})^2 = A$ as expected. In addition, $\sigma(\sqrt{A}) \subset [0, \infty)$. They can also be seen as the building blocks of a $C^*$-algebra, in the sense that each operator can be written as a linear combination of at most four positive operators. This can be seen by first writing $A = H + i K$, with both $H$ and $K$ self-adjoint, and then writing $H = H_{-} - H_{+}$ (and similarly for $K$) as in the proof of Theorem~\ref{thm:positiveop}.

	It is easy to see that $\lambda A \in \alg{A}_+$ if $\lambda \geq 0$ and $A \in \alg{A}_+$. From Lemma~\ref{lem:posnorm} and the Theorem above it follows that $A+B \in \alg{A}_+$ if both $A$ and $B$ are in $\alg{A}_+$. What is not immediately clear, but nonetheless true, is that $AB$ is positive if $A$ and $B$ commute and are both positive. One can also show (using the proof of Lemma~\ref{lem:posnorm}) that $\alg{A}_+$ is norm-closed in $\alg{A}$. From Lemma~\ref{lem:negspec} in addition it follows that $(-\alg{A}_+) \cap \alg{A}_+ = \{ 0 \}$. We say that $\alg{A}_+$ is a \emph{positive cone} in $\alg{A}$.

\subsection{Projection valued measures}
Suppose that  $A \in M_n(\mathbb{C})$. Then the spectrum $\sigma(A)$ is the set of eigenvalues of $A$. This is a special feature of the spectrum of linear operators of finite dimensional spaces: for linear operators on infinite spaces it is still true that eigenvalues are in the spectrum, but the converse is \emph{not} always true! Suppose that in addition to acting on a finite dimensional Hilbert space, $A$ is also self-adjoint. Then we know from linear algebra that $A$ can be diagonalised. This makes it easy to do the functional calculus for $A$. 
\begin{exercise}\label{ex:pvm}
	Let $A$ be as in the previous paragraph and let $f: \mathbb{R} \to \mathbb{C}$ be any function. Show that $f(A) = \sum_{\lambda \in \sigma(A)} f(\lambda) P_\lambda$, where $P_\lambda$ is the projection on the eigenspace with eigenvalue $\lambda$.
\end{exercise}
Since the spectrum of an operator on a finite dimensional vector space is discrete, it is not necessary to demand any continuity properties of $f$.

This cannot directly be generalised to arbitrary $C^*$-algebras. One of the reasons is that a $C^*$-algebra need not contain any (non-trivial) projections. However, this changes if we consider operators acting on a Hilbert space $\mathcal{H}$. In that case it is natural to try to find a natural decomposition of self-adjoint operators in terms of projections on closed subspaces of $\mathcal{H}$. This allows us to generalise the example above. To this end it is instructive to look at the previous example in a different light. For now assume again that $A$ is a self-adjoint matrix. We can consider the spectrum $\sigma(A)$ as a \emph{measurable space}, where each subset is measurable. To each $\Lambda \subset \sigma(A)$ we associate the operator
\[
	\mu(\Lambda) := \sum_{\lambda \in \Lambda} P_\lambda.
\]
Note that $\mu(\sigma(A)) = I$ and $\mu(\Lambda_1 \cup \Lambda_2) = \mu(\Lambda_1) + \mu(\Lambda_2)$ if $\Lambda_1 \cap \Lambda_2 = \emptyset$. In addition, each $\mu(\Lambda)$ is easily seen to be a projection. Hence we can interpret $\mu$ as a measure on $\sigma(A)$ taking values into the projections: it is an example of a \idx{projection valued measure}.

This is the statement that we want to generalise to (self-adjoint) bounded operators on a Hilbert space. A complete, rigorous treatment requires a good understanding of measure theory. The necessary background can be found in many functional analysis or measure theory textbooks. Here we will try to give the main ideas behind the construction and refer the reader to other textbooks for the technical details..

Suppose that $A$ is a self-adjoint bounded operator acting on a Hilbert space $\mathcal{H}$. Let $\psi \in \mathcal{H}$ be a unit vector. Then using the functional calculus we can define a map $\omega_\psi : C(\sigma(A)) \to \mathbb{C}$ via $f \mapsto \langle \psi, f(A) \psi \rangle$. Note that this is a linear map. Suppose that $f \geq 0$. Then $f(A)$ is a positive operator by Theorem~\ref{thm:positiveop}(\ref{it:posspec}) and the functional calculus. Hence there is $B \in \alg{B}(\mathcal{H})$ such that $f(A) = B^*B$. But then
\[
	\omega_\psi(f) = \langle \psi, f(A) \psi \rangle = \langle B \psi, B \psi \rangle \geq 0.
\]
We say that $\omega_\psi$ is a positive linear map of $C(\sigma(A))$. We will study positive linear maps in more detail later on, but for now we will only need the \emph{Riesz-Markov theorem}. This theorem says that if $\omega_\psi$ is such a positive linear functional of the (compactly supported) continuous functions on a locally compact Hausdorff space, there is a Borel measure\footnote{If $X$ is a topological space, the \emph{Borel measurable sets} $\mathcal{B}(X)$ is the smallest collection of subsets of $X$ such that $O \in \mathcal{B}(X)$ for all open sets $O$ and $\mathcal{B}(X)$ is closed under the complement operation and under taking countable unions. The elements of $\mathcal{B}(X)$ are called \emph{Borel sets}. A \emph{Borel measure} assigns a positive real number to each of the Borel sets in a way compatible with the structure of the Borel sets. For example, $\mu(X_1 \cup X_2) = \mu(X_1) + \mu(X_2)$ for two disjoint Borel sets $X_1$ and $X_2$. Intuitively speaking, $\mu(X)$ tells us how ``big'' the set $X$ is. Once one has a measure, it is possible to defined integration with respect to that measure.} $\mu_\psi$ such that 
\[
	\mu_\psi(f) := \int_{\sigma(A)} f(\lambda) d\mu_\psi(\lambda) = \omega_\psi(f).
\]
Hence for every (non-zero) vector in the Hilbert space, we obtain a measure on the spectrum of $A$. This measure depends on $A$ and on the choice of vector $\psi$.

In obtaining the measures $\mu_\psi$ we used the spectral calculus for \emph{continuous} functions. The idea is now to use standard measure theory to extend the spectral calculus to bounded Borel measurable functions, that is, bounded functions $f$ such that$f^{-1}(X)$ is a Borel set for every Borel set $X$. Then we can extend the integral $\int_{\sigma(A)} f(\lambda) d\mu_\psi(\lambda)$ to such functions $f$. Note that by construction, $\mu_\psi(1) = \|\psi\|^2$, where $1$ is the constant function equal to one. It follows that $\mu_\psi(f) < \infty$ if $f$ is bounded.

This allows us to define $f(A)$ for any bounded Borel measurable function $f$ on the spectrum $\sigma(A)$. Recall that with the polarisation identity, compare with equation~\eqref{eq:polarisation}, the inner product between two vectors on a Hilbert can be written as the sum of four inner products of the form $\langle \psi, \psi \rangle$. With this in mind, for $\psi, \eta \in \mathcal{H}$ and $f$ a bounded Borel function on the spectrum, we define
\[
	\langle \psi, f(A) \eta \rangle := \frac{1}{4} \sum_{k=0}^{3} i^k \mu_{\psi + i^k \eta}(f),
\]
where $\mu_{\psi + i^k \eta}(f)$ is as defined above. But this gives us the matrix elements of an operator $f(A)$, hence by the Riesz representation theorem, this defines an operator on $\alg{B}(\mathcal{H})$, which we suggestively denote by $f(A)$. One can show that the Borel functional calculus obeys similar properties as the continuous functional calculus. For example, $\overline{f}(A) = f(A)^*$ and $f(A)g(A) = (fg)(A)$.

The Borel functional calculus gives us the tools to define \emph{spectral projections}\index{spectral projection}. These are generalisations of the projections $P_\lambda$ on the eigenspaces of finite dimensional operators that we discussed before. To define them, choose a Borel subset $\Lambda \subset \sigma(A)$. Then the indicator function $\chi_\Lambda$ is Borel measurable. This means we can define $P_\Lambda := \chi_\Lambda(A)$. This is a projection, since the indicator function is real valued and squares to itself. We then define \idx{projection valued measure} by $\mu(\Lambda) = P_\Lambda$ for each Borel measurable subset $\Lambda$ of the spectrum. The following Proposition shows that its properties are very similar to that of real-valued measures, hence the name is a sensible one.

\begin{proposition}
Let $\mu$ be the projection valued measure associated to a self-adjoint operator $A$ acting on a Hilbert space. Then the following properties hold for all Borel subsets $\Lambda_i$:
\begin{enumerate}
		\item $\mu_\Lambda$ is an orthogonal projection;
		\item $\mu(\emptyset) = 0$ and there is some $\lambda \in \mathbb{R}$ such that $\mu( (-a,a) ) = I$;
		\item If $\Lambda = \bigcup_{i=1}^\infty \Lambda_i$ with $\Lambda_i \cap \Lambda_j = \emptyset$ for $i \neq j$, then $\operatorname{s-lim}_{n \to \infty} \left( \sum_{i=1}^n \mu(\Lambda_i) \right) = \mu(\Lambda)$.
\end{enumerate}
\end{proposition}
The second statement in the second property follows because the spectrum of $A$ is compact. Hence, if we make the interval big enough, the spectral projection projects onto the whole Hilbert space. In the third property, $\operatorname{s-lim}$ denotes the limit in the \idx{strong operator topology}. It means that for each $\psi \in \mathcal{H}$, $\sum_{i=1}^n \mu(\Lambda_i) \psi$ converges to $\mu(\Lambda) \psi$ as $n$ goes to infinity.\footnote{One might hope that the operators actually converge in norm, but in general this is too much to ask for.} This property corresponds to what is called $\sigma$-additivity in measure theory.

With a measure we can talk about integration: the goal is to define an (operator valued!) integral with respect to this measure. This is not much different from how one defines scalar valued integrals in measure theory, for example the Lebesgue integral. We first consider a \emph{step function} $f$: choose $n$ Borel subsets $\Lambda_i$ of the real line with $\Lambda_i \cap \Lambda_j = \emptyset$ if $i \neq j$ and such that $\sigma(A) \subset \cup_{i=1}^n \Lambda_i$. Then $f(\lambda) = \sum_{i=1}^n c_i \chi_{\Lambda_i}(\lambda)$ is a step function. For such step functions we already know what the integral should be:
\[
	\int_{\sigma(A)} f(\lambda) d\mu(\lambda) := \sum_{i=1}^n c_i \mu(\Lambda_i).
\]
Note that this defines an operator in $\alg{B}(\mathcal{H})$. The general case is then obtained by approximating arbitrary bounded measurable functions $f$ by step functions, and verifying that the corresponding integrals converge to an operator. To summarise, we have outlined how to define the integral $\int_{\sigma(A)} f(\lambda) d\mu(\lambda)$. It should be no surprise that this is related to the spectral measures associated to vectors $\psi \in \mathcal{H}$: if we compute $\langle \psi, \int_{\sigma(A)} f(\lambda) d\mu(\lambda) \psi \rangle$, we recover the integrals with respect to the spectral measures defined above. 

The results of this section can be summarised in the following variant of the spectral theorem:
\begin{theorem}[\emph{Spectral theorem}]\index{spectral theorem!projection valued measure form}\label{thm:spectral}
	Let $A \in \alg{B}(\mathcal{H})$ be a self-adjoint operator acting on some Hilbert space $\mathcal{H}$. Then there is a projection valued measure $\mu$ such that $\mu(\Omega) = \chi_{\Omega}(A)$ such that $A = \int \lambda d\mu(\lambda)$. Moreover, for every bounded Borel measurable function $f$ of the spectrum, there is a unique operator $\int_{\sigma(A)} f(\lambda) d\mu(\lambda)$ such that 
	\[
	\langle \psi, \int_{\sigma(A)} f(\lambda) d\mu(\lambda) \psi \rangle = \mu_\psi(f) = \int_{\sigma(A)} f(\lambda) d\mu_\psi(\lambda).
	\]
	for all $\psi \in \mathcal{H}$.
\end{theorem}
The measure is in fact unique, and giving such a measure defines a self-adjoint operator $A$ by the integral formula. Hence there is a one-one correspondence between the two.

\subsection{Unbounded operators}\index{operator!unbounded|(}\label{sec:unbounded}
So far we have only discussed \emph{bounded} operators. In the applications we are interested in in this book, we will always work with bounded observables. Nevertheless, unbounded operators still occur in a natural way. In particular, we are interested in systems consisting of infinitely many degrees of freedom. In such systems the energy usually is only bounded from below. Hence the Hamiltonian, the generator of the dynamics, is unbounded. Although one can avoid talking about the unbounded Hamiltonian by considering the time evolution of the bounded operators, it is sometimes convenient to have a Hamiltonian at our disposal.

\begin{definition}
	An \emph{unbounded operator} $(H, D(H))$ on a Hilbert space $\mathcal{H}$ is a linear map $H$ from the linear space $D(H) \subset \mathcal{H}$ to $\mathcal{H}$ that is not bounded. In other words, $\sup_{0 \neq \psi \in D(H)} \| H \psi \| / \|\psi\| = \infty$.
\end{definition}

The choice of the domain is important and an essential part of the definition of an unbounded operator. In applications there often is a natural choice of a dense domain, but even then one has to be careful. For example, it need not be true that if $\psi \in D(H)$ for some unbounded operator $H$, that $H\psi \in D(H)$, so that even innocent looking expressions such as $H^2$ can be problematic.

\begin{example}
	Let $\mathcal{H} = L^2(\mathbb{R})$. Consider the position operator $Q$, defined by $Q\psi(x) = x \psi(x)$ for $\psi \in D(Q)$. Note that this does not define a bounded operator on $L^2(\mathbb{R})$: it is possible to find $\psi \in L^2(\mathbb{R})$ with $x \psi(x) \notin L^2(\mathbb{R})$. Instead one can take as domain $D(Q) = C_0^\infty(\mathbb{R})$, which is dense in $L^2(\mathbb{R})$. It should be noted however that there are bigger choices of domain possible. In fact, there is a natural way in which the domain can be enlarged so that $(Q, D(Q))$ is a self-adjoint operator (see below for the definition).
\end{example}

\begin{exercise}
	Give an example of a $\psi \in L^2(\mathbb{R})$ that cannot be in the domain of $Q$.
\end{exercise}

Even though superficially the theories for bounded and unbounded operators are very similar, there are also fundamental differences, for example related to the domain issues mentioned above. Recall that on page~\pageref{page:adjoint} we defined the adjoint of a bounded linear map $A$ in terms of the inner product on $\mathcal{H}$. We cannot just replace the bounded linear map in equation~\eqref{eq:adjoint} with $H$, since $H \psi$ may not be defined. Instead, we define the adjoint of an unbounded operator as follows.

\begin{definition}
	Let $(H, D(H))$ be an unbounded operator with dense domain. The \idx{adjoint} $(H, D(H^*))$ is the unbounded operator with as domain the set of $\psi \in \mathcal{H}$ for which there is $\eta \in \mathcal{H}$ such that
\[
	\langle H \varphi, \psi \rangle = \langle \varphi, \eta\rangle \quad\textrm{ for all }\quad \varphi \in D(H).
\]
We then set $H^* \psi = \eta$.
\end{definition}
It should be noted that the domain of $H^*$ need not be dense (and can even be trivial). Note that the assumption on the denseness of the domain of $H$ is necessary to uniquely determine the vector $\eta$ from the inner products. It should be noted that not all properties of Exercise~\ref{ex:adjoint} hold true for unbounded operators. Indeed, it is only possible to define $(H^*)^*$ if $D(H^*)$ is dense. Even then it need not be true that $H = H^{**}$, but in that case $H^{**}$ is always an \idx{extension} of $H$, in the sense that $D(H) \subset D(H^{**})$ and both operators agree $D(H)$.

Again it is often interesting to study the spectrum of an unbounded operator, for example to find the energy levels of a quantum system. The definition of the spectrum is as follows:
\begin{definition}\label{def:specunbounded}
	Let $(H, D(H))$ be an unbounded operator on a Hilbert space $\mathcal{H}$. The \idx{resolvent} of $H$ is the set of $\lambda \in \mathbb{C}$ such that $H-\lambda I$ is a bijection from $D(H)$ onto $\mathcal{H}$ and the inverse $(H-\lambda I)^{-1}$ is a bounded operator. The \emph{spectrum}\index{spectrum!of unbounded operator}\index{_spec(H)@$\operatorname{spec}(H)$} is defined to be the complement of the resolvent. We also denote the spectrum by $\operatorname{spec}(H)$.
\end{definition}
\begin{remark}
The definition clearly is very similar to the definition of the spectrum of bounded operators, $\sigma(A)$, that we defined before. To avoid confusion, we use different notations to distinguish the spectra of bounded and unbounded operators.
\end{remark}

In the case of bounded operators, the spectral analysis had a particular nice form for self-adjoint operators. The same is true for unbounded operators. Just as in the bounded case, an unbounded operator $(H, D(H))$ is called self-adjoint if $H=H^*$, in particular $D(H) = D(H^*)$. Note that this implies that the domain of $D(H^*)$ is also dense. There is a weaker notion of a \emph{symmetric operator}.\index{operator!symmetric} The operator $H$ is symmetric if
\[
	\langle \psi, H \eta \rangle = \langle H \psi, \eta \rangle
\]
for all $\psi,\eta \in D(H)$. This is clearly true for self-adjoint operators, but a symmetric operator need not be self adjoint This is because the domain issues mentioned above: it implies that $D(H) \subset D(H^*)$, because of the definition of the adjoint given above, but the right hand side may in fact be larger than $D(H)$.

This means that a symmetric operator $S$ may have multiple self-adjoint extensions. That is, there can be distinct self-adjoint operators that agree with $S$ on the domain $D(S)$. If there is a \emph{unique} self-adjoint extension of $S$ we say that $S$ is essentially self-adjoint.\index{operator!essentially self-adjoint} In that case it is common to identify $S$ with its self-adjoint extension. In mathematical physics, in particular quantum mechanics, it is often important to study if, say, a Hamiltonian is essentially self-adjoint or not, but for us these questions do not play a role. We therefore end our discussion of unbounded operators (at least for now) by stating a variant of the spectral theorem for self-adjoint operators. Again there are different formulations possible, but we will only state the projection valued measure form:

\begin{theorem}[Spectral theorem for unbounded operators, measure form]\index{spectral theorem!unbounded operators}
	Let $\mathcal{H}$ be a Hilbert space $\mathcal{H}$ and $(A, D(A))$ an unbounded self-adjoint operator acting on $\mathcal{H}$. Then $\operatorname{spec}(A) \subset \mathbb{R}$ and there is a projection valued measure $\mu$ on the Borel subsets of $\mathbb{R}$ such that
	\[
		g(A) := \int g(\lambda) d\mu(\lambda)
	\]
defines a self-adjoint operator for any real-valued Borel function $g$. If $g(\lambda) = \lambda$, then $g(A) = A$.
\end{theorem}

The basic idea behind the proof is the same as for the bounded case, but one has to be careful to take into account the domains of the unbounded operator. The measure $\mu$ is also slightly different: because the spectrum of unbounded operators is no longer compact, it is no longer true that there is some interval $I = (-a, a)$ such that $\mu(I)$ is the identity. This has to be replaced by the weaker condition that $\operatorname{s-lim}_{a \to \infty} \mu( (-a, a) ) = I$. Or, alternatively, $\mu(\mathbb{R}) = I$.

\begin{remark}\label{rem:1paramexp}
The functional calculus is not restricted to real-valued Borel functions (but only in this case, $g(A)$ is self-adjoint), but can be extended to arbitrary complex-valued Borel functions. Moreover, if $g$ is a \emph{bounded} Borel function, then $g(A)$ is a bounded operator. In this case the functional calculus again leads to a $*$-homomorphism, from the bounded Borel functions to the bounded operators.

In the study of quantum mechanics this is very important: consider a (unbounded) self-adjoint $H$ and for $t \in \mathbb{R}$, the function $g_t(\lambda) = e^{i t \lambda}$. Then $g_t$ is a bounded Borel function on $\mathbb{R}$. Hence $U(t) := g_t(H)$ is a bounded operator, which we suggestively denote by $e^{it H}$. Note that $U(t)^* = \overline{g_t}(H) = g_{-t}(H)$. Moreover, $U(t)^* U(t) = g_t(H) g_{-t}(H) = g_0(H) = I$, so that $U(t)$ is unitary. In fact, it is easy to verify that $U(t+s) = U(t) U(s)$. If $H$ is the Hamiltonian of a quantum system, this gives the one-parameter group of unitaries implementing the time evolution.
\end{remark}
\index{operator!unbounded|)}

\section{Linear functionals and states}
Recall that in quantum mechanics, a (pure) state is often represented by a wave function, representing a vector in the Hilbert space of the system. Suppose that we have such a vector $\ket{\psi}$. Then the expectation value of an observable is given by
\begin{equation}
	\label{eq:state}
	A \mapsto \bra{\psi} A \ket{\psi}.
\end{equation}
This definition clearly also works for operators $A$ that are not observables in the usual sense, because they are not self-adjoint. If we allow such operators, it is clear that equation~\eqref{eq:state} is linear as a function of $A$. Moreover, it is positive: for any $A$ we have 
\[
	\bra{\psi} A^*A \ket{\psi} =  \braket{A \psi}{A \psi} = \| A \psi \|^2 \geq 0.
\]
This can be abstracted to the setting of Banach $*$-algebras.
\begin{definition}
A \idx{linear functional} on a Banach $*$-algebra $\alg{A}$ is a linear map $\omega: \alg{A} \to \mathbb{C}$. It is called \emph{positive} if $\omega(A^*A) \geq 0$ for all $A \in \alg{A}$, i.e., if it takes positive values on positive operators.
\end{definition}
Note that this is a special case of the continuous linear maps of Proposition~\ref{prop:lincont}. In particular a linear functional is continuous if and only if $\omega$ is bounded. We can also define the norm of a linear functional by
\[
	\| \omega \| := \sup_{A \in \alg{A}, \|A\| = 1} |\omega(A)|.
\]
This is just a special case of the norm on bounded linear maps: the linear functionals on a Banach $*$-algebra $\alg{A}$ can be identified with $\alg{B}(\alg{A}, \mathbb{C})$.

\begin{lemma}[Cauchy-Schwarz]\index{Cauchy-Schwarz inequality}
	\label{lem:csineq}
	Let $\omega$ be a positive linear functional on a $C^*$-algebra $\alg{A}$. Then for all $A,B \in \alg{A}$ we have the inequality
\[
	|\omega(B^*A)|^2 \leq \omega(B^*B) \omega(A^*A).
\]
This is a variant of the well-known Cauchy-Schwarz inequality. Moreover, $\overline{\omega(A^*B)} = \omega(B^*A)$.
\end{lemma}
\begin{proof}
Let $\lambda \in \mathbb{C}$. Since $\omega$ is positive, it follows that
\[
	\omega( (\lambda A + B)^*(\lambda A + B)) = |\lambda|^2 \omega(A^*A) + \overline{\lambda} \omega(A^*B) + \lambda \omega(B^*A) + \omega(B^*B) \geq 0.
\]
Because this expression has to be real for all $\lambda$, it follows that $\overline{\omega(A^*B)} = \omega(B^*A)$. If we then specialise to $\lambda$ being real, we get a quadratic function of $\lambda$. Demanding that the graph of this function lies above (or on) the real line leads to the desired inequality.
\end{proof}

Note that we have not required linear functionals to be continuous. It turns out, however, that positive linear functionals are automatically continuous, and their norm can be obtained by evaluating the linear functional on the identity of the algebra. We will not prove this result here, but it can be found in most textbooks on operator algebras (for example,~\cite[Prop. 2.3.11]{MR887100}).
\begin{theorem}\label{thm:stateineq}
Let $\omega$ be a linear functional on a unital $C^*$-algebra $\alg{A}$. Then the following are equivalent:
\begin{enumerate}
	\item $\omega$ is a positive linear map;
	\item $\omega$ is continuous and $\|\omega\| = |\omega(I)|$.
\end{enumerate}
If any of these are satisfied, we have the following properties:
\begin{enumerate}[(a)]
	\item $|\omega(A)|^2 \leq \omega(A^*A) \|\omega\|$;
	\item $|\omega(A^*BA)| \leq \omega(A^*A) \|B\|$\label{it:stateopineq}
\end{enumerate}
for all $A,B \in \alg{A}$.
\end{theorem}
Again, a similar statement is true for non-unital algebras, if one replaces the unit by an approximate unit.

\subsection{The state space of a $C^*$-algebra}
Let $\alg{A}$ be a $C^*$-algebra. Then a \idx{state} on $\alg{A}$ is a positive linear functional $\omega : \alg{A} \to \mathbb{C}$ of norm one (hence $\omega(I) = 1$ if $\alg{A}$ is unital). We will write $\mc{S}(\alg{A})$ for the set of states on $\alg{A}$.\index{_S(A)@$\mc{S}(\alg{A})$}\index{state space} It turns out that this notion of a state is the correct abstraction of states in Hilbert space quantum mechanics. We will briefly come back to this later, but for the moment note that if $\alg{A} = \alg{B}(\mc{H})$ and $\psi \in \mc{H}$, the map $A \mapsto \langle \psi, A \psi \rangle$ is a state (if $\psi$ is a unit vector). This is an example of a \idx{vector state}. Recall that if $A$ is a quantum mechanical observable, this expression gives the expectation value of the measurement outcome. This is how we can think of states on $C^*$-algebras as well.

Note that $\mc{S}(\alg{A})$ is a metric space, where the metric is induced by the norm on linear functionals. There are also other natural notions of convergence of a sequence of states. We will need the following two later.
\begin{definition}
	Let $\varphi_n$ be a sequence of states on some $C^*$-algebra $\alg{A}$. We say that $\varphi_n$ \emph{converges in norm} to a state $\varphi \in \mc{S}(\alg{A})$ if $\| \varphi_n - \varphi \| \to 0$, where the norm is the usual norm on linear functions, as defined above. We say that $\varphi_n$ converges to $\varphi$ in the \emph{weak-$*$ topology}\index{weak-* topology@weak $*$-topology} if for each $A \in \alg{A}$, one has that $|\varphi_n(A) - \varphi(A)| \to 0$. 
\end{definition}

\begin{exercise}
	Let $\varphi \in \mc{S}(\alg{A})$ and suppose that we have a sequence $\varphi_n \in \mc{S}(\alg{A})$. Show that $\varphi_n \to \varphi$ in the norm topology implies that $\varphi_n \to \varphi$ in the weak-$*$ topology. Is the converse also true? Prove your claim or give a counterexample.
\end{exercise}

The state space of a $C^*$-algebra $\alg{A}$ in general has a very rich structure, which strongly depends on the algebra itself. Nevertheless, there are some general properties that are always true. An important property is that $\mc{S}(\alg{A})$ is \emph{convex}: if $0 \leq \lambda \leq 1$ and $\omega_1$ and $\omega_2$ are states,
\[
	\omega(A) := \lambda \omega_1(A) + (1-\lambda) \omega_2(A)
\]
is a state again (since $\omega(1) = \lambda + (1-\lambda) = 1$, and $\omega$ is positive). A pure state is a state that cannot be written as a convex combination of two distinct states.
\begin{definition}
	Let $\omega$ be a state on a $C^*$-algebra $\alg{A}$. Then $\omega$ is called \emph{pure}\index{state!pure} if $\omega = \lambda \omega_1 + (1-\lambda) \omega_2$ for $\lambda \in (0,1)$ and some states $\omega_1, \omega_2$ implies that $\omega_1 = \omega_2 = \omega$. If $\omega$ is not pure, it is called \emph{mixed}\index{state!mixed}.
\end{definition}
We would like to stress that the notion of pure and mixed states depends on the algebra $\alg{A}$. If $\alg{B} \subset \alg{A}$ is a $C^*$-subalgebra (that is, $\alg{B}$ is closed in norm in $\alg{A}$ and a $C^*$-algebra) containing the unit of $\alg{A}$, we can restrict a state $\omega$ on $\alg{A}$ to a state $\omega|_{\alg{B}}$ on $\alg{B}$. Even if $\omega$ is pure, it is not necessarily so that $\omega|_{\alg{B}}$ is.

\begin{remark}\label{rem:pureequiv}
	A pure state can equivalently be defined as follows. A state $\omega$ is pure if for every positive linear functional $\varphi$ that is majorised by $\omega$, that is $0 \leq \varphi(A^*A) \leq \omega(A^*A)$ for all $A \in \alg{A}$, then $\varphi = \lambda \omega$ for some $0 \leq \lambda \leq 1$. It requires some work to show that these are indeed equivalent.
\end{remark}

\begin{exercise}
	\label{exc:density} Let $\alg{A} = M_d(\mathbb{C})$ for some positive integer $d$. Recall that a density matrix\index{density matrix} $\rho$ is a matrix such that $\rho^* = \rho$, $\rho$ is positive (i.e., all the eigenvalues are bigger than or equal to zero) and $\Tr(\rho)=1$. Show that states $\omega$ on $\alg{A}$ are in one-to-one correspondence with density matrices $\rho$ such that $\omega(A) = \Tr(\rho A)$ for all $A$. Also show that $\omega$ is pure if and only if the corresponding $\rho$ satisfies $\rho^2 = \rho$ and if and only if $\rho$ is a projection onto a one-dimensional subspace.
\end{exercise}

As a special case of this exercise, consider $d=2$. Such a system is called a \idx{qubit} in quantum information theory. A convenient basis of the two-by-two matrices is given by the \idx{Pauli matrices}\footnote{Some authors label the matrices by $\sigma_1, \sigma_2, \sigma_3$ and write $\sigma_0$ for the identity matrix. Later on we will also use a superscript instead of a subscript, since we will also need to index the site on which the matrix acts.}
\begin{equation}
	\label{eq:paulixyz}
	\sigma_x = \begin{pmatrix} 0 & 1 \\ 1 & 0 \end{pmatrix},\quad
	\sigma_y = \begin{pmatrix} 0 & -i \\ i & 0 \end{pmatrix},\quad
	\sigma_z = \begin{pmatrix} 1 & 0 \\ 0 & -1 \end{pmatrix},
\end{equation}
taken together with the identity matrix. Taking linear combinations of these matrices we see that each density matrix $\rho$ can be written as
\[
	\rho = \frac{1}{2} \begin{pmatrix}
		1 +z & x + i y \\
		x - i y	& 1-z
	\end{pmatrix}.
\]
Because $A \mapsto \Tr(\rho A)$ should be a state, it follows that $x,y$ and $z$ must be real, and that $x^2 + y^2 + z^2 \leq 1$. Hence any state on $M_2(\mathbb{C})$ can be characterized by giving a point in the unit ball of $\mathbb{R}^3$. In this context it is called the \emph{Bloch sphere}. As an exercise, one can show that the state is pure if and only if the corresponding point $(x,y,z)$ lies on the surface of the sphere, so that $x^2 + y^2 + z^2 = 1$. This example also shows that if we have a \emph{mixed} state, it generally can be written in many different ways as a convex combination of pure states. Indeed, take any two antipodal points on the sphere, such that the line connecting them goes trough the point $(x,y,z)$. Then the state $\rho$ can be written as a convex linear combination of the two pure states corresponding to the points on the sphere.

To finish this section, we note the following useful theorem (without proof).
\begin{theorem}\label{thm:statecompact}\index{state space!compactness}
	Let $\alg{A}$ be a unital $C^*$-algebra. Then the set of states $\alg{A}$ is a closed convex and compact subset (with respect to weak-$*$ topology) of the set of all linear functionals.
\end{theorem}
This implies, for example, that if $\varphi_n$ is a sequence of states converging to some $\varphi$ in the weak-$*$ topology, then $\varphi$ is also a state. The compactness property can also be used to show that any sequence of states has a converging subsequence.

\subsection{Example: classical states}\index{state!classical}
A nice thing about the algebraic approach to quantum mechanics is that \emph{classical} mechanics can be described in the same framework. We illustrate this with an example. Consider a configuration space $M$ of a classical system. To avoid technical details as much as possible, we assume that $M$ is compact. A \emph{classical observable} is a continuous function $f: M \to \mathbb{R}$. Slightly more general, we may consider \emph{complex} valued continuous functions. As we have seen before, these functions form a $C^*$-algebra $C_0(M)$. The classical observables correspond to the self-adjoint elements of this algebra.

Now let $\mu$ be a probability measure on the configuration space $M$. This measure can be interpreted as our knowledge of the system: it tells us what the state of the system is and with which probability. If we want to measure an observable $f$, the expectation value of the outcome is given by
\[
	\omega(f) := \int f(x) d\mu(x).
\]
Note that $\omega(1) = 1$ since $\mu$ is a probability measure and that $\omega(\overline{f} f) \geq 0$. Hence $\omega$ is a state on the observable algebra. Conversely, one can show that \emph{any} state on $C_0(M)$ comes from a probability measure $\mu$ (see for example~\cite[Thm. 6.3.4]{MR971256}). We already used this fact in the discussion of the spectral calculus. The pure states of $C_0(M)$ correspond to Dirac (or point) measures. If $\mu = \delta_x$, the Dirac measure in the point $x \in M$, then $\omega(f) = f(x)$ for all $f$. This is the situation in we know in which configuration $x \in M$ the system is, and hence we know the outcome of each observable.

\section{The Gel'fand-Naimark-Segal construction}\label{sec:gns}
In the usual Hilbert space setting of quantum mechanics, the observables of a system are represented as bounded or unbounded operators acting on some Hilbert space. It turns out that in the abstract setting we are considering here, where observables are elements of some $C^*$-algebra $\alg{A}$, Hilbert space is actually just around the corner. That is, given a state $\omega$ on $\alg{A}$, one can construct a Hilbert space $\mc{H}_\omega$ and a representation $\pi$ of $\alg{A}$ into $\alg{B}(\mc{H}_\omega)$. This is the content of the \emph{Gel'fand-Naimark-Segal theorem}. Before going into the details of this construction, we first give the basic definitions.

\begin{definition}
	Let $\alg{A}$ be a $C^*$-algebra and let $\mc{H}$ be a Hilbert space. A \idx{representation} of $\alg{A}$ on $\mc{H}$ is a $*$-homomorphism $\pi: \alg{A} \to \alg{B}(\mc{H})$. That is, a linear map such that $\pi(AB) = \pi(A)\pi(B)$ for all $A,B \in \alg{A}$. For a $*$-representation one has in addition that $\pi(A)^* = \pi(A^*)$. In these notes we only consider $*$-representations (without explicitly mentioning so).
\end{definition}
Thus a representation represents the elements of an abstract $C^*$-algebra as bounded operators on a Hilbert space, in such a way that the algebraic relations among the elements are preserved.

\begin{remark}\label{rem:morcont}
	A basic result in operator theory is that for representations $\pi$ of $C^*$-algebras, one automatically has $\|\pi(A)\| \leq \|A\|$ for all $A \in \alg{A}$. It follows that a representation is automatically continuous with respect to the norm topology. This can be shown using the spectral calculus. If $A \in \alg{A}$ and $\lambda \notin \sigma(A)$, then $A - \lambda I$ is invertible in $\alg{B}(\mathcal{H})$, where $\mathcal{H}$ is the space on which $\pi$ acts. Since $\pi(I) = I$ it follows that $\pi(A)-\lambda I$ is invertible, hence $\lambda \notin \sigma(\pi(A))$, or $\sigma(\pi(A)) \subset \sigma(A)$. The result then follows by noting that for the spectral radius, $\|A\|^2 = r(A^*A)$, and similarly for $\pi(A)$.
	
If $\pi$ is injective then it is in fact an isometry, meaning that $\|\pi(A)\| = \|A\|$.
\end{remark}

To avoid trivialities we will always restrict to \emph{non-degenerate} representations.\index{representation!non-degenerate} A representation $\pi : \alg{A} \to \alg{B}(\mc{H})$ is non-degenerate if the set $\pi(\alg{A}) \mc{H}$ is dense in $\mc{H}$. If a representation is degenerate, one can always obtain a non-degenerate representation by restricting to Hilbert space to $[\pi(\alg{A}) \mc{H}]$. A representation is called \emph{cyclic} if there is some vector $\Omega \in \mc{H}$ such that $\pi(\alg{A}) \Omega$ is a dense subset of $\mc{H}$. Such an $\Omega$ is called a \idx{cyclic vector}. We are now in a position to state and prove one of the fundamental tools in operator algebra, the \emph{GNS construction}\index{GNS representation}\index{representation!Gel'fand-Naimark-Segal}, after Gel'fand, Naimark and Segal.

\begin{theorem}
	\label{thm:gnsrep}
	Let $\alg{A}$ be a unital $C^*$-algebra and $\omega$ a state on $\alg{A}$. Then there is a triple $(\pi_\omega, \mc{H}_\omega, \Omega)$, where $\mc{H}_{\omega}$ is a Hilbert space and $\pi_\omega$ a representation of $\alg{A}$ on $\mc{H}_\omega$, such that $\Omega$ is cyclic for $\pi_\omega$ and in addition we have
\[
	\omega(A) = \langle \Omega, \pi_\omega(A) \Omega \rangle, \quad A \in \alg{A}.
\]
This triple is unique in the sense that if $(\pi, \mc{H}, \Psi)$ is another such triple, there is a unitary $U : \mc{H}_\omega \to \mc{H}$ such that $U \Omega = \Psi$ and $\pi(A) = U \pi_\omega(A) U^*$ for all $A \in \alg{A}$.
\end{theorem}
\begin{proof}
	Define the set $N_\omega = \{ A \in \alg{A} : \omega(A^*A) = 0 \}$. Suppose that $A \in \alg{A}$ and $B \in N_\omega$. Then the Cauchy-Schwarz inequality, Lemma~\ref{lem:csineq}, implies that $A^*B \in N_\omega$. But this means that $N_\omega$ is a left ideal of $\alg{A}$. We can then form the quotient vector space $H_\omega = \alg{A}/N_\omega$ and denote $[A]$ for the equivalence class of $A \in \alg{A}$ in this quotient. That is, $[A] = [B]$ if and only if $A = B + N_0$ for some $N_0 \in N_\omega$. We can define an inner product on $H_\omega$ by 
\[
	\langle [A], [B] \rangle := \omega(A^*B).
\]
This is clearly sesqui-linear, and by Lemma~\ref{lem:csineq} it follows that this definition is independent of the choice of representatives $A$ and $B$. Because $\omega$ is a state this linear functional is positive, and it is non-degenerate precisely because we divide out the null space of $\omega$. Indeed, if $\langle [A], [A] \rangle = 0$ it follows that $A \in N_\omega$. By taking the completion with respect to this inner product we obtain our Hilbert space $\mc{H}_\omega$.

Next we define the representation $\pi_\omega$ by defining the action of $\pi_\omega(A)$ on the dense subset $H_\omega$ of $\mc{H}_\omega$. Let $[B] \in H_\omega$, then we define $\pi_\omega(A) [B] := [AB]$. It is easy to check that this is well-defined. Note that for all $B \in \alg{A}$ we have
\[
\| \pi_\omega(A) [B] \|_{\mc{H}_\omega}^2 = \langle [AB], [AB] \rangle = \omega(B^*A^*AB) \leq \|A\|^2 \omega(B^*B) = \|A\|^2 \| [B] \|_{\mc{H}_\omega}^2,
\]
where $\|\cdot\|$ is the norm of the $C^*$-algebra. For the inequality we used Theorem~\ref{thm:stateineq}\eqref{it:stateopineq}. Hence $\pi_\omega(A)$ is bounded on a dense subset of $\mc{H}_\omega$. It is also not difficult to check that $\pi_\omega(A)$ is linear on this dense subset, hence it can be extended to a bounded operator on $\mc{H}_\omega$, which we again denote by $\pi_\omega(A)$. From the definition of $\pi_\omega$ it is apparent that $\pi_\omega(AB) = \pi_\omega(A) \pi_\omega(B)$ and $\pi_\omega(A^*) = \pi_\omega(A)^*$ is an easy check. If we define $\Omega := [I]$ it is clear that $\Omega$ is cyclic for $\pi_\omega$, by definition of $\mc{H}_\omega$. Moreover, for each $A \in \alg{A}$,
\[
	\langle \Omega, \pi_\omega(A) \Omega \rangle = \langle [I], [A] \rangle = \omega(A),
\]
which completes the construction of the GNS representation.

It remains to be shown that the construction is essentially unique. Suppose that $(\pi, \mc{H}, \Psi)$ is another such triple. Define a map $U: \mc{H}_\omega \to \mc{H}$ by setting $U \pi_\omega(A) \Omega = \pi(A) \Psi$ for all $A \in \alg{A}$. Note that this is a linear map of a dense subspace of $\mc{H}_\omega$ onto a dense subspace of $\mc{H}$. Moreover, for each $A,B$ we have
\[
	\langle U \pi_\omega(A) \Omega, U \pi_\omega(B) \Omega \rangle_\omega = \langle \pi(A) \Psi, \pi(B) \Psi \rangle = \omega(A^*B) = \langle \pi_\omega(A) \Omega, \pi_\omega(B) \Omega \rangle.
\]
It follows that $U$ can be extended to a unitary operator from $\mc{H}_\omega$ onto $\mc{H}$. The property that $U \pi_\omega(A) = \pi(A) U$ can be easily verified on a dense subset of the Hilbert space.
\end{proof}

\begin{remark}
	The theorem is still true if one drops the condition that $\alg{A}$ is unital. The proof however becomes more involved and involves so-called \emph{approximate units}. 
\end{remark}

\begin{exercise}
	Let $\alg{A}$ be a $C^*$-algebra and $\omega$ a faithful state, meaning that $\omega(A^*A) > 0$ if $A \neq 0$.\index{state!faithful} Show that the corresponding GNS representation is faithful (in the sense that the kernel is trivial).
\end{exercise}
Faithfulness of the state is not a \emph{necessary} condition for the representation to be faithful, only a sufficient one.

A representation is a special case of a \idx{morphism} between $C^*$-algebras. A ($*$-)morphism $\rho:\alg{A} \to \alg{B}$ between two $C^*$-algebras is a linear map such that $\rho(AB) = \rho(A)\rho(B)$ and $\rho(A^*) = \rho(A)^*$ for all $A,B \in \alg{A}$. Hence a representation is a morphism where $\alg{B} = \alg{B}(\mc{H})$ for some Hilbert space $\mc{H}$. Morphisms are also always continuous: $\|\rho(A)\| \leq \|A\|$.

\begin{definition}
	An \emph{automorphism}\index{automorphism!of a Cstar-algebra@of a $C^*$-algebra} $\alpha$ of a $C^*$-algebra $\alg{A}$ is a morphism $\alpha: \alg{A} \to \alg{A}$ that is invertible, such that $\alpha^{-1}$ is a morphism as well. The set of all automorphisms of a $C^*$-algebra forms a group, which is denoted by $\operatorname{Aut}(\alg{A})$.
\end{definition}
\begin{exercise}
	Show that an automorphism of a $C^*$-algebra is isometric. Also verify that if $U \in \alg{A}$ is unitary, the map $\operatorname{Ad} U$ defined by $A \mapsto UAU^*$ is an automorphism. 
\end{exercise}
An automorphism that is given by conjugation with a unitary, as in the exercise, is called \emph{inner}.\index{inner automorphism} In general, \emph{not} every automorphism of a $C^*$-algebra is inner (except when $\alg{A} = \alg{B}(\mc{H})$). Note that automorphisms preserve all the algebraic relations of the algebra. Hence they are a natural tool to model symmetries. This is the context in which we will encounter most automorphisms later in the course.

There is an important corollary that follows from the uniqueness of the GNS representation. If $\alpha$ is an automorphism of $\alg{A}$ and $\omega$ is invariant under the action of this automorphism (for example, a ground state of a physical system is invariant under some symmetry), then $\alpha$ is implemented by a unitary in the GNS representation. The precise statement is as follows:
\begin{corollary}
	Let $\alg{A}$ be a $C^*$-algebra and $\alpha$ an automorphism of $\alg{A}$. Suppose that $\omega$ is a state on $\alg{A}$ such that $\omega \circ \alpha = \omega$. Then there is a cyclic representation $(\pi, \mc{H}, \Omega)$ such that $\omega(A) = \langle \Omega, \pi(A) \Omega\rangle$ and a unitary $U \in \alg{B}(\mc{H})$ such that $\pi \circ \alpha(A) = U \pi(A) U^*$ and $U\Omega= \Omega$.
\end{corollary}
\begin{proof}
	Let $(\pi, \mc{H}, \Omega)$ be the GNS representation for the state $\omega$. Note that $(\pi \circ \alpha, \mc{H}, \Omega)$ is another GNS triple: $\Omega$ is again cyclic since $\alpha$ is an automorphism, and we have
\[
\langle \Omega, \pi(\alpha(A)) \Omega \rangle = \omega(\alpha(A)) = \omega(A).
\]
By the uniqueness statement in Theorem~\ref{thm:gnsrep} there is a unitary operator $U$ such that $\pi\circ\alpha(A) = U \pi(A) U^*$. This is the unitary we were looking for.
\end{proof}

\begin{remark}
	The automorphism $\alpha$ does not have to be inner. At first sight it may seem like a contradiction that a non-inner automorphism is implemented by a unitary operator in the GNS representation. This confusion can be resolved by noting that $U$ does \emph{not} need to be in $\pi(\alg{A})$, but only in $\alg{B}(\mc{H})$, the latter being bigger in general.
\end{remark}

Essentially the same proof can be given for groups of automorphisms. For example, let $t \mapsto \alpha_t$, $t \in \mathbb{R}$ be a group of automorphisms, that is, we have $\alpha_{t+s} = \alpha_t \circ \alpha_s$ and each $\alpha_t$ is an automorphism. Suppose that $\omega$ is an invariant state for each $\alpha_t$. Then there is a group of unitaries $t \mapsto U_t$ implementing these automorphisms in the GNS representation. As an example, one can think of $\alpha_t$ being the time evolution of a quantum mechanical system, defined on the observables (the Heisenberg picture). Then $U_t$ gives us the unitary evolution on the Hilbert space. We will study this later in more detail.

\begin{exercise}\label{ex:invstategns}
Consider a group of automorphisms $\alpha_t$ as above, together with an invariant state $\omega$.
\begin{enumerate}
	\item Show that the map $t \mapsto U_t$ indeed defines a unitary representation (so that $U(t+s) = U(t)U(s)$ and $U(t)^* = U(-t)$).
	\item Suppose that $t \mapsto \alpha_t$ is strongly continuous. This means that for each $A \in \alpha_t$, the map $t \mapsto \alpha_t(A)$ is continuous. Show that for each $\psi \in \mc{H}_\omega$ the map $t \mapsto U(t) \psi$ is continuous (with respect to the norm topology on $\mc{H}_\omega$). \emph{Hint:} first show that it is enough to show continuity as $t \to 0$.\index{group of automorphisms!strongly continuous}
\end{enumerate}
These results are actually true for \emph{any} topological group $G$.
\end{exercise}

Another application is the Gel'fand-Naimark theorem, which states that in fact \emph{every} $C^*$-algebra can be seen as some algebra of bounded operators acting on a Hilbert space. In practice this is however of limited value. The resulting representation is hard to describe explicitly, and does not give much insight. It is therefore often easier to work with the abstract $C^*$-algebra $\alg{A}$ itself.

Before we prove this result we first return to direct sums of Hilbert spaces. At the end of Section~\ref{sec:funcan} direct sums of a finite number of Hilbert spaces were defined. This can be generalised to arbitrary direct sums as follows (c.f. Exercise~\ref{ex:infsum}). Let $I$ be an index set (infinite or not) and suppose that $\mc{H}_i$ is a Hilbert space for each $i \in I$. The direct sum $\bigoplus_{i \in I} \mc{H}_i$ is defined as all those $\psi = (\psi_i \in \mc{H}_i)_{i \in I}$ such that
\[
	\sum_{i \in I} \| \psi_i \|^2_{\mc{H}_i} < \infty.
\]
An inner product can then be defined by $\langle \psi, \eta \rangle = \sum_{i \in I} \langle \psi_i, \eta_i \rangle_{i \in I}$, which converges precisely because of the condition on the norms. One can show that this indeed is an inner product and that it makes the direct sum into a Hilbert space. We will also write $\oplus_{i \in I} \psi_i$ for a vector $\psi$ as above.

If for each $i \in I$ we have a representation $\pi_i$ of a $C^*$-algebra $\alg{A}$ on $\mc{H}_i$, the direct sum representation $\pi$ can be defined by
\[
\pi(A) \psi = \bigoplus_{i \in I} \pi_i(A) \psi_i.
\]
Note that this is a generalisation of the direct sums of linear maps that were discussed earlier. An easy check shows that $\pi$ is indeed a representation. We will also write $\oplus_{i \in I} \pi_i$ for this representation.

\begin{theorem}[Gel'fand-Naimark]\label{thm:gelfandnaimark}\index{Gel'fand-Naimark theorem}
	Every $C^*$-algebra $\alg{A}$ can be isometrically represented on some Hilbert space $\mc{H}$. That is, there is a faithful isometric representation $\pi: \alg{A} \to \alg{B}(\mc{H})$.
\end{theorem}
\begin{proof}
	It is a fact that for a representation $\pi$ of $\alg{A}$, for each $A \in \alg{A}$ the inequality $\| \pi(A) \| \leq \| A \|$ holds (see Remark~\ref{rem:morcont}). It therefore remains to show that one can find a representation such that $\|\pi(A)\| \geq \|A\|$. To this end we will need the following result. Let $A \in \alg{A}$ be positive (that is, $A = B^*B$ for some $B \in \alg{A}$). Then there is a state $\omega$ on $\alg{A}$ such that $\omega(A) = \|A\|$.\footnote{\label{fn:purestate}The proof of this, which is not too difficult, relies on some facts that we have not discussed here. See for example~\cite[Thm. 4.3.4(iv)]{MR719020}. One can even show that $\omega$ can be chosen to be a \emph{pure} state.}

Now consider $B \in \alg{A}$ and let $\omega$ be a state such that $\omega(B^*B) = \|B^*B\| (=\|B\|^2)$. Construct the corresponding GNS representation $(\pi_\omega, \mc{H}_\omega, \Omega)$. Then
\[
	\|\pi_\omega(B) \Omega\|^2 = \langle \Omega, \pi_\omega(B^*B) \rangle = \omega(B^*B) = \|B\|^2.
\]
Hence $\|\pi_\omega(B)\| \geq \|B\|$ and hence $\|\pi_\omega(B)\| = \|B\|$. To conclude the proof of the theorem, consider the representation $\pi$ defined by
\[
\pi = \bigoplus_{\omega \in \mc{S}(\alg{A})} \pi_\omega,
\]
where $\pi_\omega$ is the corresponding GNS representation. In fact it is enough to only consider the pure states (c.f. Footnote~\ref{fn:purestate}). It follows that $\|\pi(A)\| = \|A\|$ for all $A \in \alg{A}$.
\end{proof}

In the theorem we constructed a direct sum of cyclic representations. It is true in general that \emph{any} non-degenerate representation $\pi : \alg{A} \to \alg{B}(\mc{H})$ can be obtained as a direct sum of cyclic representations. To illustrate this, choose a non-zero vector $\psi \in \mc{H}$. Then we can consider the subspace $\mc{K} = [\pi(\alg{A})\psi]$, the closure of the set $\pi(\alg{A})\psi$. Note that $\mc{K}$ is a Hilbert space and that we have a projection $P_{\mc{K}}$ that projects $\mc{H}$ onto this subspace. Since $\mc{K}$ is invariant, $\pi(A) \psi \in \mc{K}$ if $\psi \in \mc{K}$ and it follows that $P_{\mc{K}} \pi(A) = \pi(A) P_{\mc{K}}$ for all $A \in \alg{A}$. The same is true for the projection on the orthogonal complement, $I-P_{\mc{K}}$. This allows us to write $\pi$ is a direct sum of representations,
\[
	\pi(A) = P_{\mc{K}} \pi(A) P_{\mc{K}} \oplus (I-P_{\mc{K}}) \pi(A) (I-P_{\mc{K}})
\]
It is straightforward to check this equality if we identify $\mc{H}$ with $\mc{K} \oplus \mc{K}^{\perp}$. By construction, $P_{\mc{K}} \pi(A) P_{\mc{K}}$ is a cyclic representation. We can then start over again with the representation in the second summand, and so on, to write $\pi$ as a direct sum of cyclic representations. This argument can be made fully rigorous by an application of Zorn's Lemma.

A natural question is when it is not possible to further decompose a representation into a direct sum of representations. Such representations are called \emph{irreducible}.\index{representation!irreducible} There are in fact two natural notions of irreducibility. Let $\alg{A}$ be a $C^*$-algebra acting non-degenerately on some Hilbert space $\mc{H}$. The non-degeneracy condition means that $\alg{A} \mc{H}$ is dense in $\mc{H}$. This amounts to saying that there is no non-zero vector $\psi$ such that $A \psi = 0$ for all $A \in \alg{A}$. If $\alg{A}$ acts degenerately, one can always restrict $\alg{A}$ to (the closure of) the subspace $\alg{A} \mc{H}$. The algebra $\alg{A}$ is said to act \emph{topologically irreducible} if the only closed subspaces left globally invariant under the action of $\alg{A}$ are $\{0\}$ and $\mc{H}$ itself. It is said to be \emph{irreducible} if $\alg{A}' = \mathbb{C}I$. Here the prime denotes the \idx{commutant} of $\alg{A}$, that is
\begin{equation}
	\label{eq:commutant}
	\alg{A}' = \{ B \in \alg{B}(\mc{H}) : AB = BA \textrm{ for all } A \in \alg{A} \}.
\end{equation}
A perhaps surprising result is that for $C^*$-algebras both notions of irreducibility coincide.

\begin{proposition}\label{prop:irrep}
	Let $\pi : \alg{A} \to \alg{B}(\mc{H})$ be a representation of a $C^*$-algebra $\alg{A}$. Then the following are equivalent:
\begin{enumerate}
	\item \label{it:topirred}$\pi$ is topologically irreducible.
	\item \label{it:commutant}The commutant of $\pi(\alg{A})$ is equal to $\mathbb{C} I$.
	\item \label{it:cyclic}Every non-zero vector $\psi \in \mc{H}$ is cyclic for $\pi$.
\end{enumerate}
\end{proposition}
\begin{proof}
First note that if $\psi \in \mc{H}$ is any vector, then $\pi(\alg{A}) \psi$ is an invariant subspace under the action of $\pi(\alg{A})$ and hence so is its closure. Hence if $\pi$ is topologically irreducible and $\psi \neq 0$ it follows that this closure must be equal to $\mc{H}$, so that $\psi$ is cyclic. Conversely, suppose that $\pi$ is not topologically irreducible. Consider a non-trivial closed subspace $\mc{K}$ and $\psi \in \mc{K}$ non-zero. Since $\psi$ is cyclic, the closure of $\pi(\alg{A})\psi$ is equal to $\mc{H}$ by assumption. But by the invariance assumption, $\pi(\alg{A}) \psi \subset \mc{K}$, a contradiction. This proves the equivalence of~\ref{it:topirred} and~\ref{it:cyclic}.

Next suppose that the commutant is trivial. Let $\mc{K}$ be a closed subspace of $\mc{H}$ that is invariant under $\pi(\alg{A})$. Then the projection $P_{\mc{K}}$ onto $\mc{K}$ commutes with $\pi(\alg{A})$, as we have seen above. Hence $P_{\mc{K}} = \lambda I$ for some $\lambda \in \mathbb{C}$. But $P_{\mc{K}}^2 = P_{\mc{K}}$, from witch it follows that $\lambda = 0$ or $\lambda = 1$ and therefore $\mc{K}$ is equal to $\mc{H}$ or the zero subspace, proving that~\ref{it:commutant} implies~\ref{it:topirred}.

Finally, suppose that $\pi$ is topologically irreducible. For the sake of contradiction, let us assume that the commutant is non-trivial. Then it must contain a non-trivial self-adjoint element $T$, since if some operator $A$ is in the commutant, so is $A^*$. Hence we can apply spectral theory to $T$, to obtain a non-zero projection $P$ by choosing a non-zero spectral projection. It can be shown that this spectral projection necessarily is in the commutant $\pi(\alg{A})'$ as well. But then the range of this projection is invariant under $\pi(\alg{A})$, a contradiction.
\end{proof}

\begin{exercise}
	\label{ex:gnsrep}
	Let $\alg{A} = M_2(\mathbb{C})$. For $A \in \alg{A}$, write $A_{ij}$ for the corresponding matrix elements. Define $\rho(A) = A_{11}$ and $\sigma(A) = \lambda A_{11} + (1-\lambda) A_{22}$, where $0 < \lambda < 1$. Show that $\rho$ and $\sigma$ are states and find the corresponding GNS representations. Also show that in the first case the representation is irreducible, while in the second case the GNS representation can be written as the direct sum of two irreducible representations.
\end{exercise}

In Exercise~\ref{ex:gnsrep} we obtained GNS representations for different states of $M_2(\mathbb{C})$. Note that in the notation of that exercise, $\rho$ is a pure state, while $\sigma$ is mixed. The corresponding GNS representation of $\rho$ is irreducible, while that of $\sigma$ is reducible. This is in fact a general and very useful result:
\begin{theorem}
	Let $\alg{A}$ be a $C^*$-algebra and $\omega$ a state on $\alg{A}$. Then the corresponding GNS representation $(\pi, \mc{H}, \Omega)$ is irreducible if and only if $\omega$ is pure.\index{state!pure}\index{representation!irreducible}
\end{theorem}
\begin{proof}
	Suppose that $\omega$ is pure but $(\pi, \mc{H}, \Omega)$ is not irreducible. Then there is some non-zero $T \in \pi(\alg{A})'$. It follows that $T^*$ also commutes with every $\pi(A)$ and hence $T+T^*$ is in the commutant. Because there is a non-trivial self-adjoint element in the commutant, it follows that there must be a non-trivial projection $P$ there as well, by the spectral theorem. Consider a linear functional $\varphi$ defined by
\[
	\varphi(A) = \langle P \Omega, \pi_\omega(A) P \Omega \rangle. 
\]
This functional is non-zero because $P$ commutes with $\pi(\alg{A})$ and $\Omega$ is cyclic for $\pi_\omega$. Moreover, it is clearly positive and we have
\[
	\varphi(A^*A) = \langle P \Omega, \pi_\omega(A^*A) P \Omega \rangle = \langle \Omega, \pi_\omega(A^*) P \pi_\omega(A) \Omega \rangle \leq \|P\| \omega(A^*A),
\]
with Theorem~\ref{thm:stateineq}. Hence $\varphi$ is majorised by $\omega$. On the other hand, it is easy to check that $\varphi$ is not a multiple of $\omega$. This is in contradiction with the purity of $\omega$ (see also Remark~\ref{rem:pureequiv}).

Conversely, suppose that $\pi_\omega$ is irreducible. We sketch how to show that $\omega$ must be pure. Let $\varphi$ be a positive linear functional majorised by $\omega$. We want to show that $\varphi = \lambda \omega$ for some $0 \leq \lambda \leq 1$. Consider the dense subspace $\pi_\omega(A) \Omega$ of $\mc{H}_\omega$. We can define a new inner product on this space by
\[
	\langle \pi_\omega(A) \Omega, \pi_\omega(B) \Omega\rangle_\varphi = \varphi(A^*B).
\]
By assumption the new inner product is bounded by the old inner product. By the Riesz representation theorem there exists a bounded operator $T$ such that $\langle \pi_\omega(A) \Omega, \pi_\omega(B) \Omega \rangle_{\varphi} = \langle \pi_\omega(A) \Omega, T \pi_\omega(B) \Omega \rangle$ and $\|T\| \leq 1$ (compare with Footnote~\ref{ftn:adjoint} on page~\pageref{ftn:adjoint}). By checking on a dense subset one can show that $T$ commutes with any $\pi_\omega(A)$ and hence is of the form $\lambda I$ for some $\lambda$. It follows that $\varphi = \lambda \omega$. Since $\|T\| \leq 1$, also $\lambda \leq 1$ and the proof is complete.
\end{proof}

These results can be used to study the structure of abelian $C^*$-algebras. If $\alg{A}$ is an abelian $C^*$-algebra, a \idx{character} is defined to be a non-zero linear map $\omega:\alg{A} \to \mathbb{C}$ such that $\omega(AB) = \omega(A)\omega(B)$ for all $A,B \in \alg{A}$. One can prove that a character is automatically continuous and positive. In other words, a character is a one-dimensional representation of $\alg{A}$. It turns out that the characters are just the pure states of $\alg{A}$.

\begin{proposition}\index{C*algebra@$C^*$-algebra!commutative}\label{prop:characters}
Let $\alg{A}$ be an abelian $C^*$-algebra. Then $\omega$ is a character of $\alg{A}$ if and only if $\omega$ is a pure state.
\end{proposition}
\begin{proof}
	We first show that a state $\omega$ on $\alg{A}$ is pure if and only if $\omega(AB) = \omega(A) \omega(B)$ (that is, it is also a character). To this end, note that by the theorem above $\omega$ is pure if and only if the GNS representation $(\pi_\omega, \mc{H}_\omega, \Omega)$ is irreducible. But because $\alg{A}$ is abelian, we always have that $\pi_\omega(\alg{A}) \subset \pi_\omega(\alg{A})'$. This can only be true if $\mc{H}_\omega$ is one-dimensional. Hence we have
\[
	\omega(A) = \langle \Omega, \pi_\omega(A) \Omega \rangle = \pi_\omega(A) \langle \Omega , \Omega \rangle.
\]
Since $\pi_\omega$ is a representation, it follows that $\omega(AB) = \omega(A)\omega(B)$. Conversely, since $\omega$ is a state it follows by Theorem~\ref{thm:stateineq} that $\omega(A^*) = \overline{\omega(A)}$. Hence if $\omega(AB) = \omega(A)\omega(B)$, $\omega$ is a representation and hence by the uniqueness of the GNS construction, $\mc{H}_\omega$ is one-dimensional.

It remains to be shown that a character of $\alg{A}$ is in fact a state. We assume that $\alg{A}$ is unital, the non-unital case can be shown with the help of approximate units. Since $\omega$ is a character, it is positive and continuous by the comments before the proposition. Moreover, for all $A \in \alg{A}$ it holds that
\[
	\omega(A) = \omega(AI) = \omega(A) \omega(I),
\]
hence $\omega(I) = 1$ and $\|\omega\| = 1$.
\end{proof}

\begin{remark}
	The set of characters of an abelian $C^*$-algebra $\alg{A}$ is called the \emph{spectrum}\index{spectrum!of an algebra} of the algebra. Let us write $X$ for this set. One can endow $X$ with a locally compact topology. The abelian Gel'fand-Naimark theorem mentioned earlier says that $\alg{A}$ is isomorphic to $C_0(X)$ for $X$ the spectrum of $\alg{A}$.
\end{remark}

\section{Quantum mechanics in the operator algebraic approach}
Quantum mechanics can be formulated in terms of $C^*$-algebras and states. This is often dubbed the algebraic approach to quantum mechanics. It can be argued that this is indeed the right framework for quantum mechanics. One of the fathers of algebraic quantum theory was von Neumann~\cite{MR0223138}, who also made many important contributions to operator algebra~\cite{MR0157873}. Indeed, he can be regarded as one of the founders of the field. Later important contributions to algebraic quantum theory were made by I.E. Segal~\cite{MR0022652}, who tried to give an abstract characterization of the postulates of quantum mechanics, and Haag and Kastler~\cite{MR0165864}, who apply the ideas of algebraic quantum mechanics to quantum field theory. As mentioned before, one of the driving forces behind the push for an abstract description of quantum mechanics comes from the wish to understand better systems with infinitely many degrees of freedom. Here we give an overview of the main points in the algebraic approach.\footnote{A proper treatment would require discussion of so-called \emph{normal states} and of von Neumann algebras. Since we will not need these concepts later on, we will not go into the details here.} A much more thorough analysis can be found in~\cite{MR2542202,MR0245275} or~\cite[Ch. 14]{morettiQM}.

Recall that in quantum mechanics, the observables are self-adjoint operators $A$ (which can in principle be unbounded). The eigenvalues or the spectrum of $A$ determine the possible outcomes after measuring the observable. A (pure) state is represented by a vector (``wave function'') $\psi$ in the Hilbert space $\mc{H}$. The expectation value of the measurement outcome is then given, according to the usual probabilistic interpretation of quantum mechanics, by the quantity $\langle \psi, A\psi \rangle$. After the measurement outcome has been determined, the state ``collapses'' in accordance with the outcome.

In the abstract setting all these properties can be defined without any reference to a Hilbert space. The observables\index{observable} are modelled by the self-adjoint elements of a $C^*$-algebra $\alg{A}$.\footnote{This excludes unbounded operators. One can argue that this is no severe objection, since in actual experiments there will always be a bounded range in which a measurement apparatus can operate. If one considers an unbounded operator on a Hilbert space, one can show that its spectral projections are contained in a von Neumann algebra (a subclass of the $C^*$-algebras). These spectral projections essentially project on the subspace of states with possible measurement outcomes in a bounded range. The advantage of working with bounded operators is that they are technically much easier to handle. Nevertheless, one can incorporate unbounded operators in the framework (by passing to the GNS representation of a state, for example). We will encounter some examples of this later on.} By extension we will sometimes call \emph{non}-self-adjoint operators observables as well (and call $\alg{A}$ the algebra of observables),\index{observable algebra} even though strictly speaking they are not. Also note that the product of two observables is not necessarily an observable again (since $A^*B^* \neq B^*A^*$ in general). Alternatively one can work with so-called \emph{Jordan algebras} where this problem does not occur. Jordan algebras are less well understood, however, and do not have many of the nice properties of $C^*$-algebras. The interested reader can learn about them in~\cite{MR887100}, for example.

The set of possible outcomes of a measurement is given by the spectrum of an observable, as defined in Section~\ref{sec:positive}. Assume for the moment that our observables are acting on some Hilbert space $\mathcal{H}$ (which can always be achieved by applying the GNS construction to a suitable state). In that case, if $A$ is a self-adjoint observable, the spectral theorem, Theorem~\ref{thm:spectral} gives us a spectral projection\index{spectral projection} $E(A;I)$ for each Borel subset $I \subset \mathbb{R}$.\footnote{We changed the notation a bit to emphasize the dependence on $A$.} and an associated spectral measure $dE(A; \lambda)$ such that
\[
		A = \int \lambda dE(A; \lambda).
\]
Now assume that the system is in a state $\omega$ and that we want to measure some (self-adjoint) observable $S$. Let $I$ be some interval. The probability that the measurement outcome lies in this interval is given by the \idx{Born rule}: it is equal to $\omega(E(S;I)^*E(S;I))$. Suppose that after measurement it is found that the value lies somewhere in an interval $I$. Let $P = E(S;I)$ be the corresponding spectral projection. Then the new state after the measurement is given by
\[
	\omega'(A) = \frac{\omega(PAP^*)}{\omega(P^*P)}
\]
provided $\omega(P^*P) \neq 0$. The latter condition can be interpreted physically as well. Note that projections correspond to ``yes/no'' experiments; they have outcome $0$ or $1$ (since $\sigma(P) = \{0,1\}$). If $\omega(P) = 0$, it means that the expectation value of this experiment is zero. But this means that the event occurs with probability zero. Hence the condition that $\omega(P^*P) \neq 0$ is very natural.

It is also possible to define \emph{transition probabilities}~\cite{MR0245275}.\index{transition probability} Recall that in quantum mechanics the probability of a transition from a state $\psi$ (which we see as a vector in some Hilbert space) to a state $\xi$ is given by $|\langle \psi, \xi \rangle|^2$. We can define a similar quantity in the algebraic setting as well. If $\omega_1$ and $\omega_2$ are two pure states on some $C^*$-algebra, the transition probability is defined as
\begin{equation}
	\label{eq:trans}
	P(\omega_1, \omega_2) = 1-\frac{1}{4} \|\omega_1 - \omega_2 \|^2.
\end{equation}
One can show that for two vector states $\psi, \xi$ this reduces to the usual quantity in quantum mechanics~\cite{MR0245275}. 

There are some aspects that we have not mentioned so far. One of them is the dynamics of the system. In Hilbert space quantum mechanics the dynamics are induced (through the Schr\"odinger equation) by a (possibly unbounded) operator $H$, the Hamiltonian. In the algebraic approach a natural approach is to look at the time evolution of the observables, $t \mapsto A(t)$. This is essentially the \idxr{Heisenberg picture}. The time evolution induces a one-parameter group of automorphisms $t \mapsto \alpha_t$ of the observable algebra. This brings us to another advantage of the algebraic approach. Symmetries of the system can be described in a natural way by groups of automorphisms acting on the observable. Symmetries have many applications in quantum mechanics, but let us mention one in particular: the algebraic approach turns out to be a convenient framework to describe spontaneous symmetry breaking in a mathematically rigorous way~\cite{MR2374090}. We will come back to symmetries in more detail later.

\subsection{Inequivalent representations}\index{representation!inequivalent|(}
As mentioned, one of the main reasons behind the algebraic approach to quantum mechanics is to study systems with infinitely many degrees of freedom. One of the phenomena that occur in such systems is that there are many inequivalent ways to represent the observables on a Hilbert space, which means that one has to make choices. Here we will discuss some consequences of the existence of such inequivalent representations. In the next Chapter we will see that they appear naturally in physical systems, so that this is not some pathological example.

We will say that two representations $\pi$ and $\rho$ of a $C^*$-algebra are \emph{unitary equivalent}\index{representation!equivalent} (or simply equivalent) if there is a unitary operator $U : \mc{H}_\pi \to \mc{H}_{\rho}$ between the corresponding Hilbert spaces and in addition we have that $U \pi(A) U^* = \rho(A)$ for all $A \in \alg{A}$.\footnote{Note that unitary equivalence already appeared in the uniqueness statement of the GNS representation.} If this is the case, we also write $\pi \cong \rho$.\index{_picongrho@$\pi \cong \rho$} In general, a $C^*$-algebra has many inequivalent representations. In quantum mechanics, this happens only for systems with infinitely many degrees of freedom. For finite systems the situation is different. For example, von Neumann~\cite{vnirred} showed that there is only one irreducible representation of the canonical commutation relations $[P,Q] = i \hbar$ (up to unitary equivalence). The same is true for a spin-1/2 system. If we consider the algebra generated by $S_x, S_y, S_z$, satisfying $[S_i, S_j] =  i \varepsilon_{ijk} S_k$ and $S_x^2 + S_y^2 + S_z^2 = \frac{3}{4} I$, each irreducible representation of these relations is unitary equivalent to the representation generated by the Pauli matrices~\cite{MR0106711}. A similar results is true for a finite number of copies of such systems. We will see that the existence of inequivalent representations has consequences for the superposition principle.

If $\pi:\alg{A} \to \alg{B}(\mc{H})$ is a representation of a $C^*$-algebra, there is an easy way to obtain different states on $\alg{A}$. Take any vector $\psi \in \mc{H}$ of norm one. Then the assignment $A \mapsto \langle \psi, \pi(A) \psi \rangle$ defines a state. Such states are called \emph{vector states}\index{vector state} for the representation $\pi$. Note that by the GNS construction it is clear that \emph{any} state can be realised as a vector state in some representation. Consider two states $\omega_1$ and $\omega_2$ that are both vector states for the \emph{same} representation $\pi$. Hence there are vectors $\psi_1$ and $\psi_2$ such that $\omega_i(A) = \langle \psi_i, \pi(A) \psi_i \rangle$ for all $A \in \alg{A}$. Consider now $\psi = \alpha \psi_1 + \beta \psi_2$ with $\alpha,\beta \in \mathbb{C}$ such that $|\alpha|^2 + |\beta|^2 = 1$ and both $\alpha$ and $\beta$ are non-zero. Then $\omega(A) = \langle \psi, \pi(A) \psi \rangle$ again is a state. However, it may be the case that the resulting state is not pure (even if the $\omega_i$ are) and we have a mixture
\begin{equation}
	\omega(A) = |\alpha|^2 \omega_1(A) + |\beta|^2 \omega_2(A).
	\label{eq:incoherent}
\end{equation}
If this is the case for any representation $\pi$ in which $\omega_1$ and $\omega_2$ can be obtained as vector states, we say that the two states are not \emph{superposable}\index{state!superposable}\index{state!coherent} or not \emph{coherent}. This situation was first analysed by Wick, Wightman and Wigner~\cite{Wick:1952p9186}. The existence of such vector states is a consequence of the existence of inequivalent representations:

\begin{theorem}[{\cite[Thm 6.1]{MR2542202}}]
Let $\omega_1$ and $\omega_2$ be pure states. Then they are not superposable if and only if the corresponding GNS representations $\pi_{\omega_1}$ and $\pi_{\omega_2}$ are inequivalent.
\end{theorem}
\begin{proof}
Consider a representation $\pi$ such that $\omega_1$ and $\omega_2$ are vector states in this representation. Write $\psi_i \in \mc{H}$ for the corresponding vectors. Then we can consider the subspaces $\mc{H}_i$ of $\mc{H}$, defined as the closure of $\pi(\alg{A}) \psi_i$. The projections on these subspaces will be denoted by $P_i$.

Note that $\psi_i$ is, by definition, cyclic for the representation $\pi(\alg{A})$ restricted to $\mc{H}_i$. Let us write $\pi_i$ for these restricted representations. But since the vectors $\psi_i$ implement the state, it follows that the representation $\pi_i$ must be (unitary equivalent to) the GNS representations $\pi_{\omega_i}$. Let $U: \mc{H}_1 \to \mc{H}_2$ be a bounded linear map such that $U \pi_1(A) = \pi_2(A) U$ for all $A \in \alg{A}$. By first taking adjoints, we then see that $U^*U \pi_1(A) = \pi_1(A) U^*U$. By irreducibility of $\pi_1$ it follows that $U^*U = \lambda I$ for some $\lambda \in \mathbb{C}$. In fact, $\lambda$ must be real since $U^*U$ is self-adjoint. A similar argument holds for $U U^*$. Hence by rescaling we can choose $U$ to be unitary, unless $U^*U = 0$. Hence non-zero maps $U$ intertwining the representations only exist if $\pi_1$ and $\pi_2$ are unitarily equivalent. 

To get back to the original setting, let $T \in \pi(\alg{A})'$. Then $P_2 T P_1$ can be identified with a map $U: \mc{H}_1 \to \mc{H}_2$ such that $U\pi_1(A) = \pi_2(A) U$. Conversely, any such map can be extended to an operator in $P_2 \pi(\alg{A})' P_1$. Hence
\[
P_2 \pi(\alg{A})' P_1 = \{ 0 \}
\]
if and only if $\pi_1$ and $\pi_2$ are not unitarily equivalent.

Since $I \in \pi(\alg{A})'$ it is clear that $P_1 P_2 = P_2 P_1 = 0$ if $\pi_1$ and $\pi_2$ are not equivalent. This implies that $\mc{H}_1$ and $\mc{H}_2$ are orthogonal subspaces of $\mc{H}$. Hence
\[
	\langle \psi_2, \pi(A) \psi_1 \rangle = 0 = \langle \psi_1, \pi(A) \rangle \psi_2\rangle.
\]
Consequently, if $\psi = \alpha \psi_1 + \beta \psi_2$ with $|\alpha|^2 + |\beta|^2 = 1$, then
\[
	\omega(A) := \langle \psi, \pi(A) \psi \rangle = |\alpha|^2 \omega_1(A) + |\beta|^2 \omega_2(A).
\]
Hence $\omega_1$ and $\omega_2$ are not superposable.

Conversely, suppose that $\pi_{\omega_1}$ and $\pi_{\omega_2}$ are unitarily equivalent. Then there must be some unitary $U$ in $\pi(\alg{A})'$ such that $P_2 U P_1 \neq 0$. This is only possible if there are vectors $\varphi_i \in \mc{H}_i$ such that $\langle \varphi_2, U \varphi_1 \rangle \neq 0$. Since $\pi(\alg{A}) \psi_1$ is dense in $\mc{H}_1$ (and similarly for $\mc{H}_2$), there must be $A_1, A_2 \in \alg{A}$ such that
\begin{equation}
	\label{eq:nonortho}
	\langle \pi(A_1) \psi_2, U \pi(A_2) \psi_1 \rangle \neq 0.
\end{equation}
Set $\varphi = U \psi_1$. Since $U$ commutes with $\pi(A)$ for every $A$, it follows that
\[
	\langle \varphi, \pi(A) \varphi \rangle = \omega_1(A).
\]
Now consider the vector $\psi = \alpha \varphi + \beta \psi_2$. This induces a state $\omega$, and we find
\begin{equation}
	\label{eq:incoherent2}
\omega(A) - |\alpha|^2 \omega_1(A) - |\beta|^2 \omega_2(A) = \overline{\alpha}\beta \langle U \psi_1, \pi(A) \psi_2 \rangle + \alpha\overline{\beta} \langle \psi_2, \pi(A) U \psi_1\rangle.
\end{equation}
Consider then $A = A_2^*A_1$, where $A_1$ and $A_2$ are as above. Then the right hand side of equation~\eqref{eq:incoherent2} becomes
\[
2 \operatorname{Re}(\alpha\overline{\beta} \langle \pi(A_1) \psi_2, U \pi(A_2) \psi_1\rangle).
\]
By choosing a suitable multiple if $A$ we can make the right hand side non-zero, because of equation~\eqref{eq:nonortho}. It follows that equation~\eqref{eq:incoherent} does not hold.
\end{proof}
This result shows that as soon as a $C^*$-algebra has inequivalent representations, there are states that are not coherent. That is, there are pure states $\omega_1$ and $\omega_2$ such that a superposition of those states is never pure. The proof also makes clear that if we have vector states corresponding to inequivalent representations, there can never be a transition from one state to the other, not even by applying any operation available in $\alg{A}$, because $\langle \psi_1, \pi(A) \psi_2 \rangle$ is zero. Such a rule that forbids such transitions is called a \idx{superselection rule}. There are many different (but strongly related) notions of a superselection rule around, see for example~\cite{MR2453623} for a discussion.

We can also compare this result with the definition of a transition probability between pure states, equation~\eqref{eq:trans}. Note that if $P(\omega,\sigma) \neq 0$, we must have that $\|\omega-\sigma\| < 2$. But one can show that if $\|\omega-\sigma\| < 2$, then necessarily the GNS representations $\pi_\omega$ and $\pi_\sigma$ must be unitarily equivalent~\cite[Cor. 10.3.8]{MR1468230}. Hence these results are consistent: there can be no transitions between inequivalent pure states.\index{transition probability}

\begin{exercise}
Let $\alg{A} = M_n(\mathbb{C})$. Show that there are no non-coherent states.
\end{exercise}
\begin{remark}
The Theorem is still true for non-pure states, only then the inequivalence condition has to be replaced with the condition that the representations are \emph{disjoint}. Two representations $\pi_1$ and $\pi_2$ are disjoint if and only if no subrepresentation of $\pi_1$ is unitarily equivalent to a subrepresentation of $\pi_2$.
\end{remark}
\index{representation!inequivalent|)}

\chapter{Infinite systems}\label{ch:infinite}
With the preparations from the previous chapter we are now in a position to discuss the main topic of interest: quantum spin systems with infinitely many sites. We will first consider an explicit example, showing that inequivalent representations appear in a natural way. Then we discuss the appropriate mathematical framework to discuss spin systems with infinitely many sites. After that we discuss how to describe dynamics in this setting, and how to define equilibrium and ground states. At the end we consider another example: Kitaev's toric code in the thermodynamic limit.

Although we only consider discrete spin systems here, most of the techniques generalise to more complicated systems. For example, instead of spin systems we can consider systems with at each site a separable Hilbert space. Or we can consider continuous systems and define, for example, creation and annihilation operators to obtain Fock space. This is the appropriate setting to discuss many-body quantum systems. Furthermore, one can consider the algebraic approach to quantum field theory~\cite{MR1405610}. The goal there is to give an axiomatic and mathematically fully rigorous description of quantum field theories. The main points are reviewed in Chapter~\ref{ch:aqft}. We will also see how ideas from algebraic quantum field theory (or \emph{local quantum physics}, as which it also is known) can be applied to spin systems with infinitely many sites.

\section{(In)equivalence of representations}\label{sec:inequiv}\index{representation!inequivalent|(}
As mentioned before, there are fundamental differences between systems with finitely and infinitely many degrees of freedom. For example, the existence of inequivalent representations does not occur for finite systems, as was already mentioned at the end of Chapter~\ref{ch:opalg}. To see that this is drastically different for infinite systems, we will now consider a simple example of a system with infinitely many sites, namely that of an infinite spin chain.\footnote{This example is adapted from Chapter 2.3 of~\cite{MR1919619}.}\index{spin chain} Later we will approach such systems more systematically, but this example already shows that going to infinitely many sites gives rise to phenomena not encountered in finite systems. 

More concretely, we imagine that the sites of our system are labelled by the integers, and that at each site there is a Pauli spin system. One can think for example of the Heisenberg model. First we define the observables. At each site of the system there is a copy of the algebra generated by the usual Pauli matrices. That is, we have operators $\sigma_n^x, \sigma_n^y$ and $\sigma_n^z$ for each $n \in \mathbb{Z}$.\index{Pauli matrices} These operators fulfil the following commutation relations:
\begin{equation}
	\label{eq:paulicom}
[\sigma_n^i, \sigma_m^j] = 2 \delta_{n,m} i \sum_{k=1}^3 \varepsilon_{ijk} \sigma_n^k,
\end{equation}
where $i,j = x,y,z$ and $\varepsilon_{ijk}$ is the completely anti-symmetric (Levi-Civita) symbol, together with the condition that $(\sigma_n^k)^2 = I$. Note in particular that observables acting on different sites commute. These operators and relations generate an algebra of observables. Later we will see how we can obtain a $C^*$-algebra from this, but for now we disregard the topological structure.

The first step in defining a representation of this algebra is to give the Hilbert space on which it acts. To this end, consider the spin in the $z$-direction at each site. The spin can be either up or down after a measurement in this basis. If we do this for each site, we get a sequence $s_n$, with $n \in \mathbb{Z}$ and $s_n = \pm 1$. We denote $S$ for the set of such sequences. To define the Hilbert space, a naive first attempt would be to let each such $s \in S$ correspond to a basis vector $\ket{s}$. But this is a \emph{huge} Hilbert space, since the set $S$ is uncountable.\footnote{This can be seen by using a Cantor diagonal argument.} This makes things much more complicated, and arguably isn't very physical. For example, the algebra generated by the operators $\sigma_n$ would act far from irreducibly on this Hilbert space. Therefore, instead we look at subsets of $S$:
\[
	S^+ := \{ s_n \in S : s_n \neq 1\textrm{ for finitely many } n \in \mathbb{Z} \}.
\]
That is, the set of all sequences $s_n$ for which only finitely many spins are not pointing upwards. Note that this set \emph{is} countable. The set $S^-$ is defined analogously. The corresponding Hilbert spaces are then $\mc{H}^{\pm} = \ell^2(S^\pm)$. It turns out that it is more convenient to work with an alternative description of these Hilbert spaces. Namely, consider the set of all functions $f: S^+ \to \mathbb{C}$ such that
\[
	\sum_{s_n \in S^+} |f(s_n)|^2 < \infty,
\]
together with the usual inner product. We have seen this alternative description before in Remark~\ref{rem:ell2}. We can identify the vector $\ket{s} \in \mc{H}^+$, given by a specific configuration $s \in S^+$, with the function $f_s(s') = \delta_{s,s'}$, where $\delta_{s,s'}$ is one if $s_n = s_n'$ for all $n$, and zero otherwise. These functions form an orthonormal basis for the Hilbert space. We can think of $f \in \mc{H}^+$ as defining a vector $\ket{\psi} = \sum_{s \in S^+} f(s) \ket{s}$.

Next we define a ``spin flip'' map on $S$. Let $n \in \mathbb{Z}$. Then $\theta_n: S \to S$ is defined by 
\[
	(\theta_n(s))_k = \begin{cases} 
		-s_n &\mbox{if } n = k \\ 
		s_k &\mbox{otherwise}. 
	\end{cases} 
\]
Note that this flips the $n$-th spin, while leaving all the other spins invariant. Clearly it can be restricted to a map $\theta_n : S^+ \to S^+$. Next we define a representation of the operators $\sigma^k_n$ as follows. Let $f \in \mc{H}^+$. Then we define
\begin{equation}
	\label{eq:spinrep}
	\begin{split}
		(\pi^+(\sigma_n^x) f)(s) &= f(\theta_n(s)), \\
		(\pi^+(\sigma_n^y) f)(s) &= i s_n f(\theta_n(s)), \\
		(\pi^+(\sigma_n^z) f)(s) &= s_n f(s), \\
		(\pi^+(I)f)(s) &= f(s).
	\end{split}
\end{equation}
Note that these operators act as one would expect from the Pauli matrices: $\pi^+(\sigma^x_n)$ flips the spin at the $n$-th site and $\pi^+(\sigma^z_n)$ multiplies with the corresponding eigenvalue of the spin at the $n$-th site. Moreover, they satisfy the same relations as in equation~\eqref{eq:paulicom}. So $\pi^+$ defines a representation of the algebra generated by the $\sigma_n^k$.

We first show that the representation is irreducible.\index{representation!irreducible} Note that by Proposition~\ref{prop:irrep} there are a number of equivalent criteria. Here we will show that any vector is cyclic for the representation. Recall that the functions $f_s(s') = \delta_{s,s'}$ form a basis of the Hilbert space. Since any $s_1,s_2 \in S^+$ differ only in a \emph{finite} number of places $n_1, \dots, n_k$, we can transform $f_{s_1}$ into $f_{s_2}$ by acting with the operator
\[
\pi^+(\sigma^x_{n_1}) \cdots \pi^+(\sigma^x_{n_k})
\]
on it. Hence if we act with a finite number of operators we can transform one basis state into another. This implies that any basis vector is cyclic. This is not quite enough yet to conclude that \emph{any} vector is cyclic. Consider a vector $\psi = \sum_{i=1}^k \lambda_k f_{s_k}$ for certain $s_k \in S^+$ and $\lambda_k \in \mathbb{C}$. Such vectors are dense in $\mc{H}^+$. Define for $n \in \mathbb{Z}$ the following operators:
\[
P_n^{\pm} = \frac{\pi^+(I) \pm \pi^+(\sigma_n^z)}{2}.
\]
Note that $P_n^+$ projects onto the subspace of vectors where the $n$-th spin is up, and similarly for $P_n^{-}$. Hence if $\psi$ is above, we can construct an operator $P$ that projects onto one of the basis vectors in the sum. If we act with this operator on $\psi$, and then subsequently with other operators as above, we can span a dense subspace. It follows that $\psi$ is cyclic and therefore that $\pi^+$ is irreducible.

\begin{exercise}
Show that $\pi^+$ is irreducible by checking that only multiples of the identity commute with the representation.
\end{exercise}

In defining the representation $\pi^+$ we started with the subset $S^+$ of $S$ where all but finitely many spins are $+1$. We might as well have started with the set $S^-$, defined in the obvious way: all but finitely many $s_n$ are equal to $-1$. Similarly as above we then define the Hilbert space $\mc{H}^-$ and a representation $\pi^-$. The representation $\pi^{-}$ can be defined as in equation~\eqref{eq:spinrep}, but note that it acts on a \emph{different} Hilbert space. Even though the representation is defined in essentially the same way, we will see that $\pi^+$ and $\pi^-$ are \emph{not} unitarily equivalent.

To see this, we will look at the \idx{polarisation} of the system. For each $N \in \mathbb{N}$, define
\[
	m_N = \frac{1}{2N+1} \sum_{n=-N}^N \sigma^z_n.
\]
It measures the average spin in the $z$-direction of the $2N+1$ sites centred around the origin. Using the definitions we see that for any $s,s' \in S^+$ we have
\[
\langle f_{s'}, \pi^+(m_N) f_s \rangle = \frac{1}{2N+1} \sum_{n=-N}^N s_n \langle f_{s'}, f_s \rangle = \frac{\delta_{s,s'}}{2N+1} \sum_{n=-N}^N s_n.
\]
Since all but finitely many $s_n$ are $+1$, this converges to one as $N \to \infty$ (or zero if $s \neq s'$). By linearity this is also true for vectors that are finite linear combinations of such base vectors. To show that it is true in general, first note that $\|\pi^+(\sigma_n^z)\| = 1$. By the triangle inequality it follows that $\|\pi^+(m_N)\| \leq 1$ for all $N$. Because of this uniform bound it follows that
\[
	\lim_{N \to \infty} \langle \psi, \pi^+(m_N) \xi\rangle = \langle \psi,\xi\rangle
\]
for all $\psi,\xi \in \mc{H}^+$. We can do a similar thing for the representation $\pi^-$ acting in $\mc{H}^-$. Here we find
\[
	\lim_{N \to \infty} \langle \psi, \pi^-(m_N) \xi\rangle = -\langle \psi,\xi\rangle.
\]
Note that we obtain a minus sign (since most of the spins are pointing downwards!).

Now suppose that the representations $\pi^+$ and $\pi^-$ are unitarily equivalent. Then by definition there is a unitary map $U: \mc{H}^+ \to \mc{H}^-$ such that $U \pi^+(A) U^* = \pi^-(A)$ for all observables $\alg{A}$. In particular, choose some $f_s \in \mc{H}^+$. It follows that
\[
	\langle f_s, \pi^+(m_N) f_s \rangle = \langle f_s, U^* \pi^-(m_N) U f_s \rangle = \langle U f_s, \pi^-(m_N) U f_s \rangle.
\]
However, the left hand side of this expression tends to $+1$, while the right hand tends to $-1$ as $N$ goes to infinity. Hence such a unitary cannot exist.

This result can be understood intuitively as follows. To go from a vector in $\mc{H}^+$ to a vector in $\mc{H}^-$, we have to flip infinitely many spins around. This does not correspond to an operator in the algebra $\alg{A}$ generated by the Pauli spin operators. That is, there is no ``physical'' operation with which we can transform one Hilbert space to the other, where a physical operation is here understood as something that can be approximated by operations on a \emph{finite} number of sites.

\begin{exercise}
Argue that there are infinitely many inequivalent representations.
\end{exercise}
\index{representation!inequivalent|)}

\section{Infinite tensor products}\index{tensor product}\index{tensor product!infinite|(}
Suppose that we have a quantum system consisting of $n$ copies of, say, a spin-1/2 system. The single system is described by the Hilbert space $\mc{H}_x = \mathbb{C}^2$. For the composite system we then have according to the rules of quantum mechanics $\mc{H} = \bigotimes_{i=1}^n \mathbb{C}^2$, the tensor product of $n$ copies of the single site space. The observables at a single site are spanned by the Pauli matrices, and hence equal to $M_2(\mathbb{C})$. For the whole system we then have $\alg{A} = \bigotimes_{i=1}^n M_2(\mathbb{C}) \cong \alg{B}(\bigotimes_{i=1}^n \mathbb{C}^2)$. 

For an abstract $C^*$-algebra it is not immediately clear what the tensor product of this algebra with another $C^*$-algebra should be, and this can be a subtle issue: there are different natural choices of a norm on the algebraic tensor product, which do not always lead to the same $C^*$-algebra. Other problems arise when we do not look at systems with finitely many sites any more, but infinitely many, for example in quantum spin chains. The example of inequivalent representations above already shows a fundamental difference between systems with finitely many and infinitely many sites. In this section we will introduce the proper framework to describe such systems.

\subsection{Infinite tensor product of Hilbert spaces}\index{tensor product!of Hilbert spaces}
One way to consider quantum spin systems with infinitely many sites is to first define a Hilbert space. Heuristically speaking this should be the tensor product of infinitely many copies of the one-site Hilbert space, say $\mathbb{C}^d$ for spin systems. Let us write $I$ for the set indexing the different sites in the system. The goal is to define $\bigotimes_{i \in I} \mc{H}_i$, where in our case $\mc{H}_i \cong \mathbb{C}^d$. We can try to generalise the tensor product construction of Section~\ref{sec:tensprod}. That is, we first consider $\psi = (\psi_i)$ and $\xi = (\xi_i)$, where $\psi_i, \xi_i \in \mc{H}_i$. The inner product between these vectors should be given by
\begin{equation}
	\label{eq:infhilbert}
\langle \psi, \xi \rangle = \prod_{i \in I} \langle \psi_i, \xi_i \rangle_{\mc{H}_i}.
\end{equation}
The problem is that the product in the right hand side in general does not converge! This can already be seen if we choose $\xi_i = -\psi_i$ and $\| \psi_i \|_{\mathcal{H}_i} = 1$.

To solve this, choose a ``reference'' vector $\Omega_i \in \mc{H}_i$ of norm one, for every $i \in I$. Then consider only those sequences $\psi = (\psi_n)$ such that $\psi_i \neq \Omega_i$ for only finitely many $i \in I$. For such vectors the right hand side of equation~\eqref{eq:infhilbert} converges since there are only finitely many factors not equal to one. We can then consider the linear space consisting of finite linear combinations of such vectors. The inner product can be extended to this space by (sesqui-)linearity. Finally, we can take the completion with respect to the norm induced by this inner product to obtain a Hilbert space $\mc{H}$.

\begin{exercise}
	Consider a one-dimensional chain of spin-1/2 systems.\index{spin chain} Choose as a reference vector $\ket{\Omega_n}$ at each site the spin-up vector in the $z$-basis. Define a unitary map $U: \bigotimes_{n \in \mathbb{Z}} \mathbb{C}^2 \to \mc{H}^+$, where the tensor product is the infinite tensor product as described in this section, and $\mc{H}^+$ is the Hilbert space defined in Section~\ref{sec:inequiv}.
\end{exercise}

Once we have defined the Hilbert space we can consider the observables of the system. One could for example take $\alg{A} = \alg{B}(\mc{H})$ as the observable algebra.\index{observable algebra} This has a downside, however: locality is lost in this approach. It is reasonable to assume that an experimenter can only perform measurements on a finite number of sites, and not on the whole (infinite) system at once. Hence the experimenter only has access to observables that can be approximated by strictly \emph{local} observables.\index{local!observable} We will later consider $C^*$-algebras whose operators can be approximated by such strictly local observables. But as said, by taking $\alg{B}(\mathcal{H})$ this locality is lost.

Alternatively, one could at a similar construction of the observable algebra as was used for the tensor product. That is, consider observables of the form $A = \bigotimes A_i$, where $A_i \in \alg{B}(\mc{H}_i)$ and $A_i = I$ for all but finitely many $i \in I$. Such observables act on $\bigotimes_{i \in I} \mc{H}_i$ in the obvious way, hence they are elements of $\alg{B}(\mathcal{H})$.\footnote{Note that the operator is well defined since only finitely many $A_i$ act non-trivial.} By taking linear combinations of such operators we obtain a $*$-algebra $\alg{A}$. To get a $C^*$-algebra, we can take the double commutant $\alg{M} = \alg{A}''$.\footnote{Here $\alg{A}'$ is the set of all $X \in \alg{B}(\mathbb{H})$ such that $XA = AX$ for all $A \in \alg{A}$.} This is a so-called \idx{von Neumann algebra}.\label{p:vna} A von Neumann algebra is in particular a $C^*$-algebra, but the converse is not true. A rather surprising result is that the resulting algebra $\alg{M}$ strongly depends on the choice of reference vectors $\Omega_i$ in the construction of the tensor product~\cite{MR0244773,Powers:1967tx}! Hence in this approach one is faced with the decision of which sequence of vectors to take. The right choice is not at all obvious in general.

\subsection{Tensor products of $C^*$-algebras}\index{tensor product!of C*-algebras@of $C^*$-algebras}
The (algebraic) tensor product of two $*$-algebras is defined as usual. We recall the construction here. Suppose that $\alg{A}$ and $\alg{B}$ are $*$-algebras. Then the \idx{algebraic tensor product}\index{tensor product!algebraic} $\alg{A} \odot \alg{B}$ consists of linear combinations of elements $A \otimes B$ with $A \in \alg{A}$ and $B \in \alg{B}$, where the following identifications are made:
\[
\begin{split}
	(\lambda A) \otimes B = A \otimes & (\lambda B) = \lambda (A \otimes B), \quad \lambda \in \mathbb{C} \\
	(A_1 + A_2) \otimes B = A_1 \otimes B + A_2 &\otimes B, \quad A \otimes (B_1+B_2) = A \otimes B_1 + A \otimes B_2.
\end{split}
\]
These conditions say that there is a bilinear map $L: \alg{A} \times \alg{B} \to \alg{A} \odot \alg{B}$, defined by $L(A,B) = A \otimes B$. The tensor product becomes a $*$-algebra by setting
\begin{align*}
	&(A_1 \otimes B_1)(A_2 \otimes B_2) := A_1 A_2 \otimes B_1 B_2,\\
	&(A \otimes B)^* := A^* \otimes B^*.
\end{align*}
To obtain a $C^*$-algebra we have to define a norm on the algebraic tensor product, satisfying the $C^*$-property. The completion of the algebraic tensor with respect to this norm is a $C^*$-algebra, which is called the tensor product of $\alg{A}$ and $\alg{B}$ (with respect to the chosen norm).\footnote{\label{fn:tensnorm}Defining the tensor product $\alg{A} \otimes \alg{B}$ of two $C^*$-algebras is a delicate matter in general. The point is that there are, in general, different natural $C^*$-norms on the algebraic tensor product $\alg{A} \odot \alg{B}$. These different norms lead to different completions, and hence different tensor products. There is a special class of $C^*$-algebras for which there is a unique norm on the tensor product, the \emph{nuclear} algebras. If $\alg{A}$ is nuclear, then there is a unique $C^*$-tensor product $\alg{A} \otimes \alg{B}$ for any $C^*$-algebra $\alg{B}$. Examples of nuclear algebras are $M_n(\mathbb{C})$, the algebra of bounded operators on a Hilbert space, and abelian algebras. (Almost) all $C^*$-algebras we encounter will be nuclear.\index{C*algebra@$C^*$-algebra!nuclear}}

In general there are different norms that one can choose for the algebraic tensor product, which lead to different completions (and hence distinct $C^*$-algebras). There are some natural conditions, however. Analogous to the tensor product of bounded linear maps, defined before, it is desirable that the norm on the algebraic tensor product is a \idx{cross norm}. That is, for $A \in \alg{A}$ and $B \in \alg{B}$, we want that
\[
	\| A \otimes B \| = \|A\| \|B\|,
\]
where the norms on the right hand side are the norms of $\alg{A}$ and $\alg{B}$ respectively. One way to obtain such a norm is to take faithful representations $\pi$ and $\rho$ of $\alg{A}$ and $\alg{B}$ respectively. Then we define
\[
\left\| \sum_{i=1}^n A_i \otimes B_i \right\| = \left\| \sum_{i=1}^n \pi(A_i) \otimes \rho(B_i) \right\|.
\]
The norm on the right hand side is to be understood as the norm of the bounded linear operator on $\mc{H}_\pi \otimes \mc{H}_\rho$. This defines a norm which one can show is independent of the choice of faithful representation. The completion of $\alg{A} \odot \alg{B}$ with respect to this norm is called the \emph{minimal} or \emph{spatial} tensor product.\index{tensor product!spatial}\index{tensor product!minimal} As mentioned in Footnote~\ref{fn:tensnorm}, there may be other cross norms on $\alg{A} \odot \alg{B}$. For the applications we consider, however, this will not be the case and we can safely write $\alg{A} \otimes \alg{B}$ without creating confusion about which norm to use.
\index{tensor product!infinite|)}

\subsection{Quasi-local algebras}\label{sec:quasilocal}\index{quasi-local algebra|(}
With the help of the tensor product defined above we can now define the observable algebra for the systems we are interested in. The observable algebra generated by the $\sigma_n^i$ defined in the first section of this chapter is an example. Here we will give a general construction of such algebras and show how we can define a norm to obtain a $C^*$-algebra. Again, we restrict to the case of quantum spin systems. The theory is developed in full generality in, e.g.,~\cite{MR887100,MR1441540}.

The starting point is a set $\Gamma$ labelling the sites of the system. We will always assume that $\Gamma$ is countable. In many cases $\Gamma$ will have additional structure, for example if we take $\Gamma = \mathbb{Z}^d$. In that case there is a natural action of the translation group on the set of sites. This will induce an action of the  group of translations on the observable algebra, and allows us to talk about translation invariant systems. An important theme will always be \idx{locality}. This is concerned with operators acting only on a subset of the sites.

\begin{definition}\index{_Pf(Gamma)@$\mc{P}_f(\Gamma)$}
	Let $\Gamma$ be as above. We will write $\mc{P}(\Gamma)$ for the set of all subsets of $\Gamma$. Similarly, $\mc{P}_f(\Gamma)$ is the subset of all \emph{finite} subsets of $\Gamma$.
\end{definition}

For simplicity we will assume that at each site $x \in \Gamma$ there is a $d$-dimensional quantum spin system with observable algebra $\alg{A}(\{x\}) := M_d(\mathbb{C})$, where $d$ is independent of the site $x$. If $\Lambda \in \mc{P}_f(\Gamma)$ is a finite collection of sites, the corresponding algebra of observables is given by
\[
\alg{A}(\Lambda) = \bigotimes_{x \in \Lambda} \alg{A}(\{x\}) = \bigotimes_{x \in \Lambda} M_d(\mathbb{C}).
\]
Note that this is the usual construction if one takes $n$ copies of a quantum mechanical system. The $C^*$-algebra $\alg{A}(\Lambda)$ is understood as the algebra generated by all observables acting only on the sites in $\Lambda$ (and trivially elsewhere on the system).

In this way we obtain in a natural way a \emph{local net}\idx{local!net}. The net structure is as follows. If $\Lambda_1 \subset \Lambda_2$ with $\Lambda_2 \in \mc{P}_f(\Gamma)$, then there is a inclusion of the corresponding algebras. If $A_1 \otimes \cdots \otimes A_n \in \alg{A}(\Lambda_1)$ then $A_1 \otimes \cdots \otimes A_n \otimes I \cdots \otimes I$ is in $\alg{A}(\Lambda_2)$, where we inserted copies of the identity operator acting on the sites of $\Lambda_2\setminus\Lambda_1$. The assignment $\Lambda \mapsto \alg{A}(\Lambda)$ is \emph{local} in the following sense. If $\Lambda_1,\Lambda_2 \in \mc{P}_f(\Gamma)$ and $\Lambda_1 \cap \Lambda_2 = \emptyset$, then we have
\[
	[\alg{A}(\Lambda_1), \alg{A}(\Lambda_2)] = \{0\},
\]
where we understand $\alg{A}(\Lambda_i)$ to be embedded into a sufficiently large $\alg{A}(\Lambda)$ containing both $\Lambda_1$ and $\Lambda_2$, so that it makes sense to talk about the commutator.

The set of \emph{local} or \emph{strictly local}\index{observable!local} observables is then defined by
\[
\alg{A}_{loc} = \bigcup_{\Lambda \in \mc{P}_f(\Gamma)} \alg{A}(\Lambda).
\]
Note that $\alg{A}_{loc}$ is a $*$-algebra. There is a $C^*$-norm on this algebra, given by the norms on each $\alg{A}(\Lambda)$. The completion of $\alg{A}_{loc}$ with respect to this norm is the \idx{quasi-local algebra} $\alg{A}$.\footnote{In this construction we glossed over some mathematical technicalities. Essentially, one constructs the \emph{inductive limit} of a net of algebras. The inductive limit construction also works for non-trivial (but compatible with the net structure) inclusions of algebras. See e.g.~\cite{MR1743202,MR0074800} for more details.} This algebra can be interpreted as consisting of all those observables that can be approximated arbitrarily well (in norm) by observables in \emph{finite} regions of space. As we have mentioned before, the physical idea is that an experimenter has only access to a bounded region of the system.

If $A \in \alg{A}(\Lambda)$ for some $\Lambda \in \mc{P}_f(\Gamma)$ we say that $A$ is \emph{localised} in $\Lambda$. The smallest such $\Lambda$ such that $A \in \alg{A}(\Lambda)$ is called the \idx{support} of $A$ and denoted by $\operatorname{supp}(A)$.\index{_supp(A)@$\operatorname{supp}(A)$} For convenience we will set $\alg{A}(\emptyset) = \mathbb{C} I$, since multiples of the identity are contained in $\alg{A}(\Lambda)$ for all $\Lambda \in \mc{P}_f(\Gamma)$. Analogously we can talk about observables localised in \emph{infinite} regions. Let $\Lambda \subset \Gamma$, where $\Lambda$ is not necessarily finite. Then we can define
\[
\alg{A}(\Lambda) = \overline{\bigcup_{\Lambda_f \in \mc{P}_f(\Lambda)} \alg{A}(\Lambda_f) }^{\|\cdot\|}.
\]\index{_A(Lambda)@$\mathfrak{A}(\Lambda)$}
The bar means that we have to take the closure with respect to the norm, to obtain a $C^*$-algebra. Note that $\alg{A}(\Lambda)$ can be embedded into $\alg{A}$ in a natural way. We will always regard $\alg{A}(\Lambda)$ as a subalgebra of $\alg{A}$. Again, the interpretation of operators in $\alg{A}(\Lambda)$ is that they can be approximated arbitrarily well by strictly local operators localised in $\Lambda$. Particularly important with respect to duality are \emph{complements} of $\Lambda \in \mc{P}(\Gamma)$. We will denote the complement by $\Lambda^c$. The locality condition can be extended to the corresponding algebras, by continuity. That is, we have
\[
	[\alg{A}(\Lambda), \alg{A}(\Lambda^c)] = \{0\}
\]
for any $\Lambda \in \mc{P}(\Gamma)$.

\begin{exercise}\label{ex:alghaagdual}
	Consider a quantum spin system defined on $\mathbb{Z}^d$. Suppose $\Lambda \subset \mathbb{Z}^d$ is finite. In this exercise we will prove that $\alg{A}(\Lambda)^c = \alg{A}(\Lambda^c)$, where $\alg{A}(\Lambda)^c := \alg{A}(\Lambda)' \cap A = \alg{A}(\Lambda')$, and
\[
	\alg{A}(\Lambda^c)=\overline{\bigcup_{\Lambda_1\subset \mathbb{Z}^d\backslash \Lambda \textrm{ finite }} \alg{A}(\Lambda_1)}.
\] 
The bar means that we take the norm-completion of the algebra. Recall that $\alg{A}(\Lambda)$ is isomorphic to $M_n(\mathbb{C})$ for a suitable $n$. Hence there are $E_{ij} \in \alg{A}(\Lambda)$ such that $E_{ij}$ corresponds to the matrix with a one in the $(i,j)$ position and zero elsewhere. Define the following map, for $A \in \alg{A}$:
		\[
			\Tr_{\Lambda}(A) = \sum_{i,j=1}^n E_{ij} A E_{ji}.
		\]
This is the \idx{partial trace} over $\Lambda$.
\begin{enumerate}
	\item Let $\Lambda_f \supset \Lambda$ be finite. Show that
		\[
			\Tr_{\Lambda}(A \otimes B) = \Tr(A)B
		\]
		for all $A \in \alg{A}(\Lambda)$ and $B \in \alg{A}(\Lambda_f \setminus \Lambda)$. Here $\Tr(A)$ is the usual trace.
	\item Prove that $\operatorname{Tr}_\Lambda(\alg{A}) \subset \alg{A}(\Lambda^c)$.
	\item Use the map $\operatorname{Tr}_\Lambda$ to prove the claim.
\end{enumerate}
\end{exercise}

\begin{remark}
	The statement in the theorem can be extended to infinite regions, i.e. to the case where $\Lambda$ is not necessarily finite (cf. for example~\cite[Sect. 2.2]{MR1463825}). We not however that this is a result on the level of $C^*$-algebras. For example, if we consider an irreducible representation $\pi$ of $\alg{A}$, we might expect that $\pi(\alg{A}(\Lambda))'' = \pi(\alg{A}(\Lambda^c))'$, where we take the double commutant to obtain a von Neumann algebra (since the right-hand side always is one). However, this does \emph{not} need to be true if the region $\Lambda$ is infinite. If it is true, we say that \idx{Haag duality} holds, where usually it is required to hold only for regions $\Lambda$ of a certain shape, for example halflines. Haag duality plays an important role in the study of superselection sectors in local quantum physics, as we will discuss in the next chapter. Some results on Haag duality in specific models can be found in~\cite{haagdouble,MR2605849,MR2281418}.
\end{remark}

The construction so far is very general. In fact we see that it only depends on the dimensions of the local spin systems. Hence the algebra of observables alone does not contain much information. Rather, we need to consider additional concepts such as dynamics and states for these algebras. We will address this shortly.
\index{quasi-local algebra|)}

\subsection{Simplicity of the quasi-local algebra}
A closed two-sided ideal $\alg{I}$ (or simply an ideal\index{ideal}) of a $C^*$-algebra $\alg{A}$ is a closed subspace $\alg{I} \subset \alg{A}$ such that $AB$ and $BA$ are in $\alg{I}$ for all $A \in \alg{A}$ and $B \in \alg{I}$. A $C^*$-algebra is called \emph{simple}\index{C*algebra@$C^*$-algebra!simple} if its only closed two-sided ideals are $\alg{I} = \{0\}$ and $\alg{I} = \alg{A}$.

\begin{exercise}
	\label{ex:matsimple}
Let $\alg{A} = M_d(\mathbb{C})$. Show that $\alg{A}$ is simple. \emph{Hint:} show first that if $\alg{I}$ is a non-trivial ideal, then it contains matrices $E_{ii}$ which are zero everywhere except on the $i$-th place on the diagonal, where it is one.
\end{exercise}

The goal is to show that the quasi-local algebra defined above is simple. To this end, we first need to introduce the \idx{quotient algebra} $\alg{A}/\alg{J}$, where $\alg{J}$ is a closed two-sided ideal. As a vector space, it is the quotient of $\alg{A}$ by the vector space $\alg{J}$, hence it consists of equivalence classes $[A]$, $A \in \alg{A}$ with $[A] = [B]$ if and only if $A = B+J$ for some $J \in \alg{J}$. This can be turned into a Banach space by defining the norm as
\[
	\| [A] \| := \inf_{J \in \alg{J}} \| A + J \|.
\]
That this really defines a Banach space requires some work. The details can be found in, e.g.~\cite{MR719020}. The space can be made into a $*$-algebra by setting
\[
	[A]\cdot[B] = [AB], \quad [A]^* = [A^*].
\]
Because $\alg{J}$ is a closed two-sided ideal (in particular this implies that $\alg{J}^* = \alg{J}$), this is well defined. It turns out that $\alg{A}/\alg{J}$ is a $C^*$-algebra. That is, the $C^*$-identity holds for this $*$-operation and the norm defined above.

There is a natural $*$-homomorphism $\varphi$ from $\alg{A}$ into $\alg{A}/\alg{J}$: define $\varphi(A) = [A]$. Since it is a $*$-homomorphism between $C^*$-algebras it is automatically continuous by Remark~\ref{rem:morcont}. Note that if $\alg{J}$ is a closed two-sided ideal, we can always find a representation of $\alg{A}$ which has this ideal as its kernel. To this end, take any faithful representation $\pi$ of $\alg{A}/\alg{J}$, which always exists by Theorem~\ref{thm:gelfandnaimark}. Then $\pi \circ \varphi$ is a representation of $\alg{A}$ with kernel $\alg{J}$.

\begin{proposition}
	The quasi-local algebra $\alg{A}(\mathbb{Z}^d)$ for a quantum spin system is simple.\index{quasi-local algebra!simplicity}
\end{proposition}
\begin{proof}
	Suppose that $\alg{I}$ is an ideal of $\alg{A}(\mathbb{Z}^d)$. If $\Lambda \in \mc{P}_f(\mathbb{Z}^d)$, it is easy to check that $\alg{I} \cap \alg{A}(\Lambda)$ is an ideal of $\alg{A}(\Lambda)$. Hence the intersection is either $\{0\}$ or $\alg{A}(\Lambda)$, by Exercise~\ref{ex:matsimple} and because $\bigotimes_{i=1}^n M_d(\mathbb{C}) \cong M_{d^n}(\mathbb{C})$. In the latter case, it contains the identity and hence $\alg{I} = \alg{A}(\mathbb{Z}^d)$ and we are done. We may therefore assume that $\alg{I} \cap \alg{A}(\Lambda) = \{0\}$ for all $\Lambda \in \mc{P}_f(\mathbb{Z}^d)$.

	Let $\pi \circ \varphi$ be a representation of $\alg{A}$ with kernel $\alg{J}$ as constructed above. Note that $\alg{A}(\Lambda) \cap \alg{J} = \{0\}$ by the argument above, for all $\Lambda \in \mc{P}_f(\Lambda)$. This implies that $\pi$ restricted to $\alg{A}_{loc}$ is injective. Hence $\|\pi(\varphi(A))\| = \|A\|$ for all $A \in \alg{A}_{loc}$. In particular, $\|\varphi(A)\| = \|A\|$. Hence, by definition of the norm on $\alg{A}/\alg{J}$, for each non-zero $A \in \alg{A}_{loc}$ it follows that
\[
	\|A-J\| \geq \|A\|
\]
for all $J \in \alg{J}$. But if $J \in \alg{J}$, there is a sequence $A_n \in \alg{A}_{loc}$ such that $A_n \to J$ in norm. Since $\|A_n - J\| \geq \|A_n\|$ as above and the left-hand side converges to zero, it follows that $\|A_n\|$ must converge to zero. Hence it follows that $J$ must be zero. 
\end{proof}

Note that the kernel of a representation of a $C^*$-algebra is always a closed. We thus have the following very useful corollary, where the isometric property follows from a general result on injective $*$-homomorphisms between $C^*$-algebras.
\begin{corollary}
	Every non-zero representation of the quasi-local algebra is faithful (and hence isometric).\index{representation!faithful}
\end{corollary}
This says that for any representation $\pi$ of the quasi-local algebra $\alg{A}$ we take, we can identify $\alg{A}$ with its image $\pi(\alg{A})$. In this case it is not necessary to do the whole Gel'fand-Naimark construction as in the proof of Theorem~\ref{thm:gelfandnaimark}: it suffices to take any state of $\alg{A}$ and do the GNS representation, to identify $\alg{A}$ with a closed subalgebra of $\alg{B}(\mc{H})$ for some Hilbert space $\mc{H}$ as the image of $\alg{A}$ under this representation.

\subsection{Inequivalent representations}
As we have seen, there are many inequivalent representations of the observable algebra of quantum spin systems. Since representations are often obtained by applying the GNS construction to some state on an algebra, it is useful to have an criterion for unitary equivalence of the corresponding representations only in terms of the states. For states on a quasi-local algebra there is such an explicit criterion.
\begin{proposition}\label{prop:stateinequiv}
	Let $\alg{A} := \alg{A}(\Gamma)$ be the quasi-local observable algebra of some spin system and suppose that $\omega_1$ and $\omega_2$ are pure states on $\alg{A}$. Then the following criteria are equivalent:\index{representation!inequivalent}\index{state!pure}
\begin{enumerate}
	\item The corresponding GNS representations $\pi_1$ and $\pi_2$ are equivalent.
	\item For each $\varepsilon > 0$, there is a $\Lambda_\varepsilon \in \mc{P}_f(\Gamma)$ such that
		\[
			| \omega_1(A) - \omega_2(A) | < \varepsilon \| A \|,
		\]
		for all $A \in \alg{A}(\Lambda)$ with $\Lambda \in \mc{P}_f(\Lambda_\varepsilon^c)$.
\end{enumerate}
Here $\Lambda^c_\varepsilon$ is the complement of $\Lambda_{\varepsilon}$ in $\Gamma$.
\end{proposition}
The second condition says that for every $\varepsilon > 0$, the states $\omega_1$ and $\omega_2$ are close to each other, as long as we restrict to observables outside some (fixed!) region depending on $\varepsilon$. In other words, if we restrict to observables ``far away'', the states look the same. We will omit the proof. The statement is a special case of~\cite[Cor. 2.6.11]{MR887100}.\footnote{The statement there is more general, but for pure states, quasi-equivalence is the same as unitary equivalence. Moreover, states of the quasi-local algebra of quantum spin systems are always locally normal (a term we will not define here), because the observable algebras of finite regions are finite dimensional in that case.}

This result can be used to give an alternative proof of the inequivalence of the representations $\pi^+$ and $\pi^-$ defined at the beginning of this chapter. To see this, first note that the generators $\sigma_n^i$ satisfy the relations of the Pauli matrices. Hence the operators acting on a single site $n$ generate the observable algebra $\alg{A}(\{n\}) = M_2(\mathbb{C})$. We can then define the corresponding quasi-local algebra $\alg{A}(\mathbb{Z})$ as above. We can then define states $\omega_+$ and $\omega_-$ on $\alg{A}$ and show that the corresponding GNS representations are $\pi^+$ and $\pi^-$. To define these states, it is enough to first define them for \emph{strictly} local observables $A \in \alg{A}(\Lambda)$, and show that this can be done in a consistent way. That is, if $A \in \alg{A}(\Lambda)$ with $\Lambda$ finite, then $A$ can be identified with a matrix in $A \in M_2(\mathbb{C}) \otimes \dots \otimes M_2(\mathbb{C})$. We can then define a state on $\alg{A}(\Lambda)$ by setting $\omega_\Lambda(A) = \Tr(\rho_\Lambda A)$ for some density operator $A$. The consistency conditions means that if $\Lambda \subset \Lambda'$, then we should have $\Tr(\rho_\Lambda A) = \Tr(\rho_{\Lambda'} A \otimes I \otimes \cdots \otimes I)$. That is, if we include $\alg{A}(\Lambda)$ into a bigger algebra, the value of the state should not change. Since $\alg{A}_{loc}$ is dense in $\alg{A}$, we can extend the state defined in this way on the local algebra to a state $\omega_+$ on $\alg{A}$. Since we already know that these representations are irreducible, it follows that the states must be pure. We can then use the proposition to prove inequivalence of the representations.
\begin{exercise}
Find suitable states $\omega_+$ and $\omega_-$ and use the proposition to verify that the representations $\pi^+$ and $\pi^-$ are indeed not equivalent.
\end{exercise}

\subsection{Translation symmetry}\index{translation symmetry}\index{symmetry!translation}
In many cases there is a natural translation symmetry acting on the system. The example that will be most relevant for us is when $\Gamma = \mathbb{Z}^d$ for some positive integer $d$. Note that $\mathbb{Z}^d$ acts on itself by addition (or translation). If $\Lambda \subset \Gamma$ we write $\Lambda + x$ for the same subset, translated by $x \in \mathbb{Z}^d$. This induces an \idx{action} on the quasi-local algebra $\alg{A}(\mathbb{Z}^d)$. Assume first for simplicity that $A(x) \in \alg{A}(\{x\}) \cong M_n(\mathbb{C})$ for some $x \in \mathbb{Z}^d$. Hence the support of $A(x)$ is only the site $x$. Since by construction the observable algebra is the same at any site, there is a corresponding $A(x+y) \in \alg{A}(\{x+y\})$ for every $y \in \mathbb{Z}^d$. We define $\tau_y(A(x)) = A(x+y)$. This construction works for any local operator $A \in \alg{A}(\Lambda)$. Hence for every $y \in \mathbb{Z}^d$ we can define a map $\tau_y : \alg{A}_{loc} \to \alg{A}_{loc}$. Clearly $\tau_y$ is an automorphism of the local algebra, hence $\|\tau_y(A)\| = \|A\|$ for all $A \in \alg{A}_{loc}$. We can then extend $\tau_y$ to an automorphism of $\alg{A}(\mathbb{Z}^d)$, by continuity and since $\alg{A}_{loc}$ is dense in $\alg{A}(\mathbb{Z}^d)$. The automorphism has the property that $\tau_y(\alg{A}(\Lambda)) = \alg{A}(\Lambda+y)$ for all $\Lambda \in \mc{P}(\mathbb{Z}^d)$.

This construction can be done for any $y \in \mathbb{Z}^d$. Moreover, it is easy to check that $\tau_{x+y}(A) = \tau_x(\tau_y(A))$ and that $\tau_0(A) = A$ for all $A \in \alg{A}$. Hence $\tau: \mathbb{Z}^d \to \Aut(\alg{A}(\mathbb{Z}^d))$ is a group homomorphism into the group of automorphisms. A translation invariant state on $\alg{A}(\mathbb{Z}^d)$ can then be defined as a state $\omega$ such that $\omega(\tau_x(A)) = \omega(A)$ for all $x \in \mathbb{Z}^d$ and $A \in \alg{A}$.\index{state!invariant} As an example we can consider the infinite spin chain discussed at the beginning of this chapter. If we choose all spins in the up direction, this gives a vector in $\mc{H}^+$ and hence a state on the observable algebra. Clearly this state is translation invariant. The uniqueness of the GNS construction then implies that there is a unitary representation $x \mapsto U(x)$ implementing the translations in the GNS representations. That is, we have $\pi(\alpha_x(A)) = U(x)\pi(A)U(x)^*$ for all $x \in \mathbb{Z}^d$ and $A \in \alg{A}$, cf. Exercise~\ref{ex:invstategns}.

Another useful property is that if we move any local operator away far enough, it will commute with any other local operator. In fact we can say a bit more:
\begin{theorem}
Consider the quasi-local algebra $\alg{A}(\mathbb{Z}^d)$ with $x \mapsto \tau_x$ the natural action of the translation group. Then for each $A,B \in \alg{A}(\mathbb{Z}^d)$ we have
\begin{equation}
	\label{eq:asympab}
	\lim_{|x| \to \infty} \| [\tau_x(A), B] \| = 0.
\end{equation}
Here $|x| \to \infty$ means that $x$ goes to infinity in any direction.
\end{theorem}
\begin{proof}
We first suppose that $A$ and $B$ are strictly local. In that case, there are $\Lambda_A, \Lambda_B \in \mc{P}_f(\mathbb{Z}^d)$ such that $A \in \alg{A}(\Lambda_A)$ and $B \in \alg{A}(\Lambda_B)$. Note that $\tau_x(A) \in \alg{A}(\Lambda_A+x)$. Since both $\Lambda_A$ and $\Lambda_B$ are finite, there is some integer $N$ such that $(\Lambda_A+x) \cap \Lambda_B = \emptyset$ for all $|x| > N$. Hence by locality we have $[\tau_x(A), B] = 0$ for all $|x| > N$ and hence equation~\eqref{eq:asympab} holds.

As for the general case, let $A,B \in \alg{A}$ and let $\varepsilon > 0$. Since the local algebra is dense in $\alg{A}$, there are $A_\varepsilon$ and $B_\varepsilon$ that are strictly local
 and such that $\|A-A_\varepsilon\| < \varepsilon$ (and similarly for $B$). By the argument above, $[\tau_x(A_\varepsilon), B_\varepsilon] = 0$ for $x$ large enough. In addition we have that $\|\tau_x(A)\| = \|A\|$ since $\tau_x$ is an automorphism for each $x$. Hence we have
\[
\begin{split}
	\lim_{|x| \to \infty} \|[\tau_x(A), B]\| &= \lim_{|x|\to \infty} \| [\tau_x(A-A_\varepsilon), B] + [\tau_x(A_\varepsilon),B] \| \\
		& \leq 2 \varepsilon \|B\| + \lim_{|x|\to \infty} \| [\tau_x(A_\epsilon), B-B_\varepsilon] + [\tau_x(A_\varepsilon), B_\varepsilon] \|  \\
		& \leq 2 \varepsilon (\|A_\varepsilon\| + \|B\|).
\end{split}
\]
To obtain the last line we used that $[\tau_x(A_\varepsilon), B_\varepsilon] = 0$ for $x$ sufficiently large. The claim follows by noting that $\|A_\varepsilon\| \leq \|A\| + \varepsilon$ by the reverse triangle inequality.
\end{proof}
The property in the proof is called \idx{asymptotic abelianness}. The first applications of asymptotic abelianness were studied in~\cite{MR0205096,MR0203510}. One useful consequence is that the set of translation invariant states is actually a \emph{simplex}. This means that any translational invariant states can be obtained \emph{uniquely} as a combination of extremal translational invariant states.

\section{Dynamics}
So far the discussion has been completely general, in that we have only treated kinematical aspects up to now. For example, the quasi-local algebra essentially depends only on the number of degrees of freedom at each site. What differentiates different systems of, say, spin-1/2 particles are their \idx{dynamics}. In other words, how the observables evolve over time. One way to specify the dynamics is to specify Hamiltonian.

Recall that a Hamiltonian\index{Hamiltonian} is a (possibly unbounded) self-adjoint operator acting on some Hilbert space $\mc{H}$.\footnote{We gloss over details such as the domain of definition here.} If $\psi \in \mc{H}$ is a wave function, its time evolution is governed by the Schr\"odinger equation (in units of $\hbar = 1$)
\[
	i \partial_t \psi = H \psi.
\]
The solution to this equation is given by $\psi_t = \exp(-i t H) \psi$, where $\exp(-i t H)$ is given by the usual power series expansion
\[
\exp(-i t H) = \sum_{k=0}^{\infty} \frac{(-i t H)^k}{k!}.
\]
This converges (under suitable conditions on $H$) when acting on so-called \emph{analytic} vectors $\psi$, which span a dense set. It might not always be easy to identify a set of analytic vectors for which the power expansion converges (for example, one needs to make sense of $H^2 \psi$ for such $\psi$), so alternatively one can apply the spectral calculus for unbounded operators as explained in Remark~\ref{rem:1paramexp}.

This allows us to define a unitary $U(t) = \exp(i t H)$ in $\alg{B}(\mc{H})$.\footnote{Note the \emph{positive} sign in the exponent. The choice of sign is just a matter of convention. We choose the sign in such a way that the time evolution of an observable $A$ will be $\alpha_t(A) = U(t) A U(t)^*$. For the opposite choice this would lead to $U(t)^* A U(t)$.} Hence for the time evolved vector we have $\psi_t = U(t)^* \psi$. From the power series expansion and the self-adjointness of $H$ it follows that $U(t)^* = U(-t)$, $U(0) = I$ and $U(t+s) = U(t)U(s)$. The latter condition says that evolving over a time $t+s$ is the same as first evolving for a time $s$ and then for a time $t$. In short, we have a one-parameter group $t \mapsto U(t)$ of unitaries. Moreover, one can show that this group is strongly continuous, in the sense that\index{strongly continuous group}
\[
	\lim_{t \to 0} \| U(t) \psi - \psi \| = 0
\]
for any vector $\psi \in \mc{H}$. Hence the Hamiltonian induces a group of unitaries that implement the time evolution on the vectors in the Hilbert space. A famous result by Stone shows that in fact \emph{any} such evolution can be obtained from a ``Hamiltonian''. The result can be stated as follows:

\begin{theorem}[Stone]\label{thm:stone}\index{Stone's theorem}
Let $t \mapsto U(t)$ by a strongly continuous one-parameter group of unitaries acting on a Hilbert space $\mc{H}$. Then there is a densely defined self-adjoint operator $H$ acting on $\mc{H}$ such that $U(t) = \exp(i t H)$ for  all $t \in \mathbb{R}$.
\end{theorem}

We will not give a proof here. It can be found, for example, in~\cite[Thm. 5.6.36]{MR719020}.

Instead of looking at how the states evolve, we can also look at how the observables change (the \idx{Heisenberg picture}). Let $t \mapsto U(t)$ be a one-parameter group of unitaries as above. Let $\alg{A} = \alg{B}(\mc{H})$. Then we obtain a one-parameter group of automorphism $t \mapsto \alpha_t \in \Aut(\alg{B}(\mc{H}))$ by setting
\[
	\alpha_t(A) = U(t) A U(t)^*.
\]
Both viewpoints give the same physics. Indeed, for the expectation values of an observable $A$ we have at a time $t$
\[
\langle A \rangle_t = \langle \psi_t, A \psi_t \rangle = \langle U(t)^* \psi, A U(t)^* \psi \rangle = \langle \psi, \alpha_t(A) \psi \rangle.
\]
Hence the Schr\"odinger and the Heisenberg picture give the same results.\index{Heisenberg picture}

Ultimately we want to find a description of dynamics in the abstract setting of $C^*$-algebras. If $\alg{A} = \alg{B}(\mc{H})$ this will not lead to much new. Note in particular that $U(t) \in \alg{A}$. Automorphisms of this form are called \emph{inner}.\index{automorphism!inner} An automorphism of the algebra of bounded operators on a Hilbert space is always inner, as the next exercise shows, but in general this is \emph{not} true.
\begin{exercise}
Let $\mc{H}$ be a Hilbert space and consider the $C^*$-algebra $\alg{B}(\mc{H})$. Suppose that $\psi \in \mc{H}$ is a unit vector. We define the projection on $\psi$, which we write as $\ket{\psi}\bra{\psi}$, by
\[
	\ket{\psi}\bra{\psi} \xi := \langle \psi, \xi \rangle \psi,
\]
for $\psi \in \mc{H}$. Note that any one-dimensional projection is of this form.
\begin{enumerate}
	\item Show that a projection $P \in \alg{B}(\mc{H})$ has one-dimensional range if and only if $PAP = c_A P$ for any $A \in \alg{B}(\mc{H})$ and $c_A \in \mathbb{C}$ may depend on $A$.
	\item Let $\alpha$ be an automorphism of $\alg{A}$. Show that $\alpha(\ket{\psi}\bra{\psi})$ is a one-dimensional projection.
	\item Define a map $U : \alg{A} \psi \to \mc{H}$ by $U A \psi = \alpha(A) \psi'$, where $\psi'$ is such that $\alpha(\ket{\psi}\bra{\psi}) = \ket{\psi'}\bra{\psi'}$. Show that $U$ is well-defined and extends to a unitary in $\alg{B}(\mc{H})$.
	\item Show that $\alpha(A) = UAU^*$ for all $A \in \alg{B}(\mc{H})$.
\end{enumerate}
\end{exercise}
Suppose that $\alpha_t$ is a strongly continuous one-parameter group of automorphism on $\alg{A} = \alg{B}(\mathcal{H})$. Then the exercise above shows that for each $t \in \mathbb{R}$, there is unitary $U_t$ such that $\alpha_t(A) = U_t A U_t^*$ for all $A \in \alg{B}(\mathcal{H})$.

\begin{exercise}
	Show that we can choose $U_t$ in such a way that $U_0 = I$ and for all $t,s \in \mathbb{R}$, we have
	\begin{equation}\label{eq:projective}
		U_{s+t} = c(s,t) U_s U_t,
	\end{equation}
	where $c(t,s)$ is a \emph{2-cocycle} of $\mathbb{R}$. That is, it is a map $c: \mathbb{R} \times \mathbb{R} \to \mathbb{T}$ which satisfies
	\[
		c(r, s+t) c(s,t) = c(r,s) c(r+s, t).
	\]
	\emph{Hint:} recall that if $UA = AU$ for all $A \in \alg{B}(\mathcal{H})$, then $U = \lambda I$ for some $\lambda \in \mathbb{C}$.
\end{exercise}
A map $t \mapsto U_t$ as in the exercise is called a \emph{projective representation} of $\mathbb{R}$. If we want to apply Stone's theorem to show that $\alpha_t$ is generated by a Hamiltonian, we would have to show that we can choose $c(s,t) = 
1$ for all $s$ and $t$ (so that $U_{s+t} = U_s U_t$ and $U_t$ is a proper representation) and that $t \mapsto U_t$ is strongly continuous. It is not easy to see that this is the case, but it follows from a result of Kadison~\cite[Thm. 4.13]{MR0169073}. Hence Stone's theorem and this result tell us that we can always find a Hamiltonian description of a time evolution on $\alg{B}(\mathcal{H})$. Conversely, defining a Hamiltonian gives the time evolution on the observables via $\alpha_t(A) = e^{i t H} A e^{-i t H}$.

If $\alg{A}$ is \emph{not} of the form $\alg{B}(\mc{H})$, the matter is more complicated. For one, since we do not have a natural Hilbert space any more, it is not clear how to define the Hamiltonian at all. Of course, one can take the GNS representation of a state to obtain a Hilbert space, but the result will depend on the state. In general the Hamiltonian is not bounded, hence it cannot be an element of the $C^*$-algebra. Moreover, it is \emph{not} true any more that every automorphism of $\alg{A}$ is inner. That is, there may not exist unitaries \emph{that are contained in $\alg{A}$} such that $\alpha(A) = UAU^*$ for all $A \in \alg{A}$. Therefore we take the one-parameter group $t \mapsto \alpha_t$ as one of the fundamental objects.\index{_alphat@$\alpha_t$}\index{one-parameter group} In the next sections we will see how we can describe the dynamics in a purely $C^*$-algebraic framework, and how this relates to the picture above.

\subsection{Derivations}
The Hamiltonian can be interpreted as the generator of time translations. It turns out that in the context of $C^*$-algebras, generators of one-parameter groups (or even semi-groups) of automorphisms can be conveniently expressed by derivations.\index{derivation}\index{_delta(A)@$\delta(A)$}
\begin{definition} A (symmetric) \emph{derivation} of a $C^*$-algebra $\alg{A}$ is a linear map from a $*$-subalgebra $D(\delta)$ of $\alg{A}$ into $\alg{A}$ such that
	\begin{enumerate}
		\item $\delta(A^*) = \delta(A)^*$ for $A \in D(\delta)$,
		\item $\delta(AB) = \delta(A)B + A \delta(B)$ for $A,B \in D(\delta)$.
	\end{enumerate}
The algebra $D(\delta)$ is called the \emph{domain} of $\delta$.
\end{definition}
For the systems that we are interested in, one usually has that the domain of the derivation is given by the strictly local operators of the quasi-local algebra, i.e. $D(\delta) = \alg{A}_{loc}$. Note that $\alg{A}_{loc}$ is dense in $\alg{A}$. This is also necessary, since in the end we want to be able to define a map that acts on all of $\alg{A}$. In general $\delta$ is an unbounded map, in the sense that there is no $C > 0$ such that $\delta(A) < C \|A\|$ for all $A \in D(\delta)$, hence one encounters similar concerns as in the case of unbounded operators on Hilbert spaces, discussed in Section~\ref{sec:unbounded}. In particular, it is generally not possible to extend $\delta$ to all of $\alg{A}$.

Derivations can be obtained as the \idx{generator} of one-parameter groups of automorphisms of a $C^*$-algebra $\alg{A}$. More precisely, let $t \mapsto \alpha_t$ be a one-parameter group of automorphisms of $\alg{A}$. Moreover, assume that the action is strongly continuous, in that the map $t \mapsto \alpha_t(A)$ is continuous in the norm topology for each $A \in \alg{A}$. Define $\delta(A)$ by
\begin{equation}
	\delta(A) = \lim_{t \to 0} \left( \frac{\alpha_t(A) - A}{t} \right)
	\label{eq:derivation}
\end{equation}
if this limit exists (in the norm topology). The set of all $A \in \alg{A}$ for which this limit exists is denoted by $D(\delta)$. Note that the limit is just the derivative of $\alpha_t(A)$ at $t = 0$.

\begin{lemma}\label{lem:derivation}
The map $\delta(A)$ defined above is a symmetric derivation. Moreover, we have that $\alpha_t(D(\delta)) \subset D(\delta)$.
\end{lemma}
\begin{proof}
	We first show that $D(\delta)$ is a $*$-subalgebra. It is clear that it is a vector space, since the expression~\eqref{eq:derivation} is linear in $A$. Note that it follows that $\delta$ is a linear map. Suppose that $A \in \delta(A)$. Since $\alpha_t(A^*) = \alpha_t(A)^*$ and $\|A^*\| = \|A\|$ it follows that if $A \in D(\delta)$, then $A^* \in D(\delta)$ and $\delta(A^*) = \delta(A)^*$. To show the ``Leibniz rule'', suppose that $A,B \in D(\delta)$. Note that $\alpha_t(AB) = \alpha_t(A) \alpha_t(B)$, hence we can use the product rule for derivatives. This leads to $\delta(AB) = A\delta(B) + \delta(A) B$. It follows that $D(\delta)$ is a $*$-subalgebra of $\alg{A}$ and that $\delta$ is a derivation. 

To show the last property, suppose that $A \in D(\delta)$ and let $s \in \mathbb{R}$. Then
\[
\begin{split}
	0 &= \lim_{t \to 0} \left\| \frac{\alpha_t(A) - A}{t} - \delta(A) \right\|  \\
		& = \lim_{t \to 0} \left\| \alpha_s\left(\frac{\alpha_t(A) - A}{t}\right) - \alpha_s(\delta(A)) \right\| \\
		& = \lim_{t \to 0} \left\| \left(\frac{\alpha_t(\alpha_s(A)) - \alpha_s(A)}{t}\right) - \alpha_s(\delta(A)) \right\|, 
\end{split} 
\]
where it is used that $\alpha_s$ preserves the norm. Hence $\alpha_t(D(\delta)) \subset D(\delta)$ for any $t$. The argument shows moreover that $\alpha_s(\delta(A)) = \delta(\alpha_s(A))$.
\end{proof}

The next exercise shows how derivations are related to the Hamiltonian description discussed earlier.
\begin{exercise}
	\label{ex:derivation}
	Suppose that $\alg{A} = \alg{B}(\mc{H})$ and let $\alpha_t$ be a one-parameter group of automorphisms given by $\alpha_t(A) = e^{i t H} A e^{-i t H}$ for some self-adjoint $H$. Show that the corresponding derivation is
\[
	\delta(A) = i [H, A].
\]
Verify explicitly that it is indeed a symmetric derivation.
\emph{Note:} If $H$ is unbounded, $\delta(A)$ will not be defined for each $A$ (since the limit will not always exist). In this exercise you can ignore these convergence issues.
\end{exercise}

If we are given a derivation $\delta$, the question is if (and how) it generates a one-parameter group of automorphisms. For $A \in D(\delta)$, we can set
\begin{equation}
	\label{eq:genauto}
\alpha_t(A) := \exp(t \delta)(A) = A + t \delta(A) + \frac{t^2 \delta^2(A)}{2} + \dots,
\end{equation}
where we set $\delta^0(A) = A$. Note that this definition makes sense only if we can interpret $\delta^2(A)$. That is, $\delta(A)$ should be in the domain of $D(\delta)$ again. Even if this is the case, it is not at all clear if the expression even converges (because $\delta$ need not be bounded), and if it allows us to define an automorphism on all of $\alg{A}$. Moreover, one can ask oneself in which topology this should converge, and what the continuity properties of $t \mapsto \alpha_t$ are. Such questions are studied in great detail in~\cite[Ch. 3]{MR887100} and~\cite[Ch. 6]{MR1441540}. Here we will content ourselves with answering these questions for a relatively simple (but important) class of derivations, associated to spin systems with local interactions.

First, however, consider the case of Exercise~\ref{ex:derivation}, where $\delta(A) = i[H,A]$. We will argue, at least heuristically, why equation~\eqref{eq:genauto} is the right definition. To this end, we first collect the first few terms of the expansion of the exponential, to obtain (with $A \in D(\delta)$)
\[
\exp(t \delta)(A) = A + i t [H, A] - \frac{t^2}{2} [H, [H,A]] + \mc{O}(t^3).
\]
On the other hand, if we expand $e^{i t H} A e^{-i t H}$ up to second order in $t$, we get
\[
\begin{split}
e^{i t H} A e^{-i t H} &= (I + i t H - \frac{t^2}{2} H^2 + \dots) A (I - i t H - \frac{t^2}{2} H^2 + \dots) + \mc{O}(t^3) \\
&= A - i t A H + i t H A - \frac{t^2}{2} H^2 A + t^2 H A H  - \frac{t^2}{2} A H^2 + \mc{O}(t^3) \\
&= A + i t [H,A] - \frac{t^2}{2} [H,[H,A]] + \mc{O}(t^3),
\end{split}
\]
hence we recover the automorphism from which the derivation $\delta$ was obtained.

\subsection{Finite range interactions}
The question one has to answer is which derivation is the relevant one when studying a particular system. In case of a lattice system, $\delta(A) = i [H,A]$ for some Hamiltonian $H$ is usually not really adequate. If there are infinitely many sites, $H$ will be unbounded in general and the commutator does not make any sense \emph{a priori}. For example, consider a 1D infinite spin chain. Suppose that at each site the energy is given by the value of the spin in the $z$ direction. Hence in this case, $H = \sum_{n = -\infty}^\infty \sigma^z_n$. This is clearly not an element of the quasi-local algebra $\alg{A}$, since it cannot be approximated in norm by local operators. Another approach might be to first define a Hilbert space and a representation of $\alg{A}$ on this Hilbert space, and then define $H$ on this Hilbert space. But this defeats the purpose of the algebraic approach. 

The way out is that we can look at the interactions in a \emph{finite} part of the system. From the model one wants to study it is usually clear what the interactions should be \emph{within} this finite region. If the size of the region is taken to be bigger and bigger, one could hope that this converges in a suitable sense. For example, consider again the spin chain example above. Let $A$ be a local observable. Then there is some integer $N$ such that $[\sigma_n^z,A] = 0$ for all $n \geq N$, by locality. Hence the following definition makes sense for all local observables $A \in \alg{A}_{loc}$: 
\[
	\delta(A) = \lim_{n \to \infty} i \left[ \sum_{k=-n}^n \sigma_k^z, A \right].
\]
This defines a derivation with domain $D(\delta) = \alg{A}_{loc}$. To consider more general examples we first need the following definition.

\begin{definition}
Consider a quantum spin system defined on a set $\Gamma$ of sites, with quasi-local algebra $\alg{A}$. An \idx{interaction} $\Phi$ is a map $\Phi : \mc{P}_f(\Gamma) \to \alg{A}$ such that for each $\Lambda$ we have that $\Phi(\Lambda) \in \alg{A}(\Lambda)$ and $\Phi(\Lambda)$ is self-adjoint.
\end{definition}
As the name suggests, an interaction describes what the interactions are between the spin systems sitting at the sites of some finite set $\Lambda$. If $\Phi$ is an interaction, we define the \emph{local Hamiltonians}\index{local!Hamiltonian} by
\[
	H_\Lambda = \sum_{X \subset \Lambda} \Phi(X),
\]
where the sum is over all subsets of $\Lambda$. Note that $H_\Lambda = H_\Lambda^*$ and that $H_\Lambda \in \alg{A}(\Lambda)$.

To define what it means for an interaction to be \emph{short range}\index{interaction!short range}, we have to assume some additional structure on the set $\Gamma$ of sites. In particular, we assume that there is a metric $d(x,y)$ on $\Gamma$ such that $d(x,y) \geq 1$ if $x \neq y$ and such that for any $M > 0$ and $x \in \Gamma$, the ball of radius $M$ around $x$, that is
\[
	B_M(x) = \{ y \in \Gamma : d(x,y) \leq M \},
\]
contains at most $N_M$ elements, where $N_M$ is a constant independent of $x$. This condition ensures that each point $x$ has only finitely many neighbours and the maximum number of neighbours is bounded from above. The example we will usually consider is where $\Gamma = \mathbb{Z}^d$ and the metric is the \emph{taxicab} or \emph{Manhattan} metric, defined as
\[
d(x,y) = \sum_{i=1}^d |x_i-y_i|.
\]
It is easy to see that this metric fulfils all the conditions. If $\Lambda \subset \Gamma$ is any subset, we define the diameter of $\Lambda$ as $\operatorname{diam}(\Lambda) = \sup_{x,y \in \Gamma} d(x,y)$. An interaction $\Phi$ is then set to be of finite range if there is some constant $c_\Phi > 0$ such that $\Phi(\Lambda) = 0$ whenever $\operatorname{diam}(\Lambda) > c_\Phi$. For finite range interactions there are, as the name implies, no interactions between sites that are at least a distance $d_\Phi$ apart.\footnote{At least, these sites do not interact directly. The effect of an operation at site $x$ can still propagate through the system to site $y$, but this will take time. This will be studied extensively in Chapter~\ref{ch:lr}.} We will also consider only \emph{bounded} interactions, for which $\|\Phi \| := \sup_{x \in \Gamma} \sum_{\Lambda \in \mc{P}_f(\Gamma), x \in \Lambda} \|\Phi(\Lambda)\| < \infty$.\index{interaction!bounded}

\begin{remark}
In the remainder we will mainly deal with bounded finite range interactions. These conditions can be relaxed in many cases to interactions whose strength decays sufficiently fast as the distance between two sites increases. Depending on the precise conditions one chooses, however, this can be done only at the expense of weaker statements on the convergence of the dynamics. For example, the dynamics may only converge for certain states. We refer to Chapter 6 of Bratteli and Robinson for more details~\cite{MR1441540}.
\end{remark}

Finally, suppose that we also have a translation symmetry of $\Gamma$. This induces a group of automorphisms $\tau_x$ on $\alg{A}$, as discussed above. An interaction is set to be translation invariant if $\Phi(\Lambda+x) = \tau_x(\Phi(\Lambda))$ for all $x \in \mathbb{Z}^d$ and $\Lambda \in \mc{P}_f(\Gamma)$, where $\mathbb{Z}^d$ is the translation group. Note that a translation invariant short range interaction is automatically bounded. 

\begin{example}[The quantum Heisenberg model]\index{Heisenberg XXZ model}
Consider $\Gamma = \mathbb{Z}$ with at each site of $\Gamma$ a spin-1/2 system. Define an interaction by $\Phi({n}) = -h \sigma^z_n$, with $h$ some real constant, and 
\[
\Phi(\{n,n+1\}) = -\frac{1}{2} \left(J_x \sigma^x_n \sigma^x_{n+1} + J_y \sigma^y_n \sigma^y_{n+1} + J_z \sigma^z_n \sigma^z_{n+1}\right)
\]
for real constants $J_x,J_y$ and $J_z$. In all other cases we set $\Phi(\Lambda) = 0$. Note that $\Phi$ is a finite range interaction. Moreover, it is clear that $\Phi$ is translation invariant. For $\Lambda \subset \Gamma$ of the form $[n,m+1] \cap \mathbb{Z}$ for integers $n,m$, the corresponding local Hamiltonians are
\[
H_\Lambda = -\frac{1}{2} \sum_{k=n}^m \left(J_x \sigma^x_k \sigma^x_{k+1} + J_y \sigma^y_k \sigma^y_{k+1} + J_z \sigma^z_k \sigma^z_{k+1}\right) - h \sum_{k=n}^{m+1} h \sigma^z_k.
\]
The last term can be interpreted as an external field, while the terms with $J_i$ describe a coupling between nearest neighbour spins. If $J_x = J_y \neq J_z$, this model is called the \emph{Heisenberg XXZ model}. An interesting fact is that the behaviour of the system strongly depends on the couplings $J_i$.
\end{example}

\begin{example}
Consider again $\Gamma = \mathbb{Z}$. A generalisation of the previous example can be obtained as follows. Choose functions $j_i : \mathbb{Z} \to \mathbb{R}$. Define an interaction $\Phi$ by
\[
\begin{split}
	\Phi(\{0\}) &= h \sigma^z_0, \\
	\Phi(\{0,n\}) &= \sum_{k=1}^3 j_k(n) \sigma_0^k \sigma_n^k,
\end{split}
\]
and by requiring that the interaction is translation invariant (i.e., $\Phi(\Lambda+x) = \tau_x(\Phi(\Lambda))$ for all $\Lambda \in \mc{P}_f(\Gamma)$). This describes a model of spins in an external magnetic field, together with two-body interactions whose strength is determined by the functions $j_i$. Note that if the support of $j_i$ is not bounded, there are interactions between pairs of spins that are arbitrarily far away. If this is the case, the interaction is clearly not of finite range.
\end{example}

Now let $\Phi$ be a bounded finite range interaction of range $c_\Phi$. As before, we can define the local Hamiltonians $H_\Lambda$. Now suppose that $A \in \alg{A}_{loc}$ is localised in some set $\Lambda_A$. Because the interaction is of finite range, there are only finitely many finite subsets $\Lambda$ of $\Gamma$ such that $[\Phi(\Lambda), A] \neq 0$. For example, consider the set of all points of $\Gamma$ with distance at most $c_\Phi$ from $\Lambda_A$. Then any finite subset $\Lambda$ of $\Gamma$ such that $\Phi(\Lambda)$ does not commute with $A$ must necessarily be contained in this set, by the finite range assumption. Hence if we consider $\sum_{\Lambda \in \mc{P}_f(\Gamma)} [\Phi(\Lambda), A]$, only a finite number of terms do not vanish and this sum converges to an element in $\alg{A}_{loc}$. It follows that
\[
	\delta(A) = \lim_{\Lambda \to \infty} i [H_\Lambda, A]
\]
defines a symmetric derivation\index{derivation} with domain $D(\delta) = \alg{A}_{loc}$. Here with $\Lambda \to \infty$ we mean the following. Consider a sequence $\Lambda_1 \subset \Lambda_2 \subset \dots$ of finite subsets of $\Gamma$ such that $\bigcup_{n} \Lambda_n = \Gamma$. Then with $\lim_{\Lambda \to \infty}$ we mean that the limit $\lim_{n \to \infty}$ converges independently of the sequence $\Lambda_n$ that was chosen.\footnote{Alternatively we can consider $\mc{P}_f(\Gamma)$ as the index set of a net $\Lambda \mapsto i [H_\Lambda, A]$, where the order relation is by inclusion. The convergence described in the text is the convergence of this net in the sense of topology.}

The goal is now to show that $\exp(t \delta)(A)$ converges for $A \in D(\delta)$ and $t$ sufficiently small. Note that $\delta(A) \in D(\delta)$ for all $A \in D(\delta)$, hence expressions of the form $\delta^n(A)$ are well-defined. The problem, however, is that $\operatorname{supp}\delta(A)$ may be bigger than the support of $A$. A close look at the definition of $\delta$ and the local Hamiltonians show $H_\Lambda$ that this means that for $\delta^n(A)$, in principle more and more interaction terms can contribute. The boundedness of $\Phi$ will ensure that the growth of the support will not be too big, which in turn makes it possible to make estimates on the norm $\|\delta^n(A)\|$, which will allow us to conclude that the exponential $\exp(t \delta)$ converges. The ideas behind the proof date back to Robinson~\cite{MR0228246}.

\begin{lemma}
	Let $\Phi$ be a bounded interaction of finite range and write $\delta$ for the corresponding derivation. Then the algebra $D(\delta) = \alg{A}_{loc}$ is analytic for $\delta$, in the sense that for each $A \in D(\delta)$, $\exp(t \delta)(A)$ is analytic as a function of $t$, for $t$ sufficiently small.\index{derivation!analytic algebra}
\end{lemma}
\begin{proof}
Let $A \in \alg{A}(\Lambda)$ for some finite set $\Lambda$. By locality we have that $[\Phi(X), A] = 0$ whenever $X \cap \Lambda = \emptyset$. Hence we can write
\[
	\delta^n(A) = i^n \sum_{X_1 \cap S_1 \neq \emptyset} \cdots \sum_{X_n \cap S_n \neq \emptyset} [\Phi(X_n), [\cdots [\Phi(X_1), A] \cdots ]],
\]
where we set $S_1 = \Lambda$ and $S_k = S_{k-1} \cup X_{k-1}$ for $k > 1$. Since $\Phi$ is a bounded interaction and the number of sites in a $r$-ball around any site is uniformly bounded, there is some constant $N_\Phi$ such that $\Phi(X) = 0$ if $|X| > N_\Phi$, hence it suffices to consider sets $X_i$ such that $|X_i| \leq N_\Phi$. This implies that $|S_2| \leq N_\Phi-1 + |\Lambda|$ and by induction,
\[
	|S_n| \leq (n-1)(N_\Phi-1) + |\Lambda|.
\]
Using this observation we can make the estimate
\begin{equation}
	\|\delta^n(A)\| \leq 2^n \|A\| \prod_{k=1}^n \left( (k-1)(N_\Phi-1) + |\Lambda|\right) \|\Phi\|^n
	\label{eq:derivest}
\end{equation}
by successively estimating the norm of the commutators, and doing the sums.

Next we note the following inequality, which holds for all positive $a$ and $\lambda$: $a^n \leq n! \lambda^{-n} \exp(a \lambda)$. This inequality can be easily verified using the expansion of the exponent. Using this inequality with $a = (n-1)(N_\Phi-1) + |\Lambda|$, we obtain the estimate
\begin{equation}
	\label{eq:derivest2}
	\|\delta^n(A)\| \leq \|A\| n! e^{\lambda (|\Lambda|+1-N_\Phi)} \left( \frac{2 \|\Phi\| e^{\lambda (N_\varphi-1)}}{\lambda} \right)^n.
\end{equation}
Using this estimate we see that $\exp(t \delta)(A)$ converges for $|t| < t_\lambda := \frac{\lambda}{2 \|\Phi\| e^{\lambda(N_\varphi-1)}}$, since $\sum_n x^n$ converges for $x < 1$.
\end{proof}
We define $\alpha_t(A) = \exp(t \delta)(A)$ for $A \in \alg{A}_{loc}$ and $|t| < t_\lambda$ as defined in the proof of the Lemma. Since $\delta$ is a derivation we have that $\alpha_t(AB) = \alpha_t(A)\alpha_t(B)$ and that $\alpha_t(A^*) = \alpha_t(A)^*$. However, we cannot just yet conclude that $\alpha_t$ defines an automorphism of $\alg{A}$, since $\alpha_t$ is only defined on a \emph{subset} of $\alg{A}$. Moreover, so far we have only defined it for small $t$. The basic strategy to define $\alpha_t$ for all $t$ is to use the group property $\alpha_{s+t} = \alpha_s \circ \alpha_t$. Again, one has to be a bit careful here, since generally it is not true that $\alpha_t(A) \in D(\delta)$, even for small $t$ and $A \in D(\delta)$.

The easiest way to attack these problems is by approximating $\alpha_t$ in terms of \emph{local dynamics}\idx{local!dynamics} in the following sense. Define $\delta_\Lambda(A) = i [H_\Lambda, A]$ for $\Lambda \in \mc{P}_f(\Gamma)$. Clearly, $\lim_{\Lambda \to \infty} \delta_\Lambda(A) = \delta(A)$. We will show that $\alpha_t^\Lambda(A) := \exp(t \delta_\Lambda)(A) \to \alpha_t(A)$ as $\Lambda$ grows to $\Gamma$.
\begin{theorem}\label{thm:dynconv}\index{dynamics!existence}
	Let $\Phi$ be as in the Lemma above. Then $\alpha_t$ as defined above extends to a strongly continuous one-parameter group $t \mapsto \alpha_t$ of automorphisms of $\alg{A}$ such that for each $A \in \alg{A}$ and $t \in \mathbb{R}$ we have
	\begin{equation}
		\label{eq:autconv}
		\lim_{\Lambda \to \infty} \| \alpha_t^\Lambda(A) - \alpha_t(A) \| = 0.
	\end{equation}
\end{theorem}
\begin{proof}
	As we saw before, $\exp(t \delta_\Lambda)(A) = \exp(i t H_\Lambda) A \exp(-i t H_\Lambda)$, hence $\alpha_t^\Lambda$ is isometric. Now let $A \in \alg{A}(X)$ be a local operator. Since the interaction is strictly local, for each $\Lambda$ that contains $X$ there is an integer $N(\Lambda)$ such that  $\delta_\Lambda^n(A) = \delta^n(A)$ for all $n \leq N(\Lambda)$. Moreover, $N(\Lambda) \to \infty$ if $\Lambda \to \infty$. Hence we see that
\[
\| \alpha_t^\Lambda(A) - \alpha_t(A) \| = \left\| \sum_{n=N(\Lambda)}^\infty \left( \frac{t^n \delta^n_\Lambda(A)}{n!} - \frac{t^n \delta^n(A)}{n!} \right) \right\|.
\]
By using the estimate from equation~\eqref{eq:derivest2} we see that this goes to zero as $\Lambda \to \infty$ (and hence $N(\Lambda) \to \infty$) when $|t| < t_\lambda$. Hence for each strictly local observable $A$ it follows that equation~\eqref{eq:autconv} holds for $t$ small enough. Since the strictly local observables form a dense subset of all the observables, and $\alpha_t^\Lambda$ and hence $\alpha_t$ is isometric, we can extend $\alpha_t$ to a bounded linear map on all of $\alg{A}$. A standard argument then shows that equation~\eqref{eq:autconv} holds for all $A \in \alg{A}$ and $|t|$ small enough.

To show that $\alpha_t$ can be defined for all $t$, consider $A \in \alg{A}$ and $|s|,|t| < t_\lambda$. Then by the above argument $\alpha_s(\alpha_t(A))$ is well-defined. Moreover,
\begin{align*}
	\alpha^{\Lambda}_{s+t}&(A) - \alpha_s(\alpha_t(A)) = \alpha^{\Lambda}_s(\alpha^\Lambda_t(A)) - \alpha_s(\alpha_t(A)) \\
	&= \alpha_s^\Lambda( (\alpha_t^\Lambda - \alpha_t)(A)) + \alpha_s((\alpha_t^\Lambda-\alpha_t)(A)) + (\alpha_s^\Lambda - \alpha_s)(\alpha_t(A))
\end{align*}
where the first equality sign follows from the group property for the local dynamics. The terms in the last line all vanish as $\Lambda \to \infty$ since the $\alpha_t$ are isometric and because of equation~\eqref{eq:autconv}. Hence we can define $\alpha_{s+t}(A)$ as the limit of $\alpha_{s+t}^\Lambda(A)$ and the argument above shows that it is equal to $\alpha_s(\alpha_t(A))$. Because of the group property for the local dynamics this limit does not depend on the specific choices of $s$ and $t$, only on $s+t$. Repeating this argument one can define $\alpha_t(A)$ for all $t \in \mathbb{R}$.

To show that $\alpha_t$ is a $*$-morphism one can either check it first for small $t$ and $A$ in the dense subset $D(\delta)$, or use that the local dynamics $\alpha^\Lambda_t$ are automorphisms in combination with the approximation given above. That $\alpha_t$ is an automorphism is clear, since $\alpha_t(\alpha_{-t}(A)) = \alpha_{t-t}(A) = A$.

It remains to show that $t \mapsto \alpha_t$ is strongly continuous. Since $\alpha_{t+s}(A) = \alpha_t(\alpha_s(A))$ and $\alpha_s$ is isometric for each $s$, it is enough to show that $\alpha_t(A) \to A$ as $t \to 0$. But this follows again by using the estimates for the norm of $\delta^n(A)$.
\end{proof}

Once we introduce ground states in this formalism, we will see how we can recover a Hamiltonian implementing the time translations from the one-parameter group of automorphisms.

\begin{remark}
	Essentially the same proof also works for two more general classes of interactions. First of all, one can consider the normed vector space of bounded finite range interactions as we have discussed here, and take the completion with respect to the norm to obtain a Banach space. Alternatively, one could consider all interactions $\Phi$ that satisfy $\sup_{x \in \Gamma} \sum_{X \ni x} e^{\lambda |X|} \|\Phi(X)\| < \infty$ for some $\lambda > 0$. Note that bounded interactions of finite range satisfy this condition.
\end{remark}

\begin{definition}
	Let $\alg{A}$ be a $C^*$-algebra and $t \mapsto \alpha_t$ a strongly continuous one-parameter group of automorphisms. Such a pair $(\alg{A}, \alpha)$ is called a \emph{$C^*$-dynamical system}\index{C*dynamical system@$C^*$-dynamical system}.
\end{definition}

For more general interactions there are two different strategies than can be followed in proving the existence of dynamics. As mentioned, essentially the same proof as above can be applied to interactions that decay sufficiently fast. One can also start directly with the local dynamics $\alpha^\Lambda_t(A)$ and show that they converge as $\Lambda$ grows. To this end one can employ \emph{Lieb-Robinson bounds}\index{Lieb-Robinson!bound}, which give a way to approximate $\alpha_t^{\Lambda}(A)$ by local operators whose support does not grow too fast as a function of $t$. This in turn makes it possible to show that $\alpha^\Lambda_t$ converges to some automorphism $\alpha_t$ (in the strong topology). We will discuss this strategy in Chapter~\ref{ch:lr}. Alternatively one can work directly with the derivations $\delta$. In that case, it is necessary to show that $\delta$ is a generator of a group of automorphisms, or more precisely, it can be extended uniquely to a generator $\overline{\delta}$. The theory of generators of one-parameter groups is well studied (see e.g. Chapter 3 of~\cite{MR887100}), and there are several different criteria that can be used to verify that $\delta$ is a generator. Some examples of this strategy can be found in Chapter 6 of~\cite{MR1441540}.

\section{Ground states and thermal states}
Once we have defined dynamics on a system we can try to talk about ground states and equilibrium states of such systems. Unfortunately, a rigorous treatment of most of these results would require a lot of advanced techniques from, e.g., functional analysis and complex analysis, and falls outside the scope of this course. For this reason we will omit or only sketch most of the proofs in this section, and restrict to introducing the main definitions and results. A full treatment can be found in~\cite{MR1441540}.

A ground state\index{ground state} of a Hamiltonian $H$ acting on some Hilbert space is an eigenvector of the Hamiltonian with minimum eigenvalue. Since for physically relevant systems the energy is bounded from below, we can always normalise the Hamiltonian in such a way that $H\Omega = 0$ for a ground state $\Omega$. Note that a ground state is stationary, since $e^{-i t H} \Omega = \Omega$. Hence the first condition on ground states of a $C^*$-dynamical system $(\alg{A}, \alpha)$ is that it should be invariant with respect to each $\alpha_t$. If $\omega$ is a state on $\alg{A}$ such that $\omega\circ\alpha_t = \omega$ for all $t$, it follows from Exercise~\ref{ex:invstategns} that in the corresponding GNS representation, $(\pi_\omega, \Omega, \mc{H}_\omega)$, there is a strongly continuous one-parameter group $t \mapsto U(t)$ of unitaries implementing the automorphisms $\alpha_t$. By Stone's theorem, Theorem~\ref{thm:stone}, there is a self-adjoint densely defined operator $H_\omega$ such that $U(t) = \exp(i t H_\omega)$. We require that this $H_\omega$ is positive, denoted $H_\omega \geq 0$. This means that $\langle \psi, H_\omega \psi \rangle \geq 0$ for all $\psi$ in the domain of $H_\omega$. Hence a ground state will be defined as follows:

\begin{definition}\index{ground state}
	Let $\alg{A}$ be a $C^*$-algebra and $t \mapsto \alpha_t$ a strongly continuous one-parameter group of automorphisms. A state $\omega$ on $\alg{A}$ is said to be an $\alpha$-ground state if $\omega$ is invariant for each $\alpha_t$ and if $H_\omega \geq 0$ for the Hamiltonian in the corresponding GNS representation.
\end{definition}

This definition first requires to check that the ground state is invariant. Moreover, one has to construct the GNS representation and find a Hamiltonian implementing the dynamics in the representation. It is therefore desirable to have a purely $C^*$-algebraic characterisation of ground states. A criterion in terms of the generator $\delta$ of the time translations would be particularly useful, considering that we defined derivations that generate the dynamics in terms of interactions. Indeed, there are equivalent characterisations of ground states. For this result we first need the following lemma.
\begin{lemma}
	\label{lem:invder}
Let $\delta$ be a symmetric derivation generating a strongly continuous one-parameter group $\alpha_t$ and suppose that $\omega$ is a state on $\alg{A}$ such that
\[
i \omega(A \delta(A)) \in \mathbb{R}
\]
for all $A = A^*$ in $D(\delta)$. Then $\omega(\delta(A)) = 0$ for all $A \in D(\delta)$ and $\omega \circ \alpha_t = \omega$ for all $t$.
\end{lemma}
\begin{proof}
Let $A \in D(\delta)$ be self-adjoint. Then we have that
\[
	i \omega(A \delta(A)) = \overline{i \omega(A \delta(A))} = -i \omega(\delta(A)^* A^*) = -i \omega(\delta(A)A).
\]
Using the Leibniz property of the derivation we obtain
\[
0 = \omega(A\delta(A)) + \omega(\delta(A)A) = \omega(\delta(A^2)).
\]
Now let $X \in D(\delta)$ be any positive operator, then we can take the square root $\sqrt{X}$ and use the formula above to conclude $\omega(\delta(X)) = 0$.\footnote{\label{ftn:sqrt}If $\delta$ is the generator of a strongly continuous one-parameter group of automorphisms, one can show that the square root $\sqrt{X}$ is also in the domain of $\delta$ if $X \in D(\delta)$.} Because any element of a $C^*$-algebra can be written as linear combination of at most four positive operators, it follows that $\omega(\delta(A)) = 0$ for all $A \in D(\delta)$.

Since $\delta$ is the generator of $\alpha_t$, $\delta(A)$ can be obtained as the derivative at $t=0$, by equation~\eqref{eq:derivation} Note that if $A \in D(\delta)$, then $\alpha_t(A) \in D(\delta)$ (by Lemma~\ref{lem:derivation} and if $A$ is positive, so is $\alpha_t(A)$ (by Theorem~\ref{thm:positiveop}). Hence we have for positive $A \in D(\delta)$:
\[
	\left. \frac{d}{dt} \omega(\alpha_t(A))\right|_{t=s} = \lim_{t \to 0} \frac{\omega(\alpha_{s+t}(A))-\omega(\alpha_s(A))}{t} = \omega(\delta(\alpha_s(A))) = 0.
\]
Hence $\omega(\alpha_t(A))$ is constant (and hence equal to its value at $t=0$). Since the domain of $\delta$ is a dense subset of $\alg{A}$ and any element of $\alg{A}$ is the sum of at most four positive operators, the claim follows from continuity.
\end{proof}

Suppose that $\psi$ is a vector in $\mc{H}_\omega$. On the one hand, consider the equation
\[
\left.\frac{d}{d t}\right|_{t=0} e^{it H_\omega} \pi_\omega(A) e^{-it H_\omega} \psi = i [H_\omega, \pi_\omega(A)] \psi.
\]
On the other hand, $e^{i t H} \pi_\omega(A) e^{-i t H} = \pi_\omega(\alpha_t(A))$, and hence if we take the derivative again at $t=0$, it follows that the expression above is equal to $\pi_\omega(\delta(A)) \psi$. A proper analysis requires a bit more care (in particular, since the expressions will in general not be well defined for arbitrary $\psi$ or $A$), but this can be done (see e.g.\ the text after~\cite[Prop. 3.2.28]{MR887100}). The result is that we have the following identity for all vectors $\psi \in D(H_\omega)$ and $A \in D(\delta)$:
\begin{equation}
	\label{eq:Hderiv}
	\pi_\omega(\delta(A)) \psi = i [H_\omega, \pi_\omega(A)] \psi.
\end{equation}
In addition one can show that $\pi_\omega(D(\delta)) \Omega \subset D(H_\omega)$. Another way to state the equation above is that the derivation $\delta$ is implemented as an \emph{inner} derivation on the Hilbert space $\mc{H}_\omega$.\index{derivation!inner} With this observation we can now give an alternative definition of ground states.

\begin{theorem}\label{thm:groundeq}\index{ground state!algebraic definition}
Let $\omega$ be a state on a $C^*$-algebra $\alg{A}$ and suppose that $t \mapsto \alpha_t$ is a strongly continuous one-parameter group of automorphisms. Write $\delta$ for the corresponding generator. Then the following are equivalent:
\begin{enumerate}
	\item \label{it:aground} $\omega$ is a ground state for $\alpha$,
	\item \label{it:dground} $-i \omega(A^*\delta(A)) \geq 0$ for all $A \in \delta(A)$.
\end{enumerate}
\end{theorem}
\begin{proof}
	(\ref{it:aground}. $\Rightarrow$~\ref{it:dground}.) Let $A \in \delta(A)$. Because $H_\omega$ is a positive operator it follows that $\langle \pi_\omega(A) \Omega, H_\omega \pi_\omega(A) \Omega \rangle \geq 0$. Hence
\[
\langle \Omega \pi_\omega(A), H_\omega \pi_\omega(A) \Omega  \rangle = -i \langle \pi_\omega(A) \Omega, \pi_\omega(\delta(A)) \Omega \rangle = -i \omega(A^* \delta(A)) \geq 0,
\]
since $\pi_\omega(\delta(A)) \Omega = i H_\omega \pi_\omega(A) \Omega$ by the argument before the theorem (and since $H_\omega \Omega = 0$).

(\ref{it:dground}. $\Rightarrow$~\ref{it:aground}.) From Lemma~\ref{lem:invder} it follows that $\omega$ is invariant with respect to $\alpha_t$ and hence we can find a Hamiltonian $H_\omega$ leaving the GNS vector $\Omega$ invariant. It remains to show that $H_\omega$ is positive. Let $\psi = \pi_\omega(A)\Omega$ for $A \in D(\delta)$. Then with the help of equation~\eqref{eq:Hderiv}, we find
\[
	\langle \psi, H_\omega \psi \rangle = \langle \pi_\omega(A) \Omega, \pi_\omega(\delta(A)) \Omega \rangle =  -i \omega(A^*\delta(A)) \geq 0.
\]
The result follows by a density argument (technically, that $\pi_\omega(D(\delta))$ is a \emph{core} for $H_\omega$).
\end{proof}
Hence we can define ground states in a purely algebraic setting by specifying the interactions of the system considering the corresponding derivation. Having a definition of ground states we can then try to answer different questions on the structure of the set of ground states. For example, it is not clear if it is non-empty. Indeed, there exist $C^*$-dynamical systems $(\alg{A}, \alpha)$ without ground states. Furthermore, ground states should precisely be those states that minimize the local Hamiltonians in a suitable sense. We will come back to these points later.

\subsection{KMS states}\label{sec:kmsstates}
A basic question in statistical mechanics is how to characterise thermal states.\index{thermal state|see{KMS state}} The definition that will be given here was introduced by Hugenholtz, Haag and Winnink~\cite{MR0219283}, giving a characterisation of states satisfying the conditions studied by Kubo, Martin and Schwinger~\cite{MR0098482,MR0109702}. In their honour they are called KMS states. Besides for the description of thermal states, KMS states have turned out to be extremely useful in the theory of operator algebras as well. We first give the (or rather, one of the many equivalent) definition. Later we will discuss some results that will give support to the claim that this is the right notion of thermal equilibrium states.
\begin{definition}\label{def:kms}
Let $(\alg{A}, \alpha)$ be a $C^*$-dynamical system and let $\beta > 0$. Consider the strip
\[
	S_\beta = \{ z \in \mathbb{C} : \operatorname{Im}(z) \in [0, \hbar \beta] \}
\]
in the complex plane. A state $\omega$ on $\alg{A}$ is called a \idx{KMS state} if for each $A, B \in \alg{A}$ there is a complex function $F_{AB}$ on $S_\beta$ such that
\begin{enumerate}
	\item $F_{AB}$ is analytic on the interior of the strip and continuous on the boundaries,
	\item $F_{AB}(t) = \omega(A \alpha_t(B))$ for all $t \in \mathbb{R}$, and
	\item $F_{AB}(t+i \hbar \beta) = \omega(\alpha_t(B)A)$.
\end{enumerate}
\end{definition}
From now on we will use units where $\hbar = 1$. The parameter $\beta$ is called the inverse temperature, $\beta^{-1} = k_B T$, where $T$ is the temperature and $k_B$ is Boltzmann's constant. Intuitively speaking a ground state should be an equilibrium state at $T = 0$, and indeed one can see ground states as KMS states for $\beta = \infty$. Nevertheless, despite being similar in many aspects, there are also some fundamental differences between KMS states and ground states.

Note that the function $F_{AB}$ is given by expectation values of operators in the state $\omega$ on the boundary of the strip. This relation can be extended to values into the strip, under suitable circumstances. To discuss this, we first introduce the set of \emph{entire} elements\index{entire element} for a $C^*$-dynamical system. Note that for each $A \in \alg{A}$ we have the continuous function $t \mapsto \alpha_t(A)$ by virtue of strong continuity of $\alpha_t$. The operator $A$ is said to be \emph{entire} if there is an entire (i.e., complex analytic) function $f: \mathbb{C} \mapsto \alg{A}$ such that $f(t) = \alpha_t(A)$ for all $A \in \alg{A}$. If $A$ is entire we will simply write $\alpha_z(A)$ for the corresponding function. Note that in the expansion of $\exp(t \delta)(A)$ for a derivation $\delta$ it is easy to plug in complex numbers instead of real numbers, to obtain $\alpha_z(A)$ for $z$ complex. A useful result is that there the set of entire elements, $\alg{A}_\alpha$, is a norm dense $*$-subalgebra of $\alg{A}$~\cite[Prop. IV.4.6]{MR1239893}. We then have the following result, whose proof can be found in~\cite[Prop. 5.3.7]{MR1441540}.

\begin{proposition}\label{prop:kmsanalytic}
	Let $\omega$ be a KMS state at inverse temperature $\beta$. Then for each $A \in \alg{A}$ and $B \in \alg{A}_\alpha$ we have that, restricted to the strip $S_\beta$,
	\[
		F_{A,B}(z) = \omega(A \alpha_z(B)),
	\]
	where $F_{AB}$ is the function as in Definition~\ref{def:kms}.
\end{proposition}
Using this result we can put the claim that KMS states at $\beta = \infty$ correspond to ground states in context. Indeed, $\omega$ is a ground state if and only if the function $x \mapsto \omega(A \alpha_z(B))$ is analytic and bounded on the upper half plane (i.e., the interior of the ``strip'' $S_\infty$) for all $A,B \in \alg{A}_\alpha$~\cite[Prop. 5.3.19]{MR1441540}, and continuous on the real line.

Just as for ground states there is an equivalent formulation of KMS states using only the derivation $\delta$. Again we omit the proof.\index{KMS state!algebraic definition}
\begin{theorem}[Roepstorff-Araki-Sewell]\label{thm:autocor}
Let $(\alg{A}, \alpha)$ be a $C^*$-dynamical system and $\delta$ the corresponding derivation. Then the following are equivalent:
\begin{enumerate}
	\item $\omega$ is a KMS state at inverse temperature $\beta$;
	\item $-i \beta \omega(A^* \delta(A)) \geq \omega(A^*A) \log\left(\frac{\omega(A^*A)}{\omega(AA^*)}\right)$ for all $A \in D(\delta)$.
\end{enumerate}
\end{theorem}
Here we define $x \log(x/y)$ to be zero if $x = 0$ and $+\infty$ if $x > 0$ and $y = 0$. This function is actually \emph{lower semi-continuous}, which means that $\liminf_{(x,y) \to (x_0, y_0)} x \log(x/y) \geq x_0 \log(x_0/y_0)$.

We end this section on technical properties of KMS states with the following useful result, to the effect that limits of KMS states are again KMS states. We will see later how this result can be used to show that KMS states exist.
\begin{proposition}\label{prop:limkms}
	Let $(\alg{A}, \alpha)$ be a $C^*$-dynamical system (where $\alg{A}$ is unital) and $\delta$ the generator of $\alpha$. Suppose that for each $n \in \mathbb{N}$ we have a derivation $\delta_n$ generating a one-parameter group $t \mapsto \alpha_t^n$ of automorphisms such that $\lim_n \delta_n(A) = \delta(A)$ for all $A \in D(\delta)$. Moreover, suppose that for each $n$ we have a state $\omega_n$ which is KMS (with respect to $\alpha_t^n$) at inverse temperature $\beta_n$. Finally, suppose that $\lim \beta_n = \beta$, where we allow $\beta = +\infty$ as well. Then each weak-$*$ limit point\footnote{This means that there is a subsequence $\omega_{n_k}$ such that $\lim_{n_k} \omega_{n_k}(A) = \omega(A)$ for the limit point $\omega$.} of the sequence $\omega_n$ is a KMS state (with respect to $\alpha$) at inverse temperature $\beta$.
\end{proposition}
\begin{proof}
	Since $\alg{A}$ is unital, the state space of $\alg{A}$ is compact in the weak-$*$ topology, by Theorem~\ref{thm:statecompact}. Hence weak-$*$ limit points exist. Let $\omega_{n_k}$ be a subsequence converging to a limit point $\omega$. The goal is to show that $\omega$ is a KMS state. To this end, we will distinguish the case $\beta < \infty$ and $\beta = \infty$.

Let us first assume that $\beta < \infty$. Without loss of generality, we can assume that $\beta_{n_k} < \infty$. By Theorem~\ref{thm:autocor} we have for each $A \in D(\delta)$ the following inequality:
\[
-i \beta_{n_k} \omega_{n_k}(A^*\delta(A)) \geq \omega_{n_k}(A^*A) \log\left(\frac{\omega_{n_k}(A^*A)}{\omega_{n_k}(AA^*)}\right)
\]
We can then take the limit $k \to \infty$. Note that the right-hand side is not continuous (but only lower semi-continuous), but we can still get
\[
\begin{split}
-i \beta \omega(A^*\delta(A)) &\geq \limsup_{k \to \infty} \omega_{n_k}(A^*A) \log\left(\frac{\omega_{n_k}(A^*A)}{\omega_{n_k}(AA^*)}\right) \\
	&\geq \liminf_{k \to \infty} \omega_{n_k}(A^*A) \log\left(\frac{\omega_{n_k}(A^*A)}{\omega_{n_k}(AA^*)}\right) \\
	&= \omega(A^*A) \log\left(\frac{\omega(A^*A)}{\omega(AA^*)}\right), 
\end{split}
\]
where we used that $\limsup \leq \liminf$ and that $x \log(x/y)$ is lower semi-continuous. Again by Theorem~\ref{thm:autocor} it follows that $\omega$ is a $\beta$-KMS state.

The case $\beta = \infty$ remains. We will show that $\omega$ is a ground state. First consider the case that each $\beta_{n_k} = \infty$. Then we have
\[
	-i \omega_{n_k}(A^* \delta(A)) \geq 0
\]
for all $k$ and the result follows by taking the limit. Finally, consider the case where all $\beta_{n_k} < \infty$. In this case we have by Theorem~\ref{thm:autocor} that
\[
\begin{split}
	-i \omega(A^*\delta(A)) &= \lim_{k \to \infty} -i \omega_{n_k}(A^* \delta(A)) \\
	&\geq \limsup_{k \to \infty} \frac{1}{\beta_{n_k}} \omega_{n_k}(A^*A) \log\left(\frac{\omega_{n_k}(A^*A)}{\omega_{n_k}(AA^*)}\right) \\
	&\geq \limsup_{k \to \infty} \frac{1}{\beta_{n_k}}(\omega_{n_k}(A^*A)-\omega_{n_k}(AA^*)) = 0.
\end{split}
\]
In the last line the inequality $x \log(x/y) \geq x-y$ is used, which follows from the convexity of $x \mapsto x \log x$.
\end{proof}

\subsection{Gibbs states}
We now come to an important class of examples of KMS states, the \emph{Gibbs states}\index{Gibbs state}. Let $\Lambda \in \mc{P}_f(\Gamma)$. Then $\alg{A}(\Lambda)$ is a finite dimensional $C^*$-algebra and we can write any state $\omega_\Lambda$ on $\alg{A}(\Lambda)$ in the form $\omega_\Lambda(A) = \Tr(\rho A)$ for some density matrix $\rho$ (see Exercise~\ref{exc:density}). Consider a $H_\Lambda = H_\Lambda^*$ in $\alg{A}(\Lambda)$. Consider the state corresponding to the density matrix $\rho_\Lambda = Z^{-1} \exp(- \beta H_\Lambda)$, where $Z = \Tr(\exp(-\beta H_\Lambda))$.\footnote{The reader might recognize the partition function from statistical mechanics.} Hence we have a state $\omega_\Lambda$ given by
\[
\omega_\Lambda(A) = \frac{\Tr\left(e^{-\beta H_\Lambda} A\right)}{\Tr\left(e^{-\beta H_\Lambda}\right)}.
\]
We claim that it is a KMS state for inverse temperature $\beta$, where the time automorphism is given by $\alpha_t^\Lambda(A) = e^{i t H_\Lambda} A e^{-it H_\Lambda}$. Define $F_{AB}(z) = \omega_\Lambda(A \alpha_z^\Lambda(B))$ for $z \in \mathbb{C}$. Then we have by the cyclicity of the trace, and $t$ real,
\[
Z F_{AB}(t+i\beta) = \Tr\left(e^{-\beta H_\Lambda} A e^{(-\beta+it) H_\Lambda} B e^{(\beta-it) H_\Lambda}\right) = 
\Tr\left(e^{-\beta H_\Lambda} \alpha_t^\Lambda(B) A \right).
\]
Hence $\omega$ satisfies the KMS conditions at inverse temperature $\beta$ with respect to the action $\alpha_t^\Lambda$. It is in fact the \emph{only} state on $\alg{A}(\Lambda)$ with this property. To show this result we first need the show that each tracial state on a finite matrix algebra, that is a state for which $\tau(AB) = \tau(BA)$, coincides with the usual trace. The proof is left as an exercise.

\begin{exercise}
	Let $\alg{A} = M_n(\mathbb{C})$. Let $\omega$ be a linear functional on $\alg{A}$ such that $\omega(AB) = \omega(BA)$ for all $A,B$. Prove that $\omega(A) = c \Tr(A)$ for some constant $c$. \emph{Hint:} consider matrices $E_{ij}$ which are one on the $(i,j)$ position and zero otherwise.
\end{exercise}
Consider the algebra $\alg{A}$ as in the exercise, together with a given time evolution $\alpha_t(A) = e^{i t H} A e^{-i t H}$. It turns out that the Gibbs state at inverse temperature $\beta$ is in fact the unique KMS state with respect to this action. This can be seen as follows (following the argument in~\cite{MR1239893}). Let $\omega$ be a KMS state on $\alg{A}$. Note that since $\alg{A}$ is finite dimensional, each element is analytical. This can be seen by noting that for the derivation $\delta$ we have $\delta(A) \leq 2 \|H\| \|A\|$, that is, it is bounded. By the KMS condition and Proposition~\ref{prop:kmsanalytic} we have $\omega(AB) = \omega(B\alpha_{i\beta}(A))$. Define $\widetilde{\omega}(A) = \omega(e^{\beta H} A)$. We then have
\[
\begin{split}
	\widetilde{\omega}(AB) = \omega(e^{\beta H} AB) = \omega\left(e^{\beta H} A \alpha_{i\beta}\left(e^{\beta H} B e^{-\beta H}\right)\right) = \omega\left(e^{\beta H} B A\right) = \widetilde{\omega}(BA),
\end{split}
\]
where we used the KMS condition and $\alpha_{i \beta}\left(e^{\beta H} B e^{-\beta H}\right) = B$. Hence by the exercise above, $\widetilde{\omega}(A) = \lambda \Tr(A)$ for some constant $\lambda$. The constant can be found by plugging in the identity and it follows that $\omega$ must be the Gibbs state at inverse temperature $\beta$.

In infinite dimensions, things are more complicated. Nevertheless, we can use the Gibbs states for finite systems to obtain KMS states. It is however no longer true in general that each KMS state can be obtained in this way. This has an important physical interpretation. The argument above shows that KMS states on finite volumes are unique, so if each KMS state at temperature $\beta$ could be obtained as a limit of such states, it would also be unique. That there are other KMS states leaves open the possibility of a \idx{phase transition}, where there is a critical temperature $\beta_c$ above which there is a unique thermal state, but below which this uniqueness breaks down and there are distinct equilibrium states for the same temperature.

To see how KMS states on the infinite system can be obtained from Gibbs states, we will require a basic result in functional analysis, the \idx{Hahn-Banach theorem}. It says that we can extend $\omega_\Lambda$ to a state $\omega_\Lambda^G$ on all of $\alg{A}$. This extension is not unique in general, but for our purposes existence is enough. We write $\omega^G_\Lambda$ for any such extension. Now consider an increasing sequence $\Lambda_1 \subset \Lambda_2 \subset \dots$ of finite subsets of $\Gamma$ such that their union is equal to $\Gamma$. This gives a sequence $\omega^G_{\Lambda_n}$ of (extensions of KMS) states as above. Since the state space of $\alg{A}$ is weak-$*$ compact, there is a subsequence $\omega_{\Lambda_{n_k}}^G$ converging to a state $\omega$ on $\alg{A}$. Since the convergence is in the weak-$*$ topology, by definition for local observables $A \in \alg{A}(\Lambda)$ we have
\[
	\omega(A) = \lim_{k \to \infty} \omega_{n_k}(A),
\]
where it is understood that the limit on the right hand side is only over states corresponding to sufficiently large $\Lambda_{n_k}$ (i.e., those containing the support of $A$). We say that $\omega$ is the \idx{thermodynamic limit} of local Gibbs states.

If $\Phi$ is a strictly local and bounded interaction, we have seen before that $D(\delta) = \alg{A}_{loc}$ for the corresponding derivation. We can then argue as in Theorem~\ref{thm:dynconv} and Proposition~\ref{prop:limkms} that the thermodynamic limit of the local Gibbs states defined above is a $\beta$-KMS state for the time evolution generated by $\Phi$. In other words, we have the following result:

\begin{corollary}\index{KMS state!existence}
Let $\Phi$ be a strictly local and bounded interaction. Then there is a KMS state for each $\beta > 0$ with respect to the dynamics generated by $\Phi$. Moreover, the set of ground-states is non-empty.
\end{corollary}

\begin{remark}
	Again, the same result can be stated for more general interactions. For example, instead of the convergence of the derivations for a fixed $A$, as required in Proposition~\ref{prop:limkms}, can be loosened to demand that $\lim_{\Lambda \to \infty} \| \alpha_t^{\Lambda}(A) - \alpha_t(A)\| \to 0$. This implies that there is a sequence $A_n \to A$ such that for the corresponding derivations we have $\delta_n(A_n) \to \delta(A)$. Note that this allows for the situation where $D(\delta_n)$ is different for each $n$.
\end{remark}

\subsection{Von Neumann entropy}
In classical statistical mechanics the equilibrium states are characterised by minimising the free energy. We will give a quantum version of this statement, supporting the definition of KMS states. To define the free energy we first have to define the \idx{entropy} of a state. For us it will suffice to define this only for states that are given by density matrices in finite dimensional $C^*$-algebras. The entropy can be defined for states on more general $C^*$ (or von Neumann) algebras, but that requires considerable additional machinery. The interested reader can consult~\cite{MR1230389} for the basics.

To calculate the entropy, one has to make sense of expressions of the form $\log A$ for some operator $A$. But we have such tools available: it is the \index{functional calculus}functional calculus developed in Section~\ref{sec:funccalc}. Since $\log(x)$ is continuous on $(0,\infty)$, it allows us to define $\log(\rho)$ for all $\rho$ with $\sigma(\rho) \subset (0,\infty)$. We will also need to consider $\rho$ such that $0 \in \sigma(\rho)$, however. Fortunately, to define the entropy why only need to define $\rho \log(\rho)$, which can be done using the functional calculus applied to the function $x \mapsto x \log(x)$, where we set $0 \log(0) = 0$. Note that his function is continuous on $[0,\infty)$, so we can apply it to \emph{any} positive operator.

	It is instructive to carry out this procedure for a simple example (cf. Exercise~\ref{ex:pvm}). Consider $\alg{A} = M_n(\mathbb{C})$ and let $\rho \in \alg{A}$ be a self-adjoint matrix. Then there is a complete set of eigenvalues $\lambda_i$ with corresponding orthonormal basis $\ket{\psi_i}$. Moreover, $\rho = \sum_{i=1}^n \lambda_i \ketbra{\psi_i}$. Since the $\ketbra{\psi_i}$ are mutually commuting orthogonal projections, it is easy to calculate that for a polynomial $p$ one has $p(\rho) = \sum_{i=1}^n p(\lambda_i) \ketbra{\psi_i}$. It follows that $f(\rho) = \sum_{i=1}^n f(\lambda_i) \ketbra{\psi_i}$ for an arbitrary continuous function $f$ on the spectrum of $A$.

	Consider for the moment a finite dimensional system with observable algebra $\alg{A} = M_d(\mathbb{C})$. By Exercise~\ref{exc:density} it follows that any state on $\alg{A}$ can be uniquely written in the form $\omega(A) = \Tr(\rho A)$ for some density operator $\rho$. Von Neumann defined the entropy for such a density matrix $\rho$ as
\begin{equation}
	S(\rho) = -\Tr(\rho \log \rho).
\end{equation}
Since $\rho$ is positive, this is well defined as explained above. This coincides with the classical entropy of a probability distribution on $\{1,\dots,n\}$ with $p(i) = \lambda_i$. The von Neumann entropy plays an important role in quantum information theory, just as its classical analogue, the Shannon entropy, does in classical communication theory. See~\cite{MR1796805} for more details.

\begin{exercise}
	Let $\alg{A} = M_d(\mathbb{C})$. Show that if $\rho \in \alg{A}$ is a density matrix, $S(\rho) \in [0, \log(d)]$. Which states minimize or maximize the entropy?
\end{exercise}

Note that we have defined the entropy only for finite dimensional algebras. In the quantum spin systems we have been discussing, however, we are dealing with states on the quasi-local algebra, for which no entropy is defined (because in general they are not given by a density matrix $\rho \in \alg{A}$). Therefore we will look at the entropy in a \emph{finite} subset of the system.
\begin{definition}
	Let $\omega$ be a state on a quasi-local algebra $\alg{A}$, where the local algebras $\alg{A}(\Lambda)$ are finite dimensional matrix algebras. Let $\Lambda$ be a finite set. Then we define the entropy $S_\Lambda(\omega)$ by\index{entropy}\index{_SLambda@$S_\Lambda(\omega)$}
\[
	S_\Lambda(\omega) = - \Tr( \rho_\Lambda \log \rho_\Lambda ),
\]
where $\rho_\Lambda$ is the unique density matrix such that $\omega(A) = \Tr(\rho_\Lambda A)$ for all $A \in \alg{A}(\Lambda)$.
\end{definition}

Finally, we consider the \idx{relative entropy} of two states. If $\rho$ and $\sigma$ are the density matrices of two states, their relative entropy is defined as
\[
	S_{rel}(\sigma | \rho) = \Tr(\rho \log \rho - \rho \log \sigma).
\]
One can show that $S_{rel}(\rho | \sigma) \geq 0$, where we have equality if and only if $\rho = \sigma$. Again, one can define the relative entropy of two states on the quasi-local algebra, when restricted to a finite region $\Lambda$, in a straightforward way.

\subsection{Free energy}
Consider a finite dimensional system and a time evolution given by a Hamiltonian $H$. Moreover, suppose that $\rho$ is a density matrix giving a state of the system. We will also write $\rho$ for the corresponding state. Let $\beta > 0$. Then the \idx{free energy} is given by
\[
	F_\beta(\rho) := \rho(H) - \beta^{-1} S(\rho).
\]
Write $\rho_\beta$ for the Gibbs state at inverse temperature $\beta$. Then we have
\begin{equation}
	\label{eq:srelfree}
	S_{rel}(\rho_\beta|\rho) = \beta(F_\beta(\rho)-\Phi(\beta)).
\end{equation}
Here $\Phi(\beta)$ is the value of the free energy on the \idxr{Gibbs state}. Since $S_{rel} \geq 0$, it follows that the free energy is minimal when $\rho$ is the Gibbs state at inverse temperature $\beta$. Because $S_{rel}(\rho_\beta | \rho) = 0$ implies that $\rho_\beta = \rho$ it follows that the Gibbs state can be alternatively defined as the state that minimizes the free energy. Note that this is similar to the situation in classical statistical mechanics.

\begin{exercise}
Show that $\Phi(\beta) = -\beta^{-1} \log \Tr(\exp(-\beta H))$ and verify that equation~\eqref{eq:srelfree} indeed holds.
\end{exercise}

\subsection{Translationally invariant states}\index{state!translationally invariant|(}
It would be interesting to have an analogue of the characterisation of equilibrium states as those that minimise the free energy for infinite systems. However this result does not easily generalise to such systems: in general the entropy scales proportionally to $|\Lambda|$, so that it is not very useful to define the entropy for the whole infinite system as the limit of $S_\Lambda(\omega)$ as $\Lambda \to \infty$. One can still look at the \emph{entropy per unit of volume},\index{entropy!per unit volume} at least in the case of translationally invariant systems, defined as $S(\omega) := \lim_{\Lambda \leadsto \infty} S_\Lambda(\omega)/|\Lambda|$. Here $\Lambda \leadsto \infty$ is the limit \emph{in the sense of Van Hove}. This is defined as follows. Assume that $\Gamma = \mathbb{Z}^d$, and choose a vector $(a_1, a_2, \dots, a_d) \in \mathbb{Z}^d$, where each $a_i > 0$. This defines a box $\Lambda_a$ as follows:
\[
\Lambda_a = \{ x \in \mathbb{Z}^d : 0 \leq x_i < a_i, i = 1 \dots d\}.
\]
We can translate each box over a vector $na$ (with $n \in \mathbb{Z}$) to obtain a partition of $\mathbb{Z}^d$ into boxes. If $\Lambda$ is any finite subset of $\mathbb{Z}^d$, we write $n^+_\Lambda(a)$ for the number of (translated) boxes as above that have non-empty intersection. Similarly, $n^-_\Lambda(a)$ is the number of boxes that are contained in $\Lambda$. We then say that a sequence of sets goes to infinity in the sense of Van Hove, notation $\Lambda_n \leadsto \infty$ (or simply without the subscript $n$), if
\[
\lim_{n \to \infty} \frac{n^+_{\Lambda_n}(a)}{n^-_{\Lambda_n}(a)} = 1
\]
for all $a$. Note that this is a weaker form of convergence than $\Lambda \to \infty$, where we only allow a certain sequence of sets $\Lambda_1 \subset \Lambda_2 \subset\dots$. Remark that if $\Lambda_n \leadsto \infty$ then $\cup_n \Lambda_n = \mathbb{Z}^d$. This can be seen by taking $a$ larger and larger.

Using this notion of convergence we can make sense of the entropy per unit of volume in the case of translationally invariant systems. More precisely, we have the following proposition, whose proof can be found in~\cite{MR1441540}.
\begin{proposition}
For a translation invariant state $\omega$ the following limit exists
\[
S(\omega) = \lim_{\Lambda \leadsto \infty} \frac{S_\Lambda(\omega)}{|\Lambda|}.
\]
It is equal to $\inf_{a \in \mathbb{Z}^d} S_{\Lambda_a}(\omega)/|\Lambda_a|$.
\end{proposition}

In a similar manner we can define the other component in the free energy defined above: the energy of a state. Again we look at translationally invariant systems. In addition we assume that we are given an interaction $\Phi$ that is also translationally invariant. The interaction should satisfy the bound
\[
	\|\Phi\| = \sum_{\Lambda \ni 0} \frac{\|\Phi(\Lambda)\|}{|\Lambda|} < \infty,
\]
which is certainly the case if $\Phi$ is of finite range. Again, for such systems we can define the mean energy per unit volume. In particular we have the following proposition.
\begin{proposition}\label{prop:minfree}
Let $\Phi$ be as above and suppose that $\omega$ is an invariant state. Then the following limit exists:
\[
H_\Phi(\omega) = \lim_{\Lambda \leadsto \infty} \frac{\omega(H^\Phi_\Lambda)}{|\Lambda|},
\]
where $H^\Phi_\Lambda$ are the local Hamiltonians. The limit is equal to $\omega(E_\Phi)$, where
\[
	E_\Phi = \sum_{\Lambda \ni 0} \frac{\Phi(\Lambda)}{|\Lambda|}.
\]
\end{proposition}
It follows that for translationally invariant states and ditto interactions we can again define the free energy functional\index{free energy functional} by
\[
F_\beta(\omega) = H_\Phi(\omega) - \beta^{-1} S(\omega).
\]
Recall that in the finite dimensional case the KMS states were precisely those states that minimize the free energy functional (since the KMS states are precisely the Gibbs states in that case). This is true in the present setting as well, at least if we make a slightly stronger assumption on the behaviour of the interaction. The result can be stated as follows (again we omit the proof).
\begin{theorem}\label{thm:transinvgs}
Let $\Phi$ be a translationally invariant interaction such that there is some $\lambda > 0$ such that
\[
\|\Phi\|_\lambda = \sum_{\Lambda \ni 0} \| \Phi(\Lambda) \| e^{\lambda |\Lambda|} < \infty.
\]
Write $\alpha$ for the corresponding time evolution. Then the following are equivalent for a translationally invariant state:
\begin{enumerate}
	\item $\omega$ is a KMS state at inverse temperature $\beta$ with respect to $\alpha$;
	\item $\omega$ minimizes the free energy functional $\sigma \mapsto F_\beta(\sigma)$.
\end{enumerate}
\end{theorem}

For ground states there is a similar result. Under the same assumptions as in the theorem, one can show that a translationally invariant state is a ground state for the dynamics if and only if $\omega$ minimizes $H_\Phi(\omega)$.\index{ground state!translationally invariant}

\begin{remark}\label{rem:surface}
	We have restricted ourselves to translationally invariant states, but similar results are also true for general systems (under suitable assumptions on the interaction $\Phi$). In this case the mean energy and mean entropy lose their meaning. Instead one can look at arbitrary bipartitions of the system in a finite part and its complement. Then one can calculate certain entropic quantities and energies. In the calculation of the energies one has to be careful, since there are contributions coming from the interactions at the boundary of $\Lambda$, the so-called \emph{surface elements}. Using these quantities again a free energy functional can be defined, now depending on the region $\Lambda$. Given a state $\omega$, one can then look at the set $C_\omega^\Lambda$ of all states that agree with $\omega$ when restricted to $\Lambda^c$. If for each region $\Lambda$ the free energy functional applied to $\omega$ is the same as the infimum of the value on all states in $C_\omega^\Lambda$, then $\omega$ is a KMS state (and vice versa). Details can be found in~\cite[Sect. 6.2.3]{MR1441540}.
\end{remark}
\index{state!translationally invariant|)}

\subsection{Thermalisation and the second law of thermodynamics}\index{thermalisation}
We end the discussion of KMS states with two properties relating them to thermodynamical questions. First of all, suppose we have a large (in our case, infinite) system in equilibrium at inverse temperature $\beta$. We then consider another small (finite) system in any given state and let the systems interact. In other words, we consider a quantum system coupled to a heath bath in some way. One expects that after some time, the small system will be driven to a thermal state at the same inverse temperature. This is called the \emph{zeroth law of thermodynamics}. It turns out that this is true and in fact this condition uniquely characterizes the KMS states (under certain assumptions on the coupling between the two systems). This was first proven in~\cite{MR0468989}. An overview of the main arguments can be found in~\cite[pp. 114--117]{MR1919619}. Here we will content ourselves with a brief sketch on how to describe these systems.

As for the reservoir system, we assume that it is described by a $C^*$-dynamical system $(\alg{A}_R, \alpha_R)$ with corresponding generator $\delta_R$. The small system is described by a finite dimensional algebra $\alg{A}_f = M_d(\mathbb{C})$ with time evolution governed by a Hamiltonian $H_f$. The system as a whole is then described by the observable algebra $\alg{A} = \alg{A}_f \otimes \alg{A}_{R}$. We can define a derivation describing the evolution when there is no coupling between the systems. It is just given by
\[
	\delta(A \otimes B) = i ([H,A]) \otimes I + I \otimes \delta_R(B).
\]
One can check that this indeed generates a time evolution of the form $\alpha_t = \alpha_t^f \otimes \alpha_t^R$. To introduce a coupling between the systems, we have to add a term that acts on both sides at the same time. We assume that there is an interaction of the form $\lambda V$, where $\lambda$ is some parameter and $V$ is of the form
\[
	V = \sum_{j=1}^n A_j \otimes A_j^R.
\]
We assume that $V$ is self-adjoint and in addition that the reservoir initially is in a state $\omega$. We  demand that $\omega(A_j^R) = 0$ for all $j$. The dynamics of the perturbed system is then generated by the derivation $\delta(\cdot) + i \lambda [V, \cdot]$, where $\delta$ is the derivation of the uncoupled system as above. Now suppose that the finite system is in an initial state $\rho$, such that the combined system is in the state $\rho \otimes \omega$. Under the time evolution the state will also undergo an evolution. At a later point $t$ we might be interested in the state of the finite system. In general, the system as a whole is not in a product state any more, because of the interactions between the bath and the system. Nevertheless, we can still ``forget'' about the big system and obtain a state $\rho_t$ of the small system. This can be done with the help of a \idx{conditional expectation}.\footnote{This is a generalisation of the partial trace operation.} 

\begin{definition}
	Let $\alg{A} \subset \alg{B}$ be two $C^*$-algebras. A positive linear map $\mathcal{E}: \alg{A} \to \alg{B}$ is called a \emph{conditional expectation} if $\mathcal{E}(I) = I$ and $\mathcal{E}(ABC) = A \mathcal{E}(B) C$ for all $A,C \in \alg{B}$ and $B \in \alg{A}$.
\end{definition}

The conditional expectation maps observables in $\alg{B}$ to observables in $\alg{A}$. Hence we can interpret $\mathcal{E}$ as an operation that limits our abilities to measure the system. In the particular case that we are considering here it amounts to ``forgetting'' about the heat bath our system is coupled to.

\begin{exercise}
	Set $\alg{A}_i = M_{n_i}(\mathbb{C})$ for some positive integers $n_1, n_2$ and let $\alg{A} = \alg{A}_1 \otimes \alg{A}_2$. Define a map $\mathcal{E}$ on $\alg{A}$ via
	\[
		\mathcal{E}(A) := \int_{\mathcal{U}(I \otimes \alg{A}_2)} (UAU^*) dU.
	\]
	Here the integration is with respect to the \emph{Haar measure} on the group of unitaries in $I \otimes \alg{A}_2$. Note that since this group is compact, the Haar measure exists and can be normalised to have measure 1. It is also left (and right) invariant: $d(VU) = dU$ for every unitary $V$ in $I \otimes \alg{A}_2$.

	Show that $\mathcal{E} : \alg{A} \to \alg{A}_1 \otimes I$ is a conditional expectation. You may use the facts that every element in $\alg{A}_2$ can be written as a linear combination of unitaries, and that in this case $(I \otimes \alg{A}_2)' = \alg{A}_1 \otimes I$.\footnote{The prime denotes the commutant, see equation~\eqref{eq:commutant}.}
\end{exercise}

We can now look at the state of the small system at each time $t$ to see how it evolves. Because of the coupling this time evolution is in general not given by a unitary any more. Nevertheless, there is a map $\gamma_t$ such that $\rho_t = \gamma_t(\rho)$. Technically, $t \mapsto \gamma_t$ is a one-parameter \emph{semi}group of completely positive transformations on $\alg{A}_f$. In particular, it might not be reversible. The complete positivity has a physical interpretation. It means that if we add another (finite dimensional) quantum system $\alg{B}$ to it, that is not coupled to the rest of the system, we can extend $\gamma_t$ to a positive map on $\alg{A} \otimes \alg{B}$. Indeed, by definition of complete positivity, $\gamma_t \otimes \operatorname{id}_\alg{B}$ is a unital positive map (and hence maps states to states).

It is perhaps insightful to indicate how this map $\gamma_t$ can be obtained. Suppose that the full dynamics on the whole system are implemented by a one-parameter group of unitaries $t \mapsto U_t$ (possibly acting on a GNS space of the whole system). Recall that the reservoir is in a state $\omega$. We can then define a map $P: \alg{A}_f \times \alg{A}_r \to \alg{A}_f$ via $P(A \otimes A_R) = A \omega(A_R)$. Then $\gamma_t : \alg{A}_f \to \alg{A}_f$ can be defined by 
\[
	\gamma_t(A) := P(U_t(A \otimes I) U_t^*),
\]
with $A \in \alg{A}$. It can be checked that, for a state $\rho$ on $\alg{A}_f$, $\rho(\gamma_t(A))$ is the expectation value of $A$ at time $t$ with the state at $t=0$ given by $\rho \otimes \omega$. See~\cite{Sewell2003} and references therein for a more complete discussion.

In this situation, it is shown in~\cite{MR0468989} that, roughly speaking, the following conditions are equivalent:
\begin{enumerate}
	\item $\omega$ is a KMS state at inverse temperature $\beta$,
	\item There is a unique state $\rho$ on the finite system that is invariant under the dynamics induced by the coupled system. This is the Gibbs state at inverse temperature $\beta$. Moreover, any initial state $\rho_0$ converges to $\rho$ as $t \to \infty$.
\end{enumerate}
Hence, as expected, the state of the small system will tend to the (unique) equilibrium state at the same temperature as the big reservoir to which it is coupled.

Finally we have a look at the second law of thermodynamics, which says that one cannot extract energy from a system at thermal equilibrium in a cyclic process (that is, by returning back to the equilibrium dynamics). Again we first investigate the question for a finite dimensional system governed by a Hamiltonian $H$. Let $\omega$ be the corresponding Gibbs state for $\beta > 0$. Note that the energy of the system is given by $\omega(H)$. We then apply a unitary operation $U$ to the system, to obtain a state $\omega_U(A) := \omega(U^*AU)$. The energy difference between the new and the old state is then
\begin{equation}
	\label{eq:energydifference}
		\omega_U(H) - \omega(H) = \omega(U^*HU)- \omega(H) = -i \omega(U^*\delta(U)),
\end{equation}
where $\delta(U) = i [H,U]$ as usual. Now represent $\omega$ by a density operator $\rho$. Then $\omega_U$ is represented by a unitary equivalent density operator. But this does not change the eigenvalues of $\rho$, and it follows that $S(\omega) = S(\omega_U)$. Therefore we have the following inequality for the free energy\index{free energy}:
\[
F_\beta(\omega_U)-F_\beta(\omega) = \omega_U(H)-\omega(H) \geq 0.
\]
The inequality follows because the Gibbs state minimizes the free energy. Hence the energy of the new state cannot be lower than the energy of the equilibrium state. By equation~\eqref{eq:energydifference} this is equivalent to $-i \omega(U^*\delta(U)) \geq 0$. This condition can be generalised to infinite systems in a straightforward way, leading to the condition of \emph{passivity}.
\begin{definition}
	Let $(\alg{A}, \alpha)$ be a $C^*$-dynamical system with generator $\delta$. Suppose that $\omega$ is an invariant state. Then $\omega$ is said to be \emph{passive} if 
\[
	-i \omega(U^* \delta(U)) \geq 0
\]
for all unitaries $U$ in the domain $D(\delta)$.\index{state!passive}
\end{definition}
Note that by definition ground states are passive. In a seminal work, Pusz and Woronowicz proved that in fact any KMS state is passive~\cite{MR0471796}.

\begin{theorem}[Pusz-Woronowicz]\index{Pusz-Woronowicz theorem}
Let $(\alg{A}, \alpha)$ be a $C^*$-dynamical system. Then every $\beta$-KMS state is passive.
\end{theorem}

The condition of passivity is related to the inability to let the equilibrium system perform work, as we will outline now. Consider dynamics $\alpha_t$ generated by some derivation $\delta$. At $t=0$ we start perturbing the dynamics. One can think of turning on a magnetic field, for example. We do this until some later time $T$, after which all perturbations are turned off again and we are back with the original dynamics. Mathematically this can be described as follows. At times $0 \leq t \leq T$, the dynamics is not generated by $\delta$ alone, but by
\[
	\delta^P_t(A) = \delta(A) + i [P_t, A].
\]
Note that if the dynamics is generated by a Hamiltonian, this amounts to adding an extra term $P_t$ to the Hamiltonian. We will assume that $P_t : \mathbb{R} \to \alg{A}$ with $P_t = P_t^*$ for all $t$. Moreover, $P_t = 0$ if $t \leq 0$ or $t \geq T$. Finally, we will assume that $t \mapsto P_t$ is differentiable.

The introduction of perturbations means that the time evolution is modified: for each $t$ we have an automorphism $\alpha_t^P$ describing the time evolution. Note that $t \mapsto \alpha_t^P$ does not need to be a one-parameter group any more! Since the dynamics is now time-dependent, finding the automorphisms $\alpha_t^P$ is not as straightforward as in the time-independent case. Nevertheless, one can show that they exist and are related to the original dynamics.
\begin{theorem}
Let $(\alg{A}, \alpha)$ be a $C^*$-dynamical system with generator $\delta$. Let $t \mapsto P_t$ be as above. Then there is a unique one-parameter family of $*$-automorphisms $\alpha_t^P$ solving the differential equation
\[
\frac{d\alpha_t^P}{d t}(A) = \alpha_t^P(\delta(A) + i [P_t, A])
\]
for all $A \in \delta(D)$ and with initial condition $\alpha_0^P(A) = A$. Moreover, for each $t$ there is a unitary $\Gamma^P_t \in \alg{A}$ such that
\[
	\alpha_t^P(A) = \Gamma_t^P \alpha_t(A) (\Gamma_t^P)^*
\]
for all $A \in \alg{A}$.
\end{theorem}
One can give an explicit expression for $\Gamma_t^P$ in terms of a perturbation series (see~\cite[Sect. 5.4.4]{MR1441540}).

Perturbing the system means that work is performed by some external forces. The question is then, how much work is actually performed? Suppose that initially the system is in a state $\omega$ at $t=0$. At some later time the state has evolved to $\omega_t = \omega \circ \alpha_t^P$. To estimate the change in energy we divide the interval $[0,T]$ up into $N$ intervals $(t_i, t_{i+1})$. If $N$ is big enough we can assume that the state $\omega_t$ doesn't change much over the time of an interval. The difference in energy is then given by $\omega_{t_i}(P_{t_i}-P_{t_{i-1}})$. Summing over all intervals we obtain
\[
	L^P(\omega) = \sum_{i=1}^N \omega_{t_i}(P_{t_i}-P_{t_{i-1}}).
\]
In the limit $N \to \infty$ this can be expressed as an integral. 

Alternatively we can argue as in the finite case. The total amount of work is given by the energy of the state $\omega_T$ minus the energy of the state $\omega_0 = \omega$. Note that $\omega_T(A) = \omega(\Gamma^P_T \alpha_T(A) (\Gamma^P_T)^*)$. One then has a similar result for the energy difference of the two states as in the case of finite systems discussed above. That is, it is given by $-i \omega(\Gamma^P_T \delta ((\Gamma^P_T)^*)$. In fact it is equal to $L^P(\omega)$ defined above (where strictly speaking we should use the integral formulation obtained from $N \to \infty$), that is
\[
	L^P(\omega) = -i \omega(\Gamma^P_T \delta( (\Gamma^P_T)^*).
\]
This gives us the second law of thermodynamics for KMS states. Since KMS states are passive, $-i \omega(\Gamma^P_t \delta( (\Gamma^P_t)^*)) \geq 0$, hence the system has more (or equal) energy after the ``cycle'' of perturbing the system and going back to the original dynamics. It is therefore not possible to extract energy from the system in a cycle. This gives further evidence that the notion of a KMS state is indeed the right one to describe thermal equilibrium states.

\begin{remark}
Under some additional assumptions a converse is also true. That is, under these assumptions one can show that a passive state is a KMS state. The extra assumptions are necessary because the set of passive states is convex. For example, a convex combination of two KMS states at different temperatures is a passive state, but it is clearly not a KMS state at a certain temperature. Therefore one needs to impose additional assumptions to somehow pick out the ``pure phases''.
\end{remark}

\section{The toric code}\label{sec:toriccode}
We now apply the results and techniques developed in this chapter to an important example in quantum information theory, the \idx{toric code}. This model was first introduced by Kitaev~\cite{MR1951039}. The reason that it is called the toric code is two-fold: the model is often considered on a torus (i.e., as a finite system with periodic boundary conditions in the $x$ and $y$ direction), and it is an example of a \emph{quantum code}. Quantum codes are used to store quantum information and correct errors. We will only make some brief comments later on this aspect.

\begin{figure}
	\begin{center}
	\begin{tikzpicture}
		\draw[step=1cm,gray,very thin] (-4.4,-2.4) grid (4.4,3.4);

		\foreach \x in {-4,-3,...,4}
			\foreach \y in {-2,-1,...,3}
			{
				\filldraw[fill=gray,color=gray] (\x,\y) circle [radius=0.05];
			}
		
		\filldraw[fill=black] (-2,1) circle [radius=0.05];
		\filldraw[fill=black] (-3,1) circle [radius=0.05];
		\filldraw[fill=black] (-1,1) circle [radius=0.05];
		\filldraw[fill=black] (-2,2) circle [radius=0.05];
		\filldraw[fill=black] (-2,0) circle [radius=0.05];
		\draw[line width=2pt,dashed] (-3,1) -- (-1,1);
		\draw[line width=2pt,dashed] (-2,2) -- (-2,0);

		\filldraw[fill=black] (1,1) circle [radius=0.05];
		\filldraw[fill=black] (1,2) circle [radius=0.05];
		\filldraw[fill=black] (2,2) circle [radius=0.05];
		\filldraw[fill=black] (2,1) circle [radius=0.05];
		\draw[line width=2pt,solid] (1,1) -- (1,2);
		\draw[line width=2pt,solid] (1,2) -- (2,2);
		\draw[line width=2pt,solid] (2,2) -- (2,1);
		\draw[line width=2pt,solid] (2,1) -- (1,1);
	\end{tikzpicture}
\end{center}
\caption{Lattice on which Kitaev's toric code is defined. The spin degrees of freedom live on the \emph{edges} between the solid dots. Also indicated are a star (thick dashed lines) and a plaquette (thick solid lines).}
\label{fig:lattice}
\end{figure}
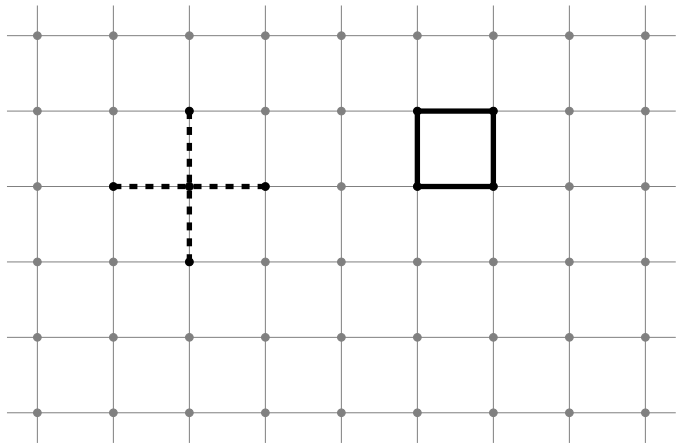
Instead of on a torus, we will consider the model on an infinite plane. That is, consider the lattice $\mathbb{Z}^2$. The set $\Gamma$ of sites is defined to consist of the \emph{edges} between nearest-neighbour points in the lattice (see Figure~\ref{fig:lattice}). At each of these edges there is a spin-1/2 degree of freedom, with corresponding observable algebra $M_2(\mathbb{C})$. We can then define the quasi-local algebra $\alg{A}(\Gamma)$ as before. Note that there is a natural action of the translation group, so that it makes sense to talk about translation invariant interactions or states.

There are two special subsets of sites that we will consider. For any vertex $v$ there are in total four edges that begin or end in that vertex (see the picture). Such a set will be called a \idx{star}. Similarly, one can define a \idx{plaquette}, as the edges around a vertex in the dual lattice (see the solid lines in the picture). To a star $s$ and a plaquette $p$ we associate the following operators:
\[
	A_s = \bigotimes_{j \in s} \sigma^x_j, \quad B_p = \bigotimes_{j \in p} \sigma^z_j.
\]
An important property is that $[A_s, B_p] = 0$ for any star $s$ and plaquette $p$. This can be seen because a star and a plaquette always have an even number of edges in common. Commuting the operators at each edge gives a minus sign, because of the anti-commutation of Pauli matrices. Since the number of minus signs is even the claim follows. Another property is that $A_s^2 = B_p^2 = I$.

The star and plaquette operators will be used to define the interactions of the model. Namely, for $\Lambda \in \mc{P}_f(\Gamma)$ we set\index{toric code!interaction}
\[
\Phi(\Lambda) = \begin{cases}
	-A_s & \textrm{$\Lambda = s$ for some star $s$} \\
	-B_p & \textrm{$\Lambda = p$ for some plaquette $p$} \\
	0 & \textrm{else}
\end{cases}.
\]
Note that the interaction is of finite range and bounded. Moreover, it is translation invariant. By Theorem~\ref{thm:dynconv} it follows that there is a one-parameter group $t \mapsto \alpha_t$ describing the time evolution with respect to these dynamics. In the remainder of this chapter we will discuss the ground states of this model. We will show in detail that there is a unique translationally invariant ground state, where we largely follow the proof of~\cite{MR2345476}. We then discuss other (non-translational invariant) ground states.

The toric code is actually only one example of a large class of models defined by Kitaev~\cite{MR1951039}, usually called the \emph{quantum double} models. For example, to each finite group $G$ one can associate such a model. The toric code corresponds to the choice $G = \mathbb{Z}_2$, the smallest non-trivial group. Using similar methods as discussed here one can find the translationally invariant ground state for any finite group $G$~\cite{haagdouble}.

To define the invariant ground state, the following lemma will be useful.
\begin{lemma}\label{lem:ststate}
	Suppose that $\alg{A}$ is a unital $C^*$-algebra and let $\omega$ be a state on $\alg{A}$. Suppose $X \in \alg{A}$ is such that $X^*=X$, the operator $I-X$ is positive, and $\omega(X) = 1$. Then for each $Y \in \alg{A}$ we have that $\omega(XY) = \omega(YX) = \omega(Y)$.
\end{lemma}
\begin{proof}
Because $I-X$ is positive we can take the square root $\sqrt{I-X}$ (compare with footnote~\ref{ftn:sqrt} on page~\pageref{ftn:sqrt}), which is also self-adjoint. Hence we have
\[
|\omega( (I-X) Y)|^2 = |\omega( \sqrt{I-X} \sqrt{I-X} Y) |^2 \leq \omega(I-X) \omega(Y^*(X-I)Y) = 0, 
\]
where we used Lemma~\ref{lem:csineq}. The equality $\omega( Y(I-X)) = 0$ is proved similarly, from which the claim follows.
\end{proof}

With the help of this lemma it is not so difficult to show that any state with $\omega(A_s) = \omega(B_p) = 1$ is a ground state.
\begin{proposition}
Let $\omega$ be a state on $\alg{A}$ such that $\omega(A_s) = \omega(B_p) = 1$ for all stars $s$ and all plaquettes $p$. Then $\omega$ is a ground state of the toric code.
\end{proposition}
\begin{proof}
Let $A \in \alg{A}(\Lambda)$ be a local operator. Write $\delta$ for the derivation corresponding to the toric code interaction. Because of locality, most terms in $\delta(A)$ vanish. We obtain
\[
\delta(A) = - i \sum_{s \cap \Lambda \neq \emptyset} [A_s, A] - i \sum_{p \cap \Lambda \neq \emptyset} [B_p, A],
\]
where the sum as over all stars $s$ (plaquettes $p$) that have non-empty intersection with the support of $A$. Hence we have
\[
\begin{split}
	-i \omega(A^*\delta(A)) =  \sum_{s \cap \Lambda \neq \emptyset} & \omega(A^*A A_s) - \omega(A^*A_s A) + 
		\\ &\sum_{p \cap \Lambda \neq \emptyset} \omega(A^*A B_p)-\omega(A^* B_p A).
\end{split}
\]
But $I-A_s$ and $I-B_p$ are positive operators, since they are positive multiples of projections, so that we can use Lemma~\ref{lem:ststate} to write $\omega(A^*A A_s) = \omega(A^*A)$ (and similarly for the $B_p$ term). 

We are done if we can show that $\omega(A^*A) - \omega(A^* A_s A)$ is positive, since the analogous result for $B_p$ follows from the same argument. But we have $\omega(A^*(I-A_s)A) \geq 0$, because $A^*(I-A_s)A$ is positive, since $I-A_s$ is positive. By Theorem~\ref{thm:groundeq} it follows that $\omega$ is a ground state.
\end{proof}
Hence any state such that $\omega(A_s) = \omega(B_p) = 1$ is a ground state. This condition can be interpreted as the condition that the operators $A_s$ and $B_p$ \emph{stabilise} the ground state. Indeed, if $(\pi_0, \mc{H}, \Omega)$ is the corresponding GNS representation, it follows that $\pi_0(A_s) \Omega = \pi_0(B_p) \Omega = \Omega$. Our next goal is to show that there in fact exist states $\omega$ such that $\omega(A_s) = \omega(B_p) = 1$, which is a priori not clear, and that there in fact is a \emph{unique} state with this property. We construct this state $\omega$ by first considering a \emph{classical} Ising model. Note that this does \emph{not} imply that there is a unique ground state, since there may be ground states that do not satisfy $\omega(A_s) = \omega(B_p) = 1$ (and indeed, we will see examples of such states later).

To construct $\omega$, first consider the $C^*$-algebra $\alg{A}_{cl}$ that is generated by the operators $A_s$ and $B_p$ (for all stars $s$ and plaquettes $p$). Since all generators mutually commute, it follows that $\alg{A}_{cl}$ is abelian. We want to map this algebra to a algebra of observables of a classical Ising system. Recall that this is defined as follows. Consider a set $\Gamma_s$ of spins. In our case, we take this set to be equal to the set of stars on the lattice. A \emph{configuration} assigns to every spin $s \in \Gamma_s$ the value +1 (up) or -1. The set of all configurations $S$, which can be identified with the set of functions $f: \Gamma_s \to \{-1,1\}$, is called the configuration space. 

As discussed before, the set of observables is the set of (real-valued) functions on the configuration space that vanish at infinity. Here we will slightly extend this to all \emph{bounded} functions. This makes no essential difference here. A particular example are the functions $\sigma_s : S \to \mathbb{R}$, with $\sigma_s(f) = f(s)$. That is, it gives the value of the spin at site $s$. Note that $\sigma_s^2 = 1$ and $\sigma_s^* = \sigma_s$. Finally, clearly all of the observables $\sigma_s$ commute. Now we add another copy of the system, associated with the plaquettes of the Kitaev model. We write $T$ for the corresponding configuration space and $\tau_p$ for the corresponding observables that measure the spin at the plaquette $p$. Hence the configuration space $S \times T$ describes two copies of a free (non-interacting) Ising model. The idea is then to map the algebra $\alg{A}_{cl}$ to a algebra of observables of the two copies of the free Ising model. But since $\sigma_s$ and $\tau_p$ have the same algebraic properties as the $A_s$ and $B_p$, we can just map $A_s$ to $\sigma_s$ and $B_p$ to $\tau_p$. This gives an isomorphism of the corresponding algebras. In this way we can obtain a state on $\alg{A}_{cl}$ by giving a state on the Ising model. Consider the state where all spins are in the up direction (this corresponds to the Dirac measure concentrated in the corresponding element in $S \times T$). Hence we obtain a state $\omega_{cl}$ with $\omega_{cl}(A_s) = \omega_{cl}(B_p) = 1$.

We now have a state $\omega_{cl}$ on a abelian subalgebra $\alg{A}_{cl}$ of $\alg{A}$ with the desired properties. By the Hahn-Banach\index{Hahn-Banach theorem} extension theorem, it follows that there is a state $\omega$ on $\alg{A}$ that extends $\omega_{cl}$. By the proposition above it follows that $\omega$ is a ground state, and we will now argue that the conditions that $\omega(A_s) = \omega(B_p) = 1$ in fact completely determine $\omega$ (so that there is a \emph{unique} extension of $\omega_{cl}$ to $\alg{A}$). To see why this is true, we first consider an example. Suppose $A = \sigma_j^y$ for some site $j$. We want to know the value of $\omega(\sigma_j)$. Let $s$ be any star that contains the site $j$. By applying Lemma~\ref{lem:ststate}, we find
\[
	\omega(A) = \omega(A_s \sigma_j^y) = - \omega(\sigma_j^y A_s) = -\omega(A),
\]
hence $\omega(A) = 0$. The second equality follows because $\sigma_j^x$ and $\sigma_j^y$ anti-commute.

Now consider a general local element $A$. By linearity it suffices to show that we can calculate $\omega(A)$ in the case that $A$ is a monomial in the Pauli matrices, that is, 
\[
	A = \prod_{j \in \Lambda} \sigma_j^{k_j},
\]
where $\Lambda$ is the support of $A$ and $k_j = x,y,z$. The idea is to use the trick above repeatedly, in a systematic way. We first try to find all sites plaquettes $p$ and stars $s$ that have only one site in common with the support $\Lambda$ of $A$. For example, let $p$ be a plaquette such that $\Lambda \cap p = \{j\}$ for some site $j$. Then, by locality, we see that $B_p$ and $A$ \emph{commute} if $k_j = z$, and anti-commute otherwise. Similarly, for star operators we have a similar result with $k_j = x$. Hence by the argument above, if for any of these sites we have anti-commutation, $\omega(A) = 0$. This forces $k_j$ for all sites $j$ as above. Continuing in this way it follows that $\omega(A)$ is zero unless $A$ is a product of star and plaquette operators. In that case, clearly we have $\omega(A) = 1$ by a repeated application of Lemma~\ref{lem:ststate}. It follows that $\omega$ is completely determined on all local operators, and hence on all of $\alg{A}$ since the local operators are dense in $\alg{A}$ and $\omega$ is continuous (because it is a state). Note that the state is manifestly translationally invariant.

\index{toric code!ground state}
The state $\omega$ that we have constructed is the \emph{only} translationally invariant ground state. This can be seen as follows. Suppose that $\rho$ is any translationally invariant state. First of all, note that $\rho(A_s)$ and $\rho(B_p)$ are real because both operators are self-adjoint. By translational invariance, the value of $A_s$ is the same for all stars $s$, and similarly for $B_p$. Now, since $A_s^2 = I$, it follows from the Cauchy-Schwarz inequality, Lemma~\ref{lem:csineq}, that $\rho(A_s) \leq 1$. But if $\rho(A_s) < 1$, it is clear that $H_\Phi(\rho) > H_\Phi(\omega)$, where $H_\Phi$ is the mean energy as before. By the paragraph following Theorem~\ref{thm:transinvgs} it follows that $\rho$ cannot be a ground state. The argument for the case that $\rho(B_p) \neq 1$ is the same.

\begin{remark}
This ground state of the Kitaev model is actually special in the sense that it minimizes each of the terms in the Hamiltonian individually. A ground state for which this is true is called \emph{completely free of frustration}.\index{interaction!frustration free} This is true here, essentially, because all the terms in the Hamiltonian commute with each other. Hence one expects that they all can be diagonalised simultaneously.
\end{remark}

The unique translationally invariant ground state $\omega$ of the system is a pure state.\index{state!pure} To show that one can argue as follows. First restrict $\omega$ to the \emph{abelian} algebra $\mathfrak{A}_{XZ}$ generated by the $A_s$ and $B_p$. Using the characterisation of $\omega$ given above, it is not difficult to see that $\omega(AB) = \omega(A) \omega(B)$ for all $A,B \in \alg{A}_{XZ}$. Hence $\omega \upharpoonright \mathfrak{A}_{XZ}$ is a pure state, since it is a character (Proposition~\ref{prop:characters}). It is known that every pure state on a unital subalgebra of a $C^*$-algebra can be extended to a pure state of the whole algebra~\cite[Prop. 2.3.24]{MR887100}, and since as we have seen above that there is a unique extension in this case, $\omega$ must be pure. The discussion can be summarised in the following theorem.

\index{toric code!ground state}
\begin{theorem}
The toric code on the plane has a unique translationally invariant ground state $\omega$, which is a pure state. It is completely determined by the conditions $\omega(A_s) = \omega(B_p) = 1$.
\end{theorem}

The situation changes when we require that the ground states need not be translationally invariant, and one can construct other ground states.

\subsubsection{Ground states of the finite model}
Before we consider other ground states of the toric code it is instructive to consider the finite version of the model. In particular, there are differences regarding the ground state degeneracy. In the finite model, one takes a compact surface of genus $g$. The genus is rougly the number of holes, so a torus has genus one. On this compact surface one can draw a (finite) graph, where the edges represent the spin degrees of freedom.\footnote{One needs to demand some additional conditions on the graph to avoid pathological cases.} Again considering the case of the torus, one can think of a square lattice, and identify the left and right edges, as well as the top and bottom edges, to get a graph that can be embedded into the torus. Once one has the graph, the edges determine the Hilbert space of the system and one can define a Hamiltonian in the model in a  similar way as above. The ground state degeneracy can then to be shown to be equal to $4^g$. Note that this is consistent with the uniqueness of the ground state on the plane, since the plane has genus zero.

We will outline why this is the case in case of the torus. Note that for the torus, there are two different non-contractible loops on the surface. If we model the torus as a square with the edges identified, these loops correspond to a horizontal and a vertical line from edge to edge. We can form such a loop by connecting edges in the graph. Let us call one such loop $\xi_Z$. To this path $\xi_Z$ we associate an operator $F_{\xi_Z} = \bigotimes_{j \in \xi_Z} \sigma^z_j$. Note that $F_{\xi_Z}$ commutes with all star operators and all plaquette operators. Hence we cannot apply the trick above, so what $\omega(F_{\xi_Z})$ is not determined. This essentially gives an extra parameter in the definition of ground states. If we combine this with a loop on the dual lattice (with operators $\sigma^x_j$), and a combination of the two, we obtain a four dimensional ground state.

It has been the hope that this degenerate ground state space could be used to store quantum information (or qubits\index{qubit}). The idea is that we would encode the information in a certain state in the ground space. Since on the torus the ground state is four-fold degenerate, this would allow us two store two qubits, i.e. $\mathbb{C}^2 \otimes \mathbb{C}^2$. The hope was that this would lead to a reliable storage. That is, the state should not change because of small perturbations due to the environment, such as thermal noise. This is a natural requirement if one wants to store information for a longer time, just as with a classical hard disk you expect to be able to read your files again if you turn on your computer the next day. The essential idea now is that to transform the state $F_{\xi_Z} \Omega$, where $\Omega$ is a state in the ground space, to another (orthogonal) state in the ground space, one has to do a \emph{non-local} operation, where locality is interpreted here in terms of the system size. Here, to undo the effect of the path operator, one has to apply another operator going around the loop. If we make the system big, such non-local operations are very unlikely to occur as a result of thermal fluctuations~\cite{MR1924451}.

A more careful analysis unfortunately shows that that this is a bit too optimistic, see for example~\cite{MR2345476,MR2525435,terhaltoric,PhysRevLett.110.090502}.\footnote{This is only a small selection on the vast amount of literature on so-called self-correcting quantum memories.} For example, what can happen is that a local (pair) of excitations gets created. This is local, and will not destroy the information in the ground state. However, there is no energy penalty for moving the excitations around. This has the effect that they may spread out very fast, to get something that is \emph{not} local any more. It is now generally accepted that the toric code is not a good quantum memory, and finding a system that is, is one of the big challenges in quantum computing at the moment.

\subsubsection{Excitations of the toric code}
The other (non-translational invariant) ground states of the toric code are related to \emph{excitations}\index{toric code!excitation} of the model. This may seem counter-intuitive at first, but we will see later why this is true.

Let $\xi$ be a (finite) non-intersecting path on the lattice, that is, a set of edges $e_i$ such that the endpoint of $e_i$ connects to the initial point of $e_{i+1}$. Then, as above, we can define an operator $F_{\xi} := \bigotimes_{j \in \xi} \sigma_j^z$. That is, we act with $\sigma^z$ on each site $j$ of the path. The operator $F_{\xi}$ will be called a \idx{path operator}. Note that $F_{\xi}$ commutes with all plaquette operators $B_p$. The same is true for all star operators $A_s$ as well, except for two cases: if $s$ is either based on the start of the path, or at the end. In that case, the star $s$ will have an odd number of edges in common with the path $\xi$, and hence $A_s$ and $F_\xi$ are easily seen to anti-commute.

Now consider a ground state $\Omega$ of the finite-dimensional model. Since all terms in the Hamiltonian commute, we can diagonalise them all at the same time, and we see that a ground state necessarily has $A_s \Omega = B_p \Omega = \Omega$. Now let us look at what happens with the energy if we apply $F_\xi$ to the ground state. By the observation above, it is easy to check that
\[
	H F_\xi \Omega = F_\xi (2 A_{s_1} + 2 A_{s_2} + H) \Omega = (4+E_0) F_\xi \Omega,
\]
where $E_0$ is the ground state energy, i.e. $H \Omega = E_0 \Omega$, and $s_1$ and $s_2$ are the stars at the endpoints of $\xi$. Hence $F_\xi \Omega$ is an excited state, and the extra energy comes precisely from the stars at the endpoints of $\xi$. It therefore makes sense to think of the state $F_\xi \Omega$ as describing a configuration where there is an excitation located at the star $s_1$, and one located at $s_2$.

A similar thing can be done with \emph{dual paths}, which are paths on the dual lattice. In other words, such paths consists of (dual) edges between the centres of the plaquettes. Note that these dual edges cross precisely one edge in the lattice. We identify the dual path $\widehat{\xi}$ with these edges in the lattice. In that way we can define the dual path operator\index{path operator} $F_{\widehat{\xi}} := \bigotimes_{j \in \widehat{\xi}} \sigma^x_j$. They commute with all star operators and with all plaquette operators, except for the $B_p$ at the ends of the path. We think of $F_{\widehat{\xi}} \Omega$ as a state with an excitation at each plaquette at the end of $\widehat{\xi}$. Hence there are two types of excitations in the toric code: those located at stars, and those located at plaquettes.

\begin{exercise}
	\label{ex:pathindep}
	Let $\Omega$ be a ground state vector of the toric code. Show that for two paths $\xi_i$ which have the same endpoints, it holds that $F_{\xi_1} \Omega = F_{\xi_2} \Omega$. (Hint: consider operators of the form $A_s F_{\xi_i}$.) Prove a similar result for dual paths.
\end{exercise}

This exercise, together with the energy calculation above, shows that we can use path operators to create pairs of excitations, and the resulting state does not depend on the specific choice of path.\footnote{It should be noted that this is only true if we just have two excitations. As soon as we use path operators to create more excitations, that is states of the form $F_{\xi_1} F_{\xi_2} \Omega$, they are only unique up to a phase. This phase depends on whether paths associated to different types of charges cross or not, and is due to the anyonic nature of the excitations. We will discuss this in more detail in Section~\ref{sec:sselecttoric}.} Also note that we can use the path operators to move the excitations around. Let $\xi_1$ and $\xi_2$ be two paths on the lattice, sharing an endpoint. If we write $\xi_1 \xi_2$ for the concatenated path, it is clear from the definition that $F_{\xi_1} F_{\xi_2} = F_{\xi_1 \xi_2}$. So if we have a two-excitation state $F_{\xi_1} \Omega$, and we apply the unitary $F_{\xi_2}$ to this state, we obtain the new two-excitation state $F_{\xi_1 \xi_2} \Omega$. Hence effectively what we have done is moving around one of the excitations. Similarly we can remove two excitations, by acting with a path operator that closes the loop. In fact, a bit of thought shows that the number of excitations of the toric code is conserved modulo 2. This is a form of charge conservation.

We will be interested in states with a \emph{single} excitation. Using the description above it is apparent that there are no local operators that create such a state. There is however something else we can do: since the excitations can be moved around, in the thermodynamic limit we can try to move on excitation to infinity. Indeed, this can be done. Choose a semi-infinite path $\xi$ going off to infinity (for simplicity, one can think of a straight line). We write $\xi_n$ for the (finite!) path consisting of the first $n$ edges of $\xi$. Then for $A \in \alg{A}$, we can define
\begin{equation}
	\label{eq:toricauto}
	\rho(A) := \lim_{n \to \infty} F_{\xi_n} A F_{\xi_n}^*.
\end{equation}
The interpretation is that this map describes how the observables of the system change in the presence of an excitation in the background. 

\begin{exercise}
	Show that $\rho(A)$ as above is well-defined and defines an automorphism of $\alg{A}$. In addition, show that $\omega_0 \circ \rho$ is a pure state, where $\omega_0$ is the translational invariant ground state of the toric code.
\end{exercise}

The automorphism $\rho$ depends on the specific choice of path. However, if we look at the state $\omega_0 \circ \rho$, this is no longer true: it only depends on the location of the single excitation. Hence the direction in which we moved the other excitation away is not observable.

\begin{exercise}
	Choose two semi-infinite paths $\xi_1$ and $\xi_2$ going off to infinity, such that their starting points coincide. Let $\rho_1$ and $\rho_2$ be the corresponding automorphisms, defined via equation~\eqref{eq:toricauto}. Let $\omega_0$ be the translation invariant ground state. Use Exercise~\ref{ex:pathindep} to show that $\omega_0 \circ \rho_1 = \omega_0 \circ \rho_2$, and conclude that there is a unitary $U$ such that $U \pi_0 \circ \rho_1(A) = \pi_0 \circ \rho_2(A) U$ for all $A \in \alg{A}$.
\end{exercise}

\subsubsection{Non-invariant ground states}
The states $\omega_0 \circ \rho$ discussed above give us new examples of ground states. It is perhaps not immediately clear why this would be a ground state. Indeed, the flipped star (or plaquette) that remains leads to an energy increase for the local Hamiltonian at that site. In other words, $\omega_0 \circ \rho(H_\Lambda) - \omega_0(H_\Lambda) = 2$ for each set $\Lambda$ containing the star or plaquette at the end of the path. Nevertheless, this is still a ground state. Heuristically this can be understood as follows. Using local operations it is possible to move the excitation around. While this lowers the energy density \emph{locally}, it raises the energy density at the new location of the excitation, leaving the total energy of the system untouched. One can show that is all one can do with local operations, so that it is for example not possible to obtain the translationally invariant ground state by only acting on a finite part of the system. This is an example where a system admits different \index{superselection sector}superselection sectors as ground states. We will study the 
superselection structure of the toric code in more detail in Section~\ref{sec:sselecttoric}.

\begin{exercise}\label{ex:noninvgs}
	Let $\omega_0$ be the translationally invariant ground state of the toric code on the $\mathbb{Z}^2$ lattice. Recall that if $A$ is a monomial of Pauli matrices, $\omega_0(A) = 1$ if and only if $A$ is a product of star and plaquette operators, and zero otherwise. Let $\xi$ be the path going from the origin to infinity along the $x$-axis and let $\rho$ be the corresponding automorphism, defined in equation~\eqref{eq:toricauto}. Show that $\omega_0 \circ \rho$ is a ground state, i.e. that $-i \omega_0 \circ \rho(A^* \delta(A)) \geq 0$ for all $A \in D(\delta)$.
\end{exercise}

It is important that one of the excitations is moved to infinity, so that we cannot remove all excitations with local operators:
\begin{exercise}
Let $\xi$ be a finite path on the lattice with distinct beginning and end point, and let $F_\xi$ the corresponding path operator. Show that $\omega(A) := \omega_0(F_\xi A F_\xi^*)$ is not a ground state, where $\omega_0$ is the translationally invariant ground state.
\end{exercise}

Exercise~\ref{ex:noninvgs} gives an example of a non-translational invariant ground state. A natural question is if there are other ground states that we have not found yet. This amounts to classifying all states $\omega$ such that $-i \omega(A^* \delta(A)) \geq 0$ for all $A \in D(\delta)$. In general this is a difficult task, and not many results of this type are known. Fortunately, the toric code is simple enough to make this problem tractable, essentially because we can reduce it to classifying ground states of the toric code Hamiltonian on a finite system with certain boundary conditions.

To state the result, recall that there are two types of excitations: of stars and of plaquettes. Consider the path $\xi_Z$ from the origin up in a straight line to infinity. Let $\widehat{\xi}_X$ be the dual path going through the plaquettes on the right of $\xi_Z$. Then we can define the corresponding automorphisms $\rho^X$ and $\rho^Z$, via equation~\eqref{eq:toricauto}. Finally, set $\rho^{Y} = \rho^X \circ \rho^Z$. We think of $\rho^{Y}$ as describing a ``composite'' excitation of both a star and a plaquette excitation. This can be seen as an ``elementary'' excitation (or charge) as well, since there are no local operators converting it into an excitation of a different type. For convenience, we set $\rho^0 = \operatorname{id}$, the trivial automorphism of $\alg{A}$.

With the notation introduced above it is now possible to give a complete description of the set of all ground states of the toric code~\cite{toricgs}:
\begin{theorem}\label{thm:toricgs}
	Let $\omega$ be a ground state of the toric code in the thermodynamic limit. Then there are states $\omega^k$, $k=0,X,Y,Z$, such that $\omega$ is a convex combination of the $\omega^k$, and each $\omega^k$ is (quasi-)equivalent to $\omega_0 \circ \rho^k$, where $\omega_0$ is the unique translational invariant ground state. In particular, if $\omega$ is pure, it is equivalent to $\omega_0 \circ \rho^k$ for some $k$.
\end{theorem}
Here quasi-equivalence of states means that the corresponding GNS representations are quasi-equivalent. Two representations $\pi_1$ and $\pi_2$ are said to be quasi-equivalent if each subrepresentation of $\pi_1$ (that is, a representation that is obtained by restricting $\pi_1$ to an invariant subspace) is unitarily equivalent to a subrepresentation of $\pi_2$, and vice versa.\index{representation!quasi-equivalent} In the case of irreducible representations this quasi-equivalence is equivalent to unitary equivalence (since there are no non-trivial invariant subspaces of the Hilbert space).

\chapter{Lieb-Robinson bounds}\label{ch:lr}
One important difference between relativistic systems and non-relativistic systems (such as the ones we consider here), is that for relativistic system one has \idx{causality}: information cannot travel faster than the speed of light. In lattice systems there is no such natural bound. Correlations can, in principle, spread arbitrarily fast. Nevertheless, under suitable conditions there \emph{is} a maximum velocity in the system, which dictates how fast information can propagate through the system. This is essentially what a \emph{Lieb-Robinson bound}\index{Lieb-Robinson!bound} is.

The first of these bounds was proven by Lieb and Robinson in 1972~\cite{MR0312860}, where the authors assumed translation invariance. This allowed them to use Fourier techniques to prove the result. In the next three decades, comparatively little work was done on this type of bounds. This changed however drastically around 2004, when people realised that such bounds are very useful in studying, for example, many body systems and quantum information. For example, Hastings~\cite{PhysRevB.69.104431} used an improved version of the bound to prove a multi-dimensional version of the Lieb-Schultz-Mattis theorem. This theorem tells us something about the low energy excitations of gapped Hamiltonians. Around the same time, Nachtergaele and Sims~\cite{MR2217299} proved a version of the Lieb-Robinson bound that does not require translational invariance. In addition, they showed that for Hamiltonians with an energy gap, one has exponential decay of correlation functions (in terms of the distance of the support of two observables). After that, many new applications and improvements of the Lieb-Robinson bounds have been found. For more information we refer to~\cite{MR2681770,lriamp}.

\section{Statement and proof for local dynamics}\label{sec:lrfinite}
Before we can state the Lieb-Robinson theorem, we first have to make clear what kind of systems we want to consider. Here we will largely follow the work of Nachtergaele and Sims~\cite{MR2217299}. The class of systems that we allow will be much wider than those considered in the previous chapter. Again, we will consider a set $\Gamma$ of sites, together with a metric $d$ on the set of sites. At each site $x \in \Gamma$ we have a finite dimensional system, as before. We allow the dimension to vary from site to site, as long as there is a (finite) upper bound $N$ to this dimension. This condition can be dropped for suitable interactions, but it will complicate the proofs.

The class of interactions that we consider will be bigger than before as well. In particular, we define the following norm of an interaction $\Phi$:\index{interaction!norm}\index{_Philambda@$\Vert\Phi\Vert_\lambda$}
\begin{equation}
	\|\Phi\|_\lambda := \sup_{x \in \Gamma} \sum_{\Lambda \ni x} |\Lambda| \|\Phi(\Lambda)\| N^{2 |\Lambda|} e^{\lambda \operatorname{diam}(\Lambda)},
\end{equation}
where $\lambda > 0$ is some constant. We consider all interactions for which $\| \Phi \|_\lambda < \infty$. Note that this includes all bounded interactions of finite range, and hence all interactions we have considered so far.

We will also need a regularity condition on the set $\Gamma$ and the metric $d$. In particular, it is important that the number of sites does not grow too fast, in the following sense. Pick a fixed site $x_0 \in \Gamma$, for example the origin in a lattice. Let $B_n(x_0)$ be the set of sites $x \in \Gamma$ with $d(x,x_0) < n$. With this notation we make the following definition:

\begin{definition}\index{lattice!lambda-regular@$\lambda$-regular}
	Let $(\Gamma, d)$ be a set of sites with a metric $d$. Let $\lambda > 0$. We say that $\Gamma$ is \emph{$\lambda$-regular} if $\sum_{n=1}^\infty f(n) e^{-\lambda n} < \infty$, where $f(n) = |B_n(x_0)\setminus B_{n-1}(x_0)|$ for some fixed $x_0$.
\end{definition}
\begin{figure}
	\begin{center}
	\begin{tikzpicture}[nodes={circle,fill,scale=0.3}, level distance=.8cm]
	  \tikzstyle{level 1}=[sibling distance=50mm] 
	  \tikzstyle{level 2}=[sibling distance=15mm] 
	  \tikzstyle{level 3}=[sibling distance=4mm] 
	\node{}
		child { node {} 
			child { node{} 
				child { node{} }
				child { node{} }
				child { node{} }
				child { node{} }
			}
			child { node{} 
				child { node{} }
				child { node{} }
				child { node{} }
				child { node{} }
			}
			child { node{} 
				child { node{} }
				child { node{} }
				child { node{} }
				child { node{} }
			}
		}
		child { node {} 
			child { node{} 
				child { node{} }
				child { node{} }
				child { node{} }
				child { node{} }
			}
			child { node{} 
				child { node{} }
				child { node{} }
				child { node{} }
				child { node{} }
			}
			child { node{} 
				child { node{} }
				child { node{} }
				child { node{} }
				child { node{} }
			}
		};
	\end{tikzpicture}
	\end{center}
	\caption{\label{fig:tree}The first four steps of the construction of a non-regular set of sites $\Gamma$. The metric is the graph metric, that is, $d(x,y)$ is the minimum number of edges to traverse to go from vertex $x$ to vertex $y$.}
\end{figure}
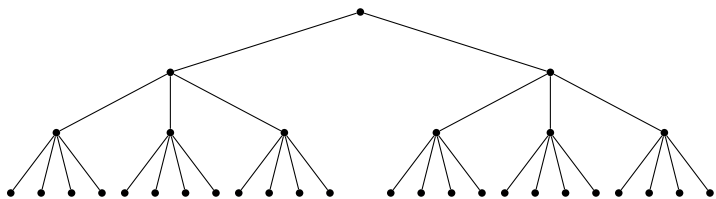
Note that the condition in fact does not depend on the choice for $x_0$. A sufficient condition is if $|B_n(x_0)|$ grows polynomially in $n$, for example in the case of $\Gamma = \mathbb{Z}^d$ with the usual metric: in that case $|B_n(x_0)| \sim n^d$. It is also not so difficult to think of a counterexample. Consider a graph $\Gamma_0$ consisting of a single vertex $x_0$, and no edges. Then define $\Gamma_1$ by adding two vertices to $\Gamma_1$ and connecting them with an edge to $x_0$. Then for each $n > 2$, define $\Gamma_n$ by adding $n+1$ new vertices \emph{for each vertex added in step $n-1$}, and connect them to that vertex. The set $\Gamma$ is the union of all $\Gamma_n$, with the obvious identifications. This is an infinite graph, and is a metric space with the usual graph metric. Since the number of points in a ball centred around any vertex $v$ grows faster than exponentially in the size of the ball, it is not $\lambda$-regular for any $\lambda$. The first few steps are display in Figure~\ref{fig:tree}.

The proof of the Lieb-Robinson bounds requires us to make many norm estimates of commutators. For this reason, it will be convenient to introduce the shorthand notation
\[
	C_{A,B}(x,t) := [\alpha_t(A), B],
\]
where $x \in \Gamma$, $A \in \alg{A}(\{x\})$ and $B \in \alg{A}$. Similarly we define\index{_CB(x,t)@$C_B(x,t)$}
\[
C_B(x,t) := \sup_{A \in \alg{A}(\{x\}), A \neq 0} \frac{\|C_{A,B}(x,t)\|}{\|A\|},
\]
that is, the maximum value of the norm of $C_{A,B}(x,t)$, taking into account irrelevant rescaling of the operator $A$. The bound that we will prove will give an estimate on the norm of $C_B(x,t)$. In particular, this norm will decrease exponentially as the distance between $x$ and the support of $B$ becomes bigger.

\begin{theorem}[Lieb-Robinson]\index{Lieb-Robinson!bound}
	\label{thm:lrbound}
	Let $\lambda > 0$ and consider a $\lambda$-regular set $\Gamma$. Then for all interactions $\Phi$ with $\|\Phi\|_\lambda < \infty$, sites $x \in \Gamma$ and $t \in \mathbb{R}$ and $B \in \alg{A}$, we have the bound
	\begin{equation}\label{eq:lrbound}
	C_B(x,t) \leq e^{2 |t| \| \Phi \|_\lambda} C_B(x,0) + \sum_{y \in \Gamma, y \neq x} e^{-\lambda d(x,y)} \left(e^{2 |t| \|\Phi\|_\lambda}-1\right) C_B(y,0).
\end{equation}
\end{theorem}
Before we go into the proof, we discuss some specific forms of the bounds. A special and very useful case is when the operator $B$ is strictly local, so that $B \in \alg{A}(\Lambda)$ for some $\Lambda \in \mc{P}_f(\Gamma)$. If in addition $x \notin \Lambda$, it follows by locality that $C_B(x,0) = 0$. Hence the first term in the bound can be dropped. The second term can be estimated as well, as is outlined below.

Another useful observation is that it is easy to extend the bounds to (local) operators $A$ that are supported on more than one site, lets say on the finite set $\Lambda$. Indeed, note that $A$ can be written in the form
\begin{equation}
	\label{eq:finexp}
	A = \sum_{i_1, \dots i_k} \lambda_{i_1,\dots,i_k} E_{i_1} \cdots E_{i_k},
\end{equation}
where the $\lambda$ are complex coefficients and where for each $l$, $\{E_{i_l} \}$ is a set of matrix units of $\alg{A}(\{l\})$ with $l \in \Lambda$. For example, one can take the matrices that have a one in exactly one entry, and zeros elsewhere. Note that by the uniform bound on the dimension of the Hilbert spaces at each site, there are at most $N^{2 |\Lambda|}$ terms in the sum. We can then repeatedly apply the inequality
\[
	\| [A_1 A_2, B] \| \leq \|A_1\| \| [A_2, B]\| + \|A_2\| \|[A_1,B] \|,
\]
which follows from the triangle inequality. Taking also the summations into account, we can make the following estimates: 
\[
\begin{split}
	\|[\alpha_t(A), B]\| &= \left\| \sum_{i_1, \dots i_k} \lambda_{i_1, \dots, i_k} [\alpha_t(E_{i_1}) \cdots \alpha_t(E_{i_k}), B] \right\| \\
		& \leq \|A\| \sum_{i_1, \dots i_k} \left\| [\alpha_t(E_{i_1}) \cdots \alpha_t(E_{i_k}), B] \right\| \\
		& \leq \| A\| \sum_{i_1, \dots i_k} \sum_{\lambda=1}^{k} \left\| [ \alpha_t(E_{i_\lambda}), B] \right\| \\
		& \leq \|A\| N^{2 |\Lambda|} \sum_{\lambda \in \Lambda} C_B(\lambda, t).
\end{split}
\]
Here we used that $\|E_{i_j}\| = 1$ in the third line, and the Lieb-Robinson theorem in the last line.

Next, if $B$ is also strictly local, we can find a more explicit bound for $C_B(x,t)$. In particular, we can bound the right-hand side of equation~\eqref{eq:lrbound} as follows. First, note that by locality $[A,B] = 0$ if the supports of $A$ and $B$ are disjoint. Hence $C_B(y,0) = 0$ if $y \notin \supp(B)$. If $y \in \supp(B)$, we can use the trivial bound on the commutator, and we obtain $C_B(y,0) \leq 2 \|B\| \chi_{\supp(B)}(y)$. Here $\chi_{\supp(B)}$ is the indicator function of the set $\supp(B)$. Also note that $\exp(-\lambda d(x,y)) \leq \exp(-\lambda d(x,\supp(B))$ if $y \in \supp(B)$. Hence we can conclude that
\[
	C_B(x,t) \leq e^{2 |t| \| \Phi \|_\lambda} C_B(x,0) + 2 | \supp(B)| \|B\| e^{-\lambda d(x,\supp(B))} \left(e^{2 |t| \|\Phi\|_\lambda}-1\right).
\]
We might as well estimate the $e^{2 |t| \|\Phi\|\lambda} -1$ term by dropping the $-1$, and take all the terms together. Combining these results we obtain the following corollary.

\begin{corollary}\label{cor:lrrough}
	Consider the same setting as in Theorem~\ref{thm:lrbound} and suppose that $A \in \alg{A}(\Lambda_1)$ and $B \in \alg{A}(\Lambda_2)$ are local operators. Then we have the bound
\[
\| [\alpha_t(A), B] \| \leq 4 \|A\| \|B\| |\Lambda_1|  |\Lambda_2| N^{2 |\Lambda_1|} e^{2 |t| \|\Phi\|_\lambda - \lambda d(\Lambda_1, \Lambda_2)}.
\]
\end{corollary}
This is a very rough bound, and with some extra work (and assumptions on the metric), it is possible to improve on it. For example, if $\Lambda_1 \cap \Lambda_2 = \emptyset$, then we can use $e^{2 |t| \|\Phi\|_\lambda}-1$ instead of $e^{2 |t| \|\Phi\|_\lambda}$. We will later discuss further improvements of this bound (at the expense of some mild additional assumptions), but first we will give the proof of Theorem~\ref{thm:lrbound}. 

As mentioned, the proof largely follows~\cite{MR2217299}, which in turn is based on the proof given in~\cite{MR1441540}, the main difference being that no Fourier theory techniques are no longer used, to be able to dispense of the translational invariance requirement. There is a technical issue, however: \emph{a priori} it is not clear that interactions such that $\|\Phi\|_\lambda < \infty$ in fact generate dynamics on the whole system, since we have only proven this result for interactions of finite range. The first step would be to define a derivation that generates the dynamics. Suppose that $A \in \alg{A}(\Lambda)$. Then we have
\[
\begin{split}
	\| \delta(A) \| &= \left\| \sum_{X \cap \Lambda \neq \emptyset} [\Phi(X), A] \right\| 
	\leq \sum_{x \in \Lambda} \sum_{X \ni x} \left\| [\Phi(X), A] \right\|  \\
	&\leq 2 \|A\| \sum_{x \in \Lambda} \sum_{X \ni x} \| \Phi(X) \| \leq 2 |\Lambda| \|A\| \| \Phi \|_\lambda.
\end{split}
\]
Hence we can define a derivation $\delta$ with domain $D(\delta) = \alg{A}_{loc}$. However, this is not quite enough to show that $\delta$ is the generator of dynamics.\footnote{Technically, $\delta$ must have a closed extension.} There are different ways of doing this, but here we follow a strategy similar to the proof of Theorem~\ref{thm:dynconv}. That is, we consider the \emph{local} dynamics $\alpha_t^\Lambda(A) := e^{i t H_\Lambda} A e^{-i t H_\Lambda}$, and show that $\alpha_t^\Lambda \to \alpha_t$ as $\Lambda$ grows. The easiest way to prove this is by first deriving a variant of the Lieb-Robinson bound for the local dynamics. Note that since we are considering only \emph{local} dynamics, it is not necessary to make any assumptions on $\Gamma$.

\begin{lemma}\label{lem:lrfinite}
	Let $\lambda > 0$. Then for all interactions $\Phi$ with $\|\Phi\|_\lambda < \infty$, sites $x \in \Gamma$ and $t \in \mathbb{R}$ and $B \in \alg{A}$, and finite subsets $\Lambda$ of $\Gamma$, we have the bound
	\begin{equation}\label{eq:lrboundfin}
		C_B^{\Lambda}(x,t) \leq e^{2 |t| \| \Phi \|_\lambda} C_B(x,0) + \sum_{y \in \Gamma, y \neq x} e^{-\lambda d(x,y)} \left(e^{2 |t| \|\Phi\|_\lambda}-1\right) C_B^{\Lambda}(y,0),
\end{equation}
where $C_B^{\Lambda}(x,t) = [\alpha_t^\Lambda(A), B],$ and $\alpha_t^\Lambda(A) = e^{it H_\Lambda} A e^{-it H_\Lambda}$.
\end{lemma}
\begin{proof}
Let $x \in \Gamma$ and $A \in \alg{A}(\{x\})$ and suppose that $B \in \alg{A}$. We first look at $C_{A,B}^\Lambda(x,t)$. Note that
\[
[\alpha_t^\Lambda(A),B] -[A,B] = \int_{0}^{t} \frac{d}{ds} [\alpha_s^\Lambda(A), B] ds = \int_{0}^{t} [\alpha_s^\Lambda(\delta_\Lambda(A)), B] ds,
\]
by the fundamental theorem of calculus. We also used that the derivative at $t=0$ of $\alpha_t^\Lambda(A)$ is equal to $\delta_\Lambda(A) := i [H_\Lambda, A]$, and the one-parameter group property of $\alpha_t^\Lambda$. Hence, writing this out we obtain
\begin{equation}
	C_{A,B}^\Lambda(x,t) = C_{A,B}^\Lambda(x,0) + i \sum_{X \subset \Lambda, X \ni x } \int_{0}^t [\alpha_s^\Lambda([\Phi(X),A]), B] ds.
\end{equation}
To get a bound for $C_B^\Lambda(x,t)$ we can take the norm of this expression, and repeatedly use the triangle inequality, to obtain
\begin{equation}
	\label{eq:cbbound1}
	C_B^\Lambda(x,t) \leq C_B^\Lambda(x,0) + \sum_{X \subset \Lambda, X \ni x} \int_0^{|t|} \sup_{A \in \alg{A}(\{x\})} \frac{\| [\alpha_s^\Lambda([\Phi(X), A]), B]\|}{\|A\|} ds.
\end{equation}
We want to estimate the terms under the integral sign. First note that if $X \in \mc{P}_f(\Gamma)$ and $x \in X$, then $[\Phi(X), A]$ is an element of $\alg{A}(X)$ (and if $x \notin X$ it is zero, by locality). We can then again write it in terms of matrix units, as we did before in equation~\eqref{eq:finexp}. Note that for the coefficients we must have $|\lambda_{i_1 \dots i_k}| \leq 2 \|\Phi(X)\| \|A\|$. We also have that
\[
	\| [\alpha_t(E_{i_1} \cdots E_{i_k}), B] \| \leq \sum_{x \in X} C_B^\Lambda(x,t),
\]
since $\|E_{i_j}\| = 1$. Combining everything with equation~\eqref{eq:cbbound1}, we arrive at the bound
\[
C_B^\Lambda(x,t) \leq C_B^\Lambda(x,0) + 2 \sum_{X \subset \Lambda, X \ni x} \| \Phi(X)\| N^{2 |X|} \int_0^{|t|} \sum_{y \in X} C_B^\Lambda(y,s) ds.
\]
Note that $C_B^\Lambda(x,t)$ appears both on the left and the right hand side. Hence we can again apply the inequality to the term under the integral, and so on. We first introduce the following shorthand notation:
\[
\varphi(x,y) := \sum_{X \ni {x,y}} \|\Phi(X)\| N^{2 |X|}.
\]
Note that the sum is over all finite sets $X$ that contain both the point $x$ and $y$. With this notation we can obtain the following bound for $C_B^\Lambda(x,t)$:
\begin{equation}
	C_B^\Lambda(x,t) \leq C_B^\Lambda(x,0) + 2 \int_0^{|t|} \sum_{y \in \Gamma} \varphi(x,y) C_B^\Lambda(y,s)ds.
	\label{eq:cbbound2}
\end{equation}
Note that this bound can be obtained from the bound above, by dropping the restriction that $X \subset \Lambda$ in the summation. This clearly gives a worse bound, but has the advantage that we lose the dependence on $\Lambda$ in the constants, and makes for nicer expressions.

We can now iterate this expression, and calculate the integrals over the constants terms, to get the following series in $|t|$:
\[
\begin{split}
	C_B^\Lambda(x,t) &\leq C_B(x,0) + 2 |t| \sum_{y \in \Gamma} \varphi(x,y) C_B^\Lambda(y, 0) \\
			&+ \frac{(2|t|)^2}{2} \sum_{y \in \Gamma} \varphi(x,y) \sum_{y' \in \Gamma} \varphi(y,y') C_B^\Lambda(y',0) + \\
			&+ \frac{(2|t|)^3}{3!} \sum_{y \in \Gamma} \varphi(x,y) \sum_{y' \in \Gamma} \varphi(y,y') \sum_{y'' \in \Gamma} \varphi(y', y'') C_B^\Lambda(y'',0) + \dots.\\
\end{split}
\]
Now set $\varphi_\lambda(x,y) := e^{\lambda d(x,y)} \varphi(x,y)$. The idea is to first rewrite the term above in terms of $\varphi_\lambda$. To get the idea, we will do this explicitly for the term proportional to $|t|^2$, the other terms are similar. We have
\[
\begin{split}
	\sum_{y \in \Gamma} \sum_{y' \in \Gamma} \varphi(x,&y) \varphi(y,y') C_B^\Lambda(y',0) \\
	&= \sum_{y \in \Gamma} \sum_{y' \in \Gamma} e^{-\lambda (d(x,y) + d(y,y'))} \varphi_\lambda(x,y) \varphi_\lambda(y,y') C_B^\Lambda(y', 0) \\
	&\leq \sum_{y \in \Gamma} \sum_{y' \in \Gamma} e^{-\lambda d(x,y')} \varphi_\lambda(x, y) \varphi_\lambda(y, y') C_B^\Lambda(y',0) \\
	&= \sum_{y \in \Gamma} \sum_{y' \in \Gamma} e^{-\lambda d(x,y)} \varphi_\lambda(x, y') \varphi_\lambda(y', y) C_B^\Lambda(y,0)
\end{split}
\]
where we used the triangle inequality to conclude that $d(x,y') \leq d(x,y) + d(y,y')$. Doing a similar rewriting with the other terms, we arrive at the following bound
\[
C_B^\Lambda(x,t) \leq \sum_{y \in \Gamma} e^{-\lambda d(x,y)} f(x,y) C_B^\Lambda(y,0),
\]
where $f(x,y)$ is the function
\[
\begin{split}
f(x,y) = &\delta_{x,y} + 2 |t| \varphi_\lambda(x,y) + \frac{(2|t|)^2}{2} \sum_{y' \in \Gamma} \varphi_\lambda(x,y') \varphi_\lambda(y', y) + \\
&\frac{(2|t|)^3}{3!} \sum_{y'' \in \Gamma} \sum_{y' \in \Gamma} \varphi_\lambda(x,y'') \varphi_\lambda(y'', y') \varphi_\lambda(y',y) + \dots.
\end{split}
\]
Note that we see that the exponential decay becomes clear from the bound now. It remains to estimate the sums over $\varphi_\lambda$. Note that for each $x \in \Gamma$,
\[
\begin{split}
\sum_{y \in \Gamma} \varphi_\lambda(x,y) &= \sum_{y \in \Gamma} \sum_{X \ni x,y} \|\Phi(X)\| N^{2|X|} e^{\lambda d(x,y)} \\
&\leq \sum_{X \ni x} \sum_{y \in X} \|\Phi(X)\| N^{2 |X|} e^{\lambda \operatorname{diam}(X)} \\
&= \sum_{X \ni x} |X| \|\Phi(X)\| N^{2 |X|} e^{\lambda \operatorname{diam}(X)} \leq \|\Phi\|_\lambda,
\end{split}
\]
where we used the definition of $\|\Phi\|_\lambda$ in the last step. Using this inequality we can estimate $f(x,y)$ by
\[
	f(x,y) \leq \delta_{x,y} + (e^{2 |t| \|\Phi\|_\lambda} -1).
\]
Plugging this in the expression for $C_B^\Lambda(x,t)$ we finally obtain
\[
C_B^\Lambda(x,t) \leq e^{2 |t| \|\Phi\|_\lambda} C_B^\Lambda(x,0) + \sum_{y \in \Gamma, y \neq x} e^{-\lambda d(x,y)} (e^{2|t| \|\Phi\|_\lambda}-1) C_B^\Lambda(y,0).
\]
This concludes the proof.
\end{proof}

\section{Existence of dynamics and proof of Theorem~\ref{thm:lrbound}}
The next goal is to show that the Lieb-Robinson bound for finite volume dynamics implies that the corresponding dynamics converge to some dynamics $\alpha_t$ on the entire system.

\begin{theorem}\label{thm:dynconvlr}
	Suppose that $\Gamma$ is $\lambda$-regular and that we have an interaction $\Phi$ with $\|\Phi\|_\lambda < \infty$. Then the local dynamics converge to a strongly continuous one-parameter group of automorphisms $\alpha_t$, i.e. for all $A \in \alg{A}$ we have
	\[
		\begin{split}
		\lim_{\Lambda \to \infty} \| \alpha_t^{\Lambda}(A) - \alpha_t^\Lambda(A) \| = 0.
	\end{split}
	\]
	Convergence is uniform for $t$ in the compact intervals $[-T, T]$.
\end{theorem}
\begin{proof}
For simplicity we will assume $t > 0$, $t < 0$ is shown by making the obvious changes. Let $A \in \alg{A}(\Lambda)$ be a local observable and suppose that $\Lambda_n \subset \Lambda_m$ are finite sets. Then
\[
	\begin{split}
	\left\| \alpha_t^{\Lambda_m}(A)-\alpha_t^{\Lambda_n}(A) \right\| &= \left\| \int_0^t \frac{d}{dt} \left[ \alpha_s^{\Lambda_m}(A)-\alpha_s^{\Lambda_n}(A)\right] ds \right\| \\
	&= \left\| \int_0^t  \alpha_s^{\Lambda_m}(\delta_{\Lambda_m}(\alpha_{t-s}^{\Lambda_n}(A))-\delta_{\Lambda_n}(\alpha_{t-s}^{\Lambda_n}(A))) ds\right\| \\
	&= \left\| \int_0^t  \alpha_s^{\Lambda_m}\left( \left[H_{\Lambda_m} - H_{\Lambda_n}, \alpha_{t-s}^{\Lambda_n}(A) \right]\right) ds\right\| \\
	&\leq \sum_{x \in \Lambda_m \setminus \Lambda_n} \sum_{X \ni x} \int_0^t \left\| [\Phi(X), \alpha_{t-s}^{\Lambda_n}(A) \right\| ds.
\end{split}
\]
We can again apply the trick of expanding $\Phi(X)$ into matrix units, as in the proof of Lemma~\ref{lem:lrfinite}. This gives the bound
\[
	\left\| \alpha_t^{\Lambda_m}(A)-\alpha_t^{\Lambda_n}(A) \right\| \leq \sum_{x \in \Lambda_m \setminus \Lambda_n} \sum_{X \ni x} \| \Phi(X) \| N^{2|X|} \sum_{y \in X} \int_0^t C_A^{\Lambda_n}(y, t-s) ds.
\]
This brings us in a position to apply Lemma~\ref{lem:lrfinite}. Since we do not need the sharpest bound possible, we take together both terms in the bound to obtain the upper bound
\[
	\sum_{x \in \Lambda_m \setminus \Lambda_n} \sum_{X \ni x} \| \Phi(X) \| N^{2|X|} \sum_{y \in X} \sum_{z \in \Lambda} e^{-\lambda d(y,z)} C_A(z,0) \int_0^t e^{2|t-s| \|\Phi\|_\lambda} ds.
\]
where we used that $C_A(y,0) = 0$ if $y \notin \Lambda$ by locality. This allowed us to get rid of the infinite sum over $\Gamma$. Now note that by the triangle inequality, $d(x,z) \leq d(x,y) + d(y,z) \leq \operatorname{diam}(X) + d(y,z)$. Hence $e^{-\lambda d(y,z)} \leq e^{-\lambda d(x,z) + \lambda \operatorname{diam}(X)}$. Also note that we can find a constant $C_T$ such that the integral in the expression above is majorised by $C_T$ for all $t \in [-T, T]$. Combining these observations it follows that
\[
\begin{split}
2 C_T \sum_{x \in \Lambda_m \setminus \Lambda_n} \sum_{X \ni x} \| \Phi(X) \| N^{2|X|} |X| e^{\lambda \operatorname{diam}(X)} \sum_{z \in \Lambda} e^{-\lambda d(x,z)}  \\
\leq 2 C_T \|\Phi\|_\lambda \|A\| \sum_{x \in \Lambda_m \setminus \Lambda_n} \sum_{z \in \Lambda} e^{-d(x,z)}.
\end{split}
\]
is an upper bound, where we first used the estimates from above and then replaced the summation over $y$ with $|X|$, and that $C_A(y,0) \leq 2 \|A\|$. By the regularity condition on $\Gamma$ the summation converges to zero as $\Lambda_n \to \infty$. The proof that $\alpha_t$ is a one-parameter group of automorphisms is the same is the proof of Theorem~\ref{thm:dynconv}.
\end{proof}

This theorem shows that we can use Lieb-Robinson bounds for \emph{local} dynamics to show that \emph{global} dynamics exist. This now finally allows us to prove Theorem~\ref{thm:lrbound}.

\begin{proof}[Proof of Theorem~\ref{thm:lrbound}]
	The above theorem gives uniform convergence of the local dynamics to the global dynamics $\alpha_t$, from which the claim follows with Lemma~\ref{lem:lrfinite}.

	Alternatively, now we know that the global dynamics exist, it can be shown that the closure of the derivation $\delta$ of Section~\ref{sec:lrfinite} generates this dynamics. The claim can then be proven directly by replacing the local dynamics with global dynamics in the proof of Lemma~\ref{lem:lrfinite}.
\end{proof}

\section{Locality of time evolution}\label{sec:lrlocality}
We are interested in the time evolution of local observables. As the time evolves, in general the size of the support will spread out. But if Lieb-Robinson bounds apply one expects that this spread will be limited. Such considerations are important, for example, when one wants to build quantum computers based on large quantum mechanical systems. For example, if one performs some operation on a small part of the system, it is important to know how long it will take before another distant part of the system is affected by this local operation. In relativistic systems this is limited by a \emph{strict} upper bound: the speed of light. So in a sense, Lieb-Robinson bounds give an analogue of the speed of light in quantum lattice systems. There is however one important difference: the Lieb-Robinson bound is usually not strict, in the sense that there will be some small, exponentially suppressed ``leakage'', as we will see.

To answer this question we first consider the following question. Consider an operator $A$ supported on some finite set $\Lambda$. Then for any operator $B$ whose support is disjoint from $\Lambda$, we have that $[A,B] = 0$. Conversely, if $A$ is any operator and $[A,B] = 0$ for any $B \in \alg{A}(\Lambda^c)$ for some finite set $\Lambda$, then it follows that $A \in \alg{A}(\Lambda)$ (see~\cite[Prop. IV.1.6]{MR1239893} or Exercise~\ref{ex:alghaagdual}). In practice, strict commutativity is a bit too much to ask, and will rule out the dynamics generated by many interesting interactions. So instead of \emph{strict} commutation, we assume only that the norm of the commutator is small. The effect is that although $A$ need not be strictly local any more, it can be approximated by an operator $A'$ that is localised in $\Lambda$. This is the content of the following lemma~\cite{localapprox}.\index{dynamics!locality}
\begin{lemma}\label{lem:loccom}
	Let $\Lambda \in \mc{P}_f(\Gamma)$. Suppose that $\varepsilon > 0$ and $A \in \alg{A}$ are such that
\[
\| [A, B] \| \leq \varepsilon \|A\| \|B\| \quad\textrm{for all } B \in \alg{A}(\Lambda^c).
\]
Then it follows that there is some $A_\Lambda \in \alg{A}(\Lambda)$ such that $\|A_\Lambda - A\| \leq \varepsilon \|A\|$.
\end{lemma}
We will not give the proof of this statement, but instead consider as an example the simpler case of finite systems. More precisely, we assume that we are given a Hilbert space $\mc{H} = \mc{H}_1 \otimes \mc{H}_2$ where $\mc{H}_2$ is finite dimensional. The observable algebra is $\alg{A} = \alg{B}(\mc{H}) = \alg{B}(\mc{H}_1) \otimes \alg{B}(\mc{H}_2)$. Suppose that we are given some $A \in \alg{A}$ and $\varepsilon > 0$ such that $\|[A,I \otimes B]\| \leq \varepsilon \|A\| \|B\|$ for all $B \in I \otimes \alg{B}(\mc{H}_2)$. The goal is to approximate $A$ by some operator $A' \in \alg{B}(\mc{H}_1) \otimes I$. To achieve this, define $A'$ by
\[
A' = \int_{\mc{U}(\mc{H}_2)} (I \otimes U) A (I \otimes U)^* dU,
\]
where the integral is over the group of unitaries in $\alg{B}(\mc{H}_2)$ and the integral is with respect to the Haar measure of this group. This measure has the property that the substitutions $U \mapsto UV$ and $U \mapsto VU$ (with $V$ also a unitary) leave the integral invariant. Hence if $V \in \alg{B}(\mc{H}_2)$, we have
\[
\begin{split}
(I\otimes V) &A' = \int_{\mc{U}(\mc{H}_2)} (I \otimes VU)  A (I \otimes U)^* dU \\
		&= \int_{\mc{U}(\mc{H}_2)} (I \otimes U)  A (I \otimes V^*U)^* dU = A' (I \otimes V),
\end{split}
\]
hence $A' \in (I \otimes \alg{B}(\mc{H}_2))' = \alg{B}(\mc{H}_1) \otimes I$. 

Finally we show that $A'$ indeed approximates $A$. Indeed, note that
\[
\| A' - A \| = \int_{\mc{U}(\mc{H}_2)} [I \otimes U, A] (I \otimes U)^* dU \leq \varepsilon \int_{\mc{U}(\mc{H}_2)} \|A\| dU = \varepsilon \|A\|,
\]
since $\|U\| = 1$ for any unitary $U$.

Now suppose that $t \in \mathbb{R}$ and $A \in \alg{A}(\Lambda)$ for some $\Lambda \in \mc{P}_f(\Gamma)$. Moreover we suppose that the time evolution $\alpha_t$ comes from some interaction satisfying the requirements for Theorem~\ref{thm:lrbound}. The goal is to show that $\alpha_t(A)$ can be approximated by a local observable with support on a set that is not too much bigger than $\Lambda$. Obviously, the support will generally grow bigger as $t$ becomes bigger, but we are interested in the case of fixed $t$. The idea is to use Lemma~\ref{lem:loccom}, so we must have a bound on commutators of that form. Corollary~\ref{cor:lrrough} gives such a bound, but the problem is that it depends on the support of the operator $B$, while the Lemma requires the bound to hold for \emph{all} $B$, independent of the size of the support. Hence as a first step we will improve the bound in Corollary~\ref{cor:lrrough}.

Recall that the site of the support of $B$ comes into the bound in Corollary~\ref{cor:lrrough} by estimating the summation in the expression $C_B(x,t)$. We proceed as before, and as a first step recall the bound
\[
	C_B(x,t) \leq 2 \|B\| \sum_{y \in \supp(B)} e^{-\lambda d(x,y) + 2 |t| \|\Phi\|_\lambda}.
\]
We estimated the sum by noting that $d(x,y) \geq d(x, \supp(B))$. This is of course a rough estimate, since the effect of the sites in the support of $B$ decays exponentially with the distance to $x$. Therefore we will make a more careful estimate. In particular, we can estimate the above expression as
\begin{equation}
	C_B(x,t) \leq 2 \|B\| \sum_{k = 0}^\infty e^{-\lambda(d(x,\supp(B)) + k) + 2 |t| \|\Phi\|_\lambda} |\supp(B) \cap (B_{k+1}(x) \setminus B_{k}(x))|,
\end{equation}
where $B_k(x)$ is the ball of size $k$ around $x$. That is, we break up the support of $B$ into bigger and bigger annuli of width $1$. If the amount of points in such an annulus does not grow exponentially (or faster) as $k$ becomes bigger, the sum can be bounded independent of the support of $B$. This is for example the case if $\Gamma = \mathbb{Z}^d$ with the usual metric. As an example we can consider $\Gamma = \mathbb{Z}$. The more general setting is discussed in, for example,~\cite{MR2256615,MR2681770}. Note that in the case $\Gamma = \mathbb{Z}$, if we increase the diameter of the ball $B_k(x)$ by one, we add two additional points. Hence in that case we can estimate
\[
C_B(x,t) \leq 4 \| B\| e^{2 |t| \|\Phi\|_\lambda - \lambda(d(x,\supp(B))} \sum_{k = 0}^\infty e^{-k \lambda}.
\]
Since $\lambda > 0$ the sum converges to $1/(1-\exp(-\lambda))$. As mentioned before, a similar result holds for $\Gamma = \mathbb{Z}^d$. With this bound we find the following improved version of Corollary~\ref{cor:lrrough}.
\begin{corollary}\label{cor:lrsharper}
	Consider the same setting as in Theorem~\ref{thm:lrbound}, with the addition that $\Gamma = \mathbb{Z}^d$.  Suppose that $A \in \alg{A}(\Lambda)$ and $B \in \alg{A}$. Then there is a constant $C > 0$ such that 
\[
\| [\alpha_t(A), B] \| \leq C \|A\| \|B\| |\Lambda| N^{2 |\Lambda|} e^{2 |t| \|\Phi\|_\lambda - \lambda d(\Lambda, \supp(B))}.
\]
\end{corollary}
Using similar considerations it is also possible to get rid of the proportionality to $|\Lambda|$, but this will not be necessary for our purposes.

Now choose an $\varepsilon > 0$, $t \in \mathbb{R}$ and $A \in \alg{A}(\Lambda)$ for some finite set $\Lambda$. Let $C$ be as in the Corollary above. Choose $d_{min}$ such that
\[
	\varepsilon > C |\Lambda| N^{2 |\Lambda|} e^{2 |t| \|\Phi\|_\lambda - \lambda d_{min}}.
\]
Then for any $B$ whose support is at least a distance $d_{min}$ from $\Lambda$, we have the bound $\| [\alpha_t(A), B ]\| < \varepsilon \|A\| \|B\|$. Hence by Lemma~\ref{lem:loccom} it follows that $\alpha_t(A) \in \alg{A}(\Lambda')$ where $\Lambda' = \cup_{x \in \Lambda} B_{d_{min}}(x)$. Also note that $d_{min}$ essentially scales as $2 |t| \|\Phi\|_\lambda$. Hence $2 \|\Phi\|_\lambda$ gives a bound on the velocity with which the support of the time-evolved operator grows: it can be interpreted as the ``group velocity'' or ``speed of sound'' in the system. This constant is sometimes called the \emph{Lieb-Robinson velocity}\index{Lieb-Robinson!velocity}, although it should be noted that this generally is not a sharp bound.

\section{Exponential decay of correlations}
In relativistic quantum field theories with a mass gap, one has that correlations decay exponentially~\cite{fredclust}. Essentially this means that if $\omega$ is the ground state, $|\omega(AB)-\omega(A) \omega(B)|$ is proportional to $\exp(-\gamma d(A,B))$, where $\gamma$ is some constant and $d(A,B)$ is the distance between the supports of $A$ and $B$. The mass gap property here means that there are only massive particles in the theory, and the lightest particle has mass $m_0$.\footnote{This property can be expressed in terms of the joint spectrum (which we will discuss later) of the generators of translations in the vacuum (that is, the momentum operators). The condition is then that (expect for 0, which corresponds to the vacuum), the spectrum lies on or above the mass shell with mass $m_0$ in the forward lightcone.} The constant $\gamma$ in the bound depends on the speed of light, which as mentioned before, also gives rise to a Lieb-Robinson bound. A natural question is then if in the context of lattice systems, it is possible to use the Lieb-Robinson bounds we have derived here as a substitute for the speed of light, and prove an analogue of the exponential clustering theorem. This is indeed the case~\cite{MR2217299}, although its proof is highly non-trivial.

To state the result we first need to give an analogue of the mass gap in the present context. This is done in terms of the spectrum of a Hamiltonian $H$. Recall that generally the Hamiltonian is not a bounded operator, and it is important that we will be careful that the domain $D(H)$ is only a dense subset, since $H$ is self-adjoint. Its spectrum, $\operatorname{spec}(H)$, was defined in Definition~\ref{def:specunbounded}. Note that a Hamiltonian is a positive operator. In particular this implies that its spectrum is contained in the positive real line. With the normalisation $H \Omega = 0$ we see that $0$ is also in the spectrum. We then say that a Hamiltonian $H$ is \emph{gapped}\index{Hamiltonian!gapped} if there is some $\gamma > 0$ such that
\[
	\operatorname{spec}(H) \cap (0, \gamma) = \emptyset.
\]
This means that there are no excited states with an energy below $\gamma$. For systems with such gapped Hamiltonians one can show that exponential clustering holds. Here we state a recent version due to Nachtergaele and Sims~\cite{nsicmp}.

\begin{theorem}[Exponential clustering]\index{exponential clustering}
Let $\Phi$ be an interaction such that $\|\Phi\|_\lambda < \infty$. Moreover, suppose that there is a unique ground state $\omega$ for the corresponding dynamics and the Hamiltonian $H$ in the ground state representation has a gap $\gamma > 0$. Then there is some constant $\mu > 0$ such that
\[
|\omega(AB) - \omega(A) \omega(B)| \leq C(A,B,\gamma) e^{- \mu d(\Lambda_1, \Lambda_2)}
\]
for all disjoint $\Lambda_1, \Lambda_2 \in \mc{P}_f(\Gamma)$ and $A \in \alg{A}(\Lambda_1), B \in \alg{A}(\Lambda_2)$, and $C(A,B,\gamma)$ is given by
\[
C(A,B,\gamma) = \|A\| \|B\| \left( 1+ \sqrt{\frac{1}{\mu d(\Lambda_1, \Lambda_2)}} + c \min(|\partial_\Phi(\Lambda_1)|, |\partial_\Phi(\Lambda_2)|) \right), 
\]
for some constant $c$ (which depends on the lattice structure and the interaction).
\end{theorem}
Here $\partial_\Phi(X)$\index{0d_Phi@$\partial_\Phi(X)$} is the set of all $x \in X$, such that there is some $\Lambda \ni x$ such that $\Lambda \cap X^c \neq \emptyset$ and $\Phi(\Lambda) \neq 0$. In other words, it consists of all points such that there is some interaction term across the boundary of $X$. The uniqueness condition on the ground state can be dropped, but one has to take a bit more care in that case. We also note that one can in fact give an explicit expression for $\mu$ in terms of the gap and the interaction.

Exponential decay of correlations says that for two observables localised far away from each other, $\omega$ is almost a product state, i.e., there are almost no correlations between $A$ and $B$. This restricts the type of system that we can have, and in particular in 1D systems some strong results on the structure of such states has been obtained. Consider a large (but finite) quantum spin system and let $\rho$ be a state of this system. Since the system is finite, $\rho$ can be represented by a density operator. Now consider a bipartition $AB$ of the system. We can restrict the state $\rho$ to part $A$, by tracing out the degrees of freedom in $B$: $\rho_A = \operatorname{Tr}_B \rho$, where $\operatorname{Tr}_B$ is the \idx{partial trace}. Then the \idx{entanglement entropy} $S_A(\rho)$ is defined as the von Neummann entropy $\rho_A$. It can be shown that $S_A(\rho)$ is zero if and only if $\rho = \rho_A \otimes \rho_B$, i.e. in the case where there are no correlations between $A$ and $B$.

In general we would expect $S_A(\rho)$ to scale proportionally to $|A|$. However, in gapped 1D systems (where we have exponential decay), it turns out that $S_A(\rho) \sim |\partial A|$, where $\partial A$ is the boundary of $A$. This is called an \idx{area law}. The first results where obtained by Hastings~\cite{HastingsAreaLaw}, and have later been improved in~\cite{Brandao2014}. This has important consequences, in particular it implies that the state can be approximated by a \idx{matrix product state}. Essentially, what this means is that we can describe the complete state by a relatively small number of parameters, called the \emph{bond dimension}. In general the amount of parameters scales exponentially in the number of sites $n$ (since the Hilbert space dimension scales exponentially), but for matrix product states this is not true. It is clear that this has significant advantages when one wants to do simulations on a computer, allowing for a much larger number of sites to be considered. See~\cite{HastingsAreaLaw,Brandao2014} and references therein for more information.

\chapter{Local quantum physics}\label{ch:aqft}

So far we have only discussed quantum spin systems, defined on a discrete lattice. In this chapter we will also look at \emph{continuous} systems, defined on Minkowski space-time. It turns out that many of the same techniques can be used in this setting. This resemblance is particularly clear in what is called \emph{algebraic quantum field theory} (AQFT).\index{AQFT} This is an attempt to formulate quantum field theory in a mathematically rigorous way, using $C^*$-algebraic techniques. One of the first formulations is due to Haag and Kastler~\cite{MR0165864}. Their goal was to give a purely algebraic description of quantum field theory, that is, without reference to any Hilbert space.

There are two key principles in the approach of Haag and Kastler: the first is locality. The other is the realisation that it is not \emph{fields} which are of fundamental importance, but rather \emph{observables}, since in the end that is what we measure in experiments. Based on this, they assign $C^*$-algebras of observables to regions $\mathcal{O}$ of space-time, in a way compatible with locality. Because locality places such a central role, AQFT is also often called \idx{local quantum physics}. We prefer this term, since these ideas can also be applied outside the realm of quantum field theory, as we have seen.

Already from the description the similarities to the local nets of quantum spin systems becomes clear. Even though there are fundamental differences between quantum spin systems and algebraic quantum field theories (one example we have already seen: the strict upper bound in the speed of light, versus approximate Lieb-Robinson bounds), there are also many similarities. In this chapter we give an introduction to the main ideas local quantum physics, and will see how many of the concepts previously discussed for quantum spin systems reappear.

Algebraic quantum field theory is about relativistic theories, hence we would like to talk about the \emph{energy-momentum spectrum}\index{energy-momentum!spectrum}. Before we discussed the spectrum of a \emph{single} unbounded operator. The energy-momentum spectrum is the \emph{joint} spectrum of the momentum operators $P_\mu$, $\mu = 0,\dots, 3$ (in 4D Minkowksi space). This notion will be developed in the next section. The energy-momentum spectrum is not specific to relativistic theories, and indeed in the next chapter we will apply the results of Section~\ref{sec:snag} to scattering of particle excitations in quantum spin systems.

After the joint spectrum has been introduced, the it is discussed how relativistic quantum field theory can be discussed in the setting of local $C^*$-algebras. For example, the results of Section~\ref{sec:snag} make it possible to elegantly describe that the energy-momentum spectrum in the ground state should be contained in the forward light-cone.

In the last section we return to quantum spin systems, and the toric code in particular. Specifically, we show that the set of superselection sectors\index{superselection sector} has a rich structure, describing the algebraic properties of elementary ``charges'' of the system. This approach originated in algebraic quantum field theory, but can be applied to quantum spin systems as well. This serves to show that the principles of local quantum physics go much further than relativistic theories alone.

This chapter is only intended as an introduction to the key ideas, to show how ideas developed for relativistic theories can also be applied to quantum spin systems. A complete treatment with full proofs is outside of the scope of these lecture notes. The interested reader will find that~\cite{MR2542202,MR1405610} are good starting points to the rich literature available on the algebraic approach to relativistic quantum field theories. An overview of recent developments can be found in~\cite{MR3381848}.

\section{The SNAG theorem and joint spectrum}\label{sec:snag}
Recall that the momentum operators can be obtained as the generator of translations. Hence, if we have more than one spatial dimension, we need to consider multiple (commuting) generators at the same time, to talk about the momentum of, say, a single particle. This is even more so in relativistic theories, where the energy (i.e., the generator of time translations) and momentum operators have to be treated on equal footing. In finite dimensional systems, we know that commuting self-adjoint operators can be diagonalised simultaneously, so we can find a basis of states that are simultaneous eigenvectors to all of them. In infinite systems this is much more delicate, for one because as we have seen, the spectrum generally does not only consist of eigenvalues.

This shows that it is necessary to find a generalisation of the spectrum of a single, unbounded operator. In particular, in relativistic theories we have momentum operators $P^\mu$, $\mu=0,\dots,3$ in four-dimensional Minkowski space-time. We would like these operators to commute. However, since the operators are unbounded, this is a tricky subject. Intuitively commutativity means that $P^\mu P^\nu - P^\nu P^\mu = 0$. The right hand side is unproblematic, but the left hand side isn't, since the operators are unbounded. Indeed, it is only defined on the \emph{intersection} of the domains of $P^\mu$ and $P^\nu$, which might be trivial. And even then, it may be the case that $P^\mu$ does not map the intersection of these domains into the common domain of $P^\mu$ and $P^\nu$. There is a way out here: since the $P^\mu$ are self-adjoint operators, we can apply the spectral theory to get a family of spectral projections $E^\mu(B)$, where $B$ is a Borel subset of $\mathbb{R}$, associated to the unbounded operator $P^\mu$. We then say that $P^\mu$ and $P^\nu$ \emph{commute}\index{unbounded operator!commuting} if and only if their spectral projections commute. 

It is well-known that for commuting self-adjoint matrices we can find a common set of eigenvectors. We want to generalise this to the spectrum of the unbounded operators $P^\mu$. Instead of looking at the unbounded operators $P^\mu$, we look at strongly continuous groups of unitaries. In particular, consider $G = \mathbb{R}^n$ and a strongly continuous unitary representation $\mathbf{t} \mapsto U(t_1, \dots, t_n)$. Note that $t \mapsto U(0, \dots, t, 0, \dots 0) := U^\mu(t)$, where $t$ is in the $\mu$-th component, gives a strongly continuous one-parameter group of unitaries. Hence, by Stone's theorem (Theorem~\ref{thm:stone}), there is a self-adjoint operator $P^\mu$ such that $U^\mu(t) = e^{i P^\mu t}$. Because $U^\mu(t)$ and $U^\nu(t)$ commute, one can in fact prove that $P^\mu$ and $P^\nu$ commute (in the sense defined above). Conversely, such commuting operators naturally lead to a strongly continuous representation of $\mathbb{R}^n$ by setting $U(\mathbf{t}) = e^{ i t_1 P^1} \cdots e^{i t_{n} P^{n}}$.

This leads to the following generalisation of Stone's Theorem, Theorem~\ref{thm:stone}. It was independently proven by Naimark~\cite{MR0010265},\footnote{Naimark's name has not always been transliterated consistently from Russian: except for Naimark, ``Neumark'' is also seen.} Ambrose~\cite{MR0011172} and Godement~\cite{MR0014589}. Consequently, it is sometimes called the \idx{SNAG theorem}.

\begin{theorem}[SNAG]\label{thm:snag}
	Let $\mathbf{t} \mapsto U(\mathbf{t})$ be a strongly continuous unitary representation of $\mathbb{R}^n$ acting on a Hilbert space $\mathcal{H}$. Then for each $\psi \in \mathcal{H}$ there is a measure $\mu_\psi$ on $\mathbb{R}^n$ such that
\[ 
	\langle \psi, U(\mathbf{t}) \psi \rangle = \int_{\mathbb{R}^n} e^{i \mathbf{t} \cdot \mathbf{\lambda}} d\mu_\psi(\lambda).
\]
Equivalently, there is a \emph{projection valued measure} $dP(\mathbf{\lambda})$ on $\mathbb{R}^n$ such that
\[
	U(\mathbf{t}) = \int_{\mathbb{R}^n} e^{i \mathbf{t} \cdot \mathbf{\lambda}} dP(\mathbf{\lambda}).
\]
\end{theorem}
We do not give a proof of this Theorem here. It can be obtained by generalising the proof of Stone's theorem and uses many of the methods we outlined in Section~\ref{sec:positive}. Note that Stone's theorem gives as $n$ generators $P^\mu$, as already remarked above. It is possible in the proof of the SNAG theorem to explicitly construct a common domain for these generators that is left invariant by $U(\mathbf{t})$.

This theorem now allows us to talk about the energy-momentum spectrum. In particular, we can consider the group $U(\mathbf{t})$ of space-time translations. Then the energy-momentum spectrum is the joint spectrum of the generators of $U$.  
\begin{definition}
	Let $U$ and $dP$ be as in the theorem above. The support of $dP$ is called the \emph{spectrum} of the unitary representation $U$, and is denoted by $\operatorname{Sp}(U)$.\index{spectrum!unitary representation of abelian group}\index{_Sp(U)@$\operatorname{Sp}(U)$} Since as mentioned above we can identify a set of $n$ generators $P^\mu$ associated to this representation, we also call this spectrum the \idx{joint spectrum} of the generators $P^\mu$.
\end{definition}

\begin{remark}
	The theorem is true more generally. One can pick any locally compact abelian group $G$. Then the dual group $\widehat{G}$, the group of characters of $G$, can be given a topology which makes it a locally compact abelian group. This is called \idx{Pontry'agin duality}. For any unitary representation $U$ of $G$ there then is a spectral measure as above, where the integral is over the dual group (with respect to the Haar measure) and the integrand is given by $\langle g, \chi \rangle := \chi(g)$. Since all characters of $\mathbb{R}$ are of the form $\chi(x) = e^{i \lambda x}$, the above theorem as stated above corresponds to $G = \mathbb{R}^n$.
\end{remark}

\section{Algebraic quantum field theory}
Quantum field theory (QFT) is arguably one of the most successful theories of the last century. Not withstanding the huge success of the ``traditional'' (mainly perturbative) methods used by physicists working in quantum field theory, these are unsatisfactory from a mathematical viewpoint, because many concepts are mathematically ill-defined. Some aspects \emph{can} be made rigorous (the reader can consult, for example, the book by Glimm and Jaffe~\cite{MR887102}), but there are still many problems. In order to study QFT in a rigorous mathematical framework, it is desirable to have an axiomatic basis for QFT as a starting point.\footnote{This section is based on Chapter 3 of~\cite{phdnaaijkens}.}

One such axiomatisation is given by the \emph{Wightman axioms}\index{Wightman axioms} which, in a nutshell, postulate that quantum fields are given by operator valued distributions. An operator valued distribution $\Phi$ is a linear map from a suitable (e.g. smooth and of compact support) class of functions (called \emph{test functions}) to the (in this case unbounded) operators on some Hilbert space $\mathcal{H}$. Physically this can be understood as follows: quantum fields cannot be localised exactly in a point. Instead, one has to ``smear'' them with some test function $f$. The result is $\Phi(f)$. The classic \emph{PCT, Spin and Statistics, and All That} by Streater and Wightman remains a good introduction to this framework~\cite{MR1884336}. Although this approach is a natural one coming from the more familiar path integral approach to quantum field theory, it also has some drawbacks. From a mathematical point of view, one has to deal with \emph{unbounded} operators, which are more difficult to deal with than bounded operators. For example, innocuously looking expressions like $\Phi^4(f)$ can be problematic, as we have seen before. Nevertheless, such terms appear in many quantum field theories in the Lagrangian of the theory. At a more conceptual level there is the criticism that the quantum fields, which in general are not observables, are like coordinates, which should not be taken as the starting point of a theory.

An alternative axiomatisation is provided by what is called \emph{algebraic quantum field theory}\index{algebraic quantum field theory|see {AQFT}} (AQFT)\index{AQFT}, based on the \emph{Haag-Kastler} axioms. This is the framework that we will discuss here. In essence, the fundamental objects are nets of $C^*$-algebras of \emph{observables}\index{local!net} that can be measured in some finite region of space-time. At first sight it is perhaps surprising that in this approach one considers only \emph{bounded} observables, since it is well known that the position and momentum operators for a single particle are unbounded. One should keep in mind, however, that in the physical world there are always limitations on the measuring equipment, and one can always only measure a bounded set of (eigen)values. In terms of the spectral theory for unbounded operators, it means that instead of dealing with the unbounded operator itself, we only consider its spectral projections to be (at least in principle) observable.

The two approaches are in fact not as unrelated as they might appear at first sight. Under certain conditions one can move from one framework to the other (and back). See, for instance~\cite{MR1199168}, and references therein. While the Wightman axioms are closer to common practice in quantum field theory, the Haag-Kastler approach is easier to work with mathematically, since one does not have to deal with unbounded operator-valued distributions.

One of the earliest works on AQFT is by Haag and Kastler~\cite{MR0165864}. By now there is a large body of work, including introductory textbooks. The monograph by Haag~\cite{MR1405610} and the book by Araki~\cite{MR2542202} are particularly recommended for a review of the physical and mathematical principles underlying this (operator) algebraic approach to quantum field theory. A review can also be found in~\cite{MR1768634}. The second edition of Streater and Wightman~\cite{MR1884336} also contains a short overview.\footnote{In June 2009 the 50 year anniversary of the birth of the theory was celebrated with a conference in G\"ottingen, where Haag, one of the founders of the subject, recollected some of the successes and problems of algebraic quantum field theory~\cite{springerlink:10.1140/epjh/e2010-10042-7}. The reader might also be interested in Haag's personal recollection of this period~\cite{haagpersonal}.}

As argued in the introduction, there are two basic principles underlying the AQFT approach. First of all, it is the \emph{algebraic} structure of the observables that is important. The second principle is locality: in relativistic QFT it makes sense to speak about observables that describe the physical properties localised in some region of space-time (for example $T \times S$, with $T$ a time interval and $S$ a bounded region of space, say a laboratory). Moreover, by Einstein causality one can argue that observables in spacelike separated regions are compatible in that they commute. As the basic regions we consider \emph{double cones}\index{double cone} $\mc{O}$, defined as the intersection of (the interior of) a forward and backward light-cone. Note that a double cone is causally complete: $\mc{O} = \mc{O}''$, where a prime $'$ denotes taking the causal complement. To each double cone $\mc{O}$ we associate a unital $C^*$-algebra $\alg{A}(\mc{O})$ of observables localised in the region $\mc{O}$. Finally, note that the Poincar\'e group $\mc{P}^{\uparrow}_+$ (generated by translations and Lorentz transformations) acts on double cones. We write $g \cdot \mc{O}$ for the image of a double cone under a transformation $g$.

The starting point of AQFT, then, is a map $\mc{O} \mapsto \alg{A}(\mc{O})$. There are a few natural properties the map $\mc{O} \mapsto \alg{A}(\mc{O})$ should have if it is to describe (observables in) quantum field theory. For example, anything that can be localised in $\mc{O}$ can be localised in a bigger region as well. This leads to the following list of axioms, now known as the \idx{Haag-Kastler axioms}.
\begin{enumerate}
	\item \emph{Isotony:} if $\mc{O}_1 \subset \mc{O}_2$ then there is an inclusion $i: \alg{A}(\mc{O}_1) \hookrightarrow \alg{A}(\mc{O}_2)$. We assume the inclusions are injective unital $*$-homomorphisms. Often the algebras are realised on the same Hilbert space, and we have $\alg{A}(\mc{O}_1) \subset \alg{A}(\mc{O}_2)$.
	\item \emph{Locality:}\index{locality} if $\mc{O}_1$ is spacelike separated from $\mc{O}_2$, then the associated local observable algebras $\alg{A}(\mc{O}_1)$ and $\alg{A}(\mc{O}_2)$ commute.
	\item \emph{Poincar\'e covariance:} there is a strongly continuous action $x \mapsto \beta_x$ of the Poincar\'e group $\mc{P}^{\uparrow}_+$ on the algebra of observables, such that 
\[
	\beta_g(\alg{A}(\mc{O})) = \alg{A}(g \cdot \mc{O}).
\]
\end{enumerate}
We will always assume that the algebras $\alg{A}(\mc{O})$ are non-trivial.\footnote{In fact, in practice one usually realises the net as a net of von Neumann algebras acting on some Hilbert space. Under physically reasonable assumptions the algebras $\alg{A}(\mc{O})$ are Type III factors. These are a specific type of von Neumann algebras. See~\cite{Yngvason:2005p1600} for a discussion of the physical significance of this.} 

\begin{remark}
Instead of Poincar\'e covariance one sometimes requires the weaker condition of translation covariance, see for example~\cite{MR660538}.
\end{remark}

Note that $\mc{O} \mapsto \alg{A}(\mc{O})$ is a \emph{net} of $C^*$-algebras, just as the local algebras of spin systems we discussed in Section~\ref{sec:quasilocal}. Just as we constructed the quasi-local observables from the strictly local observables, we can again take the union and closure of the net (or more precisely, the inductive limit), to obtain the algebra $\alg{A}$ of quasi-local observables.\index{quasi-local algebra} If the local algebras are all realised on the same Hilbert space, this amounts to taking $\alg{A} = \overline{\bigcup_{\mc{O}} \alg{A}(\mc{O})}^{\|\cdot\|}$, where the bar denotes closure with respect to the operator norm. The algebra $\alg{A}$ is called the algebra of \emph{quasi-local observables}\index{quasi-local observable}. We will usually assume that the local algebras act as bounded operators on some Hilbert space.\footnote{Recall that this can always be achieved by taking the GNS representation of a suitable state $\omega$ on $\alg{A}$.} In that case, for an arbitrary (possibly unbounded) subset $\mc{S}$ of Minkowski space, we set $\alg{A}(\mc{S}) := \overline{\bigcup_{\mc{O} \subset \mc{S}} \alg{A}(\mc{O})}^{\|\cdot\|}$, where the union is over all double cones contained in the region $\mc{S}$.

It should be noted that in this axiomatic approach some of the constructions of ``conventional'' quantum field theory can be discussed. For example, field operators, particle aspects and scattering theory can be defined in this setting. This approach is particularly suited to study structural properties of quantum field theory. Unfortunately, constructing examples of concrete (in particular interacting) examples has been proven to be very difficult, although there are some results, in particular in low dimensional space-times.

\subsection{Vacuum representation}\label{sec:vacuum}
The vacuum plays a special role in quantum field theory. Intuitively, it describes \emph{empty space}. Alternatively, it has minimal energy. To define the notion of a \emph{vacuum state}\index{vacuum}\index{state!vacuum} rigorously, one first defines energy decreasing operators. The precise details are not important for us (see e.g.~\cite[\S 4.2]{MR2542202}). In essence one considers operators of the form $Q = \int f(x) \beta_x(A) d^4x$ for some observable $A$ and smooth function $f$ whose Fourier transform has support disjoint from the forward light-cone $\overline{V}_{+}$. The $\beta_x$ are the translation automorphisms as in the Haag-Kastler axioms. A vacuum state then essentially is a state $\omega_0$ on $\alg{A}$ such that $\omega_0(Q^*Q) = 0$ for any such $Q$. We will see below that this condition has the interpretation of the vacuum having positive energy.

A vacuum state in the above sense can be shown to be \emph{translation} invariant, see Theorem 4.5 of~\cite{MR2542202}. The corresponding vacuum representation, which will be denoted by $\pi_0$, is then translation covariant.\index{representation!covariant} We denote $\alpha_x$ for the group of translation automorphisms, which is a subgroup of the Poincar\'e group. Note that there is a unitary representation $x \mapsto U(x)$ such that $\pi_0(\alpha_x(A)) = U(x) \pi_0(A) U(x)^*$ defined by $U(x) \pi_0(A) \Omega_0 = \pi_0(\alpha_x(A)) \Omega_0$ for any $x$ in Minkowski space and $A \in \alg{A}$. These translations are generated by unbounded operators $P_\mu$, which have the natural interpretations of energy ($P_0$) and momentum ($P_i$, with $i=1\cdots, d-1$).\index{energy-momentum!operators} In fact, since the group of translations is abelian, we can directly apply to the SNAG theorem and consider the joint spectrum. This spectrum can be shown to be contained in the forward light-cone, as follows from the assumptions above on $\omega_0(Q^*Q) = 0$ (for suitable $Q$). We outline the proof below. This is interpreted as ``positivity of the energy''. Finally, if $\pi_0$ is irreducible, then the vacuum vector $\Omega_0$ is the unique (up to a scalar) translation invariant vector. In fact, any factorial vacuum representation is automatically irreducible. 

Suppose that $\omega_0$ is an invariant (with respect to $\alpha_x$) vacuum state. We can now understand the condition $\omega_0(Q^*Q) = 0$. Since $\omega_0$ is invariant, the automorphisms $\alpha_x$ are implemented as a strongly continuous group of unitaries $\mathbf{x} \mapsto U(\mathbf{x})$ acting on the GNS Hilbert space $\mathcal{H}$. The vacuum vector $\Omega$ is invariant, and $\pi_0(Q) \Omega = 0$ by the vacuum condition. Hence by the SNAG theorem, there is a spectral measure $dP(\mathbf{p})$ related to $U(\mathbf{x})$. We can then calculate (a careful analysis shows each step can be justified):
\begin{align*}
	0 = \pi_0(Q) \Omega &= \int_{\mathbb{R}^4} \left(f(x) \pi_0(\alpha_x(A)) \Omega\right) dx \\
	&= \int_{\mathbb{R}^4} \left(f(x) U(\mathbf{x}) \pi_0(A) \Omega\right) dx \\
		&= \int \int \left( f(x) e^{i \mathbf{x} \cdot \mathbf{p}} dP(\mathbf{p})\right) \pi_0(A) \Omega dx \\
		&= \left(\int \widehat{f}(\mathbf{p}) dP(\mathbf{p}) \right) \pi_0(A) \Omega,
\end{align*}
where we used that $\Omega$ is invariant and $\widehat{f}$ is the Fourier transform of $f$. But this equation is true for any local $A$, and since $\Omega$ is a cyclic vector, this can only be true if the operator inside parentheses is equal to zero. But since this is true for any function $f$ whose Fourier transform is supported \emph{outside} the positive light-cone $\overline{V}_+$, it follows that the support of the spectral measure must be contained in $\overline{V}_+$. This is nothing else but saying that the energy-momentum spectrum lies in the forward light-cone, as is to be expected in a relativistic theory.

Alternatively, one can directly work with a Hilbert space, and characterise a vacuum representation\index{representation!vacuum} as a translation covariant representation such that the spectrum of the generators of these translations is contained in the forward light-cone $\overline{V}_+$. Moreover, $0$ is in the point spectrum, since the vacuum vector is invariant. A special case is a \emph{massive} vacuum representation. This is a vacuum representation where $0$ is an isolated point in the spectrum and there is some $m > 0$ such that the spectrum is contained in $\{0\} \cup \{ p: p^2 \geq m^2, p_0 > 0\}$. That is, there is a mass gap.

\subsection{Particle statistics}\index{particle statistics|(}
One of the highlights of algebraic quantum field theory is the \idx{Doplicher-Roberts theorem}~\cite{MR1062748}. This theorem allows us to construct a net of field operators creating certain (particle) excitations. They allow use to interpolate between different superselection sectors, that is, they map vectors from one irreducible representation of the observables into vectors in another irreducible representation, in a way that is compatible with the action of the observables. These field operators have the commutation properties that one would expect. For example, the field-net operators corresponding to two fermions in spacelike separated regions anti-commute, while those for bosons commute. Before we can come to this theorem, we have to explain how we describe charged excitations, and how one can define their statistics.

The idea is to find the different superselection sectors\index{superselection sector} of the theory. A \idx{charge} is then just a label for a sector. As before, these correspond to equivalence classes of irreducible representations of the observable algebra $\alg{A}$. In general there are many inequivalent ones, and most are not physically relevant. For example, we would only be interested in representations which describe a finite number of excitations. Therefore one has to introduce a criterion to select the relevant representations. One such criterion is the Doplicher-Haag-Roberts one~\cite{MR0297259,MR0334742}.\index{Doplicher-Haag-Roberts criterion} Let $\pi_0$ be a vacuum representation as above. Let $\mc{O}$ be an arbitrary double cone and write $\mc{O}'$ for its causal complement. This induces an algebra $\alg{A}(\mc{O}')$ of observables that can be measured inside $\mc{O}'$. The DHR criterion then considers all representations $\pi$ such that for each double cone $\mc{O}$ it holds that\index{superselection criterion}
\begin{equation}
	\label{eq:dhrselect}
\pi_0 \upharpoonright \alg{A}(\mc{O}') \cong \pi \upharpoonright \alg{A}(\mc{O}'),
\end{equation}
where the symbol $\upharpoonright$ means that we restrict the representations to the algebra $\alg{A}(\mc{O}')$. Intuitively speaking, it selects only those excitations that look like the vacuum when one restricts to measurements outside a double cone. It is known that this criterion rules out some physically relevant theories (for example, one can always detect an electrical charge by measuring the flux through a arbitrarily large sphere, by Gauss' law), but nevertheless this class shows many of the relevant features of a theory of superselection sectors. One can also relax the criterion. For example, Buchholz and Fredenhagen use ``spacelike cones'' instead of double cones, and show that in massive theories one always has such localisation properties~\cite{MR660538}. In recent work, a more general approach is suggested to deal with massless excitations~\cite{BuchholzRoberts}, which lead to infrared problems in the approach outlined here.

The set of representations that satisfy the criterion~\eqref{eq:dhrselect} has a rich structure.\footnote{There are some additional technical conditions necessary on the representation $\pi_0$, however, the most important one being \emph{Haag duality}.} For example, one can define a product operation $\pi_1 \times \pi_2$ on such representations. This new representation can be interpreted as first creating an excitation\index{excitation} of type $\pi_2$, then one of $\pi_1$. The ground state representation acts as the identity with respect to this product, in the sense that $\pi \times \pi_0 \cong \pi$. In this way, one can think of $n$ (particle) excitations of the class $\pi$ located in spacelike separated regions. Such configurations can be obtained as vector states in the representation $\pi^{\times n}$. Moreover, there is a canonical way to define unitary operators that interchange the order in which the excitations are created. These operators are analogous to the transposition operator in, for example, a quantum system of $n$ particles. Applying this unitary to the state vector leaves the expectation values in the state invariant, but \emph{can} change the state vector itself.

This enables a study of the exchange statistics of the excitations. The result is that the interchange of excitations is described by representations of the symmetric group, when the net of observable algebras is defined on Minkowski space with at least two spatial dimensions. This representation only depends on the equivalence class of the representation. In this way we obtain the well-known classification into bosons or fermions from first principles.\index{boson}\index{fermion} In a seminal work~\cite{MR1062748}, Doplicher and Roberts showed how one could, starting from the superselection sectors satisfying the criterion above, construct a net of field operators. These field operators can ``create'' excitations. As expected, for bosonic excitations the field operators located in spacelike regions commute, whereas for fermions they anti-commute.

In lower dimensional space-times an interesting phenomenon appears: excitations are not necessarily bosonic or fermionic anymore. For example, it may happen that the interchange of two particles leads to a non-trivial phase. In fact, even more general (non-abelian) effects are possible. Excitations with such properties (nowadays usually called \emph{anyons}\index{anyon}) are expected to appear as ``emergent'' quasi-particles in certain condensed matter systems. They play an important role in the quest of building a quantum computer that does not suffer too much from thermal noise introduced by the environment~\cite{MR2443722}. In the algebraic quantum field theory community the possibility of such ``braided'' statistics has been studied by many authors, see e.g.~\cite{MR1016869,MR1199171,MR1104414}. The name ``braided'' stems from the fact that one does not obtain a representation of the symmetric group any more, but rather of the so-called braid group.

In the next section we will carry out a similar study of superselection sectors in the context of quantum spin systems, namely in the case of the toric code. Although there are fundamental differences between relativistic theories and spin systems, up to technical differences the analysis of the sector structure is very similar. For example, we will see how in this example braided statistics appear. We therefore end our discussion on relativistic theories and return to quantum spin systems.
\index{particle statistics|)}

\section{Superselection theory of the toric code}\label{sec:sselecttoric}\index{superselection sector!toric code}
We know come back to the automorphisms $\rho^X$, $\rho^Y$ and $\rho^Z$ that we encountered in the study of the ground states of the toric code, in Section~\ref{sec:toriccode}. Above we briefly discussed the idea of superselection sectors in algebraic quantum field theory, where differently charged excitations lived in different superselection sectors. We also claimed that the study of these superselection sectors reveals the properties of the charges. Here we outline this programme in the context of the toric code. These results first appeared in~\cite{toricendo}, and were later extended to Kitaev's quantum double model for finite abelian groups in~\cite{haagdouble}. Full details can be found in these references.

\subsection{The superselection criterion}
We first encountered superselection rules at the end of Chapter~\ref{ch:opalg}, and saw that they are related to the existence of inequivalent representations. Indeed, later we defined a superselection sector to be a suitable equivalence class of suitable representations. Moreover, the different equivalence classes should correspond to different charges.

But in the toric code we already have a construction to make charged states! Moreover, Theorem~\ref{thm:toricgs} tells us that these states belong to different sectors, since states corresponding to different choices of the excitation are unitarily inequivalent. Hence, whatever selection criterion we will choose, it should encompass these cases. We therefore start our investigation with the states $\omega_0 \circ \rho^k$, and see what their properties are. We will use the same notation as before, and not repeat the definitions here.

Let us first see why the states indeed are inequivalent. Since $\omega_0$ is a pure state, as we have argued before, the GNS representation $\pi_0$ is irreducible. But $\rho^k$ is an automorphism, so $\pi_0 \circ \rho^k$ is irreducible, and it is easy to see that $(\pi_0 \circ \rho^k, \mathcal{H}, \Omega)$ is a GNS triple for $\omega_0 \circ \rho^k$. Hence $\omega_0 \circ \rho^k$ is a pure state. But in that case we have a simple criterion to check if they are inequivalent: according to Proposition~\ref{prop:stateinequiv} it is enough to show that they are not equal at infinity. That is, there is some $\varepsilon$ such that for each finite set $\Lambda$, we should be able to find an observable $A$ of norm 1, localised outside $\Lambda$, such that $|\omega_0(A) - \omega_0(\rho^k(A))| \geq \varepsilon$.

The idea behind the construction of $A$ is simple: we want to define an operator that measures the total charge inside the region $\Lambda$. These operators will measure the number of excitations of star or plaquette type, where each number is calculated modulo two. This can be understood as follows. In the toric code, each particle is its own conjugate. So if we have two excitations of a certain type, we actually have a charge and its conjugate charge, making the total charge zero.\footnote{This is somewhat like the case where you would have an electron and a positron, making total electric charge zero. The difference is of course that electrons and positrons are different particles, that is, they are not their own conjugate.} So how can we detect such an excitation? We will discuss the case of star excitations, the plaquette case is similar.

Let $F_{\xi}$ be a path operator creating excitations at a star $s_1$ and a star $s_2$. Let $\Omega$ be the GNS vector of the translation invariant ground state. Then, as we have seen before, $A_{s} F_{\xi} \Omega = \pm F_{\xi} \Omega$, where the minus sign occurs if and only if $s = s_1$ or $s = s_2$. Hence: $\omega_0(F_{\xi} A_s F_{\xi}) = \pm 1$, where we used that $F_{\xi}$ is self-adjoint. We get a minus sign if and only if $s = s_1$ or $s=s_2$. This readily generalises to states with more excitations. Now consider an operator of the form $P = \prod A_s$, where the product is over finitely many stars $s$. Then, for $F = F_{\xi_1} \cdots F_{\xi_n}$, we have that $\omega_0(F_{\xi} P F_{\xi}^*) = (-1)^{n_s}$, where $n_s$ is the number of stars that are endpoints of one of the $\xi_i$ \emph{and} appear in the product $P$. Hence $P$ detects the number of excitations modulo two at the stars appearing in the product.

We are now almost done. The construction above suggests that we can use the operators $P$ to detect the charges, and in this way distinguish the different states. However, a priori it is not clear if they have the right locality properties. Indeed, it appears that $P$ is supported on the stars $s$ in the product. However, it turns out that if we consider a ``block'' of stars, the operator is only supported at the \emph{boundary} of this block. In particular, consider a closed (for simplicity, non-intersecting) loop $\widehat{\xi}$ on the \emph{dual} lattice. We claim that $F_{\widehat{\xi}} = \prod_{s \subset \Lambda} A_s$, where $\Lambda$ are the bonds enclosed by or intersecting with $\widehat{\xi}$. To see this, first consider the smallest closed loop possible, going through four adjacent plaquettes. But in this case, $F_{\widehat{x}}$ is seen to be the star operator $A_s$ of the vertex it encloses. Now consider an adjacent star $s_1$. Then $A_s A_{s_1}$ is the path operator enclosing the two vertices. This is because on the edge in the interior of the path, we act with a $\sigma^x$ from both $A_s$ and $A_{s_1}$. Since $\sigma^x$ squares to the identity, this yields the identity on that edge and we are left with an operator acting only on the outer edges. Continuing in this way, the claim follows.

\begin{figure}
	\begin{center}
	\begin{tikzpicture}
		\draw[very thick] plot [smooth, tension=1] coordinates { (-2,0) (0.5,1) (-1,3) (0,5)};
		\draw[very thick,dashed] plot [smooth cycle, tension=1] coordinates {(-4,0) (-2,-1) (0,-1.5) (2,0) (0.5,3)};
		\filldraw[fill=black] (-2,0) circle [radius=0.05];
		\draw (-3.1,1.5) node {$\widehat{\xi}$};
		\draw (-0.35,3.5) node {$\rho^Z$};
		\draw (-2.25,0) node {$s_1$};
	\end{tikzpicture}
	\end{center}
	\caption{A path going to infinity, leading to an automorphism $\rho^Z$. Also drawn is the dual path $\widehat{\xi}$, where the corresponding path operator can be used to detect the charge $Z$ on the vertex.}\label{fig:chargedetect}
\end{figure}
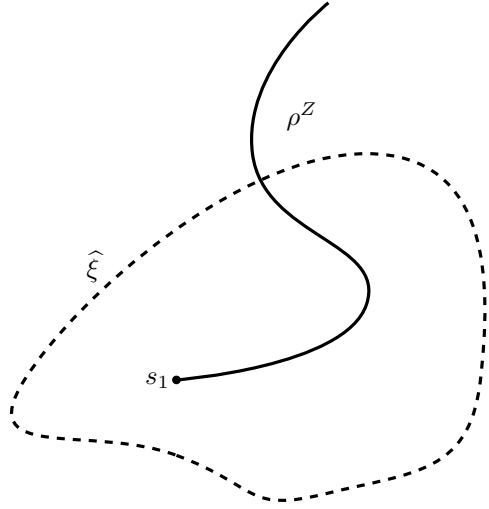

We are now in a position to show that $\omega_0$ and $\omega_0 \circ \rho^Z$ are inequivalent. Recall that $\rho^Z$ is obtained by conjugating with a path operator, and sending one end of the path to infinity. Using the commutation relations of path operators and the operators $A_s$, it follows that $\rho^Z(A_s) = (1-2 \delta_{s, s_1}) A_s$, where $s_1$ is the star at the end of the semi-infinite string used to define $\rho^Z$. Now let $\widehat{\xi}$ be a closed path on the dual lattice and let $P$ be the operator as defined above (see Figure~\ref{fig:chargedetect}). Then it follows that $\rho^Z(P) = \pm P$, where the minus sign occurs only if $s_1$ is enclosed by the dual path. Now set $\varepsilon = 1/2$ and let $\Lambda$ be any finite subset of the bonds containing the star $s_1$. Choose a dual path $\widehat{\xi}$ such that $\Lambda$ is contained in the area enclosed by the path. Then the corresponding operator $P$ is supported outside $\Lambda$, by the observation above. Moreover, repeated application of Lemma~\ref{lem:ststate} and the property that $\omega_0(A_s) = 1$, we obtain
\[
	| \omega_0(P) - \omega_0(\rho^Z(P)) | = | 2 \omega_0(P) | = 2 > \varepsilon.
\]
By Proposition~\ref{prop:stateinequiv} it follows that $\omega_0$ and $\omega_0 \circ \rho^Z$ are inequivalent.

\begin{exercise}
	Give an example of a state in each of the equivalence classes described in Theorem~\ref{thm:toricgs}, and adapt the argument above to show that they are all mutually inequivalent.
\end{exercise}

We have now constructed different examples of inequivalent states. Hence, by the GNS representation, we get representatives of different superselection sectors. Can we say more about their properties? The first observation is that $(\pi_0 \circ \rho^k, \mathcal{H}, \Omega)$ is a GNS triple for the state $\omega_0 \circ \rho^k$, hence every GNS representation for $\omega_0 \circ \rho^k$ is unitarily equivalent to this one.

But there is much more to say about the representation. Let $\Lambda$ be a cone-like region containing the path to infinity used in defining $\rho^k$, and let $\Lambda^c$ be its complement in the set of edges. Then, by locality, it follows that $\rho^k(A) = A$ if $A \in \alg{A}(\Lambda^c)$. By the uniqueness of the GNS representation, it follows that for \emph{any} GNS representation $\pi$ for $\omega_0 \circ \rho^k$, there is a unitary operator $V$ such that $V \pi(A) = \pi_0(A) V$ for all $A \in \alg{A}(\Lambda^c)$. The restriction to observables in $\alg{A}(\Lambda^c)$ is important here! Indeed, we know that $\pi$ and $\pi_0$ are inequivalent representations, so that the relation cannot hold for \emph{all} observables $A$, as this would imply that $\pi$ and $\pi_0$ are equivalent. So we arrive at the following interpretation: \emph{outside} of the cone $\Lambda$ the ``charged'' representation looks just like the translation invariant ground state representation. However, as soon as we allow observables to cross the cone, it is possible to measure the charge in the background and distinguish the representations. We say that the charge is \emph{localised} in $\Lambda$.\index{charge!localised}

There is nothing special about the choice of the cone $\Lambda$. In fact, we could have chosen any other cone $\Lambda$, and the statement about the existence of $V$ above would still hold. This is ultimately a consequence of the ``topological'' properties of the excitations, in the sense that the state $F_\xi \Omega$ only depends on the endpoints of the path $\xi$, not on the path itself. This is most easily seen if the automorphisms are defined by specifying semi-infinite paths with the same endpoints. Let us write $\rho_1$ and $\rho_2$ for the corresponding automorphisms. The claim is that $\omega_0(\rho_1(A)) = \omega_0(\rho_2(A))$. In that case, it follows from the uniqueness of the GNS representation that $\pi_0(\rho_1(A)) = U \pi_0(\rho_2(A)) U^*$ for some unitary $U$, from which the claims follows.

We first show that the two states are equal. First choose a local observable $A$. Write $\xi_1$ and $\xi_2$ for the paths defining $\rho_1$ and $\rho_2$. Then, by locality, there is some $N$ such that for all $n > N$, we have $\rho_1(A) = F_{\xi_{1,n}} A F_{\xi_{1,n}}^*$ and $\rho_2(A) = F_{\xi_{2,n}} A F_{\xi_{2,n}}^*$, where $\xi_{i,n}$ is the finite path consisting of the first $n$ segments of $\xi_i$. Now choose for each $n$ a finite path $\xi'_n$, going from the endpoint of $\xi_{1,n}$ to that of $\xi_{2,n}$ in such a way that $d(x, \xi'_n) \to \infty$ as $n \to \infty$ for some fixed site $x$. Then, if $n$ is large enough, $\rho_1(A) = F_{\xi'_n} F_{\xi_{1,n}} A (F_{\xi'_n} F_{\xi_{1,n}})^* = F_{\xi_{1,n} \xi'_n} A F_{\xi_{1,n} \xi'_n}^*$. But, by construction $\xi_{1,n} \xi'_n$ and $\xi_{2,n}$ have the same endpoints! By the independence of the state on the precise path, it follows that $\omega_0(\rho_1(A)) = \omega_0(\rho_2(A))$. Since the local algebras are dense in $\alg{A}$, the equality holds for all $A \in \alg{A}$ by continuity.  

In the case the endpoints are different, connect them by a finite path $\xi$. Then $F_{\xi} \rho_1(A) F_{\xi}$ and $\rho_2(A)$ are both automorphisms corresponding to an infinite path starting in the \emph{same} site. The result then follows.

To summarise, we see that the charged representations are unitarily equivalent to the ground state representation \emph{outside} of a cone. Moreover, we can move this cone around by applying unitary operators. We say that the charges are \emph{transportable}.\index{charge!transportable} Both conditions have a clear physical interpretation. Moreover, the charged states give natural examples of representations with these properties. Hence it is natural to posit that \emph{all} such representations correspond to a charge (or superselection sector). That is, we consider all irreducible representations $\pi$ such that\index{superselection criterion}
\begin{equation}
	\label{eq:toricselect}
	\pi_0 \upharpoonright \alg{A}(\Lambda^c) \cong \pi \upharpoonright \alg{A}(\Lambda^c),
\end{equation}
for any cone $\Lambda$. To recall: the notation $\pi_0 \upharpoonright \alg{A}(\Lambda^c)$ means that we restrict the representation to the algebra $\alg{A}(\Lambda^c)$. Again, it is important that the equivalence holds for \emph{arbitrary} cones $\Lambda$. This is called the \emph{superselection criterion}, as it selects the relevant class of superselection sectors (= classes of inequivalent representations of $\alg{A})$.

\subsection{Localised endomorphisms and charge transporters}
We want to study the class of representations satisfying the selection criterion, equation~\eqref{eq:toricselect}. As we will see, this set actually has a very rich structure. Instead of working directly with representations, it will turn out to be beneficial to work with endomorphisms of $\alg{A}$ instead, like the automorphisms we have constructed to give examples of different superselection sectors. This appears to be less general, but as we shall see in the toric code this turns out not to be the case.

The first step is to encode the properties behind the superselection criterion in properties of the endomorphisms. First of all, the charges should be localised in a cone. And, as we saw, we should be able to move the localisation region around. This leads to the following definition:
\begin{definition}\index{endomorphism!localised and transportable}\index{transportable}
	Let $\rho$ be an endomorphism of $\alg{A}$. We say that $\rho$ is \emph{localised} in a cone $\Lambda$, if $\rho(A) = A$ for all $A \in \alg{A}(\Lambda^c)$. It is called \emph{transportable}, if for any other cones $\Lambda'$, there is an endomorphism $\sigma$ that is localised in $\Lambda'$ and a unitary $V$ such that $V \pi_0 \circ \rho(A) = \pi_0\circ \sigma(A) V$ for all $A \in \alg{A}$. The unitary $V$ is called a \emph{charge transporter}\index{charge!transporter}.
\end{definition}
In general, if $T$ is an operator such that $T \rho_1(A) = \rho_2(A)T$ for two morphisms $\rho_i$ of $\alg{A}$, we call $T$ an \idx{intertwiner}. Note that the charge transporters are unitary intertwiners.

By the discussion in the previous section the automorphisms that we constructed are both localised and transportable, and give examples of representations that satisfy the selection criterion. The existence of the of the charge transporters $V$ followed from the uniqueness of the GNS representation. This definition of $V$ is rather indirect, and it turns out to be useful to have a concrete construction of $V$. Fortunately, since we explicitly constructed the automorphisms, it is possible to also construct $V$. We will do this by defining a sequence of unitaries $V_n$ that converges in the \idx{weak operator topology} to $V$. In this topology on the bounded operators on a Hilbert space $\mathcal{H}$, a net $A_\lambda \in \alg{B}(\mathcal{H})$ converges to an operator $A$ if and only if $\langle \xi, (A_\lambda-A) \psi \rangle$ converges to zero for every $\xi, \psi \in \mathcal{H}$. It is a natural topology when discussing von Neumann algebras. Indeed, a unital $*$-subalgebra of $\alg{B}(\mathcal{H})$ is a von Neumann algebra if and only if it is closed in the weak operator topology. This result, which is known as the \emph{bicommutant theorem}, relates the purely algebraic definition of von Neumann algebras given on page~\pageref{p:vna}, to a topological condition.\index{von Neumann algebra}

Recall that $\alg{A}$ is a simple algebra, hence the representation $\pi_0$ is faithful (or, injective). Hence we can identify $\pi_0(A)$ with $A$, and that is what we will do from know on.

To obtain $V$ we assume for simplicity that we have two paths $\xi_1$ and $\xi_2$ on the lattice, both extending to infinity and starting in the same vertex $x$ of the lattice. Then, for each $n$, let $\xi_i^n$ be the path consisting of the first $n$ segments of $\xi_i$. Then choose paths $\widehat{\xi}_n$, connecting the endpoints of $\xi_1^n$ and $\xi_2^n$, in such a way that $\operatorname{dist}(\widehat{\xi}_n, x) \to \infty$. In particular, this means that for any local operator $A$, there will be an $N$ such that for all $n > N$, $\widehat{\xi}_n$ is disjoint from the support of $A$.

We can now define our sequence $V_n = F_{\xi_1^n} F_{\widehat{\xi}_n} F_{\xi_2}^n$. Note that $\xi_1^n \widehat{\xi}_n \xi_2^n$ is a closed path (for the toric code, the direction of the path does not matter), so that $V_n \Omega = \Omega$. Now consider a local operator $A$ and choose $N$ large enough, such that for all $n > N$, the support of $A$ is disjoint from $\widehat{\xi}_n$ and disjoint from $\xi_i \setminus \xi_i^n$. Then, for such $n$, $\rho_i(A) = F_{\xi_i^n} A F_{\xi_i^n}$. From locality one can then easily verify that $V_n \rho_1(A) = \rho_2(A) V_n$ for all $n > N$. This already suggests that taking the limit we will obtain a charge transporter. This sequence certainly does not converge in the norm topology, but we will show that it \emph{does} converge in the weak operator topology, and in fact converges to $V$.

Recall that $V$ was obtained by uniqueness of the GNS representation, so that we have that $V \Omega = \Omega$. Now let $A$ and $B$ be local observables. Then, for $n$ large enough, we have by the argument above that
\[
	\langle A \Omega, (V-V_n) \rho_1(B) \Omega \rangle = \langle A \Omega, \rho_2(B) (V-V_n) \Omega \rangle = 0.
\]
Here we used that $V_n$ is $F_{\xi}$ for some closed loop, and hence a product of plaquette operators (or stars in the dual case), which act trivially on the ground state. Since local observables are dense, the sequence $V_n$ is uniformly bounded, and $\rho_1$ is an automorphism, the claim follows from a density argument.

We have assumed that both paths start in the same point $x$. To get the general case, we only have to move one of the endpoints away over a \emph{finite} distance. But this can be done using the path operators $F_{\xi}$, so by multiplying $V$ with such path operators, the general case follows easily from the special case.

The above discussion allows us to construct examples of representations satisfying the selection criterion (each of them implemented by localised and transportable endomorphisms), and construct the corresponding charge transporters. The fact that we can represent the sectors by endomorphisms is going to be crucial in our analysis of the charges. However, at this point it is far from clear if every representation satisfying the selection criterion is of this form. This turns out to be true, but it requires an additional technical property:

\begin{definition}
	A representation $\pi$ of $\alg{A}$ is said to satisfy \idx{Haag duality} for cones, if for any cone $\Lambda$ we have that $\pi(\alg{A}(\Lambda))'' = \pi(\alg{A}(\Lambda^c))'$.
\end{definition}
The inclusion $\pi(\alg{A}(\Lambda))'' \subset \pi(\alg{A}(\Lambda^c))'$ follows from locality, but the other inclusion is non-trivial (and can indeed fail to hold for general representations). However, for the ground state representation $\pi_0$ of the toric code it can be shown to hold:
\begin{theorem}[\cite{haagdtoric}]
Let $\pi_0$ be the ground state representation corresponding to the translational invariant ground state of the toric code. Then $\pi_0$ satisfies Haag duality.
\end{theorem}

Haag duality allows us to go from representations that satisfy the selection criterion back to localised maps. The argument is as follows. Let $\pi$ be a representation satisfying equation~\eqref{eq:toricselect}, $\Lambda$ a fixed cone, and $V$ a unitary such that $V \pi(A) = \pi_0(A) V$ for all $A \in \alg{A}(\Lambda^c)$. Then we can define a morphism $\rho: \alg{A} \to \alg{B}(\mathcal{H})$ via (we again write $\pi_0$ for clarity):
\[
	\rho(A) := V \pi(A) V^*.
\]
Note that for $B \in \alg{A}(\Lambda^c)$ we have $\rho(B) = V \pi(B) V^* = \pi_0(B) V^*V$, by the choice of $V$, so that $\rho$ is localised in $\Lambda$ (since it acts like the ground state representation outside $\Lambda$). But it is not clear if $\rho$ is actually and \emph{endomorphism} of $\alg{A}$, since we do not know if it maps $\alg{A}$ into itself. And it turns out it indeed is not an endomorphism, but something very close. If $A \in \alg{A}(\Lambda)$ and $B \in \alg{A}(\Lambda^c)$, then by locality $[A,B] = 0$. Hence we calculate
\[
	\pi_0(B) \rho(A) = V \pi(AB) V^* = \rho(AB) = \rho(A) \pi_0(B).
\]
Hence, by Haag duality, $\rho(A) \in \pi_0(\alg{A}(\Lambda^c))' = \pi_0(\alg{A}(\Lambda))''$. So $\rho$ does not map $\alg{A}(\Lambda)$ into itself, but rather into the weak-operator closure of $\alg{A}(\Lambda)$ in the ground state representation.

Fortunately there is a way around this, which will turn out to be extremely useful in the next section. The idea is to find a slightly bigger algebra $\alg{A}^a$ (containing $\alg{A}$), such that we can uniquely extend the morphism $\rho$ to an \emph{endomorphism} of $\alg{A}^a$. This algebra is called the \idx{auxiliary algebra}\index{_Aa@$\alg{A}^a$} and contains the algebras $\alg{A}(\Lambda)''$ for any cone (except those in a fixed direction that can be chosen freely). One can even show that the extended morphism (which we again denote by $\rho$) is weak-operator continuous. Here we will take the existence of $\alg{A}^a$ and the extension property as fact. The construction and proofs can be found in the literature~\cite{MR660538,toricendo}.

Note that we saw before that the charge transporters $V$ intertwining $\rho_1$ and $\rho_2$ generally are not elements of the $C^*$-algebra $\alg{A}$. In the construction of the charge transporters we see that $V$ is contained in $\alg{A}(\Lambda)''$ for any cone containing the two localisation regions of $\rho_1$ and $\rho_2$. This is true in general for intertwiners. Indeed, let $\Lambda$ be such a cone. Then, for $B \in \alg{A}(\Lambda^c)$ localisation dictates that $\rho_1(B) = \rho_2(B) = B$. In particular,
\[
	V B = V \rho_1(B) = \rho_2(B) V = B V,
\]
so that $V \in \alg{A}(\Lambda^c)' = \alg{A}(\Lambda)''$ by Haag duality. Hence the remark above shows that the intertwiners \emph{are} elements of $\alg{A}^a$, and hence we can consider expressions like $\rho(V)$, which again is in $\alg{A}^a$. This will be important in the next section.

\subsection{Fusion rules}
The irreducible representations satisfying the superselection criterion are in correspondence with the different types of excitations in the toric code. There is however more to say about them. In particular, we will be interested in what happens if we bring two excitations close together (this is called \idx{fusion}), or what happens if we interchange two of them (\idx{braiding}). Here the step from representations to localised and transportable endomorphisms pays off.

To keep track of the different properties we will make use of the language of \index{tensor category}\emph{tensor categories}.\footnote{A \idx{category} is a collection of objects and maps between these objects, together with an associative composition operation of maps. That is, if $f: \rho \to \sigma$ and $g: \sigma \to \tau$ are two maps, then there is a map $g \circ f : \rho \to \tau$. What a map is depends on the context, and maps are not necessarily functions. In addition, for each object $\rho$ there is a map $1_\rho$ which acts as the identity for composition. Standard examples are the category of sets with maps between them, the category of finite groups with group homomorphisms, or the category of topological spaces with continuous maps between them. Here we will use category theory merely as a bookkeeping device. We refer to~\cite{mmappendix} for a more in-depth discussion of tensor categories, which we need here.} We let $\Delta$\index{_Delta@$\Delta$} be the category which has (cone-)localised and transportable endomorphisms as its objects. As before, it will be convenient to consider them as maps of the auxiliary algebra $\alg{A}^a$. We still have to specify the maps between two objects. If $\rho$ and $\sigma$ are objects in $\Delta$, write\index{_Deltarhosigma@$\Delta(\rho,\sigma)$}
\[
	\Delta(\rho,\sigma) := \left\{ T : T \rho(A) = \sigma(A) T\quad\textrm{ for all } A \in \alg{A} \right\}
\]
for the set of intertwiners between them. If $S \in \Delta(\rho,\sigma)$ and $T \in \Delta(\sigma,\tau)$, it is easy to check that $TS \in \Delta(\rho,\tau)$, which gives us the composition of maps. Note that the identity operator $I$ is in $\Delta(\rho,\rho)$ for any $\rho$ and acts as the unit map in the category. An object $\rho \in \Delta$ is called \idx{irreducible} if $\Delta(\rho,\rho) = \mathbb{C} I$. Note that this is equivalent to $\pi_0 \circ \rho$ being an irreducible representation.

Suppose that we now have endomorphisms $\rho$ and $\sigma$ describing certain charges. Then we can consider $\rho \circ \sigma$, which again is an endomorphism. The interpretation is that this describes the change of the observables if we first create a charge $\sigma$, and then a charge $\rho$ in the background. We use this to turn $\Delta$ into a \idx{tensor category}: for objects $\rho, \sigma \in \Delta$, we set $\rho \otimes \sigma := \rho \circ \sigma$.\index{tensor product!of category Delta@of category $\Delta$}

\begin{remark}
This tensor product operation has nothing to do with the tensor product of vector spaces that we discussed earlier, although the notion of a tensor category originates as an abstraction of that idea. For example, one can consider the category of finite dimensional Hilbert spaces. In that category, the natural $\otimes$-operation is the usual tensor product of Hilbert spaces. Another example is obtained from the category of finite dimensional representations of a compact group $G$. The tensor product there is the tensor product of two group representations.
\end{remark}

This should be a new object in our category, so we have to show that $\rho \otimes \sigma$ is again localised and transportable. In addition, we should define the tensor product of two intertwiners. First note that if $\rho$ is localised in a cone $\Lambda_1$, and $\sigma$ in a cone $\Lambda_2$, then $\rho \otimes \sigma$ is localised in any cone $\Lambda$ that contains both $\Lambda_1$ and $\Lambda_2$. So $\rho \otimes \sigma$ is again localised. To show it is transportable, we first define $S \otimes T$ for $S \in \Delta(\rho_1, \rho_2)$ and $T \in \Delta(\sigma_1, \sigma_2)$ by:
\[
	S \otimes T := S \rho_1(T) = \rho_2(T) S.
\]
Recall that the intertwiners $S$ and $T$ generally are only elements of $\alg{A}^a$ (and not of $\alg{A}$), so that in this definition it is crucial that the endomorphisms can be extended to $\alg{A}^a$. The claim is that $S \otimes T \in \Delta(\rho_1 \otimes \sigma_1, \rho_2 \otimes \sigma_2)$. And indeed, we verify for all $A \in \alg{A}$:
\[
\begin{split}
(S \otimes T) (\rho_1 \otimes \sigma_1)(A) = & S\rho_1(T) \rho_1(\sigma_1(A)) = \rho_2(T \sigma_1(A))S \\&= \rho_2(\sigma_2(A)T)S = (\rho_2 \otimes \sigma_2)(A) (S \otimes T).
\end{split}
\]
It is straightforward to check that that the tensor product is associative, where equality holds ``on the nose''. It should be noted that this is rather special, and in many categories this is only true up to isomorphism. This is already the case for the category of vector spaces, which is due to the fact that the tensor product is unique only up to isomorphism, which in turn stems from the observation that there is no canonical basis.

With this definition we can also show that $\rho \otimes \sigma$ is transportable. Choose a cone $\Lambda$. Since $\rho$ and $\sigma$ are transportable, there are unitaries $S$ and $T$ such that $\rho_1(A) := S \rho(A) S^*$ and $\sigma_1(A) := T \sigma(A) T^*$ are localised in $\Lambda$. But then $\rho_1 \otimes \sigma_1$ is also localised in $\Lambda$. Finally, note that $S \otimes T$ is unitary and $S \otimes T \in \Delta(\rho \otimes \sigma, \rho_1 \otimes \sigma_1)$, from which the claim follows.

We now come to the fusion of two charges. What happens if we bring two charges together, what is going to be the total charge of the system? By the interpretation above, this amounts to studying the tensor product $\rho \otimes \sigma$ for two localisable and transportable endomorphisms. Let $\rho_i$ and $\rho_j$ be irreducible objects (i.e., $\Delta(\rho_i, \rho_j) \cong \mathbb{C} I$). Then we can try to decompose $\rho_i \otimes \rho_j$ into irreducibles:
\[
	\rho_i \otimes \rho_j \cong \sum_{k} N^{k}_{ij} \rho_k.
\]
This is called a \idx{fusion rule}. The sum is over a fixed set of irreducibles, one for each equivalence class, and the $N^k_{ij}$ are integers. The $\sum$ is to be interpreted as a ``direct sum'' in the category, an abstraction  of the direct sum of vector spaces. Here we will not go into details, since it is not required for the toric code, but a useful analogy is to think of representations of finite (or compact) groups. There, it is well known that if we take the tensor product of two representations, we can decompose it as a direct sum of irreducible representations. Indeed, the category of representations of finite groups is very much like our category $\Delta$.

So what are the fusion rules for the toric code? Note that the rules are given up to unitary equivalence, so it is enough to select a specific representative for each class. To this end, choose a site $x$ (for example, the origin of the lattice), and let $\xi$ be the path from $x$ going to infinity along a straight, upward line. Let $\widehat{\xi}$ be the parallel path on the dual lattice, on the plaquettes to the right of the path $\xi$. Then we can define, as before, the automorphisms $\rho^X$ and $\rho^Z$ (corresponding to the path $\xi$ and $\widehat{\xi}$, and set $\rho^Y = \rho^X \circ \rho^Z = \rho^X \otimes \rho^Z$. Finally, we let $\iota$ be the identical homomorphism of $\alg{A}^a$. Then it follows immediately that
\[
	\rho^X \otimes \iota = \rho^X, \quad \rho^Y \otimes \iota = \rho^Y, \quad \rho^Z \otimes \iota = \rho^Z.
\]
In tensor categories there always is (by definition) such a unit object for the tensor operation. Note that if we switch the order of the tensor factors, we get the same result.

By definition, $\rho^X \otimes \rho^Z = \rho^Y$. Finally, since all path operators square to the identity, we also have the following fusion rules:
\[
	\rho^X \otimes \rho^X = \iota = \rho^Z \otimes \rho^Z  = \rho^Y \otimes \rho^Y.
\]
The last rule follows from associativity of the tensor product, together with the observation that $\rho_i \otimes \rho_j \cong \rho_j \otimes \rho_i$. This can be checked directly for the toric code, but also follows from more general considerations that we will come to shortly. This completes the discussion of the fusion rule.

\begin{remark}
Note that each fusion rule has exactly one outcome, that is, we do not have to take direct sums into account. This is a special feature of \emph{abelian} superselection sectors, and is not true in general. Non-abelian models are necessarily more complicated, but this also makes them more powerful in applications, such as topological quantum computing (see~\cite{MR1943131,Wang} for a review).
\end{remark}

\begin{remark}
	For each of the irreducible charges we have that $\rho \otimes \rho \cong \iota$. More generally, if $\rho \otimes \sigma$ contains the identity map $\iota$ with multiplicity one (in other words, $N^\iota_{\rho\sigma} = 1$), we say that $\sigma$ is the \idx{conjugate} of $\rho$, and write $\overline{\rho} = \sigma$. One can think of conjugates as anti-charges, and in the type of physical models that we are discussing here, it automatically follows that also $\overline{\rho} \otimes \rho$ contains $\iota$ only once. Moreover, $\overline{\overline{\rho}} \cong \rho$. Both properties are not automatically true in arbitrary tensor categories (or more precisely, in fusion categories). 
\end{remark}

We conclude the discussion of fusion rules with an example. Let $x$ and $y$ be two sites, and let $\xi_1$ and $\xi_2$ be two paths to infinity, starting in $x$ and $y$, respectively. Let $\rho_1$ and $\rho_2$ be the corresponding automorphisms. Note that these both describe excitation of type $Z$. Recall that before we defined a charge transporter $V$ transporting the ribbon $\xi_1$ to $\xi_2$, i.e.\ such that
\[
	((V \rho_1 V^*) \otimes \rho_2)(A) = F_{\xi} A F_{\xi}^*
\]
for some path $\xi$ connecting $x$ and $y$. Note that this just is $\iota$ conjugated by a local unitary. The corresponding state $\omega_0(F_{\xi} A F_{\xi}^*)$ describes two excitations, but since they are conjugate to each other, the total charge is zero, and indeed this is a state in the trivial superselection sector, confirming $\rho_1 \otimes \rho_2 \cong \iota$.

\subsubsection{Braiding}
We now come to the question what happens if we start moving charges around. This is called \idx{braiding}. 
Equivalently, we can ask the question if it matters if we first create a charge $\rho_1$, and then a charge $\rho_2$, or do it the other way round.
More precisely, we will relate $\rho_1 \otimes \rho_2$ to $\rho_2 \otimes \rho_1$ for $\rho_1$ and $\rho_2$ in $\Delta$. Suppose that $\rho_1$ and $\rho_2$ are localised in cones $\Lambda_1$ and $\Lambda_2$.
We will define unitary operators $\varepsilon_{\rho_1, \rho_2} \in \Delta(\rho_1 \otimes \rho_2, \rho_2 \otimes \rho_1)$\index{_epsilon_rho1rho2@$\varepsilon_{\rho_1,\rho_2}$} in a way that is compatible with the rest of the structure of $\Delta$.

The general strategy is then to use transportability to move one of the charges, say $\rho_2$, away from $\rho_1$. That is, choose a cone $\widehat{\Lambda}_2$. Then there is an endomorphism $\widehat{\rho}_2$ that is localised in $\widehat{\Lambda}_2$ such that there is a unitary $V \in \Delta(\rho_2, \widehat{\rho}_2)$. Now define
\begin{equation}
	\varepsilon_{\rho_1,\rho_2} := (V \otimes I)^* (I \otimes V) = V^* \rho_1(V).
\end{equation}
Then if $A \in \alg{A}$ we calculate
\[
\begin{split}
	\varepsilon_{\rho_1, \rho_2} &(\rho_1 \otimes \rho_2)(A) = (V \otimes I)^* (\rho_1 \otimes \widehat{\rho}_2)(A) (I \otimes V) \\&= (V \otimes I)^* (\widehat{\rho}_2 \otimes \rho_1)(A) (I \otimes V) = (\rho_2 \otimes \rho_1)(A) \varepsilon_{\rho_1, \rho_2}.
\end{split}
\]
Here we used that $\rho_1 \otimes \widehat{\rho}_2 = \widehat{\rho}_2 \otimes \rho_1$, since the localisation regions are disjoint. Hence we see that $\varepsilon_{\rho_1,\rho_2} \in \Delta(\rho_1 \otimes \rho_2, \rho_2 \otimes \rho_1)$.

This gives the desired unitary intertwining $\rho_1 \otimes \rho_2$ and $\rho_2 \otimes \rho_1$, but to define it we had to choose the transported endomorphism $\widehat{\rho}_2$ and a charge transporter. For this to be a good definition, the final intertwiner should not depend on the specific choice of $\widehat{\rho}$ and intertwiner. This is indeed the case, up to the ``orientation'' of $\widehat{\Lambda}$: in our two dimensional settings, we can say that a cone is to the ``left'' of another disjoint cone.

\begin{figure}
	\begin{center}
	\begin{tikzpicture}
		\clip (-5.9,-1) rectangle (4,4.9);
		\filldraw[fill=lightgray, draw=black] (-2,6) -- (1,-0.5) -- (-0.5,6) -- (-2,6);
		\filldraw[fill=lightgray, draw=black] (3,6) -- (1.5,-1) -- (4.5,6) -- (3,6);
		\draw[thin,gray] (0,0) circle [radius=3];
		\draw[very thick] ([shift=(30.7:3)]0,0) arc (30.7:44.2:3cm);
		\draw[very thick] ([shift=(86.4:3)]0,0) arc (86.4:101.6:3cm);
		\draw[very thick] ([shift=(158.7:3)]0,0) arc (158.7:163.9:3cm);
		\draw[dashed] (-8,1) -- (0,0.75) -- (-8,1.75) -- (-8,1);
		\draw[draw=black] (-2,6) -- (1,-0.5) -- (-0.5,6) -- (-2,6);
		\draw[draw=black] (3,6) -- (1.5,-1) -- (4.5,6) -- (3,6);
		\draw (-3,1.5) node {$\Lambda_a$};
		\draw (-0.35,3.5) node {$\widehat{\Lambda}$};
		\draw (2.6,2.5) node {$\Lambda$};
	\end{tikzpicture}
	\end{center}
\caption{Two cones $\Lambda$ and $\widehat{\Lambda}$ with the auxiliary cone $\Lambda_a$. For a big enough circle, the intersection of the cones with this circle will yield three disjoint segments of the circle, making it possible to define what left and right mean. In this picture, the cone $\widehat{\Lambda}$ is to the \emph{left} of $\Lambda$.}
\label{fig:leftcone}
\end{figure}
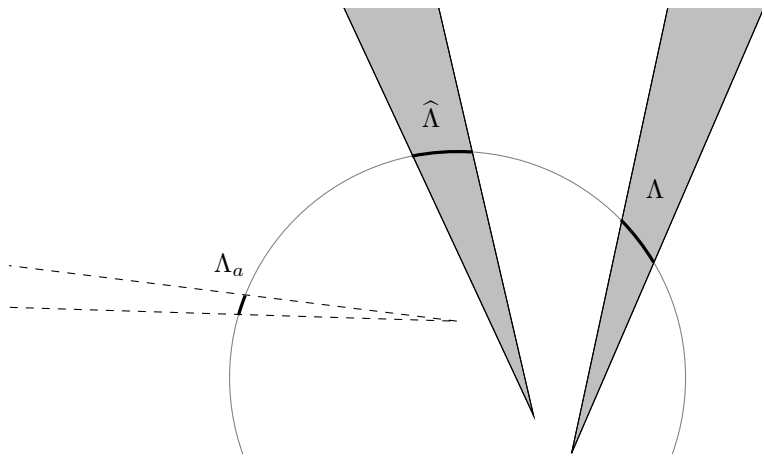

The idea for this definition can be seen from Figure~\ref{fig:leftcone}. To make this precise, Recall that we have chosen an auxiliary cone $\Lambda_a$. Let $\Lambda$ and $\widehat{\Lambda}$ be two disjoint cones (and disjoint from $\Lambda_a + x$ for some $x$). Now consider a circle centred at the origin. If we choose the radius of this circle to be big enough, we can arrange it so that the intersection of each cone with the circle is a connected segment of the circle (that is, the cone intersects the circle in only one place). Then the cone $\widehat{\Lambda}$ is to the \emph{left} of $\Lambda$ if we can rotate the segment corresponding to $\widehat{\Lambda}$ counter-clockwise until it intersects with the auxiliary cone without ever intersecting with the segment of $\Lambda$. If this is not the case, we say that it is to the right of $\Lambda$. Since the two cones are disjoint, this are all possibilities.

With this definition one can show that there are essentially two choices: we either have to chose $\widehat{\Lambda}$ to be to the \emph{left} or to the \emph{right} of $\Lambda$. Once this has been fixed, $\varepsilon_{\rho,\sigma}$ does not depend on the other choices that have to be made. This can shown by slowly moving the cone $\widehat{\Lambda}$ in a specific way, and show that $\varepsilon_{\rho,\sigma}$ stays the same along each step. This also explains why there are two choices: in two dimensions, we cannot continuously move the cone $\widehat{\Lambda}$ from the left of $\Lambda$ to the right, without it intersecting with $\Lambda$ (or the auxiliary cone). In essence, this is the reason why there are two choices. We will always chose $\widehat{\Lambda}$ to be to the left of $\Lambda$.

Note that the geometry of the system plays a role. In higher dimensions, we can no longer true define the relative orientation of the two cones, and we can always move one cone around the other without them ever intersecting (or intersecting the auxiliary cone). In other words, we can continuously deform one choice into the other, and it follows that they both agree. A consequence is that in that case, we can only have bosonic or fermionic charges, and no anyons, since then it can be shown that $\varepsilon_{\rho,\sigma} = \varepsilon_{\sigma,\rho}^*$~\cite{MR660538,MR0297259}. But this implies that $\varepsilon_{\rho,\sigma} \varepsilon_{\sigma,\rho}$, so that exchanging $\sigma$ and $\rho$ twice (or alternatively, moving $\rho$ around $\sigma$) is the trivial operation.

\begin{figure}
	\begin{center}
	\begin{tikzpicture}
		\clip (-3.9,-1) rectangle (5.4,4.9);
		\filldraw[fill=lightgray, fill opacity=0.3, draw=black] (-4,6) -- (1,-0.5) -- (-1.5,6) -- (-4,6);
		\filldraw[fill=lightgray, draw=black] (3,6) -- (1.5,-1) -- (7.5,6) -- (3,6);
		\draw[dashed] (-8,1) -- (-0.3,0.75) -- (-8,1.75) -- (-8,1);
		\draw (-3,1.5) node {$\Lambda_a$};
		\draw (0,3.5) node {$\widehat{\Lambda}$};
		\draw (3.6,0.5) node {$\Lambda$};
		\draw[very thick] (2.2,0.5) .. controls (4,2) and (2.6,6) .. (4,6);
		\draw[very thick,dashed] (3.2,1.2) .. controls (3.5,2.5) .. (6.5,6.5);
		\draw[very thick,dashed] (0.0,1.2) .. controls (-1.3,2.5) .. (-2,6.5);
		\filldraw[fill=black] (2.2,0.5) circle [radius=0.05];
		\filldraw[fill=black] (3.2,1.2) circle [radius=0.05];
		\filldraw[fill=black] (0.0,1.2) circle [radius=0.05];
		\draw[very thick, dashed] (-1.35,3.5) .. controls (0,5) and (1.5,3) .. (4.3,3.5);
		\draw (1,4.3) node {$F_{\widehat{\xi}_n}$};
	\end{tikzpicture}
	\end{center}
	\caption{An $X$ (solid line) and a $Z$ (dashed line) charge localised in the cone $\Lambda$, together with a choice for $\widehat{\rho}$. Also shown is a possible choice of paths for $F_{{\widehat{\xi}}_n}$ used in constructing the charge transporter. Note that it crosses the path used to define the $X$ charge.}
\label{fig:charges}
\end{figure}
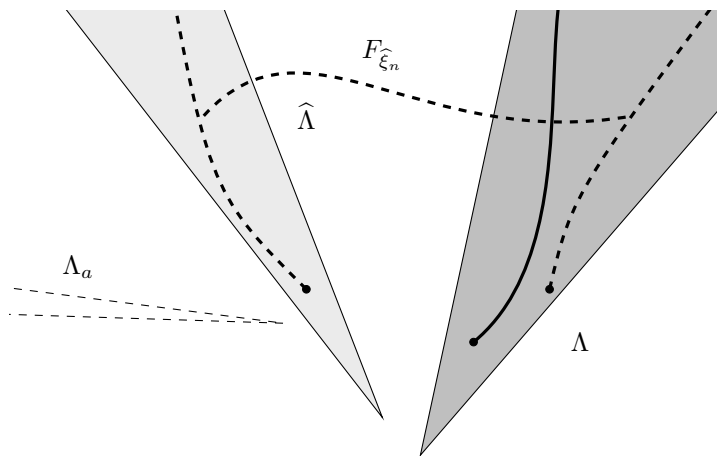

The definition of $\varepsilon_{\rho,\sigma}$ enjoys some natural properties that are to be expected in a category. For example, let $T \in \Delta(\rho,\rho')$. Then it can be shown that $\varepsilon_{\rho',\sigma} (T \otimes I) = (I \otimes T) \varepsilon_{\rho,\sigma}$ (and similar in the other argument). One says that $\varepsilon_{\rho,\sigma}$ is \idx{natural} in its arguments. Another property is that the braid relations hold:
\[
	\varepsilon_{\rho \otimes \sigma, \tau} = (\varepsilon_{\rho,\tau} \otimes 1_\sigma)(1_\rho \otimes \varepsilon_{\sigma,\tau}), \quad
	\varepsilon_{\rho, \sigma \otimes \tau} = (1_\sigma, \varepsilon_{\rho,\tau})(\varepsilon_{\rho,\sigma} \otimes 1_\tau).
\]
In other words, we could either interchange $\rho \otimes \sigma$ with $\tau$, or first interchange $\sigma$ and $\tau$, and then $\rho$ and $\sigma$ (and similarly for the second equation).
This compatible (with the $\otimes$-product) makes $\Delta$ into a \idx{braided tensor category}.

We know come back to the interpretation of moving one charge around the other. This can be made clear by considering the example of the toric code. In this case it is possibly to explicitly calculate the operators $\varepsilon_{\rho,\sigma}$, since we can easily find a net of local operators converging to a charge transporter $V$. To this end, consider a charge of type $X$ and one of type $Z$ both localised in a cone $\Lambda$, with paths extending to infinity as in Figure~\ref{fig:charges}. We will write $\rho$ and $\sigma$ for the corresponding automorphisms. By the convention chosen above we also have to pick a cone $\widehat{\Lambda}$ to the left of $\Lambda$. We choose the path corresponding to the transported $Z$-charge as in Figure~\ref{fig:charges} (the choice is not important, as long as it is inside of the cone).

To construct the charge operator $V$ transporting the $Z$ charge to $\widehat{\Lambda}$ we proceed as before. Let $\xi$ denote the ribbon for the $Z$-charge in $\Lambda$, and write $\xi'$ for the transported ribbon. Then, for each $n > 0$, we choose a path $\widehat{\xi_n}$ between the $n$th sites of $\xi$ and $\xi'$, giving us a path operator $F_{\widehat{\xi}_n}$. The charge operator $V$ can then be obtained as the weak limit of $V_n = F_{\xi_n} F_{\widehat{\xi}_n} F_{\xi'_n}$.

Note that the paths cross at some point. This means that at the intersection, one of the Pauli operators in $F_{\widehat{\xi}_n}$ will act on the same site as the path operators that we have to conjugate with to define $\rho$. By using the commutation relations for the Pauli operators we can get rid of the conjugation, at the expense of a minus sign. Hence, for big enough $n$, we will have $\rho(V_{\xi_n}) = - V_{\xi_n}$. Since we can extend $\rho$ to a weakly continuous endomorphism on the auxiliary algebra, and the construction of the sequence converging to $V$, we see that $\rho(V) = -V$. This means that
\begin{equation}
	\varepsilon_{\rho,\sigma} = V^*\rho(V) = - V^*V = -I.
\end{equation}
Other braiding operators can be found similarly.

Note that instead of $\rho \otimes \sigma$, we could also start with $\sigma \otimes \rho$. In constructing $\varepsilon_{\sigma,\rho}$ we again have to choose a path in a cone $\widehat{\Lambda}$ to the left. Since now we have to move the $X$ charge, the paths used in the converging sequence no longer cross, and it is easy to see that $\varepsilon_{\sigma,\rho} = I$. Hence we see that $\varepsilon_{\rho,\sigma} \circ \varepsilon_{\sigma,\rho} \neq I$. This is the defining feature of anyons: interchanging the two charges twice is a non-trivial operation. This will always be the case in the toric code, as long as we consider two distinct charges: either $\varepsilon_{\rho,\sigma}$ or $\varepsilon_{\sigma,\rho}$ will be equal to $-I$, the other will be equal to $I$. Charges with such non-trivial exchange behaviour are called \emph{anyons}\index{anyon}. In this case they are abelian: the result is only phase factor. For \emph{non}-abelian anyons, $\varepsilon_{\rho,\sigma} \rho_{\sigma,\rho}$ is a non-trivial unitary.\footnote{In the local quantum physics community, the term \emph{plektons} is also used for non-abelian anyons, but this does not seem to have caught on outside of this community.}

This is not true for fermions or bosons, for which we always have $\varepsilon_{\rho,\sigma} \circ \varepsilon_{\sigma,\rho} = I$. In the language of tensor categories, we say that $\Delta$ is a braided tensor category that is not \emph{symmetric}. In fact, in our case it is \emph{modular}\index{tensor category!modular}. In a sense, this is a braided tensor category that is as far away from being symmetric is possible. The category is said to be modular if for every irreducible $\rho \in \Delta$, there is an irreducible $\sigma \in \Delta$ such that $\varepsilon_{\rho,\sigma} \circ \varepsilon_{\sigma,\rho} \neq I$, but there are other (equivalent) definitions~\cite{Muger:2003p3368,MR1147467}.

\begin{remark}
It is good to contrast with the example of representations of finite groups. In that case, we can use the flip operator to relate $\pi_1 \otimes \pi_2$ to $\pi_2 \otimes \pi_1$. This defines a braiding on the category. However, it is clear that applying the flip operator twice gives back the same tensor product representation. Hence this category is an example of a \emph{symmetric} tensor category (or a symmetric fusion category).
\end{remark}

There is a more structure than we have discussed here, for example even though the fusion rules do not require it, it is possible to define a direct sum operation on $\Delta$. In the language of tensor categories, it follows that $\Delta$ is a \emph{braided tensor $C^*$-category} (BTC). The modularity can be interpreted as meaning that the braiding is as non-degenerate as possible. It implies that we can always detect a charge by creating a pair of charges (of a different type) from the ground state, move one charge of the pair around the charge we want to measure, and annihilate them again. If the category is modular, there always is a charge for which this is a non-trivial operation.

In the case of the toric code, the category $\Delta$ is completely known, and can be summarised as follows:
\begin{theorem}[\cite{toricendo}]
	The category $\Delta$ is equivalent (as a braided tensor $C^*$-category) to the category $\operatorname{Rep}_f \mathcal{D}(\mathbb{Z}_2)$, the category of finite dimensional representations of the quantum double of the group algebra of $\mathbb{Z}_2$.
\end{theorem}
The quantum double is an algebraic object that can be defined for any suitable Hopf algebra, such as the group algebra of a finite group. The upshot of this theorem is that to understand the physical properties of the charges in our model, it is enough to study the representations of an \emph{algebraic} object. The representations of the quantum doubles of Hopf algebras are well understood, so this answers the question.

\chapter{Applications of Lieb-Robinson bounds}
In Chapter~\ref{ch:lr} we already discussed some applications of Lieb-Robinson bounds. Here we will discuss two more recent developments: scattering theory for quantum spin systems and automorphic equivalence of gapped ground states. The first topic is again a confirmation that the Lieb-Robinson velocity is to quantum spin systems what the speed of light is to relativistic systems. Indeed, we will outline how one can build up a scattering theory for quantum spin systems along the lines of Haag-Ruelle theory in algebraic quantum field theory. The Lieb-Robinson bound here is key to construct incoming and outgoing particle states: it allows us to localise the operators that create single-particle states from the ground state. We can then choose such creation operators for particles with different velocities. If $|t|$ then becomes large, the excitations will be far apart from each other, which can be made precise using the localisation of the creation operators. This implies that the single-particles do not interact at large times, so that they will behave as free particles. Besides providing yet another application of Lieb-Robinson bounds, it also nicely illustrates how the SNAG theorem, discussed in Section~\ref{sec:snag}, can be used in quantum spin systems.

The other application concerns the equivalence of ``quantum phases''. A quantum phase is essentially an equivalence class of ground states. For topologically ordered systems the following definition of equivalence is prevalent in the literature: two gapped ground states are in the same phase if one can continuously deform the Hamiltonian of the first ground state to that of the second one without closing the gap along the way. This notion can be studied in the thermodynamic limit, and it will turn out that in this case the two ground states are related by an automorphism. This automorphism is local, in the sense that it satisfies a Lieb-Robinson type of bound. The construction of this automorphism again heavily relies on Lieb-Robinson bounds. The main points are reviewed in the last section of these lecture notes.

\section{Scattering theory}
One only has to look at the success of the Large Hadron Collider at CERN in obtaining information on and verifying the standard model, to see how important and useful scattering theory is. Also in many other experiments, scattering theory plays a role. It would therefore be helpful to have a scattering theory for excitations of quantum spin systems. This is what we will outline here for the case of gapped ground states of translational invariant ground states of quantum spin systems satisfying a Lieb-Robinson bound.

In a scattering experiment one prepares a number of particles (or excitations) in a known configuration $\psi^{in}$. Then, after some time, the particles will have interacted (or scattered), and we determine the outgoing configuration $\psi^{out}$. The incoming and outgoing states are related by a (typically unitary) matrix $S$. The goal of scattering theory is to find this scattering matrix $S$.\index{S-matrix@$S$-matrix}

We have already seen examples how concepts and techniques from relativistic local quantum physics can often be translated to the quantum spin setting. This will again be the case here, where we adapt the Haag-Ruelle scattering theory~\cite{Ha58,Ru62}\index{Haag-Ruelle scattering} to quantum spin systems. To discuss scattering theory it is first necessary to talk about $n$-particle excitations. To this end, we first recall the construction of Fock space, and then discuss single-particle excitations. This definition is motivated by the well-known particle classification in relativistic quantum mechanics, and makes essential use of the spectrum as defined in the SNAG Theorem~\ref{thm:snag}. Then Haag-Ruelle creation operators that create these excitations from the ground state are introduced. Using these operators it is possible to discuss scattering theory and define the $S$-matrix.

Again, Lieb-Robinson bounds provide an essential tool in the analysis. This section is also meant as an introduction to the literature, and only the big picture is discussed. We also clarify the role Lieb-Robinson bounds play in the whole theory. Full proofs and all the estimates can be found in~\cite{BDNScattering}. We also note that Haag-Ruelle theory is by no means the only way to do scattering theory in spin systems, see for example~\cite{Haegeman:2013a,Vanderstraeten:2014a}, although the setting there is slightly different.

The precise sufficient assumptions to define an $S$-matrix will be given below, but for now note that so far, in the case of quantum spin systems, only scattering of neutral excitations has been developed. That is, all bosonic excitations that can be obtained in the ground state sector (fermions would require operators that anti-commute at a distance, which we do not have in the ground state sector). Therefore, it cannot be applied to the toric code example discussed in the previous section. It should also be noted that in this section, we consider a specific type of excitations, in the sense that they lie on a momentum shell. It does not encompass the ``anyonic excitations'' discussed before (for one, they are time invariant, so that there is no scattering).

\subsection{Fock space}
What actually happens at a microscopic level when two particles scatter off of each other is often difficult to describe. Rather, what one does in scattering experiments, is describing a configuration of \emph{incoming} particles, and then look at the \emph{outgoing} particles. If we do this long before and long after the actual scattering, the particles will be far away from each other, hence it is reasonable to assume that they do not interact any more. In other words, we can regard them as \emph{free particles}\index{free particle}. It would therefore be useful to have a way to describe configurations of $n$ free particles, for each $n$. This can be done with the help of the \emph{Fock space}:

\begin{definition}
	Let $\mathcal{H}$ be a Hilbert space. Then we define the \idx{Fock space}\index{_F(H)@$\mathcal{F}(\mathcal{H})$} $\mathcal{F}(H)$ as the Hilbert space
	\[
		\mathcal{F}(\mathcal{H}) := \bigoplus_{n=0}^\infty \mathcal{H}^{\otimes n},
	\]
	where $\mathcal{H}^{\otimes n}$ is the tensor product of $n$ copies of $\mathcal{H}$, and the case $n=0$ is understood to be $\mathbb{C}$.
\end{definition}

In our application, $\mathcal{H}$ will be a single-particle Hilbert space. Then, as we have seen before, $\mathcal{H}^{\otimes n}$ would describe $n$ of such particles together. By taking the direct sum, we are able to describe $n$ particle states for any $n$.

It will be useful to extend certain operators on $\mathcal{H}$ to the Fock space. First consider the case of a unitary $U \in \alg{B}(\mc{H})$. Then for each $n \geq 1$, we can define
\[
	U^{(n)}(\psi_1 \otimes \dots \otimes \psi_n) := U\psi_1 \otimes U\psi_2 \otimes \dots \otimes U \psi_n.
\]
Since $U$ is unitary, $U^{(n)}$ is unitary. If we define $U^{(0)} = I$, it follows that $\bigoplus_{n \geq 0} U^{(n)}$ can be extended to a unitary $\Gamma(U)$ on $\mc{F}(\mc{H})$. The operator $\Gamma(U)$ is called the \idx{second quantisation} of $U$. The origin of the term is historical, and does not have much to do with quantisation in the sense of going from classical to quantum operators.

Now consider a self-adjoint (possible unbounded) operator $H$ defined on a dense domain $D(H)$. Define for $\psi_1 \otimes \cdots \otimes \psi_n \in \bigotimes_{k=1}^n D(H)$ an operator $H^{(n)}$ by
\[
	H^{(n)} (\psi_1 \otimes \cdots \otimes \psi_n) = \sum_{i=1}^n (\psi_1 \otimes \psi_2 \otimes \cdots \otimes H \psi_i \otimes \psi_{i+1} \otimes \cdots \otimes \psi_n).
\]
This yields an unbounded operator $\bigoplus_{n=0}^\infty H^{(n)}$ (even if $H$ is bounded!). One can show that this operator is closable and has a dense set of analytic vectors. The closure is usually written as $d\Gamma(H) = \overline{\bigoplus_{n=0}^\infty H^{(n)}}$.
The notation can be explained by considering a strongly continuous one-parameter group $U_t = \exp(i t H)$ for some unbounded self-adjoint generator $H$. Then one can show that
\[
	\Gamma(U_t) = e^{i t d\Gamma(H)},
\]
so that $d\Gamma(H)$ can be seen as the derivative of $\Gamma(U_t)$ at $t = 0$, just like $H$ is the derivative of $U_t$ at $t = 0$.

\subsubsection{Bosonic and fermionic Fock space}
Usually not all states in $\mathcal{F}$ are relevant. For example, we know that if we interchange two bosons, the resulting state should be the same. Hence it would be useful to restrict to bosonic (resp.\ fermionic) states that are symmetric (resp.\ anti-symmetric) under permutations. To this end, let $S_n$ be the permutation group of $n$ elements. Let $\sigma \in S_n$ and write $\sigma(k)$ for the image of $k \in \{1, \dots, n\}$ under the permutation. Then for each $\sigma$ we can define
\[
	U_\sigma (\psi_1 \otimes \cdots \otimes \psi_n) := \psi_{\sigma(1)} \otimes \dots \otimes \psi_{\sigma(n)}.
\]
It can be shown that this map extends to a unitary operator on $\mathcal{H}^{\otimes n}$ (see Exercise~\ref{ex:permutation} below). Hence we can take the average over all permutations, to obtain an operator
\begin{equation}
	P_n^s := \frac{1}{n!} \sum_{\sigma \in S_n} U_{\sigma}.
\end{equation}
\begin{exercise}\label{ex:permutation}
	Show that $U_\sigma$ is unitary, that $P_n^s$ is a projection, and that the closure of $P_n \mathcal{H}^{\otimes n}$ is precisely the subspace of permutation invariant vectors.
\end{exercise}

Because of this exercise, it is natural to define the \emph{bosonic} or (\emph{symmetric}) \emph{Fock space} as $\mathcal{F}_s(\mathcal{H}) := \bigoplus_{n=0}^\infty P_n^s \mathcal{H}^{n}$.\index{Fock space!bosonic and fermionic}\index{_FsH@$\mathcal{F}_s(\mathcal{H})$} We write $P_s$ for the projection from $\mathcal{F}(\mathcal{H})$ onto $\mc{F}_s(\mc{H})$.

Although we will not use the fermionic Fock space, it can be defined analogously. Let again $\sigma \in S_n$ be a permutation, and write $\operatorname{sign}(\sigma)$ for the sign of the permutation (i.e., it is +1 if it is an even number of transpositions, or -1 if it can be obtained by an odd number of transpositions). Then define a unitary via
\[
	U_\sigma^a(\psi_1 \otimes \cdots \otimes \psi_n) := \operatorname{sign}(\sigma) \psi_{\sigma(1)} \otimes \cdots \psi_{\sigma(n)}.
\]
Then $P^a_n := \frac{1}{n!} \sum_{\sigma \in S_n} U_\sigma^a$ defines a projection on the anti-symmetric part of $\mathcal{H}^{\otimes n}$. The fermionic Fock space is then $\mathcal{F}_a(\mathcal{H}) := \bigoplus_{n=0}^\infty P_n^a \mathcal{H}^{\otimes n}$, consisting of the anti-symmetric states.

\subsubsection{Creation and annihilation operators}
Suppose that we have an $n$-particle state $\psi_1 \otimes \cdots \otimes \psi_n$ and let $\psi$ be a single-particle state. Then it would be useful to have an operation that either ``creates'' the particle $\psi$ to get a $n+1$-particle state, or remove it from the $n$ particles. To this end, define for $\psi \in \mathcal{H}$ the \idx{creation operator} $a^*(\psi)$ and the \idx{annihilation operator} $a(\psi)$ by linear extension of
\begin{align*}
	a(\psi) (\psi_1 \otimes \cdots \otimes \psi_n) &= n^{\frac{1}{2}} \langle \psi, \psi_1 \rangle \psi_2 \otimes \cdots \otimes \psi_n, \\
	a^*(\psi) (\psi_1 \otimes \cdots \otimes \psi_n) &= (n+1)^{\frac{1}{2}} \psi  \otimes \psi_1 \otimes \cdots \otimes \psi_n.
\end{align*}
It should be noted that the operators are not bounded, so that they will only be defined on a dense domain.

\begin{exercise}
	Show that $a(\psi)$ and $a^*(\psi)$ are not bounded. Also show that $\langle a^*(\psi) \varphi_1, \varphi_2 \rangle = \langle \varphi_1, a(\psi) \varphi_2\rangle$ with $\varphi_1 \in \mathcal{H}^{\otimes n}$ and $\varphi_2 \in \mathcal{H}^{\otimes n+1}$.
\end{exercise}

In applications we are often interested in either the bosonic or fermionic Fock space, hence it is useful to have creation and annihilation operators that restrict to these subspaces. Hence can define
\[
a_+(\psi) := P_s a(\psi) P_s, \quad\quad\quad a_+^*(\psi) := P_s a^*(\psi) P_s.
\]
The fermionic creation and annihilation operators $a_{-}^*(\psi)$ and $a_{-}(\psi)$ are defined analogously.

As the name suggests, these operators can be used to create $n$-particle states. To see this, consider the state $\Omega = (1, 0, 0, \dots)$ in Fock space. This corresponds to the absence of any particle, that is, it is the vacuum state. Suppose that $\psi_1, \psi_2 \in \mc{H}$ are single-particle states. Then we calculate
\[
\begin{split}
	\varphi^{(2)} &:=  \frac{1}{\sqrt{2}} a_+^*(\psi_1) a_+^*(\psi_2) \Omega = \frac{1}{\sqrt{2}} a_+^*(\psi_1) \psi_2 \\
		&= P_s (\psi_1 \otimes \psi_2) = \frac{1}{\sqrt{2}} \left( \psi_1 \otimes \psi_2 + \psi_2 \otimes \psi_1 \right)
\end{split}
\]
Hence $\varphi^{(2)}$ is a bosonic 2-particle state, with a particle in state $\psi_1$ and one in $\psi_2$. Similarly, $a_{+}(\psi_1)$ removes a particle $\psi_1$. This construction can be straightforwardly extended to $n$ particles, and we see that by acting with products of creation operators on the vacuum state, a dense subset of Fock space can be obtained. As a final aside, note that $a_{-}^*(\psi) a_{-}^*(\psi) \Omega = 0$, due to the anti-symmetric projection. This encodes the Pauli exclusion principle: we cannot create two fermions in the same state $\psi$.

\subsection{Single-particle states}
The next step is to define single-particle states, and find operators to create them. The definition will ultimately make use of the joint spectrum\index{joint spectrum} as obtained via the SNAG theorem, but before that we make precise the assumptions that we make. 

\begin{assumption}\label{as:qspin}
	Consider a ground state $\omega$ of a quantum spin system which carries an action $\mathbf{x} \mapsto \tau_{\mathbf{x}}$, with $\mathbf{x} \in \mathbb{Z}^d$, of translations, and let $(\pi, \Omega, \mathcal{H})$ be the corresponding GNS representation. We furthermore assume that:
\begin{enumerate}
	\item the dynamics $t \mapsto \tau_t$ satisfy a Lieb-Robinson bound;
	\item the ground state is translation invariant;
	\item if $H$ is the Hamiltonian implementing the dynamics in the ground state representation, $H\Omega = 0$ and $H$ is gapped;
	\item there is a unique ground state in this representation, i.e.\ $0$ is a simple eigenvalue with eigenvector $\Omega$;
	\item the interaction is translation invariant.
\end{enumerate}
\end{assumption}
The last assumption implies that $\tau_t \tau_{\mathbf{x}} = \tau_{\mathbf{x}} \tau_t$ for all $x \in \mathbb{Z}^d$ and $t \in \mathbb{R}$. Consequently, we will write $\tau_{(t,\mathbf{x})} := \tau_t \circ \tau_{\mathbf{x}}$ for the action of space-time translations. The condition on the gap is there for technical reasons. The existence of a Lieb-Robinson bound is more fundamental, and plays a crucial role. As we have seen before, it plays the role of the speed of light in relativistic theories. This again is the case here. For example, we shall see that it gives an upper bound on the velocity of particle excitations. It will also be necessary for some locality estimates, which are used to show that at large times, particle creation operators for different velocities almost commute (and hence behave as non-interacting particles).

Note that the assumptions imply that we can find a strongly continuous group of unitaries $(t,\mathbf{x}) \mapsto U(t,\mathbf{x})$ acting on $\mathcal{H}$ which leaves the ground state $\Omega$ invariant. By Theorem~\ref{thm:snag}, the SNAG theorem, there is a spectral measure $dP$ such that
\[
	U(t,\mathbf{x}) = \int_{\mathbb{R} \times \mathbb{T}^d} e^{i(Et - \mathbf{p} \cdot \mathbf{x})} dP(E,\mathbf{p}).
\]
Here we changed the sign of the momentum part to make the expression similar to relativistic theories. The $d$-dimensional torus, $\mathbb{T}^d$, appears because it is the group dual of the group $\mathbb{Z}^d$. This duality is essentially the Fourier transform on the circle.

The energy-momentum spectrum of the theory is then defined to be $\operatorname{Sp}(U)$. Note that from the assumptions that we have made, $0$ is an isolated point in the spectrum. Because $\omega$ is a ground state, it follows that $\operatorname{Sp}(U) \subset \mathbb{R}_+ \times \mathbb{T}^d$.

We make some further assumptions on the spectrum, namely that there is a \idx{mass shell} satisfying certain regularity conditions. This is somewhat analogous to mass shells in relativistic theories. The definition is as follows:

\begin{definition}
	Let $\mathfrak{h} \subset \operatorname{Sp}(U)$. Then $\mathfrak{h}$ is a \emph{mass shell} if the following conditions are satisfied:
	\begin{enumerate}
		\item the spectral projection $P(\mathfrak{h})$ is non-zero;
		\item there is an open subset $\Delta_\mathfrak{h} \subset \Gamma$ such that $\mathfrak{h}$ is the graph of a smooth, real valued function $E \in C^\infty(\mathbb{T}^d)$:\index{_Deltah@$\Delta_{\mathfrak{h}}$}
			\[
				\mathfrak{h} = \{ (E(\mathbf{p}), \mathbf{p}) : \mathbf{p} \in \Delta_{\mathfrak{h}} \}.
			\]
			$E$ is called the \idx{dispersion relation}.
		\item The second derivatives of $E$ vanish almost nowhere, that is $\{ \mathbf{p} \in \Delta_{\mathfrak{h}} : D^2E(p) = 0\}$ has Lebesgue measure zero. Here $D^2$ is the Hessian of $E$, that is, $D^2E(\mathbf{p}) = [\partial_i \partial_j E(\mathbf{p})]_{i,j=1, \dots, d}$.
	\end{enumerate}
	If $\mathfrak{h}$ is a mass shell, $\mc{H}_{\mathfrak{h}} := P(\mathfrak{h}) \mathcal{H}$ is called the \idx{single-particle subspace}.
\end{definition}
Note that $\Delta_\mathfrak{h}$ might be smaller than $\mathbb{T}^d$, i.e., the mass shell might only cover part of the momentum range. The last condition is necessary for technical reasons. Since we do not discuss the complete proof here, it will not appear again, but the main reason is as follows. The group velocity is given by $\nabla E$, so that the last condition says that it is constant only on a set of Lebesgue measure zero. This will ensure that we can find a dense space of single-particle states by considering configurations of particles with distinct velocities.

In addition to Assumption~\ref{as:qspin}, we make the following additional assumptions about the mass shell $\mathfrak{h}$:
\begin{assumption}\label{as:massshell}
	Let $\mathfrak{h}$ be a mass shell. We assume in addition that it is \emph{isolated}\index{mass shell!isolated}, and that either condition~\ref{it:regular} or~\ref{it:pseudorel} holds:
	\begin{enumerate}
		\item $\mathfrak{h}$ is isolated, in the sense that for any $p \in \Delta_{\mathfrak{h}}$ there exists $\varepsilon > 0$ such that 
			\[
				( [E(\mathbf{p}) -\varepsilon, E(\mathbf{p}) + \varepsilon] \times \{\mathbf{p} \} ) \cap \operatorname{Sp}(U) = \{(E(\mathbf{p}), \mathbf{p}) \};
			\]
		\item\label{it:regular} it is \emph{regular}, in the sense that $\{ \mathbf{p} : \det(D^2(E)) = 0 \}$ has Lebesgue measure zero; 
		\item\label{it:pseudorel} it is \emph{pseudo-relativistic}: $(\mathfrak{h} - \mathfrak{h}) \cap \operatorname{Sp}(U) = \{0 \}$.
	\end{enumerate}
\end{assumption}
\begin{figure}\label{fig:massshells}
	\begin{center}
	\includegraphics[width=10cm]{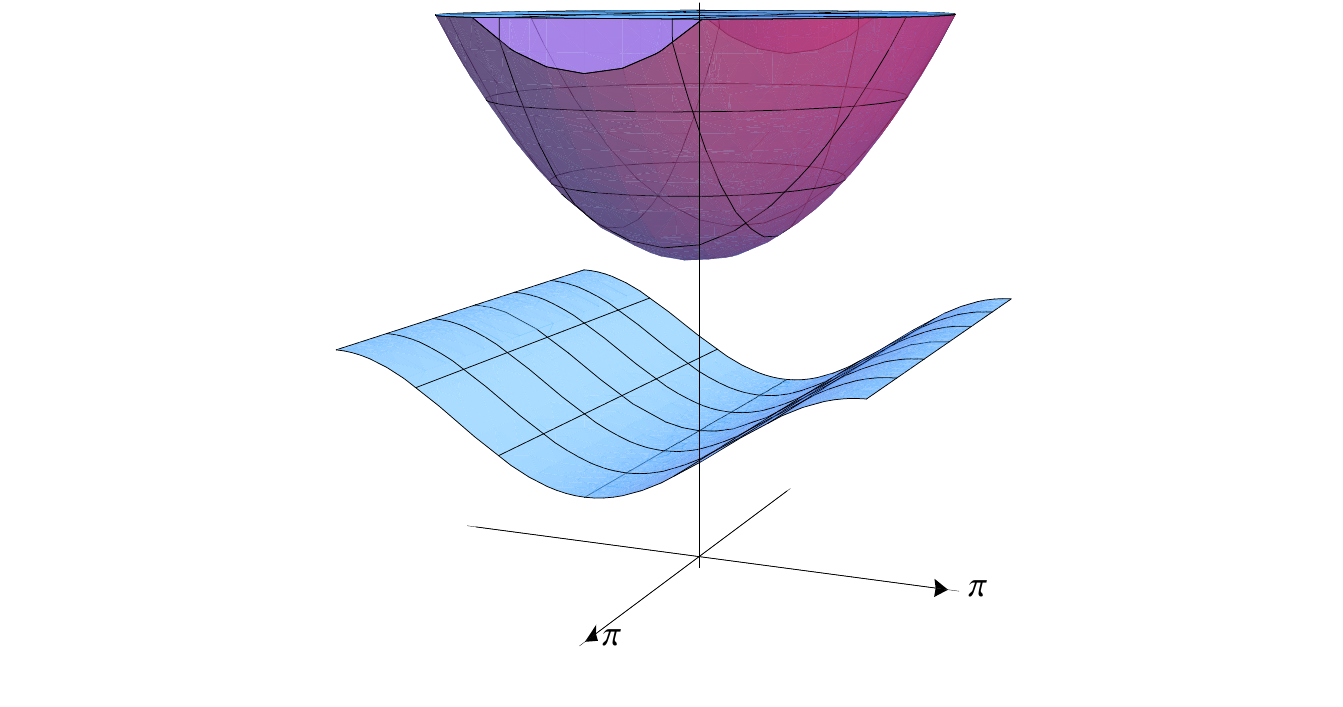}
	\includegraphics[width=7cm]{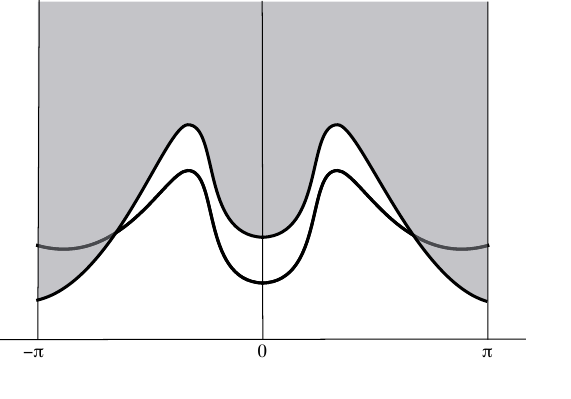}
\end{center}
	\caption{A pseudo-relativistic mass shell (top) and a regular, isolated mass shell (bottom). Note in the bottom picture the mass shell is, by definition, only the part that is isolated from the rest of the spectrum. That is, it does not cover the full momentum range. Pictures taken from~\cite{BDNScattering}. (c) 2015 Springer Basel, included with permission.}
\end{figure}

The name \emph{pseudo-relativistic} is chosen because this property is shared with relativistic mass shells. Again, the last two properties are necessary for technical reasons, and will not play much of a role in our discussion. The isolatedness of the mass shell is more important: as we will see it makes it possible to create states with momentum on the mass shell, but disjoint from the rest of the spectrum. Figure~\ref{fig:massshells} shows an example of a regular and a pseudo-relativistic mass shell.

\subsection{Operator smearing and almost local observables}
Now that we have defined single-particle states, we will look at operators that create such states from the ground state. Moreover, these operators should be local in a suitable sense. As we shall see, there is a trade-off between the locality of this operator, and localisation of the energy-momentum spectrum. In particular, it is not possible to have a strictly local operator $A$ such that $\pi(A) \Omega \in P(\Delta) \mathcal{H}$ for some compact subset $\Delta \subset \mathbb{R}^+ \times \mathbb{T}^d$. 

The following definition is useful in studying these questions.
\begin{definition}
	An operator $A \in \alg{A}$ is said to have \emph{energy-momentum transfer}\index{energy-momentum!transfer} in $\Delta_t$ if	
	\[
		\pi(A) P(\Delta) = P(\overline{\Delta + \Delta_t}) \pi(A) P(\Delta)
	\]
	for all Borel subsets $\Delta \subset \mathbb{R} \times \mathbb{T}^d$.
\end{definition}
In other words, if one has a state $\psi$ with energy-momentum in $\Delta$, then $\pi(A) \psi$ has its energy-momentum shifted by $\Delta_t$. We have already seen an example of such an operator in Section~\ref{sec:vacuum}, where we considered smeared operators $\int \beta_x(Q) f(x) dx$, with the Fourier transform of $f$ outside the forward light-cone. The calculation there shows that (at least on the vacuum), it shifts the energy-momentum outside of the forward light-cone, and hence the vacuum state is annihilated by such operators. This idea of smearing operators will also be of use here: let $f \in L^1(\mathbb{R} \times \mathbb{Z}^d)$ and $A \in \alg{A}$. Then we can define the smeared version of $A$ via\index{_tauf(A)@$\tau_f(A)$}
\begin{equation}
	\tau_f(A) := (2 \pi)^{-\frac{d+1}{2}} \sum_{\mathbf{x} \in \mathbb{Z}^d} \int_\mathbb{R} \tau_{(t,\mathbf{x})}(A) f(t,\mathbf{x}) dt.
\end{equation}
Note that since the spatial variable is discrete, we have to sum instead of to integrate over $\mathbf{x}$. The significance of this construction lies in the following lemma:

\begin{lemma}
	\label{lem:smearing}
	Let $f \in L^1(\mathbb{R} \times \mathbb{Z}^d)$ and suppose that $A \in \alg{A}$ has energy-momentum transfer in $\Delta$. Then $\tau_f(A)$ has energy-momentum transfer in $\Delta \cap \supp \widehat{f}$, where $\widehat{f}$ is the Fourier transform of $f$.
\end{lemma}
The proof, in the context of quantum spin systems, can be found in~\cite{BDNScattering}, but an essential part is to note that $\pi(\tau_{(t,\mathbf{x})}(A)) = U(t,\mathbf{x}) \pi(A) U(t, \mathbf{x})^*$. At that point one can use the SNAG theorem to write $U(t, \mathbf{x})$ as an integral over the spectral measure (see also the calculation of $\pi_0(Q) \Omega$ in Section~\ref{sec:vacuum}).

To find operators that create single-particle states we can now proceed as follows. First choose an open neighbourhood of a point $(E,p) \in \mathfrak{h}$ that is disjoint from $\operatorname{Sp}(U) \setminus \mathfrak{h}$. Note that this is always possible because $\mathfrak{h}$ is isolated. Then choose a function $f \in L^1(\mathbb{R} \times \mathbb{T}^d)$ with a Fourier transform $\widehat{f}$ supported in $\Delta$. Consider now $\psi := \pi(\tau_f(A)) \Omega$ for some $A \in \alg{A}$. Then we find
\[
	\psi = \pi(\tau_f(A)) \Omega = \pi(\tau_f(A)) P(\{0\}) \Omega = P(\supp \widehat{f}) \psi.
\]
Here we used that $0$ is an isolated point of the spectrum, together with the energy-momentum transfer lemma. It follows that $\psi \in \mathcal{H}_{\mathfrak{h}}$, the single-particle subspace. It might still happen that $\psi = 0$, but because $\Omega$ is a cyclic vector, it can be shown that this cannot be the case for all $A \in \alg{A}$. Hence we can create single-particle states.

Just being able to create single-particle states, however, is not enough. In the end we want the particles to act independently for large timescales, or in other words, there creation operators should commute for such large $t$, at least if we choose the velocities of the particles judiciously. As we have seen, locality of the operators would be very useful here. Unfortunately, these operators cannot be strictly local. A heuristic argument goes as follows. The energy-momentum of $\tau_f(A)$ will be in $\supp \widehat{f}$. This has to be a compact set. But it is well-known from Fourier analysis, that the better localised $\widehat{f}$ is, the more $f$ is spread out. Consequently, $\tau_f(A)$ will not be strictly local. This also is true because of the smearing in the time direction: we know that $\tau_t(A)$ is in general not strictly localised any more for $t \neq 0$. In fact, it cannot be for all local $A$, since this would imply a zero-velocity Lieb-Robinson bound. In that case there would be no scattering at all, as we will also see later.

\begin{remark}
	One can even show a stronger result. If we disregard the energy component for a moment, and only smear over the space directions, one can show that the momentum transfer of \emph{any} non-trivial local operator is equal to $\mathbb{T}^d$, see~\cite[Proposition 3.7]{BDNScattering}. A similar statement is true in relativistic theories~\cite{Bu90}. The proof uses the support properties of the Fourier transform.
\end{remark}

Fortunately, it is possible to to consider a slightly larger class of observables, called \emph{almost local}, which can be approximated well by strictly local observables. Of course, by construction, \emph{any} observable in $\alg{A}$ can be approximated by strictly local observables. What sets the almost local observables apart is that the error that we make decreases quickly with the size of the support. In the context of quantum spin systems, the almost local observables have first been studied by Schmitz~\cite{Sch83}, but before that they also appeared in algebraic quantum field theory. The precise definition is:

\begin{definition}
	An observable $A \in \alg{A}$ is called \emph{almost local}\index{observable!almost local} if there exists a mapping $\mathbb{R}_+ \ni r \mapsto A_r \in \alg{A}_{loc}$ such that
	\[
		\| A - A_r \| = O(r^{-\infty}),
	\]
	with the support of $A_r$ contained in a ball of radius $r$. The notation $O(r^{-\infty})$ means that the terms is of $O(r^{-n})$ for any $n \in \mathbb{N}$. In other words, $r^n \| A-A_r \|$ goes to zero as $r \to \infty$ for any $n \in \mathbb{N}$. The set of all almost local observables is a $*$-algebra, which is denoted by $\alg{A}_{a-loc}$.\index{_A_aloc@$\alg{A}_{a-loc}$}
\end{definition}

Clearly we have $\alg{A}_{loc} \subset \alg{A}_{a-loc}$. Although the almost local observables are a bit more cumbersome to work with, it turns out that they have the right locality properties for our purposes. For example, if $A_i \in \alg{A}_{a-loc}, i=1,2$, then it can be shown that, with $\mathbf{y} \in \mathbb{Z}^d$, one has for any $n \in \mathbb{N}$
\[
	[ A_1, \tau_{\mathbf{y}}(A_2) ] = O\left( \left(\frac{1}{1 + \mathbf{y}^2}\right)^{\frac{n}{2}} \right).
\]
In other words, if we translate $A_2$ far enough, the commutator will go to zero. Although it will not become identically zero in general, it will become small enough for our purposes. The proof of this statement follows from a straightforward approximation argument, by first approximating $A_i$ with local observables.

More important is the following theorem, which says that smearing of almost local observables with Schwarz class functions again gives an almost local observable. A smooth function $f: \mathbb{R} \times \mathbb{Z}^d \to \mathbb{C}$ is said to be of Schwarz class (notation: $f \in \mathcal{S}(\mathbb{R} \times \mathbb{Z}^d)$) if $\partial_t^n f = O( ((1+t)^2 + (1+\mathbf{x}^2))^{-\infty})$. Alternatively, it means that all partial derivatives of the Fourier transform $\widehat{f}(E,\mathbf{p})$ are of $O( (1+E^2)^{-\infty})$ uniformly in $\mathbf{p}$. This is the point where the assumption that Lieb-Robinson bounds hold becomes crucial.

\begin{theorem}\label{thm:localsmear}
	Let $A \in \alg{A}_{a-loc}$ and $f \in \mc{S}(\mathbb{R} \times \mathbb{Z}^d)$, and suppose that $\tau_t$ satisfies a Lieb-Robinson bound. Then we have
	\begin{enumerate}
		\item $\tau_{(t,\mathbf{x})}(A) \in \alg{A}_{a-loc}$,
		\item $\tau_f(A) \in \alg{A}_{a-loc}$.
	\end{enumerate}
\end{theorem}
To understand why this is true it is helpful to understand why $\tau_f(A)$ is generally not local, even when $A$ is strictly local. We will assume $A \in \alg{A}_{loc}$. Recall the definition of $\tau_f(A)$:
\[
	\tau_f(A) = (2 \pi)^{-\frac{d+1}{2}} \int_{\mathbb{R}} \sum_{\mathbf{x} \in \mathbb{Z}^d} f(t, \mathbf{x})\tau_{(t,\mathbf{x})}(A) dt.
\]
Hence the non-locality can be seen to be due to two parts. First of all, we take sums of translates of $A$. Since $f$ is of Schwarz class, it generally does not have finite support in the position variables. Nevertheless, $f(t, \mathbf{x})$ will decay faster than any polynomial in $|\mathbf{x}|$, so the non-locality is relatively mild. However, the time evolution $\tau_t(A)$ also increases the support. This is where Lieb-Robinson bounds come into play, since we have seen before that they can be used to show that $\tau_t(A)$ can be well approximated by a strictly local observable. More precisely, define a mapping $\mathbb{R}_+ \ni r \mapsto A_{(r)}$ via
\[
	A_{(r)} = (2 \pi)^{-\frac{d+1}{2}} \int_{|t| \leq r} \sum_{|\mathbf{x}| \leq r} f(t,\mathbf{x}) \tau_{(t,\mathbf{x})}(A) dt.
\]
This operator in general still is non-local, because of the time evolution, but this non-locality can we controlled by using the Lieb-Robinson bound. In particular, using the arguments after Corollary~\ref{cor:lrsharper}, we know that we can approximate it well by strictly local operators, and it is in fact possible to estimate the size of the support of this local observable in terms of $v_{LR} t$. A careful analysis along these lines, together with approximating elements in $\alg{A}_{a-loc}$ by strictly local operators yields the result. 

As an aside, it is interesting to note that this theorem is harder the proof for quantum spin systems than it is for relativistic theories. The reason is precisely the non-locality of the time evolution. While in relativistic theories the time evolution is strictly local, this is no longer true in spin systems.

\subsubsection{Haag-Ruelle creation operators}
The above discussion provides a method to obtain states in the single particle subspace. Moreover, they can be obtained using almost local operators. The next step is create states that correspond to \emph{free} particle, since at large times we want states that essentially behave like a collection of non-interacting particles. To this end, consider a function $g : \mathbb{Z}^d \to \mathbb{R}$ such that $\widehat{g} \in C^\infty(\mathbb{T}^d)$ with the support of $\widehat{g}$ in $\Delta_{\mathfrak{h}}$. We then define a \emph{positive energy} \idx{wave packet} via
\[
	g_t(\mathbf{x}) = (2 \pi)^{-\frac{d}{2}} \int_{\Delta_\mathfrak{h}} \widehat{g}(\mathbf{p}) e^{-i t E(\mathbf{p}) + i \mathbf{p} \cdot \mathbf{x}} d\mathbf{p}.
\]
Note that this is somewhat like a plane wave, where the energy is replaced by the dispersion relation. It describes the \emph{free} evolution of single particle (or in this case, rather the \emph{backward} evolution) governed by the dispersion relation. If $g$ is a wave packet, its \idx{velocity support} is defined as
\[
	V(g) := \{ \nabla E(\mathbf{p}) : \mathbf{p} \in \Delta_\mathfrak{h} \}.
\]
As the name suggests, it tells us which velocities occur in the wave packet.

We now have all the definitions to define Haag-Ruelle creation operators, which will create single-particle states.
\begin{definition}
	Let $A^* \in \alg{A}_{a-loc}$ have compact energy-momentum transfer $\Delta$ contained in $(0,\infty) \times \mathbb{T}^d$ such that $\Delta \cap \operatorname{Sp}(U) \subset \mathfrak{h}$, and let $g$ be a positive energy wave packet. The \idx{Haag-Ruelle creation operator} is then defined as
	\[
		B_t^*(g_t) := (2 \pi)^{-\frac{d}{2}} \sum_{\mathbf{x} \in \mathbb{Z}^d} U(t,\mathbf{x}) \pi(A^*) U(t, \mathbf{x})^* g_t(x).
	\]
	That is, we smear $\pi_0(\tau_t(A^*))$ over the \emph{spatial} directions using the wave packet.
\end{definition}
The use of $B^*_t$ instead of its adjoint is to match the notation with the creation operators on Fock space. The operator $A^*$ can be obtained by smearing local observables with suitable test functions, as we have seen before.

The idea behind $B_t^*$ is that it compares the \emph{full} evolution $\tau_t$ with the backward \emph{free} evolution described by the wave packet. Also, by construction, it creates states in the single particle subspace, when acting on the ground state.
\begin{lemma}
	Let $B_t(g_t)^*$ be a Haag-Ruelle creation operator. Then we have $B_t^*(g_t) \Omega \in \mathcal{H}_\mathfrak{h}$ and $(B_t^*(g_t))^* \Omega = 0$, and
	\[
		B_t^*(g_t)\Omega = B^*(g) \Omega = P(\mathfrak{h}) B^*(g) \Omega
	\]
	for all $t \in \mathbb{R}$, where $B^*(g) = (2 \pi)^{-\frac{d}{2}} \sum_{\mathbf{x}} g(\mathbf{x}) \pi(\tau_{\mathbf{x}}(A^*))$. 
\end{lemma}
\begin{proof}
	The first equation can be seen form the energy-momentum transfer lemma, together with a slight extension of Lemma~\ref{lem:smearing}. The second can be seen by noting that if the energy-momentum transfer of an operator $A$ is $\Delta$, then $-\Delta$ is the energy-momentum transfer of $A^*$. Note that we use the group structure of $\mathbb{R} \times \mathbb{T}^d$ here.

To see the last equation, compute
\[
	\begin{split}
		B_t^*(g_t) \Omega &= (2 \pi)^{-\frac{d}{2}} \sum_{\mathbf{x} \in \mathbb{Z}^d} g_t(\mathbf{x}) U(t,\mathbf{x}) \pi(A^*) \Omega  \\
		&= (2 \pi)^{-\frac{d}{2}} \sum_{\mathbf{x}} g_t(\mathbf{x})  \int_{\mathfrak{h}}  e^{i E t - i \mathbf{p} \cdot \mathbf{x}} dP(E, \mathbf{p}) P(\mathfrak{h}) \pi(A^*) \Omega \\ 
		&= \int_{\mathfrak{h}}  \widehat{g}(\mathbf{p}) e^{i (E - E(\mathbf{p})) t} dP(E, \mathbf{p}) P(\mathfrak{h}) \pi(A^*) \Omega.
	\end{split}
\]
In the last step we used the inverse Fourier transform to obtain $\widehat{g}(\mathbf{p})$. Since we are only integrating over the mass shell, $E = E(\mathbf{p})$ and the exponential disappears. 

Similarly, we calculate
\[
	\begin{split}
		B^*(g) \Omega &= (2 \pi)^{-\frac{d}{2}} \sum_{\mathbf{x}} g(\mathbf{x}) U(\mathbf{x}) P(\mathfrak{h}) \pi(A^*) \Omega \\
		&= \int_{\mathfrak{h}} \widehat{g}(\mathbf{p}) dP(E, \mathbf{p}) P(\mathfrak{h}) \pi(A^*) \Omega.
	\end{split}
\]
This concludes the proof.
\end{proof}

Note that this lemma tells us that on the ground state, the ``free'' evolution coincides with the full (interacting) evolution, even for arbitrary $t$ (not only in the asymptotic limit). This is not surprising, since $B_t^*(g_t) \Omega$ contains a single particle excitation, hence there is nothing to interact with and the particle behaves like a free particle. This in general is no longer true for states like $B_{1,t}^*(g_{1,t}) B_{2,t}^*(g_{2,t}) \Omega$ that contain more excitations.

We end this construction with a remark on the Lieb-Robinson velocity. As we have stressed several times, it plays the role of the speed of light in quantum spin systems. Hence a reasonable conjecture would be that the velocity of single-particle excitations is bounded by the Lieb-Robinson velocity. This indeed turns out to be the case.

\begin{proposition}
	Let $v_{LR}$ be the Lieb-Robinson velocity. Then $| \nabla E(\mathbf{p}) | < v_{LR}$ for any $\mathbf{p} \in \Delta_\mathfrak{h}$.
\end{proposition}
The proof is by contradiction. One first assumes that such states exist. A dense set of such states can be obtained as $B^*(g) \Omega$, where $\widehat{g}$ has support in an open subset of $\Delta_{\mathfrak{h}}$ on which $|\nabla E(\mathbf{p})| < v_{LR}$ is violated. Using Lieb-Robinson bounds one can then show that this vector must orthogonal to $\pi(A) \Omega$ for any local $A$. Hence, by cyclicity of $\Omega$, it must be zero. The argument that leads to this conclusion takes some work: the details can be found in~\cite[Prop. 5.3]{BDNScattering}. An alternative proof can be found in~\cite{Sch83}.

\subsection{Scattering states}
The Haag-Ruelle creation operators can be used to create single-particle states from the ground state. To create multiple particles, the idea is to apply the creation operators repeatedly. That is, choose positive energy wave packets $g_i$, $i=1, \dots, n$ and Haag-Ruelle creation operators $B_{i,t}^*(g_i)$. Then we can look at states
\[
	\Psi_t := B_{n,t}^*(g_{n,t}) \cdots B_{1,t}^*(g_{1,t}) \Omega.
\]
In particular, we are interested in the asymptotic (incoming or outgoing) states, where we take the limit $t \to \pm \infty$. An important part of scattering theory is to show that these limits exist and have the right properties.

It is good to first think about which properties should be expected. As discussed before, the incoming and outgoing particles should behave (asymptotically) as free particles. But free (bosonic) particles are well described by the symmetric Fock space over the single-particle Hilbert space $\mathcal{H}_\mathfrak{h}$. This in particular means that the state should be invariant under changing the order of the creation operators, at least in the limit $t \mapsto \pm \infty$. They should also only depend on the \emph{single}-particle states: if $B_{i,t}^*(g_{i,t}) \Omega = \widetilde{B}_{i,t}^*(\widetilde{g}_{i,t}) \Omega$, we should be able to replace $B_{i,t}^*$ with $\widetilde{B}_{i,t}$ in the definition of $\Psi_t$, at least in the limit of large $|t|$. This indeed turns out to be true. Anticipating this, we write
\begin{equation}
	\label{eq:outgoing}
	\Psi_1 \timeso \cdots \timeso \Psi_n := \Psi^{out} := \lim_{t \to \infty} B_{n,t}^*(g_{n,t}) \cdots B_{1,t}^*(g_1) \Omega.
\end{equation}
Here $\Psi_i = B_{i,t}^*(g_{i,t}) \Omega$. Similarly, the \emph{incoming} states are defined by taking the limit to negative infinity. These states indeed have all the properties we expect them to have:

\begin{theorem}\label{thm:fock}
	Suppose that Assumptions~\ref{as:qspin} and~\ref{as:massshell} hold, and moreover assume that the wave packets $g_i$ have disjoint velocity support. Then the limit in equation~\eqref{eq:outgoing} exists and is independent of the choice of $B_{i,t}^*(g_{i,t})$ as long as $B_{i,t}^*({g}_{i,t}) \Omega = \widetilde{B}_{i,t}^*(\widetilde{g}_{i,t}) \Omega$. Moreover, this limit is independent of the order of the Haag-Ruelle creation operators. Finally, let $\Psi^{out}$ and $\widetilde{\Psi}^{out}$ be two such states with $n$ and $\widetilde{n}$ particles respectively. Then
	\begin{equation}
	\langle \Psi^{out}, \Psi^{out} \rangle = \delta_{n, \widetilde{n}} \sum_{\sigma \in S_n} \langle \Psi_1, \widetilde{\Psi}_{\sigma(1)} \rangle \cdots \langle \Psi_n, \widetilde{\Psi}_{\sigma(n)} \rangle,
		\label{eq:fockinner}
	\end{equation}
where $S_n$ is again the set of all permutations of $n$ elements.

All statements equally hold for the incoming particle states.
\end{theorem}
The proof of this theorem is quite technical, and requires various bounds on how the Haag-Ruelle creation operators depend on $t$, and bounds on the commutators of such creation operators. The main idea however is easy to understand. The creation operators (and hence, in a sense, the particles) are reasonably well localised. Since the velocity supports are disjoint, for large $|t|$, the particles must be far away from each other. By the locality of creation operators, this implies that commutator bounds of such operators will become small, since they particles essentially behave as free particles.

The theorem shows that the asymptotic states have a structure resembling that of (symmetric) Fock space. For the scattering theory it will be useful to have an operator mapping the single particle (symmetric) Fock space $\mc{F}_s(\mc{H}_\mathfrak{h})$ to these scattering states. Before we give the definition of such an operator, we note that $U(t,\mathbf{x})$ can be restricted to a unitary $U(t,\mathbf{x})_\mathfrak{h}$ acting on $\mathcal{H}_{\mathfrak{h}}$, and hence this defines (via second quantisation) an action of space-time translations on the Fock space, by $(t,\mathbf{x}) \mapsto \Gamma(U(t,\mathbf{x})_\mathfrak{h})$. The mapping alluded to should be compatible with this mapping, in the following sense:

\begin{definition}
	An outgoing \idx{wave operator} is an isometry $W^{out} : \mc{F}_s(\mathcal{H}_{\mathfrak{h}}) \to \mathcal{H}$ such that
\[
	\begin{split}
	W^{out} \Omega &= \Omega, \\
	W^{out} (a_+^*(\Psi_1)a_+^*(\Psi_2) &\cdots a^*_+(\Psi_n)) \Omega = \Psi_1 \timeso \cdots \timeso \Psi_n,\\
	U(t,\mathbf{x}) \circ W^{out} &= W^{out} \circ \Gamma(U(t,\mathbf{x})_\mathfrak{h}),
	\end{split}
\]
for all states $\Psi_i \in \mathcal{H}_{\mathfrak{h}}$. The incoming wave operator $W^{in}$ is defined analogously.
\end{definition}
Hence the wave operators map configurations of incoming or outgoing free particle states, to states in the Hilbert space $\mathcal{H}$, in a way that is compatible with space-time translations. We can then define the \emph{$S$-matrix}\index{S-matrix@$S$-matrix} as the isometry
\[
	S := (W^{out})^* W^{in}.
\]
Note that this definition coincides with what we said earlier about the $S$-matrix: it maps incoming particle configurations to outgoing particle configurations.

The main result is that the wave operators exist and are unique.
\begin{theorem}
	Consider a quantum spin system satisfying Assumptions~\ref{as:qspin} and~\ref{as:massshell}. Then the wave operators and the $S$-matrix exist and are unique.
\end{theorem}
In other words, the scattering theory exists and is well-defined. The definition of a wave operator suggests how $W^{out}$ should be defined on a basis of $n$-particle states. Using Theorem~\ref{thm:fock}, it follows that this is well-defined on particle states with disjoint velocity support. The hard part is to show that such states form a dense subspace of the Fock space, so that $W^{out}$ can be extended by linearity, and that $W^{out}$ is an isometry. Indeed, a priori it is not obvious that $\Psi_1 \timeso \cdots \timeso \Psi_n$ is even non-zero in general. The proof of this is somewhat technical and will be omitted here. This is the part where it is necessary to make the regularity assumption from Assumptions~\ref{as:massshell}, to show density.

Although the above theorem tells us that, under the assumptions made, one can do scattering theory, this is still a long way from calculating the matrix $S$ explicitly in concrete examples. To this end, one first has to prove that the assumptions are satisfied. Unfortunately, in general it is very difficult to find the spectrum of a Hamiltonian, and this is even more so for the full energy-momentum spectrum. Indeed, even the question if it is gapped can be very difficult. The most promising approach appears to be considering perturbations of ``classical'' models, for example where the dynamics are generated by commuting terms in the Hamiltonian (much like the toric code). The unperturbed models will have a zero velocity Lieb-Robinson bound, and hence no scattering, but this is no longer true once we add perturbations. Using perturbation theory, it is sometimes possible to prove the existence of an isolated mass shell. This is true for the following example. The proof uses results from~\cite{Pokorny:1993wt,Yarotsky:2005dk}, which study the spectrum of the perturbed Hamiltonian.

\begin{theorem}
There is $\varepsilon_0 > 0$ such that for all $0 < \varepsilon < \varepsilon_0$, the transverse Ising model with the following local Hamiltonian
\[
	H_{\Lambda,\varepsilon} = -\frac{1}{2} \sum_{j \in \Lambda} (\sigma^z_j -1) - \varepsilon \sum_{(i,j)} \sigma^{x}_i \sigma^{x}_j,
\]
where $(i,j)$ denotes a pair of nearest-neighbours, satisfies Assumption~\ref{as:qspin} as well as Assumption~\ref{as:massshell}.
\end{theorem}

Actually finding the matrix $S$ explicitly seems to be out of reach with current techniques. However, for many interesting quantum spin models there do exist numerical simulations. Also in algebraic quantum field theory, in some cases where the scattering theory is relatively simple, explicit models can be constructed~\cite{Lechner2008}.

\section{Gapped phases and local perturbations}
The final topic we discuss is that of gapped phases. In recent there has been much interest in topologically ordered phases of ground states, of which the toric code model is an example. An important question is how to classify such phases. For any meaningful classification, it is necessary to say when we regard two phases to be the same. Coloquially, two gapped ground states $\omega_0$ and $\omega_1$ are said to be in the same phase if there is a continuous path $H(s)$ of gapped, local Hamiltonians such that $\omega_0$ is a ground state of $H(0)$ and $\omega_1$ is a ground state of $H(1)$~\cite{PhysRevB.82.155138}. This can be made precise in the thermodynamic limit, as is done in Ref.~\cite{2011arXiv1102.0842B}, whose exposition we will largely follow here.

In particular, it can be shown that weak-$*$ limits of (mixtures of) ground states states are related by an automorphism $\alpha_s$ to such states of the \emph{unperturbed} dynamics. This automorphism $\alpha_s$ is very much like a time evolution, in the sense that it also satisfies a Lieb-Robinson bound. The construction of this family of automorphisms uses the spectral flow (also called quasi-adiabatic continuation) technique. As a by-product, it is possible to show that if we add a strictly \emph{local perturbation}\index{local!perturbation} (not closing the gap) to some local Hamiltonians, the new ground state is related to the old, unperturbed, ground state by a \emph{local} unitary, up to a small error. This is outlined in the next subsection. In the final subsection, the spectral flow is applied to obtain the $\alpha_s$.

\subsection{Quasi-adiabatic continuation or the spectral flow}
We first want to study how the spectrum changes as $s$ runs in the interval $[0,1]$. This can be done using \idx{quasi-adiabatic continuation}. This technique was discussed by Hastings~\cite{PhysRevB.69.104431,PhysRevB.72.045141}, and was later studied in the thermodynamic limit in~\cite{2011arXiv1102.0842B}. Here we follow~\cite{2011arXiv1102.0842B} and call it the \idx{spectral flow}, since it deals with the spectrum.

The precise setting is as follows. Suppose that we have defined some dynamics $\alpha_t$. Consider a ground state for these dynamics. Then there is some Hamiltonian $H(0)$ implementing the dynamics $\alpha_t$ in the GNS representation. We now assume that the perturbed dynamics are of the form $H(s) = H(0) + \Phi(s)$, where $s \in [0,1]$ and $\Phi(s) = \Phi(s)^*$. The dynamics should not change too much as a function of $s$, in the sense that there is some $M > 0$ such that $\|H'(s)\| \leq M$ for all $s \in [0,1]$. We will also assume that there is a gap $\gamma > 0$ in the spectrum for all $s$, where $\gamma$ does not depend on $s$. In particular, we assume that
\[
	\operatorname{spec} H(s) = \Sigma_1(s) + \Sigma_2(s)
\]
for some sets $\Sigma_1(s)$ and $\Sigma_2(s)$ such that $\inf_{s \in [0,1]} d(\Sigma_1(s), \Sigma_2(s)) \geq \gamma$. There is an additional (technical) smoothness condition that we need to demand. For its precise formulation, see~\cite[Assumption 2.1]{2011arXiv1102.0842B}.

A key step in the result is to connect the ground state spaces of $H(s)$ to each other. It turns out that this can be done by unitary maps $U(s)$, and more importantly, it is possible to give a generator for the ``flow'' $s \mapsto U(s)$. This is the spectral flow or quasi-adiabatic continuation.
\begin{lemma}\label{lem:flow}
	There is a norm-continuous path of unitaries $U(s)$ such that for all $s \in [0,1]$,
	\begin{equation}
		\label{eq:gsproject}
		P(s) = U(s) P(0) U(s)^*,
	\end{equation}
where $P(s)$ is the projection on the eigenspace of the eigenvalues in $\Sigma_1(s)$. The $U(s)$ satisfy the differential equation
\[
	\frac{d}{ds} U(s) = i D(s) U(s),
\]
with boundary condition $U(0) = 1$. The generator $D(s)$ can be obtained as
\[
	D(s) = \int_{-\infty}^\infty W_\gamma(t) e^{i t H(s)} H'(s) e^{-i t H(s)} dt,
\]
where $W_\gamma(t)$ is a $L^1$-function which can be defined explicitly.
\end{lemma}
Note that $D(s)$ is in fact obtained by smearing $H'(s)$ with the function $W_\gamma$ using the dynamics generated by $H(s)$. Although $W_\gamma(t)$ can be given explicitly, its definition is rather complicated and its explicit form does not provide much insight for our purposes. However, this explicit form allows the authors of~\cite{2011arXiv1102.0842B} to give explicit bounds on various estimates, which prove to be crucial in the proof.

This lemma can be applied to show that ``local perturbations perturb locally'', in the sense that if we have some Hamiltonian $H(0) + \Phi(s)$, with $\Phi'(s)$ a local operator for every $s$, then the ground states of $H(s)$ are related to those of $H(0)$ by a strictly local unitary (at least, up to arbitrarily small error). More precisely, we assume that there is some finite subset $\Lambda \subset \Gamma$ and a constant $C > 0$ such that $\Phi'(s) \in \alg{A}(\Lambda)$ and $\|\Phi'(s) \| < C$ for all $0 \leq s \leq 1$. Note that this means that $\Phi(s) \in \alg{A}(\Lambda)$ up to a constant operator. Finally, we assume that there is a Lieb-Robinson bound for the perturbed dynamics.\index{Lieb-Robinson!bound} In particular, there should exist constants $\mu > 0, C>0$ and $v > 0$ such that for local operators $A$ and $B$ we have
\[
	\begin{split}
\Big\| [\tau_t^{H(s)}&(A), B] \Big\| \leq \\ &C \|A\| \|B\| \min(|\supp(A)|, |\supp(B)|) e^{-\mu(d(\supp(A), \supp(B))- v|t|)}
\end{split}
\]
for all $s \in [0,1]$. Here $\tau_t^{H(s)}(A) = e^{i t H(s)} A e^{-i t H(s)}$, the dynamics generated by $H(s)$. 

Under these assumptions one can show that the unitaries $U(s)$ can be approximated by local unitaries. The strategy is to use the Lieb-Robinson bounds. Note that in Lemma~\ref{lem:flow} the generator $D(s)$ can be obtained by smearing the \emph{local} operator $H'(s)$ (which is $\Phi'(s)$ here) using an automorphism that satisfies a Lieb-Robinson bound. We can then employ a similar technique as outlined in the text following Theorem~\ref{thm:localsmear} to approximate $D(s)$ by a strictly local operator. Since explicit bounds on $W_\gamma(t)$ can be obtained, it is in fact possible to give a precise estimate. The result is that the unitaries $U(s)$ can be approximated by strictly local unitaries $V(s)$. These local unitaries are supported on sets containing $\Lambda$, and as expected, the bigger this set is, the better the approximation will be. In particular, define the following sets for $R > 0$:
\[
	\Lambda_R = \{ y \in \Gamma : \exists x \in \Lambda \textrm{ such that } d(x,y) < R \}.
\]
Once we approximate $D(s)$ by a local operator $D_R(s)$ supported on $\Lambda_R(s)$, one can look at the corresponding unitaries $V_R(s)$. By integrating the differential equation in~\ref{lem:flow}, the estimated for $\| D(s) - D_R(s) \|$ yields an estimate for $\| U(s) - V_R(s) \|$. With this notation, the argument above leads to the following result~\cite[Thm. 3.4]{2011arXiv1102.0842B}:
\begin{theorem}
There exists a subexponential function $G(R)$ and a constant $C$ such that for any $R > 0$ and for all $s \in [0,1]$, there exists a unitary $V_R(s)$ with support in $\Lambda_R$ such that
\[
	\| U(s) - V_R(s) \| \leq C G\left(\frac{\gamma R}{2 v}\right),
\]
where $\gamma$ and $v$ are as above.
\end{theorem}
One can show that for large enough $R$ the function $G$ is proportional to $\exp\left(-2/7 \frac{R}{\log^2(R)}\right)$. This is where the explicit form for $W_\gamma(t)$ comes in.

As an example we suppose that for each $s$ we have a non-degenerate ground state, and hence it can be represented as a vector $\Omega(s)$ in Hilbert space. The ground state is separated from the rest of the spectrum by a gap of at least $\gamma$ for all $s$. Under this conditions we can show that far away from the local perturbation $\Phi(s)$ the ground states look alike. In particular, consider $R > 0$ and $A \in \alg{A}(\Lambda_R^c)$. Then $[A,V_R(s)] = 0$ for all $s$ by locality. Moreover, note that each $P(s)$ is a one-dimensional projection because of the non-degeneracy of the ground state. This implies that $\Omega(s) = U(s) \Omega(0)$, with the help of equation~\eqref{eq:gsproject}. Hence we have
\[
\begin{split}
	|\langle \Omega(s), &\,A \Omega(s) \rangle -  \langle \Omega(0), A \Omega(0) \rangle| = | \langle \Omega(0), U(s)^* [A, U(s)] \Omega(0)\rangle|  \\
		& = |\langle \Omega(0), U(s)^* [A, V_R(s)] \Omega(0) \rangle + \langle \Omega(0), U(s)^* [A, U(s) - V_R(s)]  \Omega(0) \rangle | \\
		& \leq 2 \|A\| \| U(s) - V_R(s) \| \leq 2 \|A\| C G\left(\frac{\gamma R}{2 v} \right).
\end{split}
\]
Since the right hand side goes to zero as $R$ grows, we see that the states look the same far away from the perturbation.

\subsection{Automorphic equivalence of ground states}
So far we have discussed perturbing the system with respect to a \emph{single} local perturbation. Although useful, this is not general enough for many applications. Rather, it would be useful to allow for sums of perturbations, just like the dynamics were generated by local interactions consisting of sums of local operators. More precisely, we assume that the local dynamics are generated by
\[
	H_\Lambda(s) = H_\Lambda(0) + \sum_{X \subset \Lambda} \Phi(X, s),
\]
where $\Phi(X,s) \in \alg{A}(X)$ is self-adjoint and $H_\Lambda(0)$ in turn is given by local interactions that give a Lieb-Robinson bound (for example, the uniformly bounded finite range interactions discussed before).

We are now in a position to address the question of equivalences of gapped phases. In the thermodynamic limit\index{thermodynamic limit} this situation can be described as follows. First consider for a finite $\Lambda \subset \Gamma$ the set $\mc{S}_\Lambda(s)$ of all (mixtures of) states with energy in $I(s)$, where $I(s)$ is some interval containing $\Sigma_1(s)$ (and disjoint from $\Sigma_2(s)$). We can then consider an increasing sequence $\Lambda_n$ of subsets, and look at all the weak-$*$ limit points as $n \to \infty$ of states $\mc{S}_{\Lambda_n}(s)$.\index{_S(s)@$\mathcal{S}(s)$} Note that this is the same procedure as we used before in the construction of KMS states from finite volume Gibbs states in Section~\ref{sec:kmsstates}. 

We have seen that vector states in $\mc{S}_{\Lambda_n}(s)$ are related to those in $\mc{S}_{\Lambda_n}(0)$ by a unitary operator $U_{\Lambda_n}(s)$. The question is what can be said about the weak-$*$ limits. Since for the \emph{local} sets of states there are automorphisms $\alpha^{\Lambda_n}_s$ (obtained by conjugating with $U_{\Lambda_n}(s)$) such that $\mc{S}_{\Lambda_n}(s) = \mc{S}_{\Lambda_n}(0) \circ \alpha_s^{\Lambda_n}$, we can ask the question if $\alpha^{\Lambda_n}_s$ converges to an automorphism $\alpha_s$.\footnote{In general this automorphism would not be given by conjugating with a unitary any more, however. We have already seen an example of this with the automorphisms describing charges in the toric code model.} Note that this is similar to how we obtained the global dynamics from local dynamics in Theorem~\ref{thm:dynconv}, or using Lieb-Robinson bounds as in Theorem~\ref{thm:dynconvlr}.

Indeed, the automorphisms $\alpha_s$ can be obtained in a similar way by using Lieb-Robinson bounds. For this it is necessary to assume that $\tau_t^{H_{\Lambda}(s)} := e^{i t H_\Lambda(s)} \cdot e^{-i t H_\Lambda(s)}$ satisfies a Lieb-Robinson bound (in $t$), where the Lieb-Robinson velocity (and other constants in the bound) are uniform in $s$ and $\Lambda$. This can then be used to derive a Lieb-Robinson bound for $\alpha_s^{\Lambda_n}$. Note that the situation is more complicated here than in the local dynamics case discussed before, since essentially we are dealing with a time-\emph{dependent} interaction. The proof essentially boils down to the observation that we can interpret the generator $D_\Lambda(s)$ of the spectral flow as being given by local (time-dependent) dynamics. Once this Lieb-Robinson bound has been obtained, the following theorem can be obtained using standard methods: 

\begin{theorem}
There is an automorphism $\alpha_s$ such that $\mc{S}(s) = \mc{S}(0) \circ \alpha_s$. The automorphism $\alpha_s$ can be obtained from an $s$-dependent quasi-local interaction, and satisfies a Lieb-Robinson bound.
\end{theorem}

In typical applications $\Sigma_1(s)$ will correspond to the lowest-lying energy levels, i.e.\ ground states with possibly allowing some small splitting. It should be stressed, however, that the states in $\mathcal{S}_n$ are (mixtures of) eigenvalues of $H_{\Lambda_n}$. That is, no boundary terms for the local Hamiltonians are allowed. In general not all ground states of some dynamics $\alpha_t$ can be obtained in this way, see for example Remark~\ref{rem:surface}. This is also true in the toric code: except for the translation invariant ground states, all other ground states are obtained as weak-$*$ limits of ground states of the local dynamics \emph{with boundary term}~\cite{toricgs}.

Even for the unique frustration free ground state of the toric code, the theorem above is not the end of the story. Although it does relate the ground states of the perturbed model to the unperturbed one with an automorphism of $\alg{A}$, it is not at all clear if the same follows for all the other interesting physical properties. For example, it would be interesting to know if the whole superselection structure is preserved. This is what one would expect for a meaningful definition of equivalence of phases. Answering this question is presently an active area of research, and not many results in this direction are known.

\bibliographystyle{abbrv}
\bibliography{../refs/refs}

\def\cprime{$'$}
\begin{thebibliography}{100}

\bibitem{affleck1988}
I.~Affleck, T.~Kennedy, E.~H. Lieb, and H.~Tasaki.
\newblock Valence bond ground states in isotropic quantum antiferromagnets.
\newblock {\em Commun. Math. Phys.}, 115(3):477--528, 1988.

\bibitem{MR2345476}
R.~Alicki, M.~Fannes, and M.~Horodecki.
\newblock A statistical mechanics view on {K}itaev's proposal for quantum
  memories.
\newblock {\em J. Phys. A}, 40(24):6451--6467, 2007.

\bibitem{MR2525435}
R.~Alicki, M.~Fannes, and M.~Horodecki.
\newblock On thermalization in {K}itaev's 2{D} model.
\newblock {\em J. Phys. A}, 42(6):065303, 18, 2009.

\bibitem{MR0011172}
W.~Ambrose.
\newblock Spectral resolution of groups of unitary operators.
\newblock {\em Duke Math. J.}, 11:589--595, 1944.

\bibitem{MR2542202}
H.~Araki.
\newblock {\em Mathematical theory of quantum fields}, volume 101 of {\em
  International Series of Monographs on Physics}.
\newblock Oxford University Press, Oxford, 2009.
\newblock Translated from the 1993 Japanese original by Ursula Carow-Watamura.

\bibitem{MR0244773}
H.~Araki and E.~J. Woods.
\newblock A classification of factors.
\newblock {\em Publ. Res. Inst. Math. Sci. Ser. A}, 4:51--130, 1968/1969.

\bibitem{BDNScattering}
S.~{Bachmann}, W.~Dybalski, and P.~Naaijkens.
\newblock Lieb-{R}obinson bounds, {A}rveson spectrum and {H}aag-{R}uelle
  scattering theory for gapped quantum spin systems.
\newblock {\em Ann. Henri Poincar{\'e}}, 17:1737--1791, 2016.

\bibitem{2011arXiv1102.0842B}
S.~{Bachmann}, S.~{Michalakis}, B.~{Nachtergaele}, and R.~{Sims}.
\newblock {Automorphic Equivalence within Gapped Phases of Quantum Lattice
  Systems}.
\newblock {\em Comm. Math. Phys.}, 309(3):835--871, 2012.

\bibitem{MR1199168}
H.-J. Borchers and J.~Yngvason.
\newblock From quantum fields to local von {N}eumann algebras.
\newblock {\em Rev. Math. Phys.}, 4(Special Issue):15--47, 1992.
\newblock Special issue dedicated to R. Haag on the occasion of his 70th
  birthday.

\bibitem{Brandao2014}
F.~G. S.~L. Brand{\~a}o and M.~Horodecki.
\newblock Exponential decay of correlations implies area law.
\newblock {\em Communications in Mathematical Physics}, 333(2):761--798, 2014.

\bibitem{MR887100}
O.~Bratteli and D.~W. Robinson.
\newblock {\em Operator algebras and quantum statistical mechanics. 1}.
\newblock Texts and Monographs in Physics. Springer-Verlag, New York, second
  edition, 1987.

\bibitem{MR1441540}
O.~Bratteli and D.~W. Robinson.
\newblock {\em Operator algebras and quantum statistical mechanics. 2}.
\newblock Texts and Monographs in Physics. Springer-Verlag, Berlin, second
  edition, 1997.

\bibitem{MR3381848}
R.~Brunetti, C.~Dappiaggi, K.~Fredenhagen, and J.~Yngvason, editors.
\newblock {\em Advances in algebraic quantum field theory}.
\newblock Mathematical Physics Studies. Springer, Cham, 2015.

\bibitem{Bu90}
D.~Buchholz.
\newblock Harmonic analysis of local operators.
\newblock {\em Commun. Math. Phys.}, 129:631--649, 1990.

\bibitem{MR660538}
D.~Buchholz and K.~Fredenhagen.
\newblock Locality and the structure of particle states.
\newblock {\em Comm. Math. Phys.}, 84(1):1--54, 1982.

\bibitem{MR1768634}
D.~Buchholz and R.~Haag.
\newblock The quest for understanding in relativistic quantum physics.
\newblock {\em J. Math. Phys.}, 41(6):3674--3697, 2000.

\bibitem{BuchholzRoberts}
D.~Buchholz and J.~E. Roberts.
\newblock New light on infrared problems: Sectors, statistics, symmetries and
  spectrum.
\newblock {\em Commun. Math. Phys.}, 330(3):935--972, 2014.

\bibitem{toricgs}
M.~Cha, P.~Naaijkens, and B.~Nachtergaele.
\newblock The complete set of infinite volume ground states for {K}itaev's
  abelian quantum double models.
\newblock Preprint, arXiv:1608.04449, 2016.

\bibitem{PhysRevB.82.155138}
X.~Chen, Z.-C. Gu, and X.-G. Wen.
\newblock Local unitary transformation, long-range quantum entanglement, wave
  function renormalization, and topological order.
\newblock {\em Phys. Rev. B}, 82:155138, Oct 2010.

\bibitem{terhaltoric}
S.~Chesi, D.~Loss, S.~Bravyi, and B.~M. Terhal.
\newblock Thermodynamic stability criteria for a quantum memory based on
  stabilizer and subsystem codes.
\newblock {\em New Journal of Physics}, 12(2):025013, 2010.

\bibitem{MR1924451}
E.~Dennis, A.~Kitaev, and J.~Preskill.
\newblock Topological quantum memory.
\newblock {\em J. Math. Phys.}, 43(9):4452--4505, 2002.
\newblock Quantum information theory.

\bibitem{MR0297259}
S.~Doplicher, R.~Haag, and J.~E. Roberts.
\newblock Local observables and particle statistics. {I}.
\newblock {\em Comm. Math. Phys.}, 23:199--230, 1971.

\bibitem{MR0334742}
S.~Doplicher, R.~Haag, and J.~E. Roberts.
\newblock Local observables and particle statistics. {II}.
\newblock {\em Comm. Math. Phys.}, 35:49--85, 1974.

\bibitem{MR1062748}
S.~Doplicher and J.~E. Roberts.
\newblock Why there is a field algebra with a compact gauge group describing
  the superselection structure in particle physics.
\newblock {\em Comm. Math. Phys.}, 131(1):51--107, 1990.

\bibitem{Dyson1978}
F.~J. Dyson, E.~H. Lieb, and B.~Simon.
\newblock Phase transitions in quantum spin systems with isotropic and
  nonisotropic interactions.
\newblock {\em Journal of Statistical Physics}, 18(4):335--383, 1978.

\bibitem{MR2453623}
J.~Earman.
\newblock Superselection rules for philosophers.
\newblock {\em Erkenntnis}, 69(3):377--414, 2008.

\bibitem{10.2307/1968693}
J.~v.~N. F.~J.~Murray.
\newblock On rings of operators.
\newblock {\em Annals of Mathematics}, 37(1):116--229, 1936.

\bibitem{MR1158756}
M.~Fannes, B.~Nachtergaele, and R.~F. Werner.
\newblock Finitely correlated states on quantum spin chains.
\newblock {\em Comm. Math. Phys.}, 144(3):443--490, 1992.

\bibitem{haagdouble}
L.~Fiedler and P.~Naaijkens.
\newblock {H}aag duality for {K}itaev's quantum double model for abelian
  groups.
\newblock {\em Rev. Math. Phys.}, 27:1550021:1--43, 2015.

\bibitem{fredclust}
K.~Fredenhagen.
\newblock A remark on the cluster theorem.
\newblock {\em Comm. Math. Phys.}, 97:461--463, 1985.

\bibitem{MR1016869}
K.~Fredenhagen, K.-H. Rehren, and B.~Schroer.
\newblock Superselection sectors with braid group statistics and exchange
  algebras. {I}.\ {G}eneral theory.
\newblock {\em Comm. Math. Phys.}, 125(2):201--226, 1989.

\bibitem{MR1199171}
K.~Fredenhagen, K.-H. Rehren, and B.~Schroer.
\newblock Superselection sectors with braid group statistics and exchange
  algebras. {II}. {G}eometric aspects and conformal covariance.
\newblock {\em Rev. Math. Phys.}, 4(Special Issue):113--157, 1992.

\bibitem{MR1943131}
M.~H. Freedman, A.~Kitaev, M.~J. Larsen, and Z.~Wang.
\newblock Topological quantum computation.
\newblock {\em Bull. Amer. Math. Soc. (N.S.)}, 40(1):31--38, 2003.
\newblock Mathematical challenges of the 21st century (Los Angeles, CA, 2000).

\bibitem{MR1104414}
J.~Fr{\"o}hlich and F.~Gabbiani.
\newblock Braid statistics in local quantum theory.
\newblock {\em Rev. Math. Phys.}, 2(3):251--353, 1990.

\bibitem{MR887102}
J.~Glimm and A.~Jaffe.
\newblock {\em Quantum physics}.
\newblock Springer-Verlag, New York, second edition, 1987.
\newblock A functional integral point of view.

\bibitem{MR0014589}
R.~Godement.
\newblock Sur une g\'en\'eralisation d'un th\'eor\`eme de {S}tone.
\newblock {\em C. R. Acad. Sci. Paris}, 218:901--903, 1944.

\bibitem{Ha58}
R.~Haag.
\newblock Quantum field theories with composite particles and asymptotic
  conditions.
\newblock {\em Phys. Rev.}, 112:669--673, 1958.

\bibitem{MR1405610}
R.~Haag.
\newblock {\em Local quantum physics: Fields, particles, algebras}.
\newblock Texts and Monographs in Physics. Springer-Verlag, Berlin, second
  edition, 1996.

\bibitem{springerlink:10.1140/epjh/e2010-10042-7}
R.~Haag.
\newblock Local algebras. a look back at the early years and at some
  achievements and missed opportunities.
\newblock {\em Eur. Phys. J. H}, 35:255--261, 2010.

\bibitem{haagpersonal}
R.~Haag.
\newblock Some people and some problems met in half a century of commitment to
  mathematical physics.
\newblock {\em Eur. Phys. J. H}, 35(3):263--307, 2010.

\bibitem{MR0219283}
R.~Haag, N.~M. Hugenholtz, and M.~Winnink.
\newblock On the equilibrium states in quantum statistical mechanics.
\newblock {\em Comm. Math. Phys.}, 5:215--236, 1967.

\bibitem{MR0165864}
R.~Haag and D.~Kastler.
\newblock An algebraic approach to quantum field theory.
\newblock {\em J. Mathematical Phys.}, 5:848--861, 1964.

\bibitem{Haegeman:2013a}
J.~Haegeman, S.~Michalakis, B.~Nachtergaele, T.~J. Osborne, N.~Schuch, and
  F.~Verstraete.
\newblock Elementary excitations in gapped quantum spin systems.
\newblock {\em Phys. Rev. Lett.}, 111:080401, 2013.

\bibitem{halvapp}
H.~Halvorson.
\newblock Algebraic quantum field theory.
\newblock In J.~Butterfield and J.~Earman, editors, {\em Philosophy of
  Physics}, pages 731--922. Elsevier, 2006.

\bibitem{PhysRevB.69.104431}
M.~B. Hastings.
\newblock {L}ieb-{S}chultz-{M}attis in higher dimensions.
\newblock {\em Phys. Rev. B}, 69:104431, Mar 2004.

\bibitem{HastingsAreaLaw}
M.~B. Hastings.
\newblock An area law for one-dimensional quantum systems.
\newblock {\em Journal of Statistical Mechanics: Theory and Experiment},
  2007(08):P08024, 2007.

\bibitem{PhysRevB.72.045141}
M.~B. Hastings and X.-G. Wen.
\newblock Quasiadiabatic continuation of quantum states: The stability of
  topological ground-state degeneracy and emergent gauge invariance.
\newblock {\em Phys. Rev. B}, 72:045141, Jul 2005.

\bibitem{MR0169073}
R.~V. Kadison.
\newblock Transformations of states in operator theory and dynamics.
\newblock {\em Topology}, 3(suppl. 2):177--198, 1965.

\bibitem{MR719020}
R.~V. Kadison and J.~R. Ringrose.
\newblock {\em Fundamentals of the theory of operator algebras. {V}ol. {I}:
  Elementary theory}, volume 100 of {\em Pure and Applied Mathematics}.
\newblock Academic Press Inc., New York, 1983.

\bibitem{MR1468230}
R.~V. Kadison and J.~R. Ringrose.
\newblock {\em Fundamentals of the theory of operator algebras. {V}ol. {II}:
  Advanced theory}, volume~16 of {\em Graduate Studies in Mathematics}.
\newblock American Mathematical Society, Providence, RI, 1997.

\bibitem{MR0205096}
D.~Kastler and D.~W. Robinson.
\newblock Invariant states in statistical mechanics.
\newblock {\em Comm. Math. Phys.}, 3:151--180, 1966.

\bibitem{MR2281418}
M.~Keyl, T.~Matsui, D.~Schlingemann, and R.~F. Werner.
\newblock Entanglement {H}aag-duality and type properties of infinite quantum
  spin chains.
\newblock {\em Rev. Math. Phys.}, 18(9):935--970, 2006.

\bibitem{MR1951039}
A.~Kitaev.
\newblock Fault-tolerant quantum computation by anyons.
\newblock {\em Ann. Physics}, 303(1):2--30, 2003.

\bibitem{MR0468989}
A.~Kossakowski, A.~Frigerio, V.~Gorini, and M.~Verri.
\newblock Quantum detailed balance and {KMS} condition.
\newblock {\em Comm. Math. Phys.}, 57(2):97--110, 1977.

\bibitem{MR0098482}
R.~Kubo.
\newblock Statistical-mechanical theory of irreversible processes. {I}.
  {G}eneral theory and simple applications to magnetic and conduction problems.
\newblock {\em J. Phys. Soc. Japan}, 12:570--586, 1957.

\bibitem{PhysRevLett.110.090502}
O.~Landon-Cardinal and D.~Poulin.
\newblock Local topological order inhibits thermal stability in 2d.
\newblock {\em Phys. Rev. Lett.}, 110:090502, Feb 2013.

\bibitem{Lechner2008}
G.~Lechner.
\newblock Construction of quantum field theories with factorizing {S}-matrices.
\newblock {\em Commun. Math. Phys.}, 277(3):821--860, 2008.

\bibitem{MR0312860}
E.~H. Lieb and D.~W. Robinson.
\newblock The finite group velocity of quantum spin systems.
\newblock {\em Comm. Math. Phys.}, 28:251--257, 1972.

\bibitem{MR0109702}
P.~C. Martin and J.~Schwinger.
\newblock Theory of many-particle systems. {I}.
\newblock {\em Phys. Rev. (2)}, 115:1342--1373, 1959.

\bibitem{MR2605849}
T.~Matsui.
\newblock Spectral gap, and split property in quantum spin chains.
\newblock {\em J. Math. Phys.}, 51(1):015216, 8, 2010.

\bibitem{morettiQM}
V.~Moretti.
\newblock {\em Spectral Theory and Quantum Mechanics}.
\newblock Springer Milan, 2013.

\bibitem{mmappendix}
M.~M{\"u}ger.
\newblock Abstract duality for symmetric tensor $*$-categories.
\newblock Appendix to~\cite{halvapp}.

\bibitem{Muger:2003p3368}
M.~M{\"u}ger.
\newblock On the structure of modular categories.
\newblock {\em Proc. London Math. Soc.}, 87(2):291--308, 2003.

\bibitem{toricendo}
P.~Naaijkens.
\newblock Localized endomorphisms in {K}itaev's toric code on the plane.
\newblock {\em Rev. Math. Phys.}, 23(4):347--373, 4 2011.

\bibitem{phdnaaijkens}
P.~Naaijkens.
\newblock {\em Anyons in infinite quantum systems: {QFT} in {$d=2+1$} and the
  toric code}.
\newblock PhD thesis, Radboud Universiteit Nijmegen, 2012.

\bibitem{haagdtoric}
P.~Naaijkens.
\newblock Haag duality and the distal split property for cones in the toric
  code.
\newblock {\em Lett. Math. Phys.}, 101(3):341--354, 2012.

\bibitem{nachtergaele1996}
B.~Nachtergaele.
\newblock The spectral gap for some spin chains with discrete symmetry
  breaking.
\newblock {\em Commun. Math. Phys.}, 175(3):565--606, 1996.

\bibitem{MR2256615}
B.~Nachtergaele, Y.~Ogata, and R.~Sims.
\newblock Propagation of correlations in quantum lattice systems.
\newblock {\em J. Stat. Phys.}, 124(1):1--13, 2006.

\bibitem{localapprox}
B.~Nachtergaele, V.~B. Scholz, and R.~F. Werner.
\newblock Local approximation of observables and commutator bounds.
\newblock In J.~Janas, P.~Kurasov, A.~Laptev, and S.~Naboko, editors, {\em
  Operator Methods in Mathematical Physics}, volume 227 of {\em Operator
  Theory: Advances and Applications}, pages 143--149. Springer Basel, 2013.

\bibitem{MR2217299}
B.~Nachtergaele and R.~Sims.
\newblock Lieb-{R}obinson bounds and the exponential clustering theorem.
\newblock {\em Comm. Math. Phys.}, 265(1):119--130, 2006.

\bibitem{nsicmp}
B.~Nachtergaele and R.~Sims.
\newblock Locality estimates for quantum spin systems.
\newblock In V.~Sidoravicius, editor, {\em New Trends in Mathematical Physics.
  Selected contributions of the XVth International Congress on Mathematical
  Physics}, pages 519--614, 2009.

\bibitem{MR2681770}
B.~Nachtergaele and R.~Sims.
\newblock Lieb-{R}obinson bounds in quantum many-body physics.
\newblock In {\em Entropy and the quantum}, volume 529 of {\em Contemp. Math.},
  pages 141--176. Amer. Math. Soc., Providence, RI, 2010.

\bibitem{lriamp}
B.~Nachtergaele and R.~Sims.
\newblock Much ado about something: {Why Lieb-Robinson} bounds are useful.
\newblock {\em IAMP News Bulletin}, pages 22--29, Oct. 2010.
\newblock Preprint arXiv:1102.0835.

\bibitem{MR2443722}
C.~Nayak, S.~H. Simon, A.~Stern, M.~Freedman, and S.~Das~Sarma.
\newblock Non-abelian anyons and topological quantum computation.
\newblock {\em Rev. Modern Phys.}, 80(3):1083--1159, 2008.

\bibitem{MR0010265}
M.~Neumark.
\newblock Positive definite operator functions on a commutative group.
\newblock {\em Bull. Acad. Sci. URSS S\'er. Math. [Izvestia Akad. Nauk SSSR]},
  7:237--244, 1943.

\bibitem{MR1796805}
M.~A. Nielsen and I.~L. Chuang.
\newblock {\em Quantum computation and quantum information}.
\newblock Cambridge University Press, Cambridge, 2000.

\bibitem{MR1463825}
F.~Nill and K.~Szlach{\'a}nyi.
\newblock Quantum chains of {H}opf algebras with quantum double cosymmetry.
\newblock {\em Comm. Math. Phys.}, 187(1):159--200, 1997.

\bibitem{MR1230389}
M.~Ohya and D.~Petz.
\newblock {\em Quantum entropy and its use}.
\newblock Texts and Monographs in Physics. Springer-Verlag, Berlin, 1993.

\bibitem{MR971256}
G.~K. Pedersen.
\newblock {\em Analysis now}, volume 118 of {\em Graduate Texts in
  Mathematics}.
\newblock Springer-Verlag, New York, 1989.

\bibitem{Pokorny:1993wt}
M.~Pokorny.
\newblock Continuous spectrum in the ground state of two spin-{$\frac12$}
  models in the infinite-volume limit.
\newblock {\em J. Statist. Phys.}, 72(1-2):381--403, 1993.

\bibitem{Powers:1967tx}
R.~T. Powers.
\newblock Representations of uniformly hyperfinite algebras and their
  associated von neumann rings.
\newblock {\em Ann. of Math. (2)}, 86:138--171, 1967.

\bibitem{MR0471796}
W.~Pusz and S.~L. Woronowicz.
\newblock Passive states and {KMS} states for general quantum systems.
\newblock {\em Comm. Math. Phys.}, 58(3):273--290, 1978.

\bibitem{MR1743202}
B.~V. Rajarama~Bhat, G.~A. Elliott, and P.~A. Fillmore, editors.
\newblock {\em Lectures on operator theory}, volume~13 of {\em Fields Institute
  Monographs}.
\newblock American Mathematical Society, Providence, RI, 1999.

\bibitem{MR751959}
M.~Reed and B.~Simon.
\newblock {\em Methods of modern mathematical physics. {I}: {F}unctional
  analysis}.
\newblock Academic Press, Inc. [Harcourt Brace Jovanovich, Publishers], New
  York, second edition, 1980.

\bibitem{MR1147467}
K.-H. Rehren.
\newblock Braid group statistics and their superselection rules.
\newblock In D.~Kastler, editor, {\em The algebraic theory of superselection
  sectors ({P}alermo, 1989)}, pages 333--355. World Sci. Publ., River Edge, NJ,
  River Edge, NJ, 1990.

\bibitem{MR0245275}
J.~E. Roberts and G.~Roepstorff.
\newblock Some basic concepts of algebraic quantum theory.
\newblock {\em Comm. Math. Phys.}, 11:321--338, 1969.

\bibitem{MR0228246}
D.~W. Robinson.
\newblock Statistical mechanics of quantum spin systems. {II}.
\newblock {\em Comm. Math. Phys.}, 7:337--348, 1968.

\bibitem{Ru62}
D.~Ruelle.
\newblock On the asymptotic condition in quantum field theory.
\newblock {\em Helv. Phys. Acta}, 35:147--163, 1962.

\bibitem{MR0203510}
D.~Ruelle.
\newblock States of physical systems.
\newblock {\em Comm. Math. Phys.}, 3:133--150, 1966.

\bibitem{Sch83}
M.~Schmitz.
\newblock Lokalit\"atseigenschaften von {E}inteilchenzust\"anden in
  {Q}uanten-{S}pinzust\"anden.
\newblock Master's thesis, University of Freiburg, 1983.

\bibitem{MR0022652}
I.~E. Segal.
\newblock Postulates for general quantum mechanics.
\newblock {\em Ann. of Math. (2)}, 48:930--948, 1947.

\bibitem{MR1919619}
G.~L. Sewell.
\newblock {\em Quantum mechanics and its emergent macrophysics}.
\newblock Princeton University Press, Princeton, NJ, 2002.

\bibitem{Sewell2003}
G.~L. Sewell.
\newblock Quantum theory of irreversibility: Open systems and continuum
  mechanics.
\newblock In F.~Benatti and R.~Floreanini, editors, {\em Irreversible Quantum
  Dynamics}, pages 7--30. Springer Berlin Heidelberg, Berlin, Heidelberg, 2003.

\bibitem{MR1239893}
B.~Simon.
\newblock {\em The statistical mechanics of lattice gases. {V}ol. {I}}.
\newblock Princeton Series in Physics. Princeton University Press, Princeton,
  NJ, 1993.

\bibitem{MR1884336}
R.~F. Streater and A.~S. Wightman.
\newblock {\em P{CT}, spin and statistics, and all that}.
\newblock Princeton Landmarks in Physics. Princeton University Press,
  Princeton, NJ, 2000.
\newblock Corrected third printing of the 1978 edition.

\bibitem{MR2374090}
F.~Strocchi.
\newblock {\em Symmetry breaking}, volume 732 of {\em Lecture Notes in
  Physics}.
\newblock Springer, Berlin, second edition, 2008.

\bibitem{MR0074800}
Z.~Takeda.
\newblock Inductive limit and infinite direct product of operator algebras.
\newblock {\em T\^ohoku Math. J. (2)}, 7:67--86, 1955.

\bibitem{MR1873025}
M.~Takesaki.
\newblock {\em Theory of operator algebras. {I}}, volume 124 of {\em
  Encyclopaedia of Mathematical Sciences}.
\newblock Springer-Verlag, Berlin, 2002.

\bibitem{MR1943006}
M.~Takesaki.
\newblock {\em Theory of operator algebras. {II}}, volume 125 of {\em
  Encyclopaedia of Mathematical Sciences}.
\newblock Springer-Verlag, Berlin, 2003.

\bibitem{MR1943007}
M.~Takesaki.
\newblock {\em Theory of operator algebras. {III}}, volume 127 of {\em
  Encyclopaedia of Mathematical Sciences}.
\newblock Springer-Verlag, Berlin, 2003.
\newblock Operator Algebras and Non-commutative Geometry, 8.

\bibitem{Vanderstraeten:2014a}
L.~Vanderstraeten, J.~Haegeman, T.~J. Osborne, and F.~Verstraete.
\newblock {$S$} matrix from matrix product states.
\newblock {\em Phys. Rev. Lett.}, 112:257202, 2014.

\bibitem{mpsreview}
F.~Verstraete, V.~Murg, and J.~I. Cirac.
\newblock Matrix product states, projected entangled pair states, and
  variational renormalization group methods for quantum spin systems.
\newblock {\em Adv. in Physics}, 57(2):143--224, 2008.

\bibitem{vnirred}
J.~von Neumann.
\newblock {Die Eindeutigkeit der Schr\"odingerschen Operatoren}.
\newblock {\em Mathematische Annalen}, 104(1):570--578, 1931.

\bibitem{MR0223138}
J.~von Neumann.
\newblock {\em Mathematische {G}rundlagen der {Q}uantenmechanik}.
\newblock J. Springer, Berlin, 1932.

\bibitem{MR0157873}
J.~von Neumann.
\newblock {\em Collected works. {V}ol. {III}: {R}ings of operators}.
\newblock General editor: A. H. Taub. Pergamon Press, New York, 1961.

\bibitem{Wang}
Z.~Wang.
\newblock {\em Topological Quantum Computation}, volume 112 of {\em CBMS
  Regional Conference Series in Mathematics}.
\newblock Published for the Conference Board of the Mathematical Sciences,
  Washington, DC, 2010.

\bibitem{Wick:1952p9186}
G.~Wick, A.~Wightman, and E.~Wigner.
\newblock The intrinsic parity of elementary particles.
\newblock {\em Physical Review}, 88(1):101--105, 1952.

\bibitem{MR0106711}
E.~P. Wigner.
\newblock {\em Group theory: {A}nd its application to the quantum mechanics of
  atomic spectra}.
\newblock Expanded and improved ed. Pure and Applied Physics. Vol. 5. Academic
  Press, New York, 1959.

\bibitem{CPA:CPA3160130102}
E.~P. Wigner.
\newblock The unreasonable effectiveness of mathematics in the natural
  sciences. richard courant lecture in mathematical sciences delivered at new
  york university, may 11, 1959.
\newblock {\em Communications on Pure and Applied Mathematics}, 13(1):1--14,
  1960.

\bibitem{Yarotsky:2005dk}
D.~A. Yarotsky.
\newblock Ground states in relatively bounded quantum perturbations of
  classical lattice systems.
\newblock {\em Comm. Math. Phys.}, 261(3):799--819, 2006.

\bibitem{Yngvason:2005p1600}
J.~Yngvason.
\newblock {The role of type III factors in quantum field theory}.
\newblock {\em Reports on Mathematical Physics}, 55(1):135--147, 2005.

\end{thebibliography}

\printindex
\end{document}